\documentclass[final]{lmcs}
\pdfoutput=1

\usepackage{lastpage}
\lmcsdoi{16}{3}{18}
\lmcsheading{}{\pageref{LastPage}}{}{}%
{Apr.~19,~2019}{Sep.~24,~2020}{}

\usepackage[utf8]{inputenc}

\pdfoutput=1
\usepackage{mathtools} %\xRightArrow and other stuff
\usepackage{amssymb}
\usepackage{thmtools,thm-restate}
\numberwithin{equation}{section}
\usepackage{hyperref}
\usepackage[capitalise]{cleveref}

\makeatletter
\let\c@thm\undefined%
\makeatother

\theoremstyle{plain}
\declaretheorem[name=Theorem,numberwithin=section]{thm}
\newtheorem{clm}[thm]{Claim}
\newtheorem{lem}[thm]{Lemma}
\newtheorem{cor}[thm]{Corollary}
\newtheorem{prop}[thm]{Proposition}
\newtheorem{fact}[thm]{Fact}

\theoremstyle{thmC}
\newtheorem{lemC}[thm]{Lemma}

\theoremstyle{definition}
\newtheorem{exa}[thm]{Example}
\newtheorem{defi}[thm]{Definition}

\crefname{prop}{Proposition}{Propositions}
\crefname{thm}{Theorem}{Theorems}
\crefname{clm}{Claim}{Claims}
\crefname{lem}{Lemma}{Lemmas}
\crefname{lemC}{Lemma}{Lemmas}
\crefname{cor}{Corollary}{Corollaries}
\crefname{exa}{Example}{Examples}
\crefname{defi}{Definition}{Definitions}
\crefname{fact}{Fact}{Facts}

\crefname{section}{Section}{Sections}
\crefname{figure}{Figure}{Figures}

\usepackage{stmaryrd}
\usepackage{xargs}
\usepackage{bussproofs}
\usepackage{tikz}
\usepackage{varwidth}
\usepackage{booktabs}

\newenvironment{mathprooftree}
  {\varwidth{\textwidth}\centering\leavevmode}
  {\DisplayProof\endvarwidth}

\usepackage{xstring}
\newif\ifcp%
\cptrue%

\def\colorOp{\color{OliveGreen}}
\def\colorPtp{\color{blue}}
\def\colorNode{\color{cyan}}
\def\colorMsg{\color{BrickRed}}

\newcommand{\gname}[1][i]{\ifcp{\colorNode{\scriptstyle\textsf{#1}}}{}\fi}
\newcommand{\msg}[1][m]{\mathsf{\colorMsg{#1}}}
\usetikzlibrary{
  arrows,
  backgrounds,
  chains,
  calc,
  decorations.markings,decorations.pathreplacing,
  fadings,
  fit,
  patterns,
  petri,
  positioning,
  shadows,
  shapes,automata,shapes.callouts
}
\tikzset{
  bblock/.style = {rectangle, draw=blue!50, opacity=.7, line width=.5pt, align=center, fill=white, rounded corners=0.1cm,
    minimum size=4mm, inner sep=1pt},
  ogate/.style = {
    diamond, draw, fill=orange!25,
    minimum size=4mm,
    inner sep=0pt,
    label={[red]center:$+$}
  },
  smallglobal/.style={
        node distance=1cm and 0.8cm, semithick, scale=0.8, every node/.style={transform shape}
  },
  line/.style = {draw,->, rounded corners=0.07cm,>=latex},
  source/.style={draw,circle,fill=white,
    minimum size=3mm,
    inner sep=0pt
  },
  sink/.style={draw,circle,double,fill=white,
    minimum size=3mm,
    inner sep=0pt
  },
}
\newcommandx{\mkint}[6][3=i,4=\p,5=\msg,6=\q,usedefault=@]{
  \node[bblock, #1] (#2) {$\gint[#3][#4][#5][#6]$};
}
\newcommandx{\mkbranchblock}[5][5=@]{
  \mkgateblock{ogate}{#1}{#2}{#3}{#4}[#5]
}
\newcommandx{\mkbranch}[4][2=gatenode,3=i,4=.6,usedefault=@]{
  \mkgatebegin{#1}[{\gname[{#3}]}][ogate][#4]{#2}
}
\newcommandx{\mkmerge}[4][2=gatenode,3=i,4=.5,usedefault=@]{
  \mkgateend{#1}[{\ifempty{#3}{}{\nmerge[#3]}}][ogate][#4]{#2}
}
\newcommandx{\mkgatebegin}[5][2={},3=ogate,4=.5,usedefault=@]{
  \coordinate (gatecord) at (0,0);
  \coordinate (xmax) at (0,0);
  \coordinate (xmin) at (0,0);
  \pgfgetlastxy \xmin \xmax;
  \foreach \n [count=\i] in {#1}{
    \pgfgetlastxy \xc \yc;
    \path (\n);
    \pgfgetlastxy \xn \yn;
    \ifnum \i = 1
      \coordinate (xmin) at (\xn,0);
      \coordinate (xmax) at (\xn,0);
      \coordinate (max) at (0,\yn);
    \else
      \ifdim \xn < \xmin
        \coordinate (xmin) at (\xn,0);
      \fi
      \ifdim \xn > \xmax
        \coordinate (xmax) at (\xn,0);
      \fi
      \ifdim \yn < \yc
        \coordinate (max) at (0,\yc);
      \else
        \coordinate (max) at (0,\yn);
      \fi
    \fi
  }
  \coordinate (gatecord) at ($(xmin)!.5!(xmax) + (max) + (0,#4) + (max)$); % chktex 35
  \node[#3,label={below:$#2$}] (#5) at (gatecord) {};
  \pgfgetlastxy{\xgate}{\ygate};
  \pgfmathtruncatemacro{\xgateround}{\xgate};
  \StrCount{#1,}{,}[\l] % from package xxstring
  \ifnum \l < 2 {\errmessage{#1 argument should be a comma-separated list of lenght >= 2}}
  \else{
    \foreach \n in {#1}{
      \path (\n);
      \pgfgetlastxy{\xnode}{\ynode};
      \pgfmathtruncatemacro{\xnround}{\xnode};
      \pgfmathsetmacro\tmpdiff{abs(\xnround - \xgateround)} % chktex 8 chktex 36
      \ifdim \tmpdiff pt > 1 pt \path[line] (#5) -| (\n);
      \else
        \path[line] (#5) -- (\n); % chktex 8
      \fi
    }
  }
  \fi
}
\newcommandx{\mkgateend}[5][2={},3=ogate,4=.5,usedefault=@]{
  \coordinate (gatecord) at (0,0);
  \coordinate (xmax) at (0,0);
  \coordinate (xmin) at (0,0);
  \pgfgetlastxy \xmin \xmax;
  \foreach \n [count=\i] in {#1}{
    \pgfgetlastxy \xc \yc;
    \path (\n);
    \pgfgetlastxy \xn \yn;
    \ifnum \i = 1
      \coordinate (xmin) at (\xn,0);
      \coordinate (xmax) at (\xn,0);
      \coordinate (ymin) at (0,\yn);
    \else
      \ifdim \xn < \xmin
        \coordinate (xmin) at (\xn,0);
      \fi
      \ifdim \xn > \xmax
        \coordinate (xmax) at (\xn,0);
      \fi
      \ifdim \yn > \yc
        \coordinate (ymin) at (0,\yc);
      \else
        \coordinate (ymin) at (0,\yn);
      \fi
    \fi
  }
  \coordinate (gatecord) at ($(xmin)!.5!(xmax) + (ymin)$);
  \node[#3,label={above:$#2$}] (#5) at ($(gatecord) - (0,{#4})$) {};
  \pgfgetlastxy{\xgate}{\ygate};
  \pgfmathtruncatemacro{\xgateround}{\xgate};
  \StrCount{#1,}{,}[\l] % from package xstring
  \ifnum \l < 2 {\errmessage{#1 argument should be a comma-separated list of lenght >= 2}}
  \else{
    \foreach \n in {#1}{
      \path (\n);
      \pgfgetlastxy{\xnode}{\ynode};
      \pgfmathtruncatemacro{\xnround}{\xnode};
      \pgfmathsetmacro\tmpdiff{abs(\xnround - \xgateround)} % chktex 36 chktex 8
      \ifdim \tmpdiff pt > 1 pt \path[line] (\n) |- (#5);
      \else
        \path[line] (\n) -- (#5); % chktex 8
      \fi
    }
  }
  \fi
}
\newcommand{\mkseq}[2]{\path[line] (#1) -- (#2);} % chktex 8
\newcommandx{\mkgraph}[3][1=.5cm, usedefault=@]{
  \node[source,above = #1 of {#2}] (src#2) {};
  \node[sink,below  = #1 of {#3}] (sink#3) {};
  \path[line] (src#2) -- (#2); % chktex 8
  \path[line] (#3) -- (sink#3); % chktex 8
}
\newcommandx{\gint}[4][1=i,2=A,3=\msg,4=B,usedefault=@]{
  \hellptp[#2] {\colorOp \xrightarrow{\scriptscriptstyle\gname[#1]}} \hellptp[#4] \colon {\msg[{#3}]}
}
\newcommand{\hellptp}[1][a]{\ensuremath{\mathsf{\colorPtp{#1}}}}
\newcommand{\orgateG}{
\begin{tikzpicture}[smallglobal,baseline=-.5ex, scale=0.75, every node/.style={transform shape}]
  \node [ogate] (o) {};
\end{tikzpicture}
}
\definecolor{kwordcoloraux}{RGB}{128,0,128}
\definecolor{smvocolor}{RGB}{32,44,191}
\definecolor{shnamecolor}{RGB}{199,98,52}
\definecolor{dtcolor}{RGB}{15, 138, 125}
\colorlet{kwordcolor}{kwordcoloraux}
\colorlet{pstcolor}{dtcolor}
\colorlet{headLcolor}{kwordcolor}
\colorlet{sparopcolor}{black}
\colorlet{queuecolor}{black!70}
\definecolor{qmvcolor}{rgb}{.1,.2,.3}
\definecolor{psnamecolor}{rgb}{0.2,0,0.2}
\newcommand{\psnames}[2]{\textcolor{psnamecolor}{(}\tuple{#1}\textcolor{psnamecolor}{,}{\ptp{#2}}\textcolor{psnamecolor}{)}} % chktex 9
\newcommand{\psname}[2]{\textcolor{psnamecolor}{(}{#1}\textcolor{psnamecolor}{,}{\ptp{#2}}\textcolor{psnamecolor}{)}} % chktex 9
\colorlet{atcolor}{black!70}
\colorlet{poppnamecolor}{smvocolor}
\colorlet{figureblue}{kwordcoloraux}

\newcommand{\Pif}[3]{\kword{if}~{#1}~\kword{then}~#2~\kword{else}~{#3}}
\newcommand{\truek}{\mathtt{true}}
\newcommand{\falsek}{\mathtt{false}}
\newcommand{\G}{\mathcal{G}}
\newcommand{\imrunsu}[1]{\mathcal{R}#1}
\newcommand{\mysubseteq}{\Subset}
\newcommand{\has}{\triangleright}
\newcommand{\envmv}{\Delta}
\newcommand{\imruns}[1]{\mathcal{R}(#1)} % chktex 36
\newcommand{\kword}[2][kwordcolor]{\textstyle{\sf\textcolor{#1}{#2}}}
\newcommand{\ptp}[1]{{\mathsf{\MakeLowercase{#1}}}}
\newcommand{\TT}{\mathsf{T}}
\newcommand{\smvo}[1]{\mathsf{\textcolor{smvocolor}{#1}}}
\newcommand{\smv}{\smvo{y}}
\newcommand{\typeruns}[2]{\mathcal{R}_{#1}(#2)} % chktex 36
\newcommand{\shset}{\mbox{$\mathbb{U}$}}
\newcommandx{\shname}[1][1={},usedefault=@]{\textcolor{shnamecolor}{u}^{\ifempty{#1}{}{\ptp{#1}}}}
\newcommand{\chset}{\mbox{$\mathbb{Y}$}}
\newcommand{\ptpset}{\mbox{$\mathbb{P}$}}
\newcommand{\dt}[1]{\textcolor{dtcolor}{\mathsf{#1}}}
\newcommand{\sort}{\dt{d}}
\newcommand{\GT}{\mathsf{G}}
\newcommandx{\Gchoice}[7][1=i, 2=I, 3=p, 4=q, 5=\smv,6=\sort,7=\GT,usedefault=@]{\sum_{#1\in #2} \Gint[#3][#4_{#1}][#5_{#1}][#6_{\mathit{#1}}];#7_{#1}}
\newcommandx{\Gint}[4][1=p,2=q,3=\smv,4=\sort,usedefault=@]{\ptp{#1} \rightarrowtriangle \ptp{#2} \colon #3\ \dt{#4}}
\newcommand{\Gdef}[2]{#2^{*^{#1}}}
\newcommand{\f}{f}
\newcommand{\Gend}{\mathsf{end}}
\newcommand{\range}[1]{\chn{#1}}
\newcommand{\participants}[1]{\mathcal{P}(#1)} % chktex 36
\newcommand{\chn}[1]{\mathtt{ch}(#1)} % chktex 36
\newcommand{\ready}[1]{\mathtt{rdy}(#1)} % chktex 36
\newcommandx{\Geq}[3][1=\G,2={\tuple \smv},3=\GT,usedefault=@]{{#1}({#2}) \mmdef {#3}} % chktex 36
\newcommand{\subs}[2]{\{#2 / #1\}}
\newcommand{\Tssel}[3]{\TSsel{i\in I}{#1}{#2}{#3}}
\newcommand{\Tbbra}[3]{\displaystyle{\sum_{i \in I} \Tbra{#1}{#2}{#3}}}
\newcommand{\Tdef}[1]{\tloop{#1}}
\newcommand{\Tend}{\mathsf{end}}
\newcommand{\TSsel}[4]{\displaystyle{\bigoplus_{#1} \Tsel{#2}{#3}{#4}}}
\newcommand{\Tbra}[3]{\Treceive{#1}{#2}. #3}
\newcommand{\tloop}[1]{#1^{\star}}
\newcommand{\Tsel}[3]{\Tsend{#1}{#2}\ifempty{#3}{}{. #3}}
\newcommand{\Treceive}[2]{{#1}\ {#2}}
\newcommand{\Tsend}[2]{\overline{#1}\ {#2}}
\newcommandx{\Leq}[3][1=\T,2={\tuple \smv},3=\TT,usedefault=@]{{#1}({#2}) \mmdef {#3}} % chktex 36
\newcommand{\T}{\mathcal{T}}
\newcommand{\proj}{\!\!\upharpoonright\!}
\newcommand{\TBbra}[5]{\displaystyle{\sum_{#1 \in #2} \Tbra{#3}{#4}{#5}}}
\newcommand{\GPOP}{\exG \GT {POP}}
\newcommand{\exG}[2]{{#1}_{\resizebox{!}{3pt}{$\mathtt{#2}$}}}
\newcommand{\varset}{\mbox{$\mathbb{X}$}}
\newcommand{\valset}{\mbox{$\mathbb{V}$}}
\newcommand{\val}[1]{\mathsf{#1}}
\newcommand{\bop}{\mathtt{~bop~}}
\newcommand{\uop}{\mathtt{~uop~}}
\newcommand{\emptylist}{\varepsilon}
\newcommand{\headL}[1]{\kword[headLcolor]{hd}(#1)} % chktex 36
\newcommand{\tailL}[1]{\kword[headLcolor]{tl}(#1)} % chktex 36
\newcommand{\var}[1]{\mathrm{var}({#1})} % chktex 36
\newcommand{\Preq}[4]{\lreq{#1}{#2}{#3}.#4}
\newcommand{\Pacc}[4]{\lacc{#1}{#2}{#3}.#4}
\newcommand{\Psend}[2]{\overline{#1}\,{#2}}
\newcommand{\Pfor}[3]{\kword{for}~{#1}~\kword{in}~{#2}~\kword{do}~{#3}}
\newcommand{\PloopU}[2]{\kword{repeat}~ #1~\kword{until}~ #2} % chktex 39
\newcommand{\Pechoice}[4]{\displaystyle{\sum_{#1}} \Preceive{#2}{#3}{#4}}
\newcommand{\sparop}{\textcolor{sparopcolor}{\mid}}
\newcommand{\lreq}[3]{\overline{#1}^{\ptp{#2}}(#3)} % chktex 36
\newcommand{\lacc}[3]{{#1}^{\ptp{#2}}(#3)} % chktex 36
\newcommand{\Preceive}[3]{{#1}\ifthenelse{\equal{#2}{}}{}{(#2)}.#3}
\newcommand{\queue}[2]{{#1}\textcolor{queuecolor}{[}#2\textcolor{queuecolor}{]}} % chktex 9
\newcommand{\qmv}{\textcolor{qmvcolor}{\tuple{\val v}}}
\newcommand{\At}{\textcolor{atcolor}{@}}
\newcommand{\AT}[2]{#1{\textcolor{atcolor}{@}}\, {\ptp #2}}
\newcommand{\Pend}{\kword{0}}
\newcommand{\fn}[1]{\textcolor{black}{\mathtt{fn}(#1)}} % chktex 36
\newcommandx{\fU}[1][1=\_,usedefault=@]{\mathtt{fu}(#1)} % chktex 36
\newcommandx{\fY}[1][1=\_,usedefault=@]{\mathtt{fy}(#1)} % chktex 36
\newcommandx{\fX}[1][1=\_,usedefault=@]{\mathtt{fx}(#1)} % chktex 36
\newcommand{\bn}[1]{\mathtt{bn}(#1)} % chktex 36
\newcommandx{\bY}[1][1=\_,usedefault=@]{\mathtt{by}(#1)} % chktex 36
\newcommandx{\bX}[1][1=\_,usedefault=@]{\mathtt{bx}(#1)} % chktex 36
\newcommandx{\bU}[1][1=\_,usedefault=@]{\mathtt{bu}(#1)} % chktex 36
\newcommand{\upd}[2]{[{#2}\mapsto {#1}]}
\newcommand{\eval}[1]{{#1}\downarrow}
\newcommand{\lsend}[2]{\overline{#1}#2}
\newcommand{\lreceive}[2]{#1#2}
\newcommand{\lcond}[2]{#1\vdash #2}
\newcommand{\C}{e}
\newcommandx{\state}[2][1=\_,2=\sigma,usedefault=@]{\langle {#1} , {#2} \rangle} % chktex 26
\newcommand{\Ptrans}[1]{\xrightarrow{#1}}
\def\mathrule #1#2#3{
  \begin{array}{l}%
    \ifempty{#3}{}{\hspace{0em}\mbox{\footnotesize$\mathsf{[#3]}$}\\}
    \irule{#1}{#2}
  \end{array}
}
\newcommand{\irule}[2]{\frac{\textstyle\rule[-1.3ex]{0cm}{3ex}#1}{\textstyle\rule[-.5ex]{0cm}{3ex}#2}}
\def\mathaxiom #1#2{
  \begin{array}{l}%
    \ifempty{#2}{}{\hspace{0em}\mbox{\footnotesize$\mathsf{[#2]}$}\\}
    { #1}
  \end{array}
}
\newcommand{\isemptyL}[1]{#1 = \emptylist}
\newcommand{\myrule}[1]{\mbox{\footnotesize$\mathsf{[#1]}$}}
\newcommand{\poppname}[1]{\mathit{\color{poppnamecolor}#1}}
\newcommand{\popcname}[1]{\smvo{\MakeLowercase{#1}}}
\newcommand{\ctx}{\mathtt{C}}
\newcommandx{\Gimp}[4][1={},2=\iota,3={\shname},4={\tuple \smv},usedefault=@]{{{#1}{#2}_{{#4}\At{#3}}}}
\newcommandx{\vjdg}[5][1=\envmv,2=u,3=\G,4=\vdash,5=\colon,usedefault=@]{{#1} \  {#4} \ {#2} {#5} {#3}}
\newcommand{\updi}[2]{[{#2}\mapsto {#1}]}
\newcommandx{\names}[1][1=\_,usedefault=@]{\mathtt{n}(#1)} % chktex 36
\newcommand{\Wtrans}[1]{\xhookrightarrow{ #1 }}
\newcommand{\imrunsaux}[1]{\widetilde{\mathcal{R}}(#1)} % chktex 36
\newcommand{\myequal}{\lessdot}
\newcommand{\pst}{\textcolor{pstcolor}{\mathbb{T}}}
\newcommand{\Ipst}[5][\C]{{\displaystyle{\bigoplus_{#2} \guard[{#1}]{\Tsel{#3}{#4}{#5}}}}}
\newcommand{\Epst}[5][\C]{{\displaystyle{\sum_{#2} \guard[{#1}]{\Tbra{#3}{#4}{#5}}}}}
\newcommand{\guard}[2][\C]{{#1} \Yleft {#2}} %\sphericalangle \leftslice \rightarrowtriangle \shortrightarrow
\newcommand{\rmg}[1]{\Game\big({#1}\big)}
\newcommand{\peace}[1][\C]{{\bowtie}}
\newcommand{\nform}[2][\C]{\mathsf{nf}({#1, #2})} % chktex 36
\newcommand{\nfp}[1]{\mathsf{nf}({#1})} % chktex 36
\newcommandx{\pjdg}[5][1=\C,2=\envmv,3=P,4={\envmv'},5=\Gamma,usedefault=@]{
  {#1}\ \text{\textvisiblespace}  \ {#5} \ \proves \ {#3} \ \has \ {#2}
}
\newcommand{\proves}{\vdash}
\newcommand{\vreqrule}{
  \mathrule
  {
    \envmv(\shname) \equiv \G(\tuple \smv)
    \qquad
    \pjdg[@][\envmv, \psnames \smv 0:  \pst]
    \qquad
    \rmg{ \nfp{\pst}} = \nfp{\G(\tuple \smv)\proj{\ptp 0}}
  }
  {
    \pjdg[@][@][\Preq{\shname}{n}{\tuple \smv}{P}]
  }{VReq}
}
\newcommand{\rcvrule}{
  \mathrule
  {
    \forall i \in I \qst \smv_i \in \tuple \smv \qquad
    \pjdg[@][\envmv,\psnames \smv p : {\pst_i}][P_i][@][\Gamma,x_i : \sort_i] % chktex 26
  }
  {
    \pjdg[@][\envmv, \psnames \smv p : \displaystyle{\sum_{i\in I} \Tbra{\guard\smv_i}{\sort_{\mbox{\scriptsize$i$}}}{\pst_i}}][\Pechoice{i\in I}{\smv_i}{x_i}{P_i}] % chktex 26
  }
  {VRcv}
}
\newcommand{\vsendrule}{
  \mathrule
  {
    \vjdg[\Gamma][e'][\sort]
    \qquad
    \smv \in \tuple \smv
    \qquad
    \envmv \text{ $\Tend$-only}
  }
  {
    \pjdg[@][\envmv, \psnames \smv p: \guard{\Tsend \smv \sort};\guard\Tend][\Psend{\smv}{e'}][\envmv \upd \TT {\psnames \smv p}]
  }
  {VSend}
}
\newcommand{\vendrule}{
  \mathrule{
    \envmv \text{ $\Tend$-only}
  }{
    \pjdg[@][@][\Pend][\envmv]
  }
  {VEnd}
}
\newcommand{\vseqrule}[1]{
  \mathrule{
    \pjdg[@][\envmv_1][P_1]
    \qquad
    \pjdg[@][\envmv_2][P_2]      }
  {
    \pjdg[@][\envmv_1;\envmv_2][P_1;P_2]
  }
  {#1}
}
\newcommand{\vifrule}[1]{
  \mathrule{
    \pjdg[\C \land  \C'][\envmv_1][P_1][@][\Gamma]
    \qquad
    \pjdg[\C\land\neg \C'][\envmv_2][P_2][@][\Gamma]
  }
  {
    \pjdg[@][ \envmv_1 \peace \envmv_2][\Pif{\C'}{P_1}{P_2}][\envmv_1'\restriction_{\dom \envmv}]
  }
  {#1}
}
\newcommand{\vfor}{
  \mathrule{
    \begin{array}{c}
      \consistent[\C][\ell \not = \emptylist]
      \quad
      \vjdg[\Gamma][\ell][[\sort]]
      \quad
      \pjdg[\C][@][@][@][\Gamma, x : \sort]
      \quad
      \envmv \mbox{ active}
      \quad
      x \not\in \var \envmv
    \end{array}
  }{
    \pjdg[@][{\Tdef\envmv}][\Pfor{x}{\ell}{P}]
  }{VFor}
}

\newcommand{\vforend}{
  \mathrule{
    \consistent[\C][\ell  = \emptylist]
    \qquad
    \vjdg[\Gamma][\ell][[\sort]]
    \qquad
    \envmv \text{ $\Tend$-only}
  }{
    \pjdg[@][\envmv][\Pfor x \ell P]
  }
  {VForEnd}
}
\newcommand{\vlooprule}[1]{
  \mathrule{
    \begin{array}{c}
      \pjdg[@][\envmv_1][N]
      \qquad
      \pjdg[@][\envmv_2][M]
      \qquad
      \envmv_1 \mbox{ and } \envmv_2 \mbox{ passively compatible}
    \end{array}
  }
  {
    \pjdg[@][{\Tdef{\envmv_1}};\envmv_2][\PloopU{N} M][\envmv']
  }
  {#1}
}
\newcommand{\vparrule}{
  \mathrule{
    \pjdg[\C_1][\envmv_1][S_1][\envmv_1'][\Gamma]
    \qquad
    \pjdg[\C_2][\envmv_2][S_2][\envmv_2'][\Gamma]
    \qquad
    \envmv_1\text{ and }\envmv_2\text{ independent}
  }{
    \pjdg[\C_1 \land \C_2][\envmv_1 \cup \envmv_2][S_1 \sparop S_2][@][\Gamma]
  }{VPar}
}
\newcommand{\vnewrule}[1]{
  \mathrule{
    \pjdg[@][\envmv][S][][\Gamma]
  }{
    \pjdg[@][\envmv|_{-\tuple \smv}][(\nu \tuple\smv\At\shname)S][@][\Gamma] % chktex 36
  }{#1}
}
\newcommand{\vqueuerule}{
   \mathrule
   {
     \vjdg[][\val v][\sort]
     \qquad
     \pjdg[@][\smv : \queue{}{\tuple \sort}][\queue{\smv}{\tuple{\val v}}][@][\Gamma] % chktex 26
   }
   {
     \pjdg[@][\smv : \queue{}{\tuple \sort \cdot \sort}][\queue{\smv}{\qmv \cdot \val v}][@][\Gamma] % chktex 26
   }{VQueue}
}
\newcommand{\vemptyrule}{
  \mathaxiom{\pjdg[@][\smv : \queue{}{}][\smv : \queue{}{}][@][\Gamma]} % chktex 26
  {VEmpty}
}
\newcommand{\vaccrule}{
  \mathrule
  {
    \envmv(\shname) \equiv \G(\tuple \smv)
    \qquad
    \pjdg[@][\envmv, \psnames \smv p :  \pst] % chktex 26
    \qquad
    \rmg{\nfp \pst} = \nfp{\G(\tuple \smv) \proj {\ptp p}}
  }
  {
    \pjdg[@][@][\Pacc{\shname}{p}{\tuple \smv}{P}]
  }
  {VAcc}
}
\newcommand{\spec}[2]{#2}
\newcommandx{\consistent}[2][1=\C,2=\C',usedefault=@]{{#1}\land{#2} \ \not\proves\ \bot}
\newcommand{\Ttrans}[1]{\xrightarrow{#1}}
\newcommandx{\consExp}[2][2=\sigma,1=\C, usedefault=@]{#1\&#2}
\newcommandx{\consPst}[2][2=\sigma,1=\pst, usedefault=@]{{#2} \ltimes {#1} }
\newcommandx{\consSt}[5][1=\sigma,2=S, 5 = \Gamma, 4= \C, 3=\envmv, usedefault=@]{#1 \ltimes {(#4; #5; #2; #3)} }
\newcommand{\Strans}[1]{\, \stackrel{#1}{\Longrightarrow}\, }
\newcommand{\R}{\mathbb{R}}
\newcommand{\CWSRel}[2]{(#1,#2)\in\R}
\newcommand{\MAsend}[2]{\lsend {#1}  {#2}}
\newcommand{\varpst}[1]{\mathrm{var}({#1})} % chktex 36
\newcommand{\lemref}[1]{\cref{#1}}% \newcommand{\lemref}[1]{Lem.~\ref{#1}}
\newcommand{\CC}{E}
\newcommand{\seq}[1]{\widetilde#1}

\newcommand{\Ipstaux}[2]{\displaystyle{\bigoplus_{#1} #2}}
\newcommand{\eg}{\text{e.g.}}
\def\cf{\emph{cf.}\ }
\usepackage{xifthen}
\newcommand{\ifempty}[3]{%
  \ifthenelse{\isempty{#1}}{#2}{#3}%
}
\newcommand{\tuple}[1]{\vec{#1}}
\newcommand{\sep}{\;\bnfmid\;}
\newcommand{\sst}{\;\big|\;}

\newcommand{\bnfmid}{\;\ \big|\ \;}
\newcommand{\dom}[1]{\mathrm{dom} {(#1)}}
\def\finex{{\unskip\nobreak\hfil
\penalty50\hskip1em\null\nobreak\hfil$\diamond$
\parfillskip=0pt\finalhyphendemerits=0\endgraf}}
\newcommand{\mmdef}{\eqdef}
\newcommand{\eqdef}{\,\triangleq\,}
\newcommand{\bnfdef}{\ ::=\ } % chktex 26
\newcommand{\qst}{\;\colon\;} %such that
\newcommand{\nat}{\mathbb{N}}
\newcommand{\defref}[1]{Definition~\ref{#1}}

\begin{document}

\title{On Resolving Non-determinism in Choreographies}

\author[L. Bocchi]{Laura Bocchi\rsuper{a}}
\address{\lsuper{a}School of Computing, University of Kent, UK}

\author[H. Melgratti]{Hern\'an Melgratti\rsuper{b}}
\address{\lsuper{b}Instituto de Ciencias de la Computaci\'on - Universidad de Buenos Aires - Conicet, Argentina} % chktex 8

\author[E. Tuosto]{Emilio Tuosto\rsuper{c}}
\address{\lsuper{c}Gran Sasso Science Institute, IT \and Department of Computer Science, University of Leicester, UK}

\keywords{Choreography, multiparty session types, process algebras, whole-spectrum implementation, message-passing, non-determinism}
\subjclass{D.1.3: Concurrent Programming, D.2.2: Design Tools and Techniques, D.2.4: Software/Program Verification, D.3.1: Formal Definitions and Theory}

\thanks{
  Research partly supported by the European Unions Horizon 2020
  research and innovation programme under the Marie
  Sk\l{}odowska-Curie grant agreement No 778233, by the UBACyT
  projects 20020170100544BA and 20020170100086BA, and by the PIP
  project 11220130100148CO\!
}

\begin{abstract}
  Choreographies specify multiparty interactions via message passing.
  A \emph{realisation} of a choreography is a composition of
  independent processes that behave as specified by the choreography.
  Existing relations of correctness/completeness between
  choreographies and realisations are based on models where choices
  are non-deterministic.
  Resolving non-deterministic choices into deterministic choices
  (e.g., conditional statements) is necessary to correctly
  characterise the relationship between choreographies and their
  implementations with concrete programming languages.
  We introduce a notion of realisability for choreographies --called
  \emph{whole-spectrum implementation}-- where choices are still
  non-deterministic in choreographies, but are deterministic in their
  implementations.
  Our notion of whole spectrum implementation rules out deterministic
  implementations of roles that, no matter which context they are
  placed in, will never follow one of the branches of a
  non-deterministic choice.
  We give a type discipline for checking whole-spectrum
  implementations. As a case study, we analyse the
  POP protocol under the lens of whole-spectrum implementation.
\end{abstract}

\maketitle
\section{Introduction}\label{sec:intro}
% !TEX root = main.tex

\newcommand{\ATM}{\emph{ATM}}

\noindent\textbf{The context} A \emph{choreography} describes the
expected interactions of a system in terms of the messages exchanged
among its components (aka \emph{roles}):
\begingroup
\addtolength\leftmargini{-2ex}
\begin{quote}\sl
  \lq\lq Using the Web Services Choreography specification, a contract
  containing a global definition of the common ordering conditions and
  constraints under which messages are exchanged, is produced
  [...]. Each party can then use the global definition to build and % chktex 11
  test solutions that conform to it. The global specification is in
  turn realised by combination of the resulting local systems
  [...]\rq\rq % chktex 11
\end{quote}
The first part of the excerpt above taken from~\cite{w3c:cho}
envisages a choreography as a global contract regulating the exchange
of messages; the last part identifies a distinctive element of
choreographies: the global definition can be used to verify local
components (correctly) realise the global contract.
A choreography allows for the combination of independently developed
distributed components (\eg, services) while hiding implementation
details.

\begin{figure}
  \begin{tikzpicture}[node distance=.5cm and .6cm, scale =1.1, every node/.style={transform shape}]
    \mkint{}{cbauth}[][c][login][b];
    \mkint{below = 1cm of cbauth, xshift=-2cm}{cbdeposit}[][c][deposit][b];
    \mkint{below = 1cm of cbauth, xshift=2cm}{cboverdraft}[][c][overdraft][b];
    \mkint{below = 1cm of cboverdraft,xshift=-1cm}{bcko}[][b][ko][c];
    \mkint{below = 1cm of cboverdraft,xshift=1cm}{bcok}[][b][ok][c];
    \mkbranch{bcko,bcok}[ok+ko][][-2.3];
    \mkmerge{bcko,bcok}[-ok+ko][][.5];
    \mkseq{cboverdraft}{ok+ko};
    \mkbranch{cbdeposit,cboverdraft}[d+o][][-.7];
    \mkmerge{cbdeposit,-ok+ko}[-d+o][][.3];
    \mkseq{cbauth}{d+o};
    \mkgraph{cbauth}{-d+o};
  \end{tikzpicture}
\caption{\ATM\ Choreography}%
\label{fig:atm-chor}
\end{figure}
Moreover, the communication pattern specified in the choreography
yields sufficient information to be projected so to check each
component implementing one of the roles.
For illustration, take a simple choreography, hereafter called
\ATM\ and expressed as the global graph~\cite{lty15,gt17}
on~\cref{fig:atm-chor}, involving a customer $\ptp c$ and the cash
machine of a bank $\ptp b$.
The nodes labelled by \orgateG\ represent branching or merging points
of choices.
After successful authentication, $\ptp b$ offers a deposit and an
overdraft service to $\ptp c$.
In the global graph this choice corresponds to the topmost
\orgateG\ branching over the next interactions between the
participants.
The bank $\ptp b$ can either grant or deny overdrafts asked by the
customer $\ptp c$.

\smallskip\noindent\textbf{On realisations}
A set of processes is a \emph{realisation} of a choreography when the
behaviour emerging from their distributed execution matches the
behaviour specified by the choreography.
A choreography is \emph{realisable} when it has a realisation.

A realisation of \ATM\ can, for example, be given using two CCS-like
processes~\cite{mil89} (augmented with internal $\_\oplus\_$ and
external $\_+\_$ choice operators) for roles $\ptp b$ and $\ptp c$:
\[
\begin{array}{r@{\ = \ }l}
  \hspace{1.5cm}\TT_{\ptp b}  &  \smvo{login}.(\smvo{deposit}
  \ + \ \smvo{overdraft}.(\overline{\smvo{ok}}\
  \oplus  \ \overline{\smvo{ko}}))
\\[1em]
  \TT_{\ptp c}  &  \overline{\smvo{login}}.(\overline{\smvo{deposit}}
  \ \oplus\ \overline{\smvo{overdraft}}.(\smvo{ok}\
  +  \ \smvo{ko}))
  \end{array}
\]
In words, $\TT_{\ptp b}$ specifies that, after $\ptp c$ logs in, $\ptp
B$ waits to interact either on $\smvo{deposit}$ or on
$\smvo{overdraft}$; in the latter case, $\ptp b$ non-deterministically
decides whether to grant or deny the overdraft; $\TT_{\ptp c}$ is the
dual of $\TT_{\ptp b}$.
Note that \ATM\ uses non-determinism to avoid specifying the criteria
for $\ptp b$ to grant or deny an overdraft.
The use of non-determinism is also reflected in realisations, in fact
$\TT_{\ptp b}$ uses the internal choice operator $\_ \oplus \_$ to
model the reaction when $\ptp c$ requests an overdraft.

Choreographies can be interpreted either as \emph{constraints} or as
\emph{obligations} of distributed
interactions~\cite{DBLP:conf/icsoc/LohmannW11}.
The former interpretation
(aka~\emph{partial}~\cite{DBLP:conf/icsoc/LohmannW11} or
weak~\cite{DBLP:conf/wsfm/SuBFZ07}) admits a realisation if it
exhibits a subset of the behaviour.
We propose WSI as a criterion to assess if an implementation takes
into account all the execution cases prescribed by a choreography.
The theory can be used in practice, to implement a static verification
checker targeted at specific languages, and possibly embedded in a
(behavioural) type checker (e.g., for Java~\cite{event} or
Haskel~\cite{NeubauerThiemann04}), or an API generation tool (e.g.,
for Java~\cite{HuY16} or F\#~\cite{NeykovaHYA18}).  The aforementioned
tools are generally targeted at checking soundness of implementations,
and checking for WSI would add guarantees of `completeness'.  In this
paper, we use a process algebra as an abstraction of implementations,
instead of a full fledged programming language, to simplify the
presentation and the development of the theory.
For instance, take
\[
  \TT_{\ptp b}' =  \smvo{login}.(\smvo{deposit}
  \ + \ \smvo{overdraft}. \overline{\smvo{ko}})
\]
then $\TT_{\ptp b}'$ and $\TT_{\ptp c}$ form a partial realisation of
\ATM\ where overdraft requests are consistently denied.
On the contrary, when interpreting choreographies as obligations, a
realisation is admissible if it is able to exhibit \emph{all}
interaction sequences (hence such realisations are also referred to as
\emph{complete} realisations~\cite{DBLP:conf/icsoc/LohmannW11}).

For instance, $\TT_{\ptp b}$ and $\TT_{\ptp c}$ form a complete
realisation of \ATM\@.

\medskip\noindent\textbf{The problem}
Typically, realisations are non-deterministic specifications; here we
explore the problem of resolving their non-determinism.
In fact, despite being a valuable abstraction mechanism,
non-determinism has to be implemented using deterministic constructs
such as conditional statements.

Using again \ATM, we illustrate that traditional notions of complete
realisation are not fully satisfactory.
The non-deterministic choice in $\TT_{\ptp b}$ abstracts away from the
actual conditions used in implementations to resolve the choice.
This permits a bank to adopt different policies depending \eg, on the
type of the clients' accounts.
Consider the (deterministic) implementations $B_1$ and $B_2$ of
$\TT_{\ptp b}$ below written as value-passing CCS processes:
\[\begin{array}{lcll}
    B_i &::=& \smvo {login}(c).\left( \smvo {deposit} (x).Q + \smvo {overdraft}(x).P_i(c) \right)
    & \text{for } i = 1,2 \quad (Q \text{ is immaterial})
    \\[.1pc]
    P_1(c) &::=& \Pif{\mathit{check}(c)}{\overline{\smvo{ok}}}{\overline{\smvo{ko}}}
    \\[.1pc]
    P_2(c) &::=& \overline{\smvo{ko}}
\end{array}\]
Both $B_1$ and $B_2$ expect to receive the login credentials $c$ of a
client on channel $\smvo {login}$.
After that, they offer the services $\smvo {deposit}$ and $\smvo
{overdraft}$.
The implementations differ in the way they handle overdraft requests,
which are respectively defined by $P_1(c)$ and $P_2(c)$.
The expression $\mathit{check}(c)$ in $P_1(c)$ deterministically
discriminates if the overdraft should be granted depending on the
login credentials $c$ provided by the client. Differently, $P_2(c)$
refuses any overdraft request.
It is not hard to see that both $B_1$ and $B_2$ are suitable
implementations of $\TT_{\ptp b}$ in partial realisations of the
choreography\footnote{For instance, both $B_1$ and $B_2$ type-check
  against $\TT_{\ptp b}$ considered as a session type due to the fact
  that subtyping for session types~\cite{GH05} is contra-variant with
  respect to internal choices (and covariant with respect to external
  choices).}  (as
\eg\ in~\cite{DBLP:conf/wsfm/Dezani-Ciancaglinid09}).

Conversely, neither $B_1$ nor $B_2$ can be used in a complete
realisation.
This is straightforward for $B_2$ (unable to interact over
$\mathsf{ok}$ after receiving an overdraft request), but not so
evident for $B_1$.
Depending on the credentials $c$ sent by the customer at login time,
$\mathit{check}(c)$ will evaluate either to $\truek$ or to $\falsek$.
Therefore, $B_1$ will execute only one of the branches.
This will be the case for any possible deterministic implementation of
\ATM:\@ only one branch will be matched. Consequently, there is not a
complete, deterministic realisation for \ATM\@.

We prefer $B_1$ to $B_2$ arguing that they are not equally appealing
when interpreting choreographies as obligations.
In fact, $B_2$ consistently precludes one of the alternatives while
$B_1$ guarantees only one or the other alternative (provided that
$\mathit{check}$ is not a constant function) depending on the
deterministic implementation of the role $\TT_{\ptp c}$.

\medskip\noindent\textbf{Contributions and synopsis}
We introduce \emph{whole-spectrum implementation} (WSI), a new
interpretation of choreographies as interaction obligations.
A WSI of a role $\ptp R$ guarantees that, whenever the choreography
allows $\ptp R$ to make an internal choice, there is a context (i.e.,
an implementation of the remaining roles) for which (the
implementation of) $\ptp R$ chooses such alternative.
Through the paper, we illustrate the use of WSI to analyse the POP2
protocol (i.e., choreography \cref{sec:runGT}, implementation
\cref{sec:runGTp}, and verification \cref{sec:exproof}).

In the following, we use an elaboration of global types
from~\cite{lt12} to model choreographies. Implementations of
choreographies are abstracted as {systems}, which are parallel
compositions of processes in an asynchronous calculus.
\emph{Whole-spectrum implementation} (WSI) is defined in terms of the
denotational semantics of global types and systems.
A key point on the characterisation of WSI is the distinction between
mandatory and optional behaviour arising from the implementation of
loops.
The denotational semantics of a global type $\G$ is given by the set
$\imrunsu(\G)$ of traces describing mandatory and optional behaviour
(\cref{def:runGTT} in \cref{sec:hnice}).
The denotation of a system $S$ is also a set $\imrunsu(S)$ of traces
(\cref{def:gimp} in \cref{sec:hnice}), which however do not
discriminate optional behaviour.
WSI is defined as a \emph{covering} relation $\mysubseteq$ between the
traces of implementations $\imrunsu(S)$ and those of global types
$\imrunsu(\G)$ (\cref{def:wsr}, \cref{sec:hnice}).

We devise a behavioural typing framework for checking WSI\@.
Our global and local types are in \cref{sec:types}; the language for
implementations is in \cref{sec:systems}.
Our typing discipline is introduced in \cref{sec:typing}.
As usual, a type judgement $ \vdash S \has \envmv$ says that an
implementation $S$ has local types $\envmv$, obtained by projecting a
global type $\G$ (cf. \cref{sec:types}).
We show a theorem of subject reduction and one of conformance
(\cref{thm:sr,thm:conformance} in \cref{sec:conformance}) which
guarantee that well-typed systems do not deviate from the behaviour
specified by their type.
Both results rely on the operational semantics of systems
(\cref{sec:systems}) and of local types (\cref{sec:conformance}).
We show the adequacy of the denotational semantics of systems to their
operational semantics in \cref{prop:run-vs-red} (\cref{sec:hnice}).
Analogously, we establish the correspondence between the operational
and denotational semantics of local types
(\cref{lem:corr-semantics-local-types-back,lem:corr-semantics-local-types},
\cref{sec:wsi-by-typing}).

The proof that our type system ensures WSI is given in two steps: we
prove that
\begin{itemize}
\item $\imruns\G \mysubseteq \imruns\envmv$
  (\cref{lem:runs-types-covers-runs-global} in
  \cref{sec:wsi-by-typing}), namely that the traces of the local types
  projected from a global type $\G$ cover the traces of $\G$; and that
\item $\imruns\envmv \mysubseteq \imruns S$
  (\cref{thm:impl-covers-local} in \cref{sec:wsi-by-typing}), namely
  that the traces of well-typed systems cover the traces of their
  local types.
\end{itemize}
By transitivity of $\mysubseteq$, we obtain $\imruns\G \mysubseteq
\imruns S$, which entails WSI for well-typed implementations
(\cref{cor:impl-covers-local} in \cref{sec:wsi-by-typing}).

The main contributions in this paper can be diagrammatically presented
as follows:
\[
  \begin{tikzpicture}[node distance = 1.2cm and 3cm,scale=1.3,transform shape]
    \tikzset{legend/.style = {font=\fontsize{6}{6.5}\selectfont, opacity=.9, color = poppnamecolor, align=center, text width = 1.5cm}}
    \node (j') {$ \vdash$};
    \node[right = .1pt of j'] (S') {$S'$};
    \node[right = .1pt of S'] (x') {$\has$};
    \node[right = .1pt of x'] (Env') {$\envmv'$};
    \node[below = of j'] (j) {$ \vdash$};
    \node[below = of S'] (S) {$S$};
    \node[below = of S] (RG) {$\imruns{\G}$};
    \node[left = of RG,xshift=2.5cm,yshift=-.6cm] (RS) {$\imrunsu(S)$};
    \node[below = of x'] (x) {$\has$};
    \node[below = of Env'] (Env) {$\envmv$};
    \node[right = of RG,xshift=-2.5cm,yshift=-.6cm] (REnv) {$\typeruns{}{\envmv}$};
    \node[below = of Env,xshift = 2cm,yshift=1.2cm] (G) {$\G$};
    \draw[->,dotted] (G) -- (Env) node[legend,midway,sloped] (GEnv) {projection\\\S \ref{sec:types}};
    \draw[->,dashed] (G) -- (RG) node[legend,midway,sloped] (GRG) {Def. \ref{def:gimp}\\ \S \ref{sec:hnice}};
    \draw[->,poppnamecolor,dashed] (S) .. controls (-60pt,-70pt) and (-80pt,-100pt) .. (RS) node[legend,midway,sloped,xshift=0.5cm,yshift=-.1cm,rotate=-20] {Def. \ref{def:runGTT},\\ \S \ref{sec:hnice}};
    \path[->,poppnamecolor,draw,dashed] (Env) .. controls (150pt,-30pt) and (200pt,-90pt).. node[legend,midway,sloped,xshift=-.3cm,yshift=-.07cm,rotate=27] {Fig. \ref{fig:def-runs-local-types}\\ \S \ref{sec:hnice}} (REnv);
    \path[->,figureblue,draw] (S) -- (S') node[midway] (SS') {} node[legend,near end,sloped,below] {\S \ref{sec:systems}};
    \path[->,figureblue] (SS') -- coordinate[midway] (SS'G) (G);
    \draw (SS') edge[bend right, opacity = .2, double,-implies] node[legend,midway,sloped] (SS'SRS) {Thm. \ref{prop:run-vs-red},\\\S \ref{sec:hnice}} (-50pt,-70pt);
    \path[->,figureblue,draw] (Env) -- (Env') node[legend,sloped,above,near end] {\S \ref{sec:conformance}} node[midway] (EnvEnv') {};
    \path[->,red] (REnv) -- coordinate[midway] (REnvRG) (RG);
    \path[->,red] (RG) -- coordinate[midway] (RGRS) (RS);
    \node[rotate=-45] at (REnvRG) {\Large $\mysubseteq$};
    \node[rotate=225] (wsi) at (RGRS) {\Large $\mysubseteq$};
    \node[xshift=.6cm,yshift=-.25cm,rotate=180] (covers) at (RGRS) {\Large $\mysubseteq$};
    \node[legend, xshift=-.4cm, yshift=.2cm] at (wsi) {Def. \ref{def:wsr},\\\S \ref{sec:hnice}};
    \draw (EnvEnv') edge[bend left, opacity = .2, double] node[legend,midway,sloped] (EnvEnv'REnv) {Lem. \ref{lem:corr-semantics-local-types-back}-\ref{lem:corr-semantics-local-types}\\ \S \ref{sec:wsi-by-typing}} (135pt,-50pt);
    \node[legend, yshift=.35cm] at (j) {\S \ref{sec:typing}};
    \node[legend, yshift=.1cm,xshift=.5cm] at (REnvRG) {Thm. \ref{lem:runs-types-covers-runs-global},\\\S \ref{sec:wsi-by-typing}};
    \node[legend, below = of covers,yshift=.8cm] {Thm. \ref{thm:impl-covers-local}\\ \S \ref{sec:wsi-by-typing}};
  \end{tikzpicture}
\]

\bigskip

This is an extended version of~\cite{ESOP14}.
Besides giving detailed proofs, we simplified some definitions and
typing rules.
In particular, \cref{sec:pseudo} was not included in~\cite{ESOP14} and
was added here to remove communication effects from judgements.
We also improved the general presentation and refined the running
example.

%%% Local Variables:
%%% mode: latex
%%% TeX-master: "main"
%%% End:

\section{Global and Local Types}%
\label{sec:types}
% !TEX root = main.tex

Our types elaborate from~\cite{lt12} and use a form of iteration,
which is more tractable than recursion.
We fix three countably infinite and pairwise disjoint sets:
\begin{itemize}
\item $\shset$ of \emph{shared} names (ranged over by $\shname$),
\item $\chset$ of \emph{session} channels (ranged over by $\smv, \smvo
  s, \ldots$), and
\item $\ptpset$ of \emph{(participants) roles} ranged over by $\ptp p,
  \ptp q, \ptp r, \ldots$; it is convenient to assume that $\ptpset$
  is (isomorphic to) the set of natural numbers.
\end{itemize}
Basic data types, called \emph{sorts}, (\eg, booleans $\dt{Bool}$,
integers $\dt{Int}$, strings $\dt{Str}$, singleton $\dt{Unit}$, record
types, etc.) are assumed; $\sort$ ranges over sorts.
We use \emph{sorted channels} $\smv \ \sort$ to specify that channel
$\smv \in \chset$ is used to exchange data of sort $\sort$.
We write $\tuple z = (z_1,\ldots,z_n)$ for a finite sequence of
elements $z_1, \ldots, z_n \in Z$ for a given set $Z$; when
no confusion arises, $\tuple z$ may also denote the set
$\{z_1,\ldots,z_n\}$ (\eg, we write $z_2 \in \tuple z$).

A \emph{global type term} (GTT, for short) $\GT$ is derived by the
following grammar:
\begin{eqnarray*}
  \GT  & ::= & \Gchoice
               \qquad \sep \qquad
               \GT ; \GT
               \qquad \sep \qquad
               \Gdef{\f}{\GT}
               \qquad \sep \qquad
               \Gend
\end{eqnarray*}
A GTT $\sum_{i\in I} \Gint[p][q_i][\smv_i][\sort_{\mathit i}];\GT_i$
denotes the branching of interactions from a unique selector $\ptp p$
to participants $\ptp q_i$ for $i \in I \neq \emptyset$; we tacitly
assume that $\ptp p \neq \ptp q_i$ for all $i \in I$, and that $\ptp
q_i = \ptp q_j$ implies $\smv_i \neq \smv_j$ for all $i \neq j \in I$
(namely that channels of a same receiver in two different branches are
different).
For a singleton $I = \{n\}$,
$\Gint[@][q_n][\smv_n][\sort_{\mathit n}];\GT_n$ shortens
$\sum_{i \in I} \Gint[p][q_i][\smv_i][\sort_{\mathit i}];\GT_i$.
A communication over a channel of sort $\dt{Unit}$ is a pure
synchronisation and usually we write $\Gint[@][q][\smv][];\GT$ instead
of $\Gint[@][q][\smv][\dt{Unit}];\GT$.
A GTT can also be the sequential ($\_ ; \_$) composition of two GTTs,
the iteration of a GTT ($\Gdef \_ \_$), or the empty term $\Gend$.
We omit trailing occurrences of $\Gend$ and write $\Gint[@][q][\smv]$
instead of $\Gint[@][q][\smv];\Gend$.

As in~\cite{DBLP:journals/corr/abs-1203-0780}, we adopt iteration in
global types.
Function $\f$ in $\Gdef{\f}{\GT}$ is an injective map that assigns
sorted channels to roles.
The mapping is used to indicate how the termination of the iteration
is communicated to the roles. More precisely, if $\f(\ptp p) =
\smv\ \sort$ then the participant $\ptp p$ will receive a message of
type $\sort$ on channel $\smv$ when the iteration terminates.
We use $\range f$ for the set of channels in the image of $f$,
namely
\begin{eqnarray*}
  \range{\f} & = & \{\smv \sst \f(\ptp p) = \smv\ \sort \quad \text{and} \quad \ptp p\in\dom{\f}\}
\end{eqnarray*}

\begin{exa}[Iterative GTT]%
  \label{ex:iterativeGTT}
  We revisit the scenario introduced in \cref{sec:intro}, which
  involves a client $\ptp c$ and a bank $\ptp b$. The GTT $\GT$ below
  defines a protocol in which $\ptp c$, after being logged in, may
  perform zero or more deposits and overdrafts.
  \[
    \begin{array}{l@{\ }ll}
      \GT = \Gint[c][b][\smvo {login}][\dt{Str}];
      \Gdef{\ptp b \mapsto \smvo{quit}}{(
      \\
      & \Gint[c][b][\smvo {deposit}][\dt{Int}]
      \\
      &+
      \\
      &
        \Gint[c][b][\smvo {overdraft}][\dt{Int}];(\Gint[b][c][\smvo {ok}][] + \Gint[b][c][\smvo {ko}][])\quad )}
    \end{array}
  \]
  The protocol starts with the interaction
  $\Gint[c][b][\smvo {login}][\dt{Str}]$, i.e., $\ptp c$ sends to
  $\ptp b$ its login information (of sort $\dt{Str}$) over channel
  $\smvo{login}$.
  Then, the protocol continues as an iterative GTT in which $\ptp c$
  decides to either perform a deposit, ask for an overdraft, or
  finalise the protocol. The protocol can be finalised by $\ptp c$
  by sending to $\ptp b$ a message (of omitted sort $\dt {Unit}$) on
  channel $\smvo {quit}$, as indicated by the function
  $\ptp b \mapsto \smvo{quit}$ used to decorate the iterative type.
  In the interaction $\Gint[c][b][\smvo {deposit}][\dt{Int}]$,
  $\ptp c$ requests a deposit by sending the amount to be deposited
  over channel $\smvo {deposit}$.
  In the interaction
  $\Gint[c][b][\smvo {overdraft}][\dt{Int}]$, $\ptp c$ requests an overdraft.
  Note that in the case of overdraft request $\ptp c$ waits for a
  notification from $\ptp b$ on whether the request is granted or
  denied (a message over $\smvo {ok}$ or $\smvo {ko}$, respectively,
  of the omitted sort $\dt{Unit}$).
  Once the chosen branch has been completed, the iteration can be
  restarted.
   \finex
\end{exa}

We introduce some useful auxiliary notions.
For a GTT $\GT$
  \begin{itemize}
  \item the set of \emph{participants $\participants{\GT}$ of $\GT$}
    is
    \begin{eqnarray*}
      \participants{\sum_{i\in I} \Gint[p][q_i][\smv_i][\sort_{\mathit i}];\GT_i}
      & = &
      \{\ptp p\} \cup \bigcup_{i\in I} (\{\ptp q_i\} \cup \participants{\GT_i})
      \\
      \participants{\GT ; \GT'}
      & = &
      \participants{\GT} \cup \participants{\GT'}
      \\
      \participants{\Gdef{\f}{\GT}}
      & = & \participants{\GT}
      \\
      \participants{\Gend}
      & = & \emptyset
    \end{eqnarray*}

  \item the set of participants ready to send a message, or
    \emph{ready participants $\ready{\GT}$ of $\GT$} is
    \begin{eqnarray*}
      \ready{ \sum_{i\in I} \Gint[p][q_i][\smv_i][\sort_{\mathit i}];\GT_i}
      & = &
      \{\ptp p\}\\
      \\
      \ready{\GT ; \GT'}
      & = &
      \begin{cases}
        \ready{\GT} & \text{ if } \ready{\GT} \neq \emptyset \\
        \ready{\GT'} & \text{ if } \ready{\GT} = \emptyset
      \end{cases}
      \\
      \ready{\Gdef{\f}{\GT}}
      & = &  \ready{\GT}
      \\
      \ready{\Gend}
      & = &
      \emptyset
    \end{eqnarray*}

  \item the set of \emph{channel names $\chn{\GT} \subseteq \chset$ of
    $\GT$} is
    \begin{eqnarray*}
      \chn{\sum_{i\in I} \Gint[p][q_i][\smv_i][\sort_{\mathit i}];\GT_i}
      & = &
      \bigcup_{i\in I} \{\smv_i\}
      \\
      \chn{\GT ; \GT'}
      & = &
      \chn{\GT}\cup\chn{\GT'}
      \\
      \chn{\Gdef{\f}{\GT}}
      & = & \chn{\GT} \cup \range{\f}
      \\
      \chn{\Gend}
      & = &
      \emptyset
    \end{eqnarray*}
  \end{itemize}

\begin{exa}%
  \label{ex:Gf}
  The set
  $\chn{\GT} = \{\smvo {login}, \smvo {deposit}, \smvo{overdraft},
  \smvo{ok}, \smvo{ko}, \smvo{quit}\}$ is the set of channels used by
  the choreography $\GT$ in \cref{ex:iterativeGTT}.
  Also, we have the set of participants
  $\participants{\GT} = \{ \ptp b, \ptp c \}$, where $\ptp c$ is the
  one that can initiate the interaction, i.e.,
  $\ready{\GT} = \{\ptp c\}$.
  \finex
\end{exa}

A \emph{global type} is defined by an equation
\[
\Geq \qquad \text{where} \qquad
\chn \GT \subseteq \tuple \smv  \subseteq \chset
\quad \text{and} \quad
\tuple \smv \text{ are pairwise distinct names}
\]
We abbreviate $\Geq$ with $\G(\tuple \smv)$ when $\GT$ is immaterial;
we write $\G$ or $\GT$ instead of $\G(\tuple \smv)$ when parameters
are understood.

For technical reasons (\cf\ \cref{sec:typing}), our global types are
explicitly parameterised on session channels; however, they will be
considered equivalent up-to renaming of parameters.
More precisely, let $\equiv_\text{GTT}$ be the structural congruence
relation on GTTs such that
\begin{itemize}
\item $\_;\_$ and $\_+\_$ form monoids with identity $\Gend$
\item and $\_+\_$ is commutative;
\end{itemize}
we say that $\Geq[\G_1][\tuple \smv][\GT_1]$ and
$\Geq[\G_2][\tuple{\smvo z}][\GT_2]$ are structurally equivalent
(written $\G_1 \equiv \G_2$) when $\GT_1 \equiv_\text{GTT}
\GT_2\subs{\tuple \smv}{\tuple{\smvo z}}$ where, as usual,
$\subs{\tuple \smv}{\tuple{\smvo \smv}}$ is the capture avoiding
substitution replacing the $i$-th element of $\tuple \smv$ with the
$i$-th element of $\tuple{\smvo z}$ (for which we assume $\tuple \smv$
to be a tuple of pairwise distinct names of the same length as
$\tuple{\smvo z}$).

The extensions of $\participants\_$ and $\ready\_$ to $\Geq$ are
straightforward: $\participants\G = \participants{\GT}$ and $\ready\G
= \ready\GT$.

As customary in session types, we restrict our attention to
\emph{well-formed} global types in order to rule out specifications
that cannot be implemented distributively.
We borrow from~\cite{HondaYC08,DBLP:journals/corr/abs-1203-0780} the
standard wellformedness conditions of \emph{knowledge of choice} and
\emph{linearity}.

Knowledge of choice requires a unique-selector in a choice and that
any other participant can determine the chosen branch from the
received messages.
For instance, the global type
\[
  \Gint[p][q][\smv_1][];\GT_1 ~~ + ~~\Gint[p][r][\smv_2][];
  \Gint[s][q][\smv_1][]; \GT_2
\]
violates the knowledge of choice condition: although there is a
unique-selector $\ptp p$ in the choice (hence the first part of the
condition is satisfied), there is a participant $\ptp q$ that cannot
determine the chosen branch.
More precisely, $\ptp q$ is ready to receive a message on $\smv_1$
either from $\ptp p$ in the first branch or from $\ptp s$ on the
second one; hence $\ptp q$ cannot determine which branch has been
selected after the input on $\smv_1$.

Linearity requires absence of communication races on channels.
Races occur when causally unrelated interactions happen on a same
channel.
Consider
\[
  \Gint[p][q][\smv][];\Gint[r][q][\smv][]
\]
where the two sent actions yield a race on $\smv$ since they are
concurrent because performed by different senders.
On the contrary,
\[
  \Gint[p][q][\smv][];\Gint[q][r][\smv][]
\]
satisfies linearity because the two interactions on $\smv$ are
executed sequentially, hence causally related.

We introduce an additional wellformedness condition
(Definition~\ref{def:well-formed}) that is specific to our form of
iteration.
\begin{defi}[Well-formed iteration]\label{def:well-formed}
  A GTT of the form $\GT^\f$ is a well-formed iteration if:
  \begin{enumerate}
  \item\label{it:wfi} $\range{\f}\cap \chn{\GT} = \emptyset$,
  \item\label{it:wfii} $\ready{\GT}$ is a singleton (we call the participant in
    $\ready{\GT}$ the \emph{iteration-controller of $\GT$}) and
    $\dom{f} = \participants{\GT}\backslash\ready{\GT}$,
  \item\label{it:wfiii} for any proper subterm $\Gdef{\f_1}{\GT_1}$ of $\GT$,
    $\range{\f_1} \cap \range{\f} = \emptyset$.
  \end{enumerate}
\end{defi}

\noindent
By condition~\eqref{it:wfi}, the channels used to terminate an
iteration are disjoint from those in the body $\GT$.
This is fundamental to avoid confusion when implementing iterations
and resembles the condition about knowledge of choices.
Consider the global type
\[
  \Gdef{\ptp q \mapsto \smv}{ (\Gint[p][q][\smv][];\GT_1)};\GT_2
\]
that violates condition~\eqref{it:wfi}.
Note that after receiving a message from $\ptp p$ on $\smv$, $\ptp q$
is unable to determine whether it should behave as specified by
$\GT_1$ (i.e., to perform the body of the iteration) or $\GT_2$
(i.e., to exit the iteration).

Condition~\eqref{it:wfii} requires a unique role (the
iteration-controller) to be the one deciding whether to execute the
iteration body or to terminate by notifying all other participants in
global type.
This condition avoids situations in which a participant is unaware of
the fact that some iteration has been finalised.
For example, in the following GTT
\[
  \Gdef{\ptp q \mapsto \smv_3}{ (\Gint[p][q][\smv_1][];\Gint[q][r][\smv_2][])};\GT_2
\]
$\ptp r$ will not receive any message when $\ptp p$ decides to
conclude the iteration.

Finally, condition~\eqref{it:wfiii} prevents interference between
terminations of nested iterations. Consider
\[
 \Gdef{\ptp q \mapsto \smv_2, \ptp r \mapsto \smv_4}{ (
 \Gdef{\ptp q \mapsto \smv_2}{(\Gint[p][q][\smv_1][])};\Gint[q][r][\smv_3][])};\GT_2
\]
Note that after receiving a message on $\smv_2$, $\ptp q$ is unable to
determine if $\ptp p$ has concluded the inner or the outer iteration.

Hereafter, we will assume that for every $\Geq$, $\GT$ satisfies
knowledge of choice and linearity (form~\cite{HondaYC08,DBLP:journals/corr/abs-1203-0780}), and that every
iterative GTT appearing in $\GT$ is a well-formed iteration.

A \emph{local type term} (LTT for short) $\TT$ is derived from the
following grammar:
\begin{gather*}
  \TT \ ::=  \ \Tssel{\smv_i}{\sort_i}{\TT_i}
  \quad \sep \quad
  \Tbbra{\smv_i}{\sort_i}{\TT_i}
  \quad \sep \quad
  \TT_1 ; \TT_2
  \quad \sep \quad
  \Tdef \TT
  \quad \sep \quad
  \Tend
\end{gather*}
An LTT is either an internal ($\bigoplus$) or external ($\sum$)
guarded choice on non-empty index sets $I$, the sequential composition
$\_ ; \_$, an iteration $\Tdef \_$, or the empty term $\Tend$.
We usually omit trailing occurrences of $\Tend$ and write
$\Tsend{\smv_n}{\sort_n}; \TT_n$ (resp.
$\Treceive{\smv_n}{\sort_n};\TT_n$) for
$\Tssel{\smv_i}{\sort_i}{\TT_i}$
(resp. $\Tbbra{\smv_i}{\sort_i}{\TT_i}$) when $I = \{n\}$.

We assume that all channels appearing in the guards of an internal or
an external choice are pairwise different.  Similarly to GTT, the set
$\chn \TT$ of \emph{channels} of an LTT $\TT$ is defined as
\[
\begin{array}{l}
  \chn{\Tssel{\smv_i}{\sort_i}{\TT_i}}
  \quad = \quad
  \chn{ \Tbbra{\smv_i}{\sort_i}{\TT_i} }
  \quad = \quad \{\smv_i \sst i\in I \} \cup \bigcup_{i \in I} \chn{\TT_i}
  \\[2em]
  \chn{\TT_1 ; \TT_2}  =  \chn{\TT_1} \cup \chn{\TT_2}
  \qquad\qquad
  \chn{\Tdef{\TT}}  =  \chn{\TT}
  \qquad\qquad
  \chn{\Tend}  =  \emptyset
\end{array}
\]

A \emph{local type} is defined by an equation
\[
\Leq
\qquad \text{where} \qquad
\chn \GT \subseteq \tuple \smv  \subseteq \chset
\quad \text{and} \quad
\tuple \smv \text{ are pairwise distinct names}
\]
We abbreviate $\Leq$ with $\T(\tuple \smv)$ when $\TT$ is immaterial
and we write $\T$ or $\TT$ instead of $\T(\tuple \smv)$ when
parameters are understood.

The structural congruence on LTTs is defined as the smallest
congruence $\equiv_\text{LTT}$ such that
\begin{itemize}
\item internal and external choice operators are associative,
  commutative and have $\Tend$ as identity,
\item and $\_ ; \_$ is associative and has $\Tend$ as neutral element.
\end{itemize}
Two local types $\Leq[\T_1][\tuple \smv_1][\TT_1]$ and
$\Leq[\T_2][\tuple \smv_2][\TT_2]$ are structurally equivalent
(written $\T_1 \equiv \T_2$) when
$\TT_1 \equiv_\text{LTT} \TT_2\subs{\tuple \smv_2}{\tuple \smv_1}$.
In the following, we consider types up-to structural congruence.

The \emph{projection}\label{page:projection} operation extracts local
types from a global type; we restrict such operation on well-formed
global types.
Given a participant $\ptp r \in \ptpset$, the \emph{projection} of a
well-formed GTT $\GT$ on $\ptp r$, denoted as $\GT\proj{\ptp{r}}$, is
defined as follows:
{\small
  \begin{eqnarray*}
    \GT\proj\ptp r
    =
	 \begin{cases}
           \Tend  & \text{if } \ptp{r} \not\in \participants{\GT}
           \\
           \Tssel{\smv_i}{\sort_i}{\GT_i\proj\ptp r} & \text{if }\GT =  \displaystyle{\sum_{i\in I} \Gint[r][q_i][\smv_i][\sort_{\mathit i}];\GT_i }
            \\
            \TBbra{i}{I_{\ptp r}}{\smv_i}{\sort_i}{\GT_i\proj\ptp r}  + \sum_{i\in I \setminus I_{\ptp r}} {\GT_i\proj\ptp r} &
              \text{if }\GT = \displaystyle{\sum_{i\in I_{\ptp r}} \Gint[p][q_i][\smv_i][\sort_{\mathit i}];\GT_i + \sum_{i\in I \setminus I_{\ptp r}} \Gint[p][q_i][\smv_i][\sort_{\mathit i}];\GT_i } \\
              & \ptp r \neq \ptp p, \text{ and } I_{\ptp r} = \{i \in I \sst \ptp q_i = \ptp r\}
              \\
              \GT_1\proj\ptp{r}; \GT_2\proj\ptp{r} & \text{if } \GT = \GT_1; \GT_2
              \\
              \Tdef{(\GT_1\proj\ptp{r})};  \Tsend{b_1}{\sort_1};\ldots ; \Tsend{b_n}{\sort_n} &  \text{if }\GT = \Gdef{\f}{\GT_1},~\bigcup\{\f(\ptp p)\, | \, \ptp{p}\in \range{\f}\}=\{b_1\ \sort_1, \ldots, b_n \ \sort_n\}, \\
              & \text{and } \ptp{r} \in \ready{\GT_1}
              \\
              \Tdef{(\GT_1\proj\ptp{r})}; \Treceive{b}{\sort} &  \text{if }\GT =\Gdef{\f}{\GT_1},~\f(\ptp{r})= b\ \sort
            \end{cases}
  \end{eqnarray*}}
  When projecting a global type $\GT$ over a participant name $\ptp r$ that does not appear
  in it, produces the idle local type $\Tend$. The remaining cases implicitly assume that
  $\ptp r\in \participants \GT$.

  The projection of the (unique) selector of a branch results in the
  internal choice on the session channels.
  Dually, the projection on a receiver in a branch results in an
  external choice; observe that branches where the receiver is not
  $\ptp r$ (i.e., $i \in I \setminus I_{\ptp r}$) are treated
  differently from those where the receiver is $\ptp r$ (i.e.,
  $i \in I_{\ptp r}$).
  We remark that the projected local type is well-defined when $\GT$
  is well-formed: the condition about knowledge of choices ensures
  that $\ptp r\in\participants{\GT_i}$ for all
  $i\in I \setminus I_{\ptp r}$; moreover, $\ptp r$ is the receiver in
  the first interaction of each branch $\GT_i$.

  The projection of a sequential composition is self-explanatory.
  Iterative GTTs are projected depending on whether the role is a
  controller or not. For $\GT = \Gdef{\f}{\GT_1}$, recall that (by
  well-formedness) there is an iteration-controller
  $\ptp r \in \ready{\GT_1}$; the projection of $\GT$ on the
  controller generates an iterative local type, which corresponds to
  the projection of $\GT_1$, followed by the messages that signal the
  termination of the iteration to the remaining participants (i.e.,
  the messages $\Tsend{b_1}{\sort_1};\ldots ; \Tsend{b_n}{\sort_n}$).
  Dually, the projection on the other participants waits for the
  signal to exit the iteration (i.e., $\Treceive{b}{\sort}$).
  The \emph{projection} $\G(\tuple \smv) \proj \ptp r$ of a global type
  $\Geq$ with respect to $\ptp r$ is a local type $\Leq$ where $\TT =
  \GT \proj \ptp r$.
  \begin{exa}\label{ex:projection}
	 Consider the GTT $\GT$ introduced in \cref{ex:iterativeGTT}.
	 Since $\GT$ consists of a single branch where $\ptp c$ is the
	 selector, the projection on $\ptp c$ is an internal choice with a
	 single branch that sends a message on channel $\smvo{login}$ and
	 follows with the projection of the iterative type in the
	 continuation.
	 Since $\ptp c$ is the interation-controller of the continuation,
	 its projection consists of the iteration of the body followed by
	 the termination message $\smvo {quit}$, as shown below.
  \begin{align*}
    \GT\proj{\ptp{c}}
     = &
     \Tsend{\smvo{login}}{\dt{Str}}. \Tdef{(
    \Tsend{\smvo{deposit}}{\dt{Int}}
    +
    \Tsend{\smvo{overdraft}}{\dt{Int}}.
    (\Treceive{\smvo{ok}}{} + \Treceive{\smvo{ko}}{}))};\Tsend{\smvo{quit}}{}
    \\
    \GT\proj{\ptp{b}}
    = &
    \Treceive{\smvo{login}}{\dt{Str}}. \Tdef{(
    \Treceive{\smvo{deposit}}{\dt{Int}}
    +
    \Treceive{\smvo{overdraft}}{\dt{Int}}.
    (\Tsend{\smvo{ok}}{} + \Tsend{\smvo{ko}}{}))};\Treceive{\smvo{quit}}{}
  \end{align*}
  Note that $\GT\proj{\ptp{b}}$ and $\GT\proj{\ptp{b}}$ are each other's dual.
  \finex
\end{exa}

%%% Local Variables:
%%% mode: latex
%%% TeX-master: "main"
%%% End:

\section{Types for the POP2 Protocol}%
\label{sec:runGT}
% !TEX root = main.tex

We illustrate our approach on the Post Office Protocol Version 2
(POP2)~\cite{pop2} between a client $\ptp c$ and a mail server $\ptp
s$.
We describe POP2 with the following global type:
\[
\begin{array}{lll}
  \GPOP & \eqdef &
  \Gint[c][s][\smvo{quit}][];\exG \GT {EXIT} \
  +\  \Gint[c][s][\smvo{helo}][Str]; \exG \GT {MBOX}
  \\[0.4mm]
  \exG \GT {EXIT} & \eqdef &
  \Gint[s][c][\smvo{bye}][]
\end{array}
\]
The global type $\GPOP$ starts with $\ptp c$ sending a message to
$\ptp s$ either on channel $\smvo{quit}$ (of the omitted sort
$\dt{Unit}$) or on $\smvo{helo}$ to communicate its password (of sort
$\dt{Str}$).
In the first case, the interaction ends after $\ptp s$ sends a message
on $\smvo {bye}$ as per $\exG \GT {EXIT}$; in the latter case the
protocol follows as per $\exG \GT {MBOX}$ below.
\[
\begin{array}{lll}
  \exG \GT {MBOX} & \eqdef &
  \Gint[s][c][\smvo{e}][];\exG \GT {EXIT} \
  + \  \Gint[s][c][\smvo{r}][Int];{\exG \GT {NMBR}}
\end{array}
\]
In $\exG \GT {MBOX}$, $\ptp s$ either signals an error on $\smvo e$
before terminating or sends a message on $\smvo r$ containing the
number of messages in the default mailbox and then continues as $\exG
\GT {NMBR}$:
\[
\begin{array}{l@{}l@{}r@{\ }l}
  \exG \GT {NMBR}
  & \eqdef
  (&&\Gint[c][s][\smvo{fold}][Str]; \Gint[s][c][\smvo{r}][Int]
  \\
  &&+&
  \Gint[c][s][\smvo{read}][Int];\Gint[s][c][\smvo{r}][Int];\exG \GT {size} )^{*\ptp s\mapsto \smvo{quit}}; % chktex 3
     {\exG \GT {EXIT}}
\end{array}
\]
where $\ptp c$ repeatedly requests either (a) the number of messages
available in a folder, or (b) the length of a particular message in
the current folder.
The iteration-controller is $\ptp c$ and it uses $\smvo {quit}$ to
communicate the termination of the loop to $\ptp s$.
In case (a), $\ptp c$ sends the folder's name over the channel
$\smvo{fold}$ and waits for the answer on $\smvo r$; after this, the
body of the iteration is completed and the loop can be repeated
again. In case (b), $\ptp c$ sends the index corresponding to the
selected message on channel $\smvo{read}$ and waits for the answer on
channel $\smvo r$; after this, the interaction continues as $\exG \GT
{size}$ specified below:
\[
\begin{array}{l@{}l@{}r@{\ }l}
  \exG \GT {size}
  & \eqdef
  (&& \Gint[c][s][\smvo{read}][Int]; \Gint[s][c][\smvo{r}][Int]
  \\
  &&+&
  \Gint[c][s][\smvo{retr}][];\Gint[s][c][\smvo{msg}][Data].
  \exG \GT {xfer})^{*\ptp s\mapsto \smvo{fold} \ \dt{Str}}; % chktex 3
  \Gint[s][c][\smvo{r}][Int]
\end{array}
\]
In $\exG \GT {size}$, another loop controlled by $\ptp c$ lets the
client either (a) ask for another message by interacting again on
$\smvo{read}$ as described above or (b) retrieve a message.
In the latter case, $\ptp c$ signals on $\smvo{retr}$ that it is ready
to receive the message, which is then sent back on $\smvo{msg}$ by
$\ptp s$ (sort $\dt{Data}$ abstracts away the format of messages
specified in~\cite{rfc0822}).
Finally, $\ptp c$ acknowledges the reception of the requested message
as follows:
\[
\begin{array}{lll}
  \exG \GT {xfer}
  &\eqdef&
  \Gint[c][s][\smvo{acks}][]; \Gint[s][c][\smvo{r}][Int] \\
  &+&
  \Gint[c][s][\smvo{ackd}][]; \Gint[s][c][\smvo{r}][Int]
  \\
  &+&
  \Gint[c][s][\smvo{nack}][]; \Gint[s][c][\smvo{r}][Int]
\end{array}
\]
Basically $\ptp c$ may use one among three alternative channels:
$\smvo{acks}$ to acknowledge the reception of the message,
$\smvo{ackd}$ to keep the message or, $\smvo{nack}$ to notify that the
message has not been received properly (in which case the message is
kept in the mailbox).
In each case, $\ptp s$ sends back to $\ptp c$ the length of the next
message over channel $\smvo r$.

The local type $ \TT_{\ptp s}$ obtained by projecting $\GPOP$ on the
participant $\ptp s$, i.e., $\TT_{\ptp s} = {\GPOP}\proj{\ptp s}$, is
below.
\[
\begin{array}{l@{\ \eqdef\ } l}
  \TT_{\ptp s}
  &
  \Treceive{\smvo{quit}}{}.{\exG \TT {EXIT}}\
  +\ \Treceive{\smvo{helo}}{\dt{Str}}. \exG \TT {MBOX}
  \\[1mm]
  \exG \TT {EXIT}
  &
  \Tsend{\smvo{bye}}{}
  \\[1mm]
  \exG \TT {MBOX}
  &
  \Tsend{\smvo{e}}{}.{\exG \TT {EXIT}}\ \oplus\
  \Tsend{\smvo{r}}{\dt{Int}}.{\exG \TT {NMBR}}
  \\[1mm]
  \exG \TT {NMBR}
  &
  (  \Treceive{\smvo{fold}}{\dt{Str}}.\Tsend{\smvo{r}}{\dt{Int}}\
  +\  \Treceive{\smvo{read}}{\dt{Int}}.\Tsend{\smvo{r}}{\dt{Int}}.\exG \TT {size} )^{*}; % chktex 3
  {\Treceive{\smvo{quit}}{};{\exG \TT {EXIT}}}
  \\[1mm]
  \exG \TT {size}
  &
  (\Treceive{\smvo{read}}{\dt{Int}}. \Tsend{\smvo{r}}{\dt{Int}}
  \ +\ \Treceive{\smvo{retr}}{}.\Tsend{\smvo{msg}}{\dt{Data}}.\exG \TT {xfer} )^{*}; % chktex 3
  \Treceive{\smvo{fold}}{\dt{Str}};\Tsend{\smvo{r}}{\dt{int}}
  \\[1mm]
  \exG \TT {xfer}
  &
  \Treceive{\smvo{acks}}{}.\Tsend{\smvo{r}}{\dt{Int}}\
  +\ \Treceive{\smvo{ackd}}{}.\Tsend{\smvo{r}}{\dt{Int}}\
  +\  \Treceive{\smvo{nack}}{}.\Tsend{\smvo{r}}{\dt{Int}}
\end{array}
\]
Note that the messages in $\TT_{\ptp s}$ are as in $\GPOP$. We remark
that $\ptp s$ does not control any of the two iterations (i.e.,
$\exG \GT {NMBR}$ and $\exG\GT {size}$), hence the projections iterate
until $\ptp s$ receive a signal on the termination channels:
$\smvo{quit}$ in $\exG \TT {NMBR}$ and $\smvo{fold}$ in
$\exG \TT{size}$, respectively.

The projection ${\GPOP}\proj{\ptp c}$ of $\GPOP$ on $\ptp c$ is
obtained analogously; the resulting local type is the dual of
$\TT_{\ptp s}$, i.e., the one obtained by substituting internal
choices by external ones and vice versa.

The projection ${\GPOP}\proj{\ptp c}$ of $\GPOP$ onto $\ptp c$ is
obtained analogously; the resulting local type is the dual of
$\TT_{\ptp s}$, i.e., the one obtained by substituting internal
choices by external ones and vice versa.

For illustrative purpose, in the next example we present a multiparty
variant of $\GPOP$, where the authentication is outsourced.

\begin{exa}\label{ex:1btypes}
A multiparty variant of POP2 is defined by the global type $\GPOP'$
below.
\[
\begin{array}{lll}
  \GPOP'
  & \eqdef &
  \Gint[c][s][\smvo{quit}][];{\exG \GT {EXIT}}~ +~
  \Gint[c][s][\smvo{helo}][Str]; \exG {\GT'}{MBOX}
  \\[0.4mm]
  \exG {\GT'}{MBOX}
  & \eqdef &
  \Gint[s][a][\smvo{req}][Str];\Gint[a][s][\smvo{res}][Bool];
  \big(\Gint[s][c][\smvo{e}][];{\exG \GT {EXIT}}~+~
  \Gint[s][c][\smvo{r}][Int];{\exG \GT {NMBR}} \big)
\end{array}
\]
In this version, $\ptp s$ uses a third-party authentication service
$\ptp a$, which is contacted immediately after the server receives a
$\smvo {helo}$ message from the client.
The server sends to $\ptp a$ an authentication request over
$\smvo{res}$ and waits for the authorisation on $\smvo{res}$ ($\exG
\GT {NMBR}$ and $\exG \GT {EXIT}$ remain unchanged).

The following equations
\[
\begin{array}{lll}
  \TT'_{\ptp s}
  & \eqdef &
  \Treceive{\smvo{quit}}{}.\exG \TT {EXIT} ~+ ~
  \Treceive{\smvo{helo}}{\dt{Str}}.\exG \TT {AUTH}
  \\
  \exG \TT {AUTH}
  & \eqdef &
  \Tsend{\smvo{req}}{\dt{Str}}.\Treceive{\smvo{res}}{\dt{Bool}}.\exG{\TT'}{MBOX}
  \\
  \exG{\TT'}{MBOX}
  & \eqdef &
  \Tsend{\smvo{e}}{}.\exG\TT{EXIT} ~\oplus~
  \Tsend{\smvo{r}}{\dt{Int}}.\exG \TT {NMBR}
\end{array}
\]
yield
the  projection $ \TT'_{\ptp s} {\GPOP'}\proj{\ptp s}$ of $\GPOP'$ on $\ptp s$.
  \finex
\end{exa}

%%% Local Variables:
%%% mode: latex
%%% TeX-master: "main"
%%% End:

\section{Processes and Systems}%
\label{sec:systems}
% !TEX root = main.tex

Choreographies specify distributed applications that we refer to as
\emph{systems}. Concretely, systems are the parallel composition of
\emph{processes} realising the roles in a choreography.

Processes manipulate and exchange values obtained by evaluating
expressions.
Let $\varset$ and $\valset$ be two infinite disjoint sets of
\emph{variables} and \emph{basic values} respectively which are both
disjoint from the sets of shared names $\shset$, session channels
$\chset$, and participants $\ptpset$.
Values are specified by \emph{expressions} having the following
syntax:
\[
e \ \bnfdef \  x \sep \val v \sep e_1 \bop e_2 \sep \uop e
\qquad\qquad
\ell \ \bnfdef \ [e_1, \ldots, e_n]  \sep e_1..e_2
\]
An expression $e$ is either a variable $x \in \varset$ or a value
$\val v \in \valset$, or else the composition $e_1 \bop e_2$ of two
expressions through a binary operator $\bop$\!, or the application of
a unary operator $\uop$ to an expression (operators are left
unspecified and can be thought of as the usual logical-arithmetic
operators of programming languages).
We assume that expressions are implicitly sorted and, for simplicity,
our expressions do not include binders of variables, names, or test
for definiteness.
Lists $[e_1, \ldots, e_n]$ and numerical ranges $e_1..e_2$ are used
for iteration; in the latter case, both expressions $e_1$ and $e_2$
are of sort integer.
The empty list is denoted as $\emptylist$ and the operations
$\headL{\ell}$ and $\tailL{\ell}$ respectively return the head and
tail of $\ell$ (defined as usual).
Given an expression $e$ or a list $\ell$, the sets $\var{e}$ and
$\var{\ell}$ of \emph{variables of $e$ and $\ell$} respectively, are
defined as
\[
  \begin{array}{c}
    \var{x} = \{x\}
    \quad
    \var{\val v} = \emptyset
    \quad
    \var{\uop e} = \var{e}
    \quad
    \var{e_1 \bop e_2} = \var{e_1} \cup \var{e_2}
  \\[1em]
    \var{[e_1, \ldots, e_n] } = \bigcup_{i=1}^n\var{e_i}
    \qquad
    \var{e_1..e_2} = \var{e_1} \cup \var{e_2}
  \end{array}
\]

\begin{figure}[t]
  \centering\small
  \begin{tabular}[c]{p{.42\linewidth}p{.5\linewidth}}
    $\begin{array}{l@{\ }c@{\ }l@{\quad}l}
      P,Q & \bnfdef & & \hspace{-.5cm}\textbf{processes}
      \\ & & \Preq \shname n {\tuple \smv} P & \text{request}
      \\ & \sep & \Pacc \shname p {\tuple \smv} P & \text{accept}
      \\[.2em] & \sep & N & \text{choice}
      \\[.2em] & \sep & \Psend \smv e  & \text{send}
      \\[.2em] & \sep & P;Q & \text{sequential}
      \\[.2em] & \sep & \Pif e P Q & \text{conditional}
      \\[.2em] & \sep & \Pfor x \ell P & \text{for}
      \\[.2em] & \sep & \PloopU N N & \text{repeat}
    \end{array}$
    &
    \begin{minipage}[c]{.9\linewidth}
    $\begin{array}{l@{\ }c@{\ }l@{\quad}l}
        N & \bnfdef && \hspace{-.5cm}\textbf{input-guarded processes}
        \\ & &\Pechoice{i \in I}{\smv_i}{x_i}{P_i}& \text{choices}\\[20pt]
        S & \bnfdef & &
        \hspace{-.5cm}\textbf{systems} \\ & & P \\ & \sep & S \sparop
        S & \text{parallel}\\ & \sep & \queue \smv \qmv & \text{queue}
        \\ & \sep & (\nu \tuple \smv \At \shname) S &
        \text{restriction} \\
      \end{array}$
    \end{minipage}
  \end{tabular}
  \caption{\label{fig:syssyn}Syntax of processes and systems.}
\end{figure}
The syntax of processes $P$, input-guarded non-deterministic
sequential processes $N$, and systems $S$ is given in
\cref{fig:syssyn}.
A process $\Preq \shname n {\tuple \smv} P$ requests a new session on
a shared name $\shname$ and then behaves as $P$; dually, process
$\Pacc \shname p {\tuple \smv} P$ accepts the request of a new session
from another process and then behaves as $P$.
A message $e$ is sent on a session channel $\smv$ by the process
$\Psend{\smv} e$.
Sequential composition and conditional are standard.
An input-guarded non-deterministic sequential processes $\Pechoice{i
  \in I}{\smv_i}{x_i}{P_i}$ (conventionally denoted as $\Pend$ when $I
= \emptyset$) can branch over $P_i$ when a message is received on the
session channel $\smv_i$; we assume $\smv_i \neq \smv_j$ when $i \neq
j \in I$.
Our language for processes provide two different constructs for
iterations
\[
\Pfor x \ell P
\qquad\text{and}\qquad
\PloopU N {\Pechoice{i\in I}{\smv_i}{x_i}{P_i}}
\]
Intuitively, the former realises the controllers of iterative global
types, while the latter is used for the remaining roles (\cf\
\cref{sec:typing}).
A for-loop iterates the body $P$ on the list $\ell$.
The body $N$ in a repeat-until-loop is a process of the form
$\Pechoice{i \in I}{\smv_i}{x_i}{P_i}$ and it is repeated until a
message on one of the channels $\smv_i$ on the until guard is
received.

We set the following precedence rules: $\Pif{\_}{\_}{\_}$,
$\Pfor{\_}{\_}{\_}$ and $\PloopU{\_}$ have the lowest precedence while
$\_.\_$ has precedence over $\_;\_$ so that, e.g., the term
$\Pif{e}{P}{ \Pacc \shname i {\tuple \smv} Q}$ reads $\Pif{e}{P}{
  (\Pacc \shname i {\tuple \smv} Q)}$ and $\Pacc \shname i {\tuple
  \smv} P;Q$ reads $(\Pacc \shname i {\tuple \smv} P);Q$.

Systems consist of a parallel composition of process together with the
queues $\queue \smv \qmv$ that store the values $\qmv$ sent over the
session channels $\smv$. Given $\smv_1, \ldots, \smv_h$ pairwise
distinct session channels, we write
$\queue{(\smv_1,\ldots,\smv_h)}{\qmv_1,\ldots,\qmv_h}$ to denote
$\queue{\smv_1}{\qmv_1} \mid \ldots \mid \queue{\smv_h}{\qmv_h}$.
Names $\tuple \smv$ are bound in $(\nu \tuple{\smv} \At \shname) S$
and related to the shared name $\shname$.

The definition of the set $\fn \_$ of \emph{free names} is standard
but for shared names $\shname$ which are also decorated to keep track
of roles; formally we define $\fn \_$ on systems as
\[
\begin{array}c
  \fn{\queue  \smv \qmv}  = \{\smv\}
  \qquad\qquad\quad
  \fn{(\nu \tuple \smv \At \shname) S } =  \fn{S} \setminus \tuple{\smv}
  \qquad\qquad\quad
  \fn{S \sparop S'} =   \fn{S} \cup  \fn{S'}
\end{array}
\]
while for processes we have
\[
\begin{array}{l@{\ }l}
  \fn{\Preq \shname n {\tuple \smv} P}
  =
  \{\shname, \shname[0]\} \cup \fn{P} \setminus \tuple \smv
  &
  \fn{\Pacc \shname p {\tuple \smv} P}
  =
  \{\shname, \shname[p]\} \cup \fn{P} \setminus \tuple \smv
  \\
  \fn{\Psend{\smv} e}
  =
  \{\smv\} \cup \var e
  &
  \fn{\Pechoice{i\in I}{\smv_i}{x_i}{P_i}}
  =
  \bigcup_{i\in I} ( \{\smv_i\} \cup \fn{P_i} \setminus \{x_i\})
  \\
  \fn{P;Q}
  =
  \fn{P}\cup\fn{Q}
  &
  \fn{\Pif e P Q}
  =
  \fn{P} \cup \fn{Q} \cup \var e
  \\
  \fn{\Pfor x \ell P}
  =
  \fn{P} \setminus \{ x \} \cup \fn{\ell}
  &
  \fn{\PloopU N {N'}}
  =
  \fn{N} \cup \fn{N'}
\end{array}
\]
where, in the first equation, a process requesting a new session on
$\shname$ plays the role $0$ and, in the second equation, a process
accepting on $\shname[p]$ plays role $\ptp p$.

The set $\fU[S]$ of \emph{free shared names} of $S$ is defined as
$\fn{S}\cap \shset$. Similarly the set of \emph{free session names}
$\fY[S]$ (resp. \emph{free variables} $\fX[S]$) of $S$ is defined as
$\fn{S}\cap \chset$ (resp. $\fn{S}\cap \varset$).  The set $\bn{\_}$
of \emph{bound names} is defined as
\[
\begin{array}{l@{\quad\qquad}l}
  \bn{\Preq \shname n {\tuple \smv} P} = \tuple \smv  \cup \bn{P}
  &
  \bn{\Pacc \shname p {\tuple \smv} P} = \tuple \smv \cup \bn{P}
  \\
  \bn{\Psend{\smv} e} = \emptyset
  &
  \bn{\Pechoice{i\in I}{\smv_i}{x_i}{P_i}} = \bigcup_{i\in I} (\{x_i\} \cup \bn{P_i})
  \\
  \bn{P;Q} = \bn{P} \cup \bn{Q}
  &
  \bn{\Pif e P Q} = \bn{P} \cup \bn{Q}
  \\
  \bn{\Pfor x \ell P} = \{x\} \cup \bn{P}
  &
  \bn{\PloopU N {N'}} = \bn{N} \cup \bn{N'}
  \\[10pt]
  \bn{\queue \smv \qmv} = \emptyset
  &
  \bn{(\nu \tuple \smv \At \shname) S } = \tuple{\smv} \cup \bn{S}
  \\
  \bn{S \sparop S'} = \bn{S} \cup  \bn{S'}
\end{array}
\]
The set $\bY[S]$ of \emph{bound session names} of $S$ is defined as
$\bn{S}\cap \chset$ and the set $\bX[S]$ of \emph{bound variables} of
$S$ is $\bn{S}\cap \varset$. Note that $\bn{S}\cap \shset=\emptyset$
for all $S$.

\begin{figure}[t]
  \begin{tabular}{lll}
    \multicolumn{3}{l} {\bf Sets}
    \\
    \toprule
    Set & Elements & Description
    \\
    \midrule
    $\varset$ & $x, x_1, \ldots$ & variables
    \\
    $\valset$ & $\val v, {\val v}_1, \ldots$ & basic values
    \\
    $\shset$ & $\shname, \shname_1, \ldots$ & shared names
    \\
    $\chset$ & $\smv, \smvo s, \ldots$ & session channels
    \\
    $\ptpset$ & $\ptp p, \ptp q, \ptp r, \ldots$ & participants
    \\
    $e$ & $e, e_1, \ldots$ & expressions
    \\
    $\ell$ & $\ell, \ell_1, \ldots$ & lists
    \\
    \bottomrule
    \multicolumn{3}{l} {}
    \\
    \multicolumn{3}{l} {\bf Functions}
    \\
    \toprule
    Function & \multicolumn{2}{l}{Description}
    \\
    \midrule
    $\var e, \var \ell$
    & \multicolumn{2}{l}{variables of a expression, a list}
    \\
    $\fn P, \fn S$ & \multicolumn{2}{l}{free names of a process, a system}
    \\
    $\bn P, \fn S$ & \multicolumn{2}{l}{bound names of a process, a system}
    \\
    $\bY[P], \bY[S]$
    & \multicolumn{2}{l}{bound session names of a process, a system}
    \\
    $\bX[P],\bX[S]$
    & \multicolumn{2}{l}{of bound variables of a process, a system}
    \\
    \bottomrule
  \end{tabular}
  \caption{Summary of notation for processes}%
  \label{fig:notation-processes}
\end{figure}

\cref{fig:notation-processes} summarises the notation introduced so
far for the syntax of processes.

As customary, we rely on a \emph{structural congruence relation}
$\equiv$ defined as the least congruence over systems closed with
respect to $\alpha$-conversion and the following axioms
\[
\begin{array}{l@{\hspace{1cm}}l}
  (\nu \tuple \smv \At  \shname)(\nu \tuple \smv ' \At  \shname{'}) S
  \equiv
  (\nu \tuple \smv' \At  \shname{'})(\nu \tuple{\smv} \At  \shname) S
  &
  (\nu \tuple \smv \At  \shname) \Pend \equiv \Pend
  \\
  (\nu \tuple \smv \At  \shname) (S \sparop S')
  \equiv
  S\ \sparop\ (\nu \tuple{\smv} \At  \shname) S',
  \text{ when } \tuple{\smv}\not\subseteq \fY[S]
\end{array}
\]
and such that $\_\sparop\_$ and $\_ ; \_$ form monoids with identity
$\Pend$ and the former is commutative.

The operational semantics of systems is given by the LTS inductively
defined by the rules in \cref{fig:LTSproc,fig:LTSsys} where
\begin{itemize}
\item a store $\sigma$ records both the values assigned to variables
  and the session channels created by a process,
\item  $\sigma\upd{\val v}{x}$ is the update of
  $\sigma$ at $x$ with $\val v$ (and likewise for
  $\sigma\upd \shname \smv$), and
\item $\eval e \sigma$ is the evaluation of $e$ (defined if
  $\var e \subseteq \dom\sigma$ and undefined otherwise).
  We assume that an expression $e$ depends only on its variables, that
  is, for all stores $\sigma$ and $\sigma'$:
  \begin{align*}
	 \sigma|_{\varset} = \sigma'|_{\varset}
	 \implies \eval e \sigma = \eval e \sigma'
  \end{align*}
  where $\cdot | \_$ is the standard restriction of a function on a
  subset of its domain.
\end{itemize}

\noindent
Labels are given by the following productions
\begin{align}
  \alpha & \bnfdef \lreq \shname n {\tuple \smv} \quad \sep \quad
  \lacc \shname p {\tuple \smv}  \quad \sep \quad
  \lsend \smv {\val v} \quad \sep \quad
  \lreceive \smv {\val v} \quad \sep \quad
  \tau \quad \sep \quad
  \lcond \C \alpha \label{eq:ltslabel}
\end{align}
that respectively represent the request of initialisation of a session
on $\shname$, the acceptance of joining a session on $\shname$ with
role $\ptp p$, the sending of a value on $\smv$, the reception of a
value on $\smv$, the silent step $\tau$, and \emph{conditional}
actions $\lcond \C \alpha$ where $\C$ is a boolean expression.
A conditional action labelling a transition $\state[S] \Ptrans{\lcond
  \C \alpha} \state[S'][\sigma']$ denotes that $\state[S]$ performs
the action $\alpha$ and moves to $\state[S'][\sigma']$ because $\eval
\C \sigma$ holds.
We may write $\alpha$ instead of $\lcond{\truek}{\alpha}$ and
$\lcond{\C \land \C'}{\alpha}$ instead of $\lcond \C {(\lcond{\C'}
  \alpha)}$.

Functions $\fY[\_]$ and $\bn \_$ extend to labels as follows:
\[
\begin{array}{l@{\quad\ }l@{\quad\ }l}
  \fY[\lreq \shname n {\tuple \smv} ]
  =
  \fY[\lacc \shname p {\tuple \smv} ] = \{\shname\}
  &
  \fY[\lsend \smv {\val v}]
  =
  \fY[\lreceive \smv {\val v}] = \{\smv\}
  \quad
  \fY[\tau]
  =
  \emptyset
  &
  \fY[\lcond \C \alpha]
  =
  \fY[\alpha]
  \\[1em]
  \bn{\lreq \shname p {\tuple \smv}} =
  \bn{\lacc \shname p {\tuple \smv}} = \tuple \smv
  &
  \bn{\lsend{\smv}{\val{v}}} =
  \bn{\lreceive{\smv}{\val v}} =
  \bn\tau = \emptyset
  &
  \bn{\lcond \C \alpha} = \bn \alpha
\end{array}
\]

\newcommand{\smartrule}[3]{ \irule{#1}{#2} \quad \myrule{#3}}

\begin{figure}[t]\small
  \[
  \begin{array}{l@{\hspace{1em}}r}
    \multicolumn 2 c{
      \mathrule{
        \tuple \smv \not\in \dom \sigma
      }{
        \state[\Preq{\shname}{n}{\tuple \smv}{P}]
        \Ptrans{\lreq{\shname}{n}{\tuple \smv}}
        \state[P][\sigma\upd{\shname}{\tuple \smv}]
      }{SReq}
      \hfill
      \mathrule{
        \tuple \smv \not\in \dom \sigma
      }{
        \state[\Pacc \shname p {\tuple \smv} P]
        \Ptrans{\lacc \shname p {\tuple \smv}}
        \state[P][\sigma\upd{\shname}{\tuple \smv}]
      }{SAcc}
      \hfill
      \mathrule{
        \eval \C \sigma = \val v
      }{
        \state[\Psend \smv \C {}] \Ptrans{\lsend \smv {\val{v}}} \state[\Pend]
      }{SSend}
    }
    \\[3em]
    \mathaxiom{
      \state[\Pechoice{i\in I}{\smv_i}{x_i}{P_i}]
      \Ptrans{\lreceive \smv_j {\val v}}
      \state[P_j][\sigma\upd{\val v}{x_j}]\, j \in I
    }{SRcv}
    &
    \mathrule{
      \state[P] \Ptrans{\lcond \C \alpha} \state[P'][\sigma']
    }{
      \state[P;Q] \Ptrans{\lcond \C \alpha} \state[P';Q][\sigma']
    }{SSeq}
    \\[3em]
    \mathrule{
      \eval \C \sigma = \truek
      \quad
      \state[P]
      \Ptrans{\lcond{\C'}{\alpha}}
      \state[P'][\sigma']
    }{
      \state[\Pif{\C}{P}{Q}]
      \Ptrans{\lcond{\C \land \C'}{\alpha}}
      \state[P'][\sigma']
    }{SThen}
    &
    \mathrule{
      \eval \C \sigma = \falsek
      \quad
      \state[Q] \Ptrans{\lcond{\C'} \alpha} \state[Q'][\sigma']
    }{
      \state[\Pif \C P Q] \Ptrans{\lcond{\neg \C \land \C'}{\alpha}} \state[Q'][\sigma']
    }{SElse}
    \\[3em]
    \mathrule{
      \eval \ell \sigma \neq \emptylist\
      \qquad
      \state[P][\sigma\upd{\headL{\eval \ell \sigma}}{x}]
      \Ptrans{\lcond \C \alpha}
      \state[P'][\sigma']
    }{
      \state[\Pfor{x}{\ell}{P}]
      \Ptrans{\lcond \C \alpha}
      \state[P';\Pfor{x}{\tailL{\ell}}{P}][\sigma']
    }{SFor}
    &
    \mathrule{
      \isemptyL{\eval \ell \sigma}
    }{
      \state[\Pfor{x}{\ell}{P}] \Ptrans{\tau} \state[\Pend]
    }{SForEnd}
    \\[3em]
    \multicolumn 2 c {
      \mathrule{
        \state[N] \Ptrans{\lcond \C \alpha} \state[P][\sigma']
        \quad
        M = \Pechoice{i \in I}{\smv_i}{x_i}{P_i}
        \quad
        \forall i\in I. \smv_i \not\in \fY[\alpha]
      }{
        \state[\PloopU{N}{M}]
        \Ptrans{\lcond \C \alpha}
        \state[P;\PloopU N M][\sigma']
      }{SLoop}
      \
      \mathrule{
        \state[M] \Ptrans{\lcond \C \alpha}\state[P][\sigma']
      }{
        \state[\PloopU{N} M]\Ptrans{\lcond \C \alpha} \state[P][\sigma']
      }{SLoopEnd}
    }
  \end{array}
  \]
  \caption{\label{fig:LTSproc}Labelled transitions for processes}
\end{figure}

\begin{figure}[t]
  \begin{center}\footnotesize
    \[
    \begin{array}{c}
      \mathaxiom{
        \state[\Preq \shname n {\tuple \smv}{P_0}
          \sparop \Pacc \shname 1 {\tuple \smv}{P_1}
          \sparop \ldots \sparop \Pacc \shname n {\tuple \smv}{P_n}]
        \Ptrans{\tau}
        \state[(\nu \tuple \smv \At \shname )(
          P_0
          \sparop \ldots \sparop P_n
          \sparop \tuple \smv : \emptyset
          )][\sigma\upd{\shname}{\tuple \smv}]
        \qquad
        \tuple \smv \not\in \dom \sigma
      }{SInit}
      \\[2em]
      \mathrule{
        \state[P]
        \Ptrans{\lcond e {\lsend{\smv}{\val v}}}
        \state[P'][\sigma']
      }{
        \state[P \ \sparop\  \queue \smv \qmv]
        \Ptrans{\lcond e \tau}
        \state[P' \ \sparop\ \queue \smv {\qmv \cdot \val v}][\sigma']
      }{SCom_1}
      \hfill
      \mathrule{
        \state[P]
        \Ptrans{\lcond e {\lreceive{\smv}{\val v}}}
        \state[P'][\sigma']
      }{
        \state[P\ \sparop\  \queue \smv {\val v \cdot \qmv}]
        \Ptrans{\lcond e \tau}
        \state[P'\ \sparop\ \queue \smv \qmv][\sigma']
      }{SCom_2}
      \\[3em]
      \mathrule{
        \state[S_1]
        \Ptrans{\lcond e \alpha}
        \state[S_1'][\sigma']
        \qquad
        \bn \alpha \cap \dom{\sigma} = \emptyset
        \qquad
        \fX[S_2] \cap (\dom{\sigma'} \setminus \dom{\sigma}) = \emptyset
      }{
        \state[S_1 \sparop S_2]
        \Ptrans{\lcond e \alpha}
        \state[S'_1 \sparop S_2][\sigma']
      }{SPar}
      \\[3em]
      \mathrule{
        \state[S]
        \Ptrans{\lcond e \alpha}
        \state[S'][\sigma']
        \qquad
        \tuple \smv \cap \fY[\alpha] = \emptyset
      }
      {
        \state[(\nu \tuple \smv \At \shname) S]
        \Ptrans{\lcond e \alpha}
        \state[(\nu \tuple \smv \At \shname) S'][\sigma']
      }{SNew}
      \qquad\hfill\qquad
      \mathrule{
        S \equiv S_1
        \qquad
        \state[S_1]
        \Ptrans{\lcond e \alpha}
        \state[S_2][\sigma']
        \qquad
        S_2 \equiv S'
      }
      {
        \state[S]
        \Ptrans{\lcond e \alpha}
        \state[S'][\sigma']
      }{SStr}
    \end{array}
    \]
  \end{center}
  \caption{\label{fig:LTSsys}Labelled transitions for systems}
\end{figure}

We comment on the rules in \cref{fig:LTSproc} and~\ref{fig:LTSsys}.
In~\cref{fig:LTSproc}, rules \myrule{SReq} and \myrule{SAcc} deal with
the initialisation of new sessions; the store is updated to keep track
of the fresh session channels $\tuple \smv$ used in the choreography
$\shname$ (implicitly $\alpha$-converting $\tuple \smv$ when $\tuple
\smv \in \dom \sigma$). Rule \myrule{SSend} is for sending values.
Rule \myrule{SRcv} is for receiving messages in an early style
approach (variables are assigned when firing an input prefix); the
store is updated by recording that $\val v$ is is assigned to $x$.
Rule \myrule{SSeq} is for sequential composition.
Rules \myrule{SThen} and \myrule{SElse} handle conditional statements
as expected; their only peculiarity is that the guard is recorded on
the label of the transition, which is instrumental for establishing
the correspondence between systems and their types
(\cf\ \cref{sec:properties}).
Rules \myrule{SFor}, \myrule{SForEnd}, \myrule{SLoop}, and
\myrule{SLoopEnd} unfold the corresponding iterative process as
expected.

We now comment on the rules in \cref{fig:LTSsys}.
Rule \myrule{SInit} synchronises $n$ roles with the process $\Preq
\shname n {\tuple \smv_0} {P_0}$ initialising the session; this
creates a new session with (initially empty) queues on fresh session
names $\tuple \smv$.
These queues are used to exchange values as prescribed by rules
\myrule{SCom_1} and \myrule{SCom_2}.
Communication actions of processes become silent at system level
capturing the fact that each action is performed over a session queue.
Rule \myrule{SPar} stands for those transitions involving just some of
the components in a system.
By the condition $\bn \alpha \cap \dom{\sigma} = \emptyset$ in the
premiss of \myrule{SPar}, $\alpha$ should contain fresh session names
when it corresponds to the creation of a new session (i.e., it is
either $\lreq \shname n {\tuple \smv}$ or
$\lacc \shname p {\tuple \smv}$).
The rightmost condition in the premiss of rule \myrule{SPar} ensures
that each process has its own local (logical) store (i.e., there is no
confusion between bound variables of different processes).
Rule \myrule{SNew} is standard and allows an action $\alpha$ to be
observed only if it does not involve restricted names.
Rule \myrule{SStr} is standard.

%%% Local Variables:
%%% mode: latex
%%% TeX-master: "main"
%%% End:

\section{Processes of the POP2 Protocol}%
\label{sec:runGTp}
% !TEX root = main.tex

We now present an implementation for the role $\ptp s$ of the global
type $\GPOP$ introduced in \cref{sec:runGT}.
To ease the presentation, we abstract away from the concrete
representation of folders and use the following auxiliary abstract
operations:
\begin{center}
\begin{tabular}{l@{\ :\ }lp{.7\columnwidth}} % chktex 26
    $\dt{auth}$ & $ \dt{Str} \to \dt{Bool}$
    & the authentication predicate,
    \\
    $\dt{msgs}$ & $\dt{Str} \to \dt{Int}$
    & maps a  folder name into the number of messages  in that folder,
    \\
    $\dt{len}$ & $\dt{Int} \to \dt{Int}$
    & maps a message index into the length of the message,
    \\
    $\dt{data}$ & $\dt{Int} \to \dt{Data}$
    & maps a message index into  its content,
    \\
    $\dt{next}$ & $\dt{Int} \to \dt{Int}$
    & maps a message index to the next index  in the folder,
    \\
    $\dt{del}$ & $\dt{Int} \to \dt{Int}$
    & maps a message index to the next index in the folder after deletion.
\end{tabular}\end{center}
As specified for POP2~\cite{pop2}, the value $\dt{inbox}$ of sort
${\dt{Int}}$ denotes the default folder.

Process $\poppname{Init}$ below gives an implementation of the role
$\ptp s$ in the protocol POP2 described in \cref{sec:runGT}.
\begin{align}\label{ex:1proc}
  \poppname{Init} & \eqdef \Pacc{\shname}{ s}{\tuple \smv} \poppname{Srv}
\end{align}
$\poppname{Init}$ starts by joining a session on the shared channel
$\shname$; it plays role $\ptp s$ over the session channels $\tuple
\smv = (\popcname{quit},
\popcname{helo},
\popcname{bye},
\popcname{r},
\popcname{e},
\popcname{fold},
\popcname{read},
\popcname{retr},
\popcname{msg},
\popcname{acks},
\popcname{ackd},
\popcname{nack})$.
Once the session is initiated, the continuation $\poppname{Srv}$
implements the local type $\TT_{\ptp s}$ in \cref{sec:runGT}.
\begin{align}
  \poppname{Srv}
  &
  \eqdef
  \Preceive{\popcname{quit}}{}{\poppname{Exit}}~+~
  \Preceive{\popcname{helo}}{c}{\poppname{Mbox}(c)}
  \qquad \text{where} \qquad
  \poppname{Exit} \eqdef \Psend{\popcname{bye}}{}
\end{align}
As specified by $\TT_{\ptp s}$, $\poppname{Srv}$ waits for a message
on either $\popcname{quit}$ or $\popcname{helo}$. After receiving a
message on $\popcname{quit}$, it sends a message on $\popcname{bye}$
and terminates, as defined by $\poppname{Exit}$.  If the client sends
instead its credentials over channel $\popcname{helo}$, then the
implementation follows with $\poppname{Mbox}(c)$:
\begin{align}
\label{proc:mbox}
  \poppname{Mbox}(c) & \eqdef  \Pif{(\dt{auth}\ c)}
           {\Psend{\popcname{r}}{(\dt {msgs}\ \dt{inbox})};{\poppname{Nmbr}}}
           {(\Psend{{\popcname{e}}}{};{\poppname{exit}})}
\end{align}
$\poppname{Mbox}(c)$ resolves the non-deterministic choice in $\exG
\TT {MBOX} = \Tsend{\popcname{e}}{}.{\exG \TT
  {EXIT}}\ \oplus\ \Tsend{\smvo{r}}{\dt{Int}}.{\exG \TT {NMBR}}$ with
a conditional statement that evaluates the credentials provided by the
client. When they are valid, the implementation sends the number of
messages in the default folder $\dt{inbox}$ over $\popcname{r}$ and
proceeds as $\poppname{Nmbr}$ below.
On the contrary, the process closes the session after sending messages
over $\popcname{e}$ and $\popcname{bye}$.

Process $\poppname{Nmbr}$ implements the iterative behaviour defined
by the local type $\exG \TT {NMBR}$:
\begin{align}\label{pop:nmbr}
\poppname{Nmbr} & \eqdef
\PloopU{\Preceive{{\popcname{fold}}}{f}{}
  \Psend{{\popcname{r}}}{(\dt{msgs}\ f)}~+~\Preceive{{\popcname{read}}}{m}{}
  \Psend{{{\popcname{r}}}}{(\dt{len}\ m)};\poppname{Size}(m)}
{\Preceive{{\popcname{quit}}}{}{\poppname{Exit}}}
\end{align}
Since $\ptp s$ is not the iteration-controller, its implementation
uses a repeat-until-loop in which a client can repeatedly ask for the
length of a folder (by using $\popcname{fold}$) or retrieve messages
(by using $\popcname{read}$) until it terminates the interaction by
sending a message over $\popcname{quit}$.

Processes $\poppname{Size}(m)$ and $\poppname{Xfer}(m)$ are the
implementations of the local types $\exG \TT {Size}$ and $\exG \TT
{Xfer}$:
\begin{align*}
  \poppname{Size}(m)
  &
  \eqdef
  \PloopU{
    \Preceive{\popcname{retr}}{}{}
    \Psend{\popcname{msg}}{(\dt{data}\ m)};
    \poppname{Xfer}(m)~+~\Preceive{\popcname{read}}{m}{}
    \Psend{\popcname{r}}{(\dt{len}\ m)}
  }{
    \Preceive{\popcname{fold}}{f}{}
    \Psend{\popcname{r}}{(\dt{msgs}\ f)}
  }
  \\
  \poppname{Xfer}(m)
  &
  \eqdef
  \Preceive{\popcname{acks}}{}{}
  \Psend{\popcname{r}}{(\dt{len}\ (\dt{next}\ m))}
  ~+~\Preceive{\popcname{ackd}}{}{}
  \Psend{\popcname{r}}{(\dt{len}\ (\dt{del}\ m))
  }~+~\Preceive{\popcname{nack}}{}{}
  \Psend{\popcname{r}}{(\dt{len}\ m)}
\end{align*}

Let $\GPOP'$ be the multiparty variant of POP2 introduced in
\cref{ex:1btypes}.  The process $\poppname{Init}'$ below is a possible
implementation of $\ptp s$ in $\GPOP'$ (i.e., $\TT'_{\ptp s}$).

\begin{align}\label{ex:1bproc}
  \poppname{Init}'
  &
  \eqdef
  \Pacc{\shname}{ s}{\tuple \smv} \poppname{Srv}'
  \\\nonumber
  \poppname{Srv}'
  & \eqdef
  \Preceive{\popcname{quit}}{}{\poppname{Exit}}
  ~+~\Preceive{\popcname{helo}}{c}{\poppname{Auth}(c)}
  \\\nonumber
  \poppname{Auth}(c)
  &
  \eqdef
  \Psend{{\popcname{req}}}{c}; \Preceive{\popcname{res}}{a}{\poppname{Mbox'}}
  \\\nonumber
  \poppname{Mbox'}(c,a)\
  &
  \eqdef
  \Pif{
    (\dt{auth}\ c)\land a
  }{
    \Psend{\popcname{r}}{(\dt{msgs}\ \dt{inbox})};{\poppname{Nmbr}}
  }{
    \Psend{\popcname{E}}{}.{\poppname{Exit}}
  }
\end{align}
$\poppname{Init}'$ is analogous to $\poppname{Init}$ in~\eqref{ex:1proc}; we remark that $\tuple \smv$ now includes also the
session channels $\popcname{req}$ and $\popcname{res}$ used for
interacting with the authorisation authority.
After receiving the credentials of the client on session channel
$\popcname{helo}$, the server interacts with the authorisation
authority as defined in $\poppname{Auth}(c)$: it forwards the
credentials over session channel $\popcname{req}$ and awaits for the
authorisation outcome $a$ on session channel $\popcname{res}$.
Finally, $\poppname{Mbox}'(c,a)$ resolves the non-deterministic choice
in $\exG \TT {Mbox}$ by taking into account both the authorisation
outcome $a$ and the client's credentials $c$.  In this variant, a
client can access the inbox only if the credentials satisfy
\emph{both} the local authentication function \emph{and} the external
authentication service.

%%% Local Variables:
%%% mode: latex
%%% TeX-master: "main"
%%% End:

\section{Whole-Spectrum Implementation}%
\label{sec:hnice}
% !TEX root = main.tex
%
In this section we formally characterise the whole-spectrum
implementations of a role in a global type.
We start by introducing the notion of (candidate) implementation of a
global type, that is, a system in which each role of the global type
is implemented by a process.
The following definition syntactically characterises the processes
that can play a specific role $\ptp p$ in the implementation of a
global type, i.e., those processes that are able to open a session to
play role $\ptp p$.

\begin{defi}[Unique role]\label{def:uniquerole}
  A process $P$ \emph{uniquely plays role $\ptp p$ in $\shname$} if
  either of the following cases holds
\begin{itemize}
\item $P = \Preq \shname n {\tuple \smv} Q$, $\shname \not\in \fn Q$,
  and $\ptp p = 0$
\item $P = \Pacc \shname p {\tuple \smv} Q$ and
  $\shname \not\in \fn Q$
\item $P = \Pechoice{i\in I}{\smv_i}{x_i}{Q_i}$ and $Q_i$ uniquely
  plays role $\ptp p$ in $\shname$ for each $i \in I$
\item $P = \Pif e {Q_1} {Q_2}$ and both $Q_1$ and $Q_2$ uniquely play
  role $\ptp p$ in $\shname$
\item $P = Q_1 ; Q_2$ and either $Q_1$ uniquely plays role $\ptp p$ in
  $\shname$ and $\shname \not \in \fn{Q_2}$ or $Q_2$ uniquely plays
  role $\ptp p$ in $\shname$ and $\shname \not\in \fn{Q_1}$
\item $P = \PloopU {Q_1} {Q_2}$ and $Q_2$ uniquely plays role $\ptp p$
  in $\shname$ and $\shname \not \in \fn{Q_1}$.
\end{itemize}
\end{defi}
\noindent
For technical simplicity, we require a process playing role $\ptp p$
in $\shname$ to open just one session over the shared channel
$\shname$ (note the restriction $\shname \not\in \fn Q$ in the first
two items of the definition); a process playing different roles in
several instances of the same global type can be handled by using
different shared names associated to the same global type.
For branches and conditional forms we require the process to play role
$\ptp p$ in $\shname$ regardless of the chosen branch (e.g., in every
continuation $Q_j$ of a branching process).
The case for sequential composition is straightforward. We remark that
$\Psend \smv e$ does not play a role in a shared name because it
cannot open any session over any shared name.
We also exclude processes like $\Pfor x \ell Q$, which could
potentially open several sessions of a global type (once in any
iteration of the loop).  The condition for $\PloopU {Q_1} {Q_2}$ is
analogous when requiring $\shname \not \in \fn{Q_1}$.

To introduce the notion of implementations of a global type it is
convenient to use \emph{contexts}, that is terms derived from the
following productions:
\[
  \ctx[\_] ::= \_ \sep \ctx[\_] \sparop S \sep \sep S \sparop \ctx[\_]
  \sep (\nu \tuple \smv \At \shname) \ctx[\_]
\]
\begin{defi}[Implementation]\label{def:realisation}
  Let $\iota$ be a mapping assigning a process to each
  $\ptp p \in \{\ptp p_0, \ldots, \ptp p_n\} \subseteq \ptpset$,
  $P = \iota(\ptp p_0) \mid \ldots \mid \iota(\ptp p_n)$,
  $\tuple \smv$ a tuple of pairwise disjoint session channels in
  $\chset$, and $\shname \in \shset$.
  A system $\Gimp$ is an \emph{$\iota$-implementation of $\tuple \smv$
	 at $\shname$ for $\{\ptp p_0, \ldots, \ptp p_n\}$} if there is a
  context
  $\ctx[\_] = (\nu \tuple \smv_1 \At \shname_1) \cdots (\tuple \smv_h
  \At \shname_h)(\_ \mid S)$ such that
  \begin{enumerate}
  \item\label{iota:names}
    $\shname \not\in \{\shname_1,\ldots\shname_h\}$,
    $(\{\shname\} \cup \tuple \smv) \cap \fn S = \emptyset$, and
    $\tuple \smv \cap \big(\bigcup_{i=1,\ldots,h}{\tuple \smv_i}
    \big) = \emptyset$
  \item\label{iota:start} if
    $\Gimp \equiv \ctx[P]$ then, for all $0 \leq j \leq n$, $\iota(\ptp p_j)$
    uniquely plays role $\ptp p_j$ in $\shname$
  \item\label{iota:queue} if
    $\Gimp \equiv \ctx[(\nu \tuple \smv \At \shname)(P) \mid \queue{\tuple
      \smv}{\qmv_1,\ldots,\qmv_k}]$ for some $\qmv_1,\ldots,\qmv_k$
    then
    $\shname \not\in \fn{P}$.
  \end{enumerate}
  Given a global type $\G(\tuple \smv)$, an
  \emph{$\iota$-implementation of $\G(\tuple \smv)$ at $\shname$} is
  an $\iota$-implementation at $\shname$ for $\participants \G$.
\end{defi}
\noindent
Intuitively, a system $\Gimp$ is an implementation of $\G(\tuple
\smv)$ if $\Gimp$ is built-up from processes that implement all the
roles in $\G(\tuple \smv)$; the association between roles and
processes is given by the function $\iota$. In addition, $\Gimp$ may
contain other processes, possibly running different sessions.
Technically, we require $\Gimp$ to be written in terms of a context
$\ctx[\_] = (\nu \tuple \smv_1 \At \shname_1) \cdots (\tuple \smv_h
\At \shname_h)(\_ \mid S)$, which describes the part of the system
that does not directly implement $\G(\tuple \smv)$.
The conditions $\shname \not\in \{\shname_1,\ldots\shname_h\}$,
$\tuple \smv \cap \big(\bigcup_{i=1,\ldots,h}{\tuple \smv_i} \big) =
\emptyset$, and
$(\{\shname\} \cup \tuple \smv) \cap \fn S = \emptyset$ ensure that
the context does not interfere with the names used for implementing
$\G(\tuple \smv)$.
Then, an implementation has two different shapes depending on whether
the session for $\G(\tuple \smv)$ has been initiated or not.
Condition~\eqref{iota:start} stipulates that, before starting the
session, each process $\iota(\ptp p)$ uniquely plays the role $\ptp p$
in $\shname$ (i.e., $\iota(\ptp p)$ is able to open a session on
$\shname$ for the role $\ptp p$).
Condition~\eqref{iota:queue} characterises the case in which the
session has been initiated, and therefore the system contains the
message queues for the initiated session.

\begin{exa}
  Consider the global type $\GPOP$ in \cref{sec:runGT} and take $\Gimp
  = \iota(\ptp s)\ |\ \iota(\ptp c)$ where $\iota$ is such that
  $\iota(\ptp s) = \poppname{Init}$ and $\iota(\ptp c) = \poppname C$
  with $\poppname{Init}$ the process in~\eqref{ex:1proc}
  (\cf page~\pageref{ex:1proc}) and $\poppname{C} \eqdef
  \Preq{\shname}{1}{\tuple \smv} \Psend{\popcname{quit}};
  \popcname{bye}$.
  It is easy to check that $\poppname{Init}$ uniquely plays $\ptp s$
  in $\shname$ while $\poppname C$ uniquely plays $\ptp c$ (after
  assuming that $\ptp c$ is $\ptp 0$).
  Hence, it is straightforward that $\Gimp$ is a
  \emph{$\iota$-implementation} of $\G(\tuple \smv) = \GPOP$ at
  $\shname$ (it is enough to consider the identity context $\ctx[\_] =
  \_$).

  Consider now a more involved situation in which a process
  $\poppname{C'}$ that implements $\ptp c$ also interacts with another
  process $\poppname{D}$ over a different session, e.g.,
  \[
  \poppname{C'}
  \eqdef
  \Pacc{\shname_0}{a}{{\smvo z}}
  \Preq{\shname}{1}{\tuple \smv}
  \Psend{\popcname{quit}};
  \popcname{bye}.\poppname{C''}
  \qquad\qquad
  \poppname{D} \eqdef \Preq{\shname_0}{1}{{\smvo z}}{\poppname{D'}}
  \]
  In this case, $\iota'$ is such that $\iota'(\ptp s) =
  \poppname{Init}$ and $\iota'(\ptp c) = C'$.  Then, $\Gimp[@][\iota']
  = \poppname{Init}\ |\ \poppname{C}\ |\ \poppname{D}$ is an
  \emph{$\iota'$-implementation} of $\G(\tuple \smv)$ at $\shname$ (it
  suffices to consider the context $\ctx[\_] = \_ \ |\ D$). Note that
  \[
  \Gimp[@][\iota']
  \Ptrans{\tau}
  (\nu {\smvo z} \At \shname_0)
  (
  \poppname{Init}\ |\ \poppname{C'''}\ |\ \poppname{D'} \
  |\ \queue{\tuple {\smvo z}}{}
  )
  =
  \Gimp[@][\iota'']
  \]
  with $\poppname{C'''}\ \eqdef\ \Preq{\shname}{1}{\tuple \smv}
  \Psend{\popcname{quit}}; \popcname{bye}.\poppname{C''}$ and $\iota''
  = \iota'[\ptp c \mapsto \poppname{C'''}]$. Then, we can conclude
  that $\Gimp[@][\iota'']$ is an $\iota''$-implementation of
  $\G(\tuple \smv)$ at $\shname$ by considering the context $\ctx[\_]
  = (\nu {\smvo z} \At \shname_0)(\_ |\ \poppname{D'}
  \ |\ \queue{\tuple {\smvo z}}{})$.

  Consider now the transition
  \[
  \Gimp[@][\iota'']
  \Ptrans{\tau}
  (\nu {\smvo z} \At \shname_0)(
  (\nu {\tuple \smv} \At \shname)(
  \poppname{Srv}\ |\ \poppname{C''''}\ |\
  \queue{\tuple {\smv}}{[]\ldots[]})\ |\ \poppname{D'} \ |\
  \queue{ {\smvo z}}{})
  =
  \Gimp[@][\iota''']
  \]
  where $\iota'''(\ptp s) = \poppname{Srv}$ and $\iota'''(\ptp c) =
  C''''$. In this case $\Gimp[@][\iota''']$ is an
  $\iota'''$-implementation of $\G(\tuple \smv)$ at $\shname$; the
  sub-term
  $(\nu {\tuple \smv} \At \shname)
  (\poppname{Srv}\ |\ \poppname{C''''}\ |\ \queue{\tuple {\smv}}{[]\ldots[]})$
  stands for the session corresponding to the global type $\G(\tuple
  \smv)$, while the context represents the rest of the system.  \finex
\end{exa}

We characterise WSI as a relation between the execution traces of a
global type $\G$ and its implementations $\Gimp$.
An execution trace of a system $\Gimp$ is a sequence of events of the
form $\state[\ptp p][\Tsend \smv \sort ]$ and
$\state[\ptp p][\Treceive \smv \sort]$, which respectively represent
an output and an input action performed by $\ptp p$ over the channel
$\smv$.

\begin{figure*}[t]
  \begin{center}
    \[
    \begin{array}{c}
      \mathrule{
      \forall \ptp p \in \participants \G \qst \iota(\ptp p) = \Pend
      \qquad
      \forall \smv \in \tuple \smv \qst \text{ queue on } \smv
      \text{ is empty in } \Gimp
      }{
      \epsilon \in \imrunsu{\state[\Gimp]}
      }
      {REnd}
      \\[2em]
      \mathrule{
      \begin{array}{c}
        \state[\iota(\ptp p)]
        \Ptrans{\lcond{e}{\lsend \smv {\val v}}}
        \state[P][\sigma]
        \hfill
        \smv \in \tuple \smv
        \hfill
        \vjdg[{}][\val v][\sort]
        \\[5pt]
        \state[\Gimp]
        \Ptrans{\lcond{e} \tau}
        \state[{\Gimp[@][\iota\updi{P}{\ptp p}]}][\sigma]
	\qquad\qquad\qquad
        r \in \imrunsu{\state[{\Gimp[@][\iota\updi{P}{\ptp p}][@]}][\sigma] }
      \end{array}
      }{
        \state[\ptp p][\lsend \smv \sort] r \in \imrunsu{\state[\Gimp]}
      }{RSnd}
      \\[2em]
      \mathrule{
        \begin{array}{c}
          \state[\iota(\ptp p)]
          \Ptrans{\lcond e{\lreceive \smv {\val v}}}
          \state[P][\sigma']
          \hfill
          \smv \in \tuple \smv
          \hfill
          \vjdg[][\val v][\sort]
          \\[5pt]
          \state[\Gimp]
          \Ptrans{\lcond {e} \tau}
          \state[{\Gimp[@][\iota\updi{P}{\ptp p}]}][\sigma']
          \qquad\qquad\qquad
          r \in \imrunsu{\state[{\Gimp[@][\iota\updi{P}{\ptp p}][@]}][\sigma']
          }
        \end{array}
      }{
      \state[\ptp p][\lreceive \smv \sort] r \in \imrunsu{\state[\Gimp]}
      }{RRcv}
      \\\\
      \mathrule{
        \begin{array}{c}
          \state[\iota(\ptp p)]\Ptrans{\lcond e \alpha}\state[P] [\sigma']
          \qquad
          \names[\alpha] \cap \tuple \smv = \emptyset
          \qquad
          \state[\Gimp]\Ptrans{\lcond {e'} \beta}
          \state[{\Gimp[@][\iota\updi{P}{\ptp p}][@]}][\sigma']
          \\[5pt]
          \shname \not\in \fn{\beta}
          \qquad
          r \in \imrunsu{\state[{\Gimp[@][\iota\updi{P}{\ptp p}][@]}][\sigma']}
        \end{array}
      }{
        r \in \imrunsu{\state[\Gimp]}
      }{RExt_1}
      \\\\
      \mathrule{
      \begin{array}{c}
        \state[\Gimp]\Ptrans{\lcond {e'} \beta}
         \state[({\Gimp})'][\sigma']
        \qquad\qquad
        r \in \imrunsu{\state[({\Gimp})'][\sigma'] }
      \end{array}
      }{
      r \in \imrunsu{\state[\Gimp]}
      }
      {RExt_2}
      \\\\
      \mathrule{
        \begin{array}{c}
          \state[\iota(\ptp p_0)]
          \Ptrans{\lcond {e_0}{\lreq \shname n {\tuple \smv}}}
          \state[P_0][\sigma_0]
          \qquad\qquad
          \forall 1 \leq i \leq n \qst
          \state[\iota(\ptp p_i)]
          \Ptrans{\lcond {e_i}{\lacc \shname i {\tuple \smv}}}
          \state[P_i][\sigma_i]
          \\[5pt]
          \participants{\G(\tuple \smv)} =
          \{\ptp p_0, \ptp p_1, \ldots, \ptp p_n\}
          \hfill
          \iota' = \iota\updi{P_0}{\ptp p_0} \ldots \updi{P_n}{\ptp p_n}
          \\[5pt]
          \state[\Gimp]
          \Ptrans{\lcond e \tau}
          \state[{\Gimp[@][\iota'][@]}][\sigma\updi{\shname}{\tuple \smv}]
          \hfill
          r \in \imrunsu{\state[{\Gimp[@][\iota'][@]}][
              \sigma\updi{\shname}{\tuple \smv}]
          }
        \end{array}
      }{
      r \in \imrunsu{\state[\Gimp]}
      }{ROpen}
    \end{array}
    \]
    \caption{Runs of implementations\label{fig:runs-realizations}}
  \end{center}
\end{figure*}

\begin{defi}[Runs of implementations]\label{def:gimp}
  Let $\Gimp$ be an $\iota$-implementation of $\G(\tuple \smv)$.
  The set of \emph{runs of $\Gimp$ initiated on $\shname$ with store
    $\sigma$}, written $\imrunsu{\state[\Gimp]}$, is the set
  inductively defined by the rules in \cref{fig:runs-realizations}.
  We write $\imrunsu{\Gimp}$ for $\imrunsu{\state[\Gimp][\emptyset]}$
  and extend the notion to sets of implementations $\mathbb{I}$ as
  $\imrunsu{\mathbb{I}} = \cup_{I\in\mathbb{I}}\imrunsu{I}$.
\end{defi}

Rules in \cref{fig:runs-realizations} rely on the semantics of
\cref{fig:LTSproc} and \cref{fig:LTSsys}. Rule~\myrule{REnd}
establishes that a completed session, i.e., one in which all processes
are terminated and the session queues are empty, contains the empty
run $\epsilon$.
Non-empty runs of $\Gimp$ are defined in terms of the input and output
actions that processes $\iota(\ptp p)$ perform over the session
channels $\tuple \smv$, as described by the rules~\myrule{RSnd}
and~\myrule{RRcv}.
In rule~\myrule{RSnd}, $\iota(\ptp p)$ performs an output over a
session channel associated with $\shname$; which is formally captured
by the conditions
$\state[\iota(\ptp p)] \Ptrans{\lcond{e}{\lsend \smv {\val v}}}
\state[P][\sigma]$ and $\smv \in \tuple \smv$ in the premiss.
When $\iota(\ptp p)$ evolves to $P$ by performing
$\lsend \smv {\val v}$, $\Gimp$ evolves to $\Gimp[@][\iota']$ with
$\iota'= \iota\updi{P}{\ptp p}$, i.e., $\iota'$ coincides with $\iota$
in all roles but ${\ptp p}$. This is stated by the condition
$\state[\Gimp] \Ptrans{\lcond{e'} \tau}
\state[{\Gimp[@][\iota\updi{P}{\ptp p}]}][\sigma]$ in the premiss of
the rule.
Hence, $\Gimp$ contains a run $\state[\ptp p][\lsend \smv \sort] r$
(see the conclusion of the rule) when $r$ is a run of the state
$\state[{\Gimp[@][\iota\updi{P}{\ptp p}]}][\sigma']$ reached after
$\iota(\ptp p)$ performs $\lsend \smv {\val v}$.
We remark that runs abstract away from the particular values sent by
the processes and keep instead the sorts of sent value (i.e.,
condition $\val v : \sort$).
Input events are handled analogously in rule~\myrule{RRcv}; in this
case also $\sigma$ evolves to $\sigma'$ when $\iota(\ptp p)$ performs
an input.

Rule~\myrule{RExt_1} accounts for the computation steps of $\iota(\ptp
p)$ that do not involve session channels in $\tuple \smv$ (condition
$\names[\alpha] \cap \tuple \smv = \emptyset$), which can be an
internal transition $\tau$ in a role, a communication over a channel
not in $\tuple \smv$, or a session initiation.
This rule allows each process to freely initiate a session that does
not correspond to the global type $\GT$, i.e., over a shared name
different from $\shname$ (condition $\shname \not\in \fn{\beta}$).
Rule~\myrule{RExt_2} handles the cases in which the transition of
$\Gimp$ does not involve any process $\iota(\ptp p)$. This is captured
by the fact that the continuation $\Gimp[@][\iota']$ uses the same
mapping $\iota$. By the definition of $\Gimp$, the reduction does not
interfere with the names of the session, i.e., $\fn \beta \cap (\tuple
\smv \cup \{\shname\})$ holds.

Rule~\myrule{ROpen} allows for the initiation of a new session on
$\shname$ and requires all roles to participate in the synchronisation
(as stated by the three first premisses).
We assume that any role in the implementation will execute exactly one
action over the channel $\shname$ which also matches the role assigned
by $\iota$.
Nested sessions are handled by assuming that all sessions are created
over different channels that have the same type.
This is just a technical simplification analogous to the possibility
of having annotations to indicate the particular instance of the
session under analysis.

The runs of an implementation abstractly capture the traces of
communications of the processes in the system.
This can be easily formalised by using a more concrete relation on
systems.
More precisely, we define $\Wtrans{\_}$ as the relation induced by the
rules such as those in \cref{fig:LTSproc} once the $\tau$ in the
conclusion of rules \myrule{SCom_1} and \myrule{SCom_2} is replaced
with the output and input action respectively.
Then we can state the following proposition:
\begin{thm}\label{prop:run-vs-red}
  Given an implementation $\Gimp$ and a store $\sigma$, if $r \in
  \imrunsu{\state[\Gimp]}$ is a run of length $m$ then there is a
  sequence $\Gimp \Wtrans{\alpha_1} \state[S_1][\sigma_1] \cdots
  \Wtrans{\alpha_n} \state[S_n][\sigma_n]$ such that we can find an
  order preserving bijection $\chi$ between the sets $\{1,\ldots,m\}$
  and $\{1 \leq j \leq n \sst \alpha_i \text{ not a } \tau\}$ such
  that the $i$-th element in $r$ is an input of sort $\sort$ iff so is
  the $\alpha_{\chi(i)}$ and the value in $\alpha_{\chi(i)}$ has the
  sort $\sort$.
\end{thm}
\begin{proof}
  Straightforward induction on the derivation of $r \in
  \imrunsu{\state[\Gimp]}$.
\end{proof}

We now introduce the notion of runs associated to a global type.  Our
notion of WSI will allow us to implement an iterative type, which
accounts for an unbounded number of repetitions, with a process
exhibiting a bounded number of iterations. For this reason, we deviate
from the previous definition of traces of global
types~\cite{ChenH12,DBLP:journals/corr/abs-1203-0780,gt17,gt16} and
use \emph{annotated traces} to distinguish mandatory from optional
events.
Annotating optional events is instrumental to the comparison of traces
of iterative types (which is defined below).
Syntactically, an optional sequence $r$ of events is written $[r]$.
As usual, we consider an asynchronous communication model and a trace
implicitly denotes the equivalence class of all traces obtained by
permuting causally independent events, that is events executed by
different participants on different channels.

\begin{defi}[Runs of a global type]\label{def:runGTT}
  Given a global type term $\GT$, the set $\imruns{\GT}$ denotes the
  runs allowed by $\GT$ and is defined as the least set closed under
  the rules in \cref{fig:def-runs-global-type}.
\end{defi}
\begin{figure*}[t]
  \[
  \begin{array}{l}
    \mathaxiom{
      \epsilon \in \imruns{\Gend}
    }{RGEnd}
    \hfill
    \mathrule{
      j\in I \qquad r \in \imruns{\GT_j}
    }{
      \state[\ptp p][\lsend \smv_j \sort_j]
      \state[\ptp q_j][\lreceive \smv_j \sort_j] r  \in \imruns{\Gchoice}
    }{RGCom}
    \\[2em]
    \mathrule{
      r_1\in\imruns{\GT_1} \quad r_2\in\imruns{\GT_2}
    }{
      r_1 r_2 \in\imruns{\GT_1; \GT_2}
    }{RGSeq}
    \quad\ \
    \mathrule{
      r\in\imruns{\GT}
    }{
      r\in\imrunsaux{\Gdef{\f}{\GT}}
    }{RG^*_1}
    \hfill
    \mathrule{
      r_1\in\imruns{\GT}  \quad r_2\in\imrunsaux{\Gdef{\f}{\GT}}
    }{
      r_1 [r_2]\in \imrunsaux{\Gdef{\f}{\GT}}
    }{RG^*_2}
    \\[2em]
    \mathrule{
      r \in \imrunsaux{\Gdef \f \GT}
      \quad
      \ready \GT = \{\ptp p\}
      \quad
      \participants{\GT} = \{\ptp p,\ptp p_1,\ldots, \ptp p_n\}
      \quad
      \forall 1\leq i \leq n \qst f(\ptp p_i) = \smv_i \sort_i
    }{
      r\state[\ptp p][\lsend {\smv_1} {\sort_1}]
      \state[\ptp p_1][\lreceive {\smv_1} {\sort_1}]
      \ldots\state[\ptp p][\lsend {\smv_n} {\sort_n}]
      \state[\ptp p_n][\lreceive {\smv_n} {\sort_n}]
      \quad \in \quad \imruns{\Gdef{\f}{\GT}}
    }{RGIter}
  \end{array}
  \]
  \caption{Runs of a global type}%
  \label{fig:def-runs-global-type}
\end{figure*}
The first three rules of \cref{fig:def-runs-global-type} are
straightforward.
The runs of an iterative type $\Gdef{\f}{\GT}$ are given by the
rule~\myrule{RGIter}, whose premiss uses the set
$\imrunsaux{\Gdef{\f}{\GT}}$ to unfold $\Gdef{\f}{\GT}$ (as defined by
the rules~\myrule{RG^*_1} and \myrule{RG^*_2}).
Optional events are introduced when unfolding an iterative type (rule
\myrule{RG^*_2}).  The main motivation is that an iterative type
$\Gdef{\f}{\GT}$ denotes an unbounded number of repetitions of (traces
of) ${\GT}$ (i.e., an infinite number of traces).  Note that
${\imrunsaux{\Gdef{\f}{\GT}}} =\{ r_1, r_1[r_2],
r_1[r_2[r_3]],\ldots\}$ with $r_i\in\imruns{\GT}$.
Rule \myrule{RGIter} adds the events associated to the termination of
an iteration: $(i)$ the ready role $\ptp p$ sends the termination
signal to any other role by using the dedicated channels specified by
$f$ (i.e., $\state[\ptp p][\lsend {\smv_1} {\sort_1}],
\ldots,\state[\ptp p][\lsend {\smv_n} {\sort_n}] $), and $(ii)$ all
roles but the ready one receive the corresponding termination message
(i.e., $\state[\ptp p_1][\lreceive {\smv_1} {\sort_1}],
\ldots,\state[\ptp p_n][\lreceive {\smv_n} {\sort_n}]$).
We just consider one of the possible interleavings of termination
events because we consider traces up-to the permutation of causally
independent events.

\begin{defi}[Trace preorder]
  The \emph{trace preorder} $\myequal$ is the least preorder on
  annotated traces satisfying the following axioms and rules.
  \[
  \mathrule{-}{[r] \myequal \epsilon}{\myequal drop}
  \qquad\qquad
  \mathrule{-}{\epsilon \myequal r}{\myequal emp}
  \qquad\qquad
  \mathrule{
    r_1 \myequal r_1' \quad r_2 \myequal r_2'
  }{
    r_1r_2 \myequal r_1'r_2'
  }{\myequal cmp}
  \]
\end{defi}

We say $r'$ \emph{covers} $r$ when $r\myequal r'$, i.e., when $r'$
matches all mandatory actions of $r$.
Analogously, we say a set $R_2$ of annotated traces \emph{covers}
another set $R_1$, written $R_1\mysubseteq R_2$, if for all $r\in R_1$
there is $r' \in R_2$ such that $r\myequal r'$.

\begin{defi}[Whole-spectrum implementation]%
  \label{def:wsr}
  A set $\mathbb{I}$ of implementations \emph{covers} a global type
  $\Geq$ if $ \imruns{\GT} \mysubseteq \imrunsu{\mathbb{I}}$.
   A process $P$ is a \emph{whole-spectrum implementation of $\ptp p
    \in \participants{\GT}$} when there exists a set $\mathbb{I}$ of
  $\iota$-implementations of $\G(\tuple \smv)$ at $\shname$ that
  covers $\GT$ such that $\Gimp \in\mathbb{I}$ implies $\iota(\ptp p)
  = P$.
\end{defi}
A whole-spectrum implementation (WSI) of a role $\ptp p$ is a process
$P$ such that any expected behaviour of the global type can be
obtained by putting $P$ into a proper context.
For iteration types, the comparison of annotated traces implies that
the implementation has to be able to perform the iteration body at
least once, but it can arbitrarily choose the number of iterations.
\begin{rem}[Trace semantics for WSI]
WSI is based on a trace semantics. WSI implementations are
deterministic programs, hence their behaviour if faithfully captured
by the set of their execution traces.  A stronger notion like
bisimulation is not necessary. In fact, the branching specified in a
global type becomes deterministic in a process, therefore the set of
all possible executions of a process (in a given implementation)
results in a single trace.
\end{rem}

%%% Local Variables:
%%% mode: latex
%%% TeX-master: "main"
%%% End:

\section{Typing}%
\label{sec:typing}
% !TEX root = main.tex

In this section we introduce a typing discipline to guarantee that a
well-typed process is a WSI of the role it plays in a global type.
Technically we rely on an enriched version of local types, dubbed
\emph{pseudo-types}, that takes into account branching enabling
conditions.

\subsection{Pseudo-types \& Typing judgements}%
\label{sec:pseudo}
The scaffolding of our typing discipline is standard but for the need
of making the typing depending on the expressions the processes use to
render choices.
This requires to revisit the usual definition of mergeability
(cf. \cref{def:mergeability} below) that now relies on a notion of
\emph{normalisation} of local types.

The syntax of pseudo-types is given by the following grammar:
\begin{gather*}
  \pst ::=
  \Ipst[e_i]{i \in I}{\smv_i}{\sort_i}{\pst_i}
  \quad \sep \quad
  \Epst[e_i]{i \in I}{\smv_i}{\sort_i}{\pst_i}
  \quad \sep \quad
  \pst_1 ; \pst_2
  \quad \sep \quad
  \Tdef\pst \quad \sep \guard \Tend
\end{gather*}
We call \emph{guards} the expressions $e_i$ occurring in a
pseudo-type.
Guards keep track of the conditions that have to be satisfied in order
to enable a certain behaviour.
For instance, the pseudo-type
\[\pst_{\oplus} = \Ipst[\C_i]{i \in \{1,2\}}{\smv_i}{\dt{int}}{\Tend}\]
where the guards are $ \C_1 = x > 0 $ and $\C_2 = x \leq 0$.
By $\C_1$, $x$ needs to be strictly positive in order to choose the
first branch.
Local types from \cref{sec:types} can be thought of as pseudo-types
where all guards are $\truek$. Hereafter we may omit guards $\truek$
and have e.g.,
\[
  \guard [\truek] \Tend \text{ abbreviated as } \Tend
  \qquad \text{and} \qquad
  \Ipst[\truek]{i \in I}{\smv_i}{\sort_i}{\pst_i} \text{ abbreviated as }
  \Tssel{\smv_i}{\sort_i}{\pst_i}
\]

The notions of free and bound names straightforwardly extend to
pseudo-types.
We will write $\var \pst$ for the set of variables occurring in the
expressions of $\pst$ and $\fY[\pst]$ for the set of session channels
in $\pst$; for instance, $\var{\pst_{\oplus}} = \{x\}$.

\begin{figure}\small
  \begin{math}
    \begin{array}{lrcl@{\hspace{-1.3cm}}r}
      (1)
      &
      \nform{\guard[\C'] \Tend} & = & \guard[\C \land \C']\Tend
      \\[1em]
      \multicolumn 5 l {
        \text{let } J = \{i \in I \sst (\C \land \C_i) \iff \falsek\}
        \text{ in the clauses $(2)$ and $(3)$ below}
      }
      \\[1em]
      (2)
      & \nform{\Ipst[\C_i]{i \in I}{\smv_i}{\sort_i}{\pst_i}}
      & =
      &
      \begin{cases}
        \guard[\falsek]\Tend & \textit{if}\ \ I = J
        \\
        \Ipst[\C_i\land\C]{i \in I\setminus J}{\smv_i}{\sort_i}{\nform[\C_i\land\C]{\pst_i}}
        &
        \textit{if}\ \ I \neq J
      \end{cases}
      \\
      (3)
      &
      \nform{\Epst[\C_i]{i \in I}{\smv_i}{\sort_i}{\pst_i}}
      &
      =
      &
      \begin{cases}
        \guard[\falsek]\Tend & \textit{if}\ \ I = J
        \\
        \Epst[\C_i\land\C]
             {i \in I\setminus J}
             {\smv_i}{\sort_i}
             {\nform[\C_i\land\C]{\pst_i}}
             & \textit{if}\ \ I \neq J
      \end{cases}
      \\[2em]
      (4)
      &
      \nform{\guard[\C'] \Tend;\pst}
      & = &
      \nform[\C\land\C']{\pst}
      \\
      (5)
      &
      \nform{\big(\Ipst[\C_i]{i \in I}{\smv_i}{\sort_i}{\pst_i}\big);\pst}
      & = &
      \nform{\Ipst[\C_i]{i \in I}{\smv_i}{\sort_i}{(\pst_i;\pst)}}
      \\
      (6)
      & \nform{\big(\Epst[\C_i]{i \in I}{\smv_i}{\sort_i}{\pst_i}\big);\pst}
      & = &
      \nform{\Epst[\C_i]{i \in I}{\smv_i}{\sort_i}{(\pst_i;\pst)}}
      \\
      (7)
      &
      \nform{(\pst;\pst');\pst''}
      & = &
      \nform{\pst;{(\pst';\pst'')}}
      \\
      (8)
      &
      \nform{\Tdef{\pst};\pst'}
      & = &
      \begin{cases}
      	\guard[\C']\Tend
        &
        \textit{if}\ \nform{\pst} =  \guard[\C']\Tend
        \\
       	\nform{\Tdef{\pst}};\nform{\pst'}
        & \textit{if}\ \nform{\pst} \neq \guard[\C']\Tend
      \end{cases}
      \\
      (9)
      &
      \nform{\Tdef{\pst}}
      & = &
      \begin{cases}
        \guard[\C']\Tend
        & \textit{if}\ \nform{\pst} =  \guard[\C']\Tend
        \\
        \Tdef{\nform{\pst}}
        & \textit{if}\ \nform{\pst} \neq \guard[\C']\Tend
      \end{cases}
    \end{array}
  \end{math}
  \caption{Normalisation procedure for pseudo-types\label{fig:normalization-pseudo}}
\end{figure}

Given a pseudo-type $\pst$, the \emph{normal form of $\pst$}, written
$\nfp{\pst}$, is defined as $\nfp{\pst} =\nform[\truek]\pst$ where
$\nform[\_]\_$ given by the equations in
\cref{fig:normalization-pseudo}.
Intuitively, the normalisation of a pseudo-type propagates the guards
of branches to their continuations while removing those alternatives
with inconsistent guards.
We just remark that $\nfp{\pst}$ is defined for any $\pst$ (details
are in \cref{app:normal-form}).

The notion of normalisation is instrumental to adapt the standard
merge operation $\_ \peace \_$ of session
types~\cite{DBLP:journals/corr/abs-1203-0780} to pseudo-types.
Our definition of $\_ \peace \_$ requires the \emph{mergeability} of
pseudo-types, which amounts to have branches with the same
communication prefix guarded by mutually exclusive conditions.
\begin{defi}[Mergeable pseudo-types]\label{def:mergeability}
  Two pseudo-types $\pst_1$ and $\pst_2$ in normal form are
  \emph{mergeable}, if
  \begin{itemize}
  \item $\pst_1 = \guard\Tend$ and $\pst_2 =\guard[\C'] \Tend$
  \item $\pst_1 = \Epst[\C_i]{i \in I}{\smv_i}{\sort_i}{\pst_i}$,
    $\pst_2 = \Epst[\C_i']{i \in I}{\smv_i}{\sort_i}{\pst_i'}$ and for
    all $i \in I$,
    $\pst_i$ and  $\pst_i'$ are mergeable,
    and $e_i \land e_i' \iff
    \falsek$
  \item
    $\pst_1 = \Ipst[\C_i]{i \in I \cup J}{\smv_i}{\sort_i}{\pst_i}$
    and
    $\pst_2 = \Ipst[\C_i']{i \in I \cup K}{\smv_i}{\sort_i}{\pst_i'}
    $ with $I \cap J = I \cap K = \emptyset$ and sets
    ${\{\smv_i\}}_{i \in I}$, ${\{\smv_j\}}_{j \in J}$, and
    ${\{\smv_k\}}_{k \in K}$ pairwise disjoint, and for all $i \in I$,
    $\pst_i$ and $\pst_i'$ are mergeable,
    and $e_i \land e_i' \iff \falsek$
  \item $\pst_1 = \pst'_1;\pst''_1$, $\pst_2 = \pst'_2;\pst''_2$, and
    $\pst'_1$ and $\pst'_2$ as well as $\pst''_1$ and $\pst''_2$ mergeable
  \item $\pst_1 = \Tdef{(\pst'_1)}$, \ \ $\pst_2 =\Tdef{(\pst'_2)}$
    with  $\pst'_1$ and $\pst'_2$ mergeable.
  \end{itemize}
\end{defi}
\noindent
Basically, $\pst_1$ and $\pst_2$ are mergeable when they have the same
structure and at choice points branches either use different channels
or they use mutually exclusive guards.
When types are mergeable, operation $\_ \peace \_$ \lq\lq
glues\rq\rq\ branches that use the same channel.
\begin{defi}[Merge]\label{def:merge}
  The \emph{merge $\pst_1 \peace \pst_2$} of two mergeable
  pseudo-types $\pst_1$ and $\pst_2$ is defined as:
  \[
  \pst_1 \peace \pst_2
  =
  \begin{cases}
    \guard[\C\lor \C']  \Tend
    &
    \text{if } \pst_1 = \guard\Tend \text{ and } \pst_2 =\guard[\C'] \Tend
    \\[5pt]
    \Epst[\C_i \lor \C_i']{i \in I}{\smv_i}{\sort_i}{\pst_i' \peace \pst_i'}
    &
    \text{if } \pst_1 = \Epst[\C_i]{i \in I}{\smv_i}{\sort_i}{\pst_i}
    \text{ and } \pst_2 = \Epst[\C_i']{i \in I}{\smv_i}{\sort_i}{\pst_i'}
    \\[5pt]
    \Ipst[\C_i]{i \in I \cup J}{\smv_i}{\sort_i}{\pst_i}
    &
    \text{if } \pst_1 =   \Ipst[\C_i]{i \in I}{\smv_i}{\sort_i}{\pst_i}
    \text{ and } \pst_2 = \Ipst[\C_j]{j \in J}{\smv_j}{\sort_j}{\pst_j}
    \\[5pt]
    (\pst'_1\peace\pst'_2);(\pst''_1\peace\pst''_2)
    & \text{if } \pst_1 = \pst'_1;\pst''_1
    \text{ and } \pst_2 = \pst'_2;\pst''_2
    \\[5pt]
    \Tdef{(\pst'_1\peace\pst'_2)}
    &
    \text{if } \pst_1 = \Tdef{(\pst'_1)}
    \text{ and }\pst_2 =\Tdef{(\pst'_2)}
\end{cases}\]
\end{defi}

\bigskip

As we will see, our typing discipline keeps track of the assumptions
(i.e., the guards) necessary to reach a particular point in the
processes.
In fact, systems are typed by judgements of the form
\begin{gather}\label{eq:jdg}
  \pjdg[@][@][S][@][@]
\end{gather}
stipulating that, under the assumption $\C$ and the type assignment of
variables $\Gamma$, the system $S$ is typed as $\envmv$.
In~\eqref{eq:jdg}, $\Gamma$ and $\envmv$ are (possibly empty) partial
functions.
We adopt the usual syntactic notation for environments:
\begin{eqnarray*}
  \Gamma & ::= & \emptyset \quad \sep \quad \Gamma, x : \sort
  \\
  \envmv & ::= & \emptyset \quad \sep \quad
                 \envmv,\shname:\G \quad \sep \quad
                 \envmv, \psnames \smv p : \pst \quad \sep \quad
                 \envmv, \smv : \queue{}{\sort}
\end{eqnarray*}
Environments $\Gamma$ assign sorts $\sort$ to variables $x$.
Environments $\envmv$, called \emph{specifications}, map ($i$) shared
names $\shname$ to global types $\G$, ($ii$) participants' sessions
$\psnames \smv p$ to pseudo-types $\pst$, and ($iii$) session names
$\smv$ to queues of sorts $\sort$.
As usual, we implicitly assume that in a judgement of the form~\eqref{eq:jdg} the following holds:
\begin{itemize}
\item $x \not\in \dom \Gamma$ when writing $\Gamma, x : \sort$,
  and
\item $\shname \not\in \dom \envmv$ when writing
  $\envmv, \shname : \G$ (likewise for participants' sessions and for
  sessions' queues).
\end{itemize}
For judgements of the form~\eqref{eq:jdg} we also assume that
\begin{itemize}
\item $\fX[e] \subseteq \dom \Gamma$ and
  $\fn S \subseteq \dom \envmv$, and
\item in $\envmv, \psnames \smv p : \pst$ it is
  $\fY[\pst] \subseteq \tuple \smv$ and
  \begin{gather}\label{eq:sanity}
    \forall \psnames{\smv_1}{p_1}, \psnames{\smv_2}{p_2} \in
    \dom{\envmv}, \tuple \smv_1 \cap \tuple \smv_2 \neq \emptyset
    \implies \tuple \smv_1 = \tuple \smv_2
  \end{gather}
\end{itemize}
Condition~\eqref{eq:sanity} states that a session channel can be used
only in one session.
We sometimes write $\tuple \smv \in \dom\envmv$ when there exists
$\ptp p$ such that $\psnames \smv p \in \dom\envmv$.
Similarly, $\smv \in \dom\envmv$ stands for $\exists \tuple \smv \in
\dom \envmv \qst \smv \in \tuple \smv$.
The extension of $\var \_$ and $\fY $ to environments is
straightforward.

\begin{figure}[t]
\begin{tabular}{lll}
\multicolumn{3}{l} {\bf Predicates}
\\
\toprule
Notation & Arity & Description
\\
\midrule
$\envmv$ $\Tend$-\emph{only} & 1 & every session in $\envmv$ is terminated
\\
$\envmv$ \emph{active} & 1 &
every session in  $\envmv$ is an internal guarded-choice
\\
$\envmv_1$ and  $\envmv_2$ \emph{passively compatible} & 2 &
$\envmv_1$ and  $\envmv_2$  agree on the termination of an iteration
\\
$\envmv_1$ and $\envmv_2$  \emph{independent} & 2 &
disjoint sessions and agreement on shared names
\\
\bottomrule
\multicolumn{3}{l} {}
\\
\multicolumn{3}{l} {\bf Operations}
\\
\toprule
Notation & Arity & Description
\\
\midrule
$\envmv_1;\envmv_2$ & 2 & sequential composition
\\
$\envmv_1\cup\envmv_2$ & 2 & union
\\
$\Tdef\envmv$ & 1 & closure
\\
$\envmv|_{-\tuple \smv}$ & 2 &
restriction of $\envmv$ to names not in $\tuple \smv$
\\
\bottomrule
\end{tabular}
\caption{Summary of operations and predicates for environments}%
\label{fig:oper-pred-env}
\end{figure}

Our typing relies on
an operation that recovers types from pseudo-types by removing guards. Formally,
guard removal $\rmg \cdot$ is defined as
\[\begin{array}c
    \rmg{\guard \Tend} = \Tend
    \qquad\qquad
    \rmg{\pst;\pst'} = \rmg{\pst} ; \rmg{\pst'}
    \qquad\qquad
    \rmg{\Tdef \pst} = \Tdef{\rmg \pst}
    \\[10pt]
    \rmg {\Epst[\C_i]{i \in I}{\smv_i}{\sort_i}{\pst_i}}
    = \TBbra{\alpha}{I /_\sim}{\smv_\alpha}{\sort_\alpha}{\rmg{\peace_{i \in \alpha} \pst_i}}
    \\[10pt]
    \rmg {\Ipst[\C_i]{i \in I}{\smv_i}{\sort_i}{\pst_i}}
    = \TSsel{\alpha \in I /_\sim}{\smv_\alpha}{\sort_\alpha}{\rmg{\peace_{i \in \alpha} \pst_i}}
\end{array}\]
where in the last two equations $\sim$ is the equivalence relation on
$I$ defined as $i \sim j \iff \smv_i = \smv_j$ and
$\smv_\alpha \sort_\alpha = \smv_i \sort_i$ for $i \in \alpha$.
Other auxiliary operations and predicates on environments are listed
in~\cref{fig:oper-pred-env} and formally defined in the following
section.

\subsection{Typing rules}
The typing rules for processes and systems are grouped in
\cref{fig:typingproc} and \cref{fig:typingsys}.
For the sake of readability, we restate the typing rules as we comment
them so to introduce notation and concepts appearing in the rules as
we present them.

To type a request for a new session $\Preq \shname {n}{\tuple
  \smv}{P}$ we use the following rule
\begin{align*}
  \vreqrule
 \end{align*}
The premiss checks that the continuation $P$ can be typed with
$\envmv$ extended with an assignment of the pseudo-type $\pst$ to the
participant's session $\psname {\tuple \smv} p$, for some $\pst$
\emph{matching} the projection of the global type $\envmv(\shname)$ on
the corresponding role.
Intuitively, the type obtained by removing guards from $\pst$
coincides with the projection of the global type.

The rule $\myrule{VAcc}$ for typing the acceptance for the $\ptp p$-th
role $\Pacc{\shname}{p}{\tuple \smv}{P}$ is defined analogously.

An external choice is checked by
\[
  \rcvrule
\]
that types each branch $P_i$ against the respective continuation of
the type $\psnames \smv p : \pst_i$ (once $\Gamma$ is extended with
the type assignment on the bound name $x_i$); the first condition in
the premiss permits to branch only over a subset of the session
channels.

An output $\Psend \smv {e'}$ and the idle process are respectively
typed as follows
\[
  \vsendrule
  \qquad\qquad
  \vendrule
\]
The expression $e'$ in \myrule{VSend} has to be of the sort expected
on channel $\smv$; moreover, no further actions should occur on
session channels (rendered with the condition \emph{$\envmv$
  $\Tend$-only} abbreviating $\forall \psnames {\smv'} q
\in\dom{\envmv} \qst \nfp{\envmv \psnames {\smv'} q} = \guard\Tend$);
for the idle process we simply require $\envmv$ to map each session
channel to the $\Tend$ type.

The typing of sequential compositions is handled by the rule
\[
\vseqrule{VSeq}
\]
that requires to decompose the specification into $\envmv_1$ and
$\envmv_2$ to respectively type each part of the sequential
composition.
In the conclusion of the rule, the partial operation
$\_; \_$ on specifications requires that
$\dom {\envmv_2} \subseteq \dom {\envmv_1}$ and
$\envmv_1 |_{\shset \cup \chset} = \envmv_2 |_{\shset \cup \chset}$
and it is defined as follows\label{pag:envseq}
\[
\begin{array}{rcl}
  (\envmv_1;\envmv_2)|_{\shset \cup \chset}  & = & \envmv_1|_{\shset \cup \chset}
  \\
  (\envmv_1;\envmv_2) \psnames \smv p & = &
  \begin{cases}
    \envmv_1 \psnames \smv p ; \envmv_2 \psnames \smv p
    & \psnames \smv p \in \dom {\envmv_2}
    \\
    \envmv_1 \psnames \smv p
    & \psnames \smv p \in \dom{\envmv_1} \setminus \dom {\envmv_2}
    \\
    \text{undef} & \text{otherwise}
  \end{cases}
\end{array}
\]

The notion of mergeable pseudo-types is extended to specification
pairs in order to type conditionals.
Specifications $\envmv_1$ and $\envmv_2$ are \emph{mergeable} when the
local types they assign to sessions are mergeable.
Formally, $\envmv_1$ and $\envmv_2$ are mergeable iff
$\envmv_1|_{\shset \cup \chset} = \envmv_2|_{\shset \cup \chset}$,
$\dom{\envmv_1} = \dom{\envmv_2}$ and, for all $\psnames \smv p \in
\dom{\envmv_1}$, $\envmv_1 \psnames \smv p$ and $\envmv_2 \psnames
\smv p$ are mergeable.
When $\envmv_1$ and $\envmv_2$ are mergeable, $\envmv_1 \peace
\envmv_2$ merges the local types of sessions:
\[
(\envmv_1 \peace \envmv_2)|_{\shset \cup \chset} = \envmv_1|_{\shset \cup \chset}
\qquad \text{ and } \qquad
(\envmv_1 \peace \envmv_2) \psnames \smv p  =
\envmv_1 \psnames \smv p \peace \envmv_2 \psnames \smv p
\]
We remark that the merge operation on specifications is idempotent,
associative, and commutative.
For conditionals we have
\[
\vifrule{Vif}
\]
that requires to decompose the specification into two specifications
$\envmv_1$ and $\envmv_2$ so to type conditional processes with the
merge (\cf\ \cref{def:merge}) of $\envmv_1$ and $\envmv_2$.
The premiss checks that the then-branch is typed in $\envmv_1$ after
extending the assumption $e$ with the guard $e'$ of the conditional
while the else-branch is typed in $\envmv_2$ after extending $e$ with
the negation of $e'$.
Recall that judgements require consistency of their assumptions, hence
rule \myrule{Vif} is not applicable if $e \land e'$ or $e \land \neg
e'$ are inconsistent.

\begin{exa}
  Our typing distinguishes between $B_1$ and $B_2$ in \cref{sec:intro}
  because $B_1$ is validated while $B_2$ is not.
  This is due to the rule \myrule{VIf}.
  In fact, after a few verification steps on $B_1$ we can apply rule
  \myrule{VIf} and prove the following judgement:
  \[
  \pjdg[\truek][\envmv][
    \Pif{({check}\ c)}{\Psend{\popcname {ok}}{}}{\Psend{\popcname {ko}}{}}
  ]
  \]
  where $\Gamma$ assigns some sort to $\val c$ and
  $\envmv =
  ((\popcname {login}, \popcname {deposit}, \popcname {overdraft},
  \popcname {ok}, \popcname {ko}), \ptp b):
  \Tsend{\popcname {ok}}{} \oplus \Tsend{\popcname{ko}}{}$.
  Instead, for $B_2$ we would have to prove
  \[
    \pjdg[\truek][\envmv][\Psend{\popcname {ko}}{}]
  \]
  which makes the validation of $B_2$ fail; this is due to the fact
  that the only rule for typing a sending process is \myrule{VSend},
  which cannot be applied against the specification $\envmv$ that
  assigns $\Tsend{\popcname {ok}}{} \oplus \Tsend{\popcname{ko}}{}$ to
  the participant's session.
  \finex
\end{exa}

For-loops are typed with the following two rules
\[\begin{array}c
  \vforend
  \\[2em]
  \vfor
  \end{array}
\]
Rule \myrule{VForEnd} handles the case in which the expression $\ell$
denotes an empty list, that is when the for-loop should be skipped.
For this reason, the typing is similar to the typing of the idle
process (rule \myrule{VEnd}).
When the expression $\ell$ denotes a non-empty list under the
assumption $\C$ (i.e., $\vjdg[\Gamma][\ell][[\sort]]$ and
$\consistent[\C][\ell \not = \emptylist]$), we apply rule
\myrule{VFor} to validate for-loops.
The conclusion of the rule types the for-loop with $\Tdef \envmv$
which introduces iterative pseudo-types and is defined as follows:
\[
  \Tdef{\envmv}|_{\shset \cup \chset} = \envmv|_{\shset \cup \chset}
  \qquad\text{and}\qquad
  \Tdef{\envmv}\psnames \smv p  = \Tdef{(\envmv \psnames \smv p)}
\]
Note that $\envmv$ has to type the body $P$ under the context $\Gamma$
extended with $x : \sort$ because $x$ can occur free in $P$.
Moreover, $P$ has to inform all other peers that the iteration
continues.
This is checked by the condition $\envmv$ \emph{active}, namely that
for all $\psnames \smv p \in\dom\envmv$, $\nfp{\envmv \psnames \smv p}
= \Ipst[\C_i]{i \in I}{\smv_i}{\sort_i}{\pst_i}$.
Condition $x \not\in \var \envmv$ ensures that guards of $P$ do not
depend on the iteration variable $x$, making each choice available at
each iteration.

Rule \myrule{VLoop} below types passive processes of iterations.
\[
\vlooprule{VLoop}
\]
The premiss of the rule types the iteration body $N$ and the loop exit
$M$ with specifications $\envmv_1$ and $\envmv_2$ respectively.
Both specifications are required to be {\em passively compatible},
i.e., there is just one session that decides whether to continue or
terminate the iteration, and use different channels to communicate
such choice. Formally, specifications $\envmv_1 = \envmv, \psnames
\smv p :\pst$ and $\envmv_2 = \envmv, \psnames \smv p :\pst'$ are \emph{passively compatible} iff
\[
\dom\envmv \subseteq {\shset \cup \chset},
\qquad
\pst = \Epst[\C_i]{i \in I}{\smv_i}{\sort_i}{\pst_i}
\qquad \text{and} \qquad
\pst' = \Epst[\C_j]{j \in J}{\smv_j}{\sort_j}{\pst_j}
\]
with $y_i \neq y_j$ for all $i \in I, j \in J$.

We now consider the typing rules for systems in \cref{fig:typingsys},
which essentially deal with parallel composition, restriction of
shared names, and queues.
For parallel composition
\[
\vparrule
\]
requires to split the specification into two \emph{independent}
specifications $\envmv_1$ and $\envmv_2$ that respectively type each
side of the parallel.
Specifications are independent when they agree on shared names and are
disjoint on queues and participants' sessions; more precisely,
$\envmv_1$ and $\envmv_2$ are \emph{independent} when
\begin{itemize}
\item $\envmv_1|_{\shset} = \envmv_2|_{\shset}$ and
  $\dom{\envmv_1|_{\chset}} \cap \dom{\envmv_2|_{\chset}} = \emptyset$
\item for all $\psnames \smv p \in \dom{\envmv_1}$ and
  $\psnames{\smv'}{p'} \in \dom{\envmv_2}$, if $\tuple \smv \cap
  \tuple \smv' \neq \emptyset$ then $\tuple \smv = \tuple \smv'$ and
  $\ptp p \neq \ptp{p'}$
\end{itemize}
The union of independent specifications enjoys the sanity
condition~\eqref{eq:sanity}.

A restricted session is typed by
\[
\vnewrule{VNew}
\]
that removes participants' sessions and sessions' queues referring to
the restricted names $\tuple \smv$ from the specification $\envmv$
typing the scope $S$ (this restricted specification is denoted as
$\envmv|_{-\tuple \smv}$).

The typing of queues is straightforwardly handled by the following two
rules
\begin{align*}
  \vqueuerule
  &\hspace{2cm}
  \vemptyrule
\end{align*}
where $\_ \cdot \_$ denotes the concatenation operation on sequences
and \myrule{VEmpty} permits to type empty queues.
\begin{figure*}[tp]
  \begin{center}
    \[
    \begin{array}{l}
      \vreqrule
      \\
      \vaccrule
      \\[2em]
      \rcvrule
      \\[2em]
      \vsendrule
      \hfill
      \vendrule
      \\[2em]
      \vseqrule{VSeq}
      \\[2em]
      \vifrule{VIf}
      \\[2em]
      \vforend
      \\[2em]
      \vfor
      \\[2em]
      \vlooprule{VLoop}
  \end{array}\]
  \end{center}
  \caption{Typing rules for processes\label{fig:typingproc}}
\end{figure*}

\begin{figure}[t]
  \begin{center}
  \[
  \begin{array}{l}
      \vparrule
      \\
      \vqueuerule
      \hfill
      \vemptyrule
      \\[5pt]
      \vnewrule{VNew}
    \end{array}
    \]
  \end{center}
  \caption{Typing rules for systems\label{fig:typingsys}}
\end{figure}

%%% Local Variables:
%%% mode: latex
%%% TeX-master: "main"
%%% End:

\subsection{Typing the POP2 Protocol}%
\label{sec:exproof}
% !TEX root = main.tex
We now apply our type system to the implementations of our running
example.
We start by considering the process $\poppname{Init}$ in~\ref{ex:1proc} (page~\pageref{ex:1proc}) and the specification $\envmv
= \shname : \GPOP$ with $\GPOP$ from \cref{sec:runGT}. We recall that
its projection on $\ptp s$ is $\exG \TT {\ptp s} = {\GPOP}\proj{\ptp
  s}$ (also from~\ref{sec:runGT}).
Then, the typing judgement for $\poppname{Init}$ is obtained by using
rule \myrule{VAcc} as follows
\begin{prooftree}
\AxiomC{$\vdots$}
\UnaryInfC{
  $
  \envmv(u)\equiv\GPOP
  \quad
  \pjdg[\truek][\envmv, \psnames \smv s : \exG \pst {\ptp s}]
       [\poppname{Srv}][][\emptyset]
  \quad
  \rmg{\exG \pst {\ptp{s}}} = \exG \TT {\ptp s}
  $
}
\RightLabel{\myrule{VAcc}}
\UnaryInfC{
  $
  \pjdg[\truek][@][\Pacc{\shname}{s}{\tuple \smv}
    \poppname{Srv}][][\emptyset]
  $
}
\end{prooftree}
where the continuation $\poppname{Srv}$ is typed against the
specification $\envmv$ extended with a new participant's session
$ \psnames \smv s$ whose type $\exG \pst {\ptp s}$ matches $\exG \TT
{\ptp s}$, i.e., $\rmg{\exG \pst {\ptp{s}}} = \exG \TT {\ptp s}$.
Such $\exG \pst {\ptp s}$ is obtained from the judgement
$\pjdg[\truek][\envmv, \psnames \smv s : \exG \pst {\ptp
    s}][\poppname{Srv}][@][\emptyset]$.
Since
$\poppname{Srv} \eqdef
\Preceive{\popcname{quit}}{}{\poppname{Exit}}\ +\
\Preceive{\popcname{helo}}{c}{\poppname{Mbox}(c)}
$
is an input-guarded process, the judgement is obtained by applying
rule $\myrule{VRcv}$ as follows:
\begin{gather}
  \label{inf-typing-ex-1}
  \begin{mathprooftree}
    \AxiomC{
      $\vdots$
    }
    \UnaryInfC{
      $
      \pjdg[\truek][\envmv, \psnames \smv s : {\exG \pst {EXIT}}] % chktex 26
      [{\poppname{Exit}}][@][\emptyset]
      $
    }
    \AxiomC{$\vdots$}
    \UnaryInfC{
      $
      \pjdg[\truek][\envmv, \psnames \smv s : \exG \pst {MBOX}] % chktex 26
           [\poppname{Mbox}(c)][@][c : \dt{Str}] % chktex 26 chktex 36
     $
    }
    \RightLabel{\myrule{VRcv}}
    \BinaryInfC{
      $
      \pjdg [\truek]
	    [\envmv, \psnames \smv s : \exG \pst {\ptp s}] % chktex 26
	    [\Preceive{\popcname{quit}}{}{\poppname{Exit}}  +
              \Preceive{\popcname{helo}}{c}{\poppname{Mbox}(c)}][@][\emptyset] % chktex 36
      $
    }
  \end{mathprooftree}
\end{gather}
where
  \begin{align*}
    \exG \pst {\ptp s}
    &
    \eqdef
    \guard[{\truek}]{\Treceive{\popcname{quit}}{}}.{\exG \pst {EXIT}}\
    +\
    \guard[{\truek}]{\Treceive{\popcname{helo}}{\dt{Str}}}.\exG \pst {MBOX}
   \end{align*}
	with $\rmg{\exG \pst {EXIT}} = \exG \TT {EXIT}$ and
	$\rmg{\exG \pst {MBOX}} = \exG \TT {MBOX}$ so to satisfy
	$\rmg{\exG \pst {\ptp{s}}} = \exG \TT {\ptp s}$.
	The first premiss in~\eqref{inf-typing-ex-1} is derived as follows
	by taking
	$\exG \pst {EXIT} =
	\guard[\truek]{\Tsend{\popcname{bye}}{}}.\guard[\truek]{\Tend}$
	(recall that $\poppname{Exit}\eqdef \Psend{\popcname{bye}}{}$ and
	$\Psend{\popcname{bye}}{}$ is a shorthand for
	$\Psend{\popcname{bye}}{()}$).
\begin{prooftree}
  \AxiomC{$
    \emptyset \vdash ():\dt{Unit}
    \qquad
    \popcname{bye} \in\tuple\smv
    \quad\qquad
    (\envmv, \psnames \smv s : \guard[\truek]{\Tend})$ end-only}
  \RightLabel{\myrule{VSend}}
  \UnaryInfC{$
    \pjdg[\truek][\envmv, \psnames \smv s : {\exG \pst {EXIT}}]
         [{\Psend{\popcname{bye}}{()}}][@][\emptyset]
    $}
\end{prooftree}

The second premiss in~\eqref{inf-typing-ex-1} follows by using rule
$\myrule{VIf}$ because $\poppname{Mbox}(c)$ in~\eqref{proc:mbox} is a
conditional process (hereafter, we write $\C$ as shorthand for
$\dt{auth}\ c$).
{\small\begin{gather}
\label{inf-typing-ex-2}
\begin{mathprooftree}
  \AxiomC{
    $
    \begin{array}{l}
      \begin{array}{c}
	\vdots
	\\
	\hline % chktex 44
	\pjdg [\C]
	      [\envmv,\psnames \smv s: \exG \pst {then}]
 	      [\Psend{\popcname{r}}{(\dt {mn}\ \dt{inbox})};{\poppname{Nmbr}}]
	      [@]
	      [c:\dt{Str}]
      \end{array}
      \\[-15pt]
      \hspace{6cm}
      \begin{array}{c}
	\vdots
	\\
	\hline % chktex 44
	\pjdg [\neg \C]
	      [\envmv, \psnames \smv s: \exG \pst {else}]
	      [\Psend{{\popcname{e}}}{};{\poppname{Exit}}]
	      [@]
	      [c:\dt{Str}]
      \end{array}
    \end{array}
    $
  }
  \RightLabel{\myrule{VIf}}
  \UnaryInfC{$
    \begin{array}{ll}
      \pjdg [\truek]
	    [\envmv, \psnames \smv s:
	      {\exG \pst {MBOX}}
	    ]
	    [\poppname{Mbox}(c)]
	    [@]
	    [c:\dt{Str}]
    \end{array}
    $
  }
\end{mathprooftree}
\end{gather}
}
with $\exG \pst {MBOX} = \exG \pst {then} \ \peace\ \exG \pst {else}$.
The second premiss above can be shown by using the rules
$\myrule{VSeq}$ and $\myrule{VSend}$ and by taking
\[\exG \pst {else} = \guard[\neg \C]{\Tsend{\popcname{e}}{}};\guard[\neg \C]{\Tend};\guard[\neg \C]{\Tsend{\popcname{bye}}{}};\guard[\neg \C]{\Tend}\]

For the first premiss  in~\eqref{inf-typing-ex-2}  we use rule $\myrule{Seq}$ as follows
{\small\begin{gather}
\label{inf-typing-ex-3}
\begin{mathprooftree}
\AxiomC{$\vdots$}
\UnaryInfC{$
	\pjdg [\C]
		[\envmv,\psnames \smv s:{\exG \pst {then_1}}]
 		[\Psend{\popcname{r}}{(\dt {mn}\ \dt{inbox})}]
		[@]
		[c:\dt{Str}]
$}
\AxiomC{$\vdots$}
\UnaryInfC{$
	\pjdg [\C]
		[\envmv,\psnames \smv s:{\exG \pst {then_2}}]
 		[\poppname{Nmbr}]
		[@]
		[c:\dt{Str}]
$}
\RightLabel{\myrule{VSeq}}
\BinaryInfC{$
\pjdg [\C]
		[\envmv,\psnames \smv s: \exG \pst {then}]
 		[\Psend{\popcname{r}}{(\dt {mn}\ \dt{inbox})};{\poppname{Nmbr}}]
		[@]
		[c:\dt{Str}]$}
\end{mathprooftree}
\end{gather}}
with $\exG \pst {then} = \exG \pst {then_1};\exG \pst {then_2}$.
By applying $\myrule{VSend}$, we conclude that
$\exG \pst {then_1} = \guard[\C]{\Tsend{\popcname{r}}{\dt
	 {Int}}};\guard[\C]\Tend$.
Since $\poppname{Nmbr}$ is a repeat-until loop, the second premiss in~\eqref{inf-typing-ex-3} is obtained by using rule $\myrule{Vloop}$ (we
write $\poppname{Nmbr_{body}}$ for the body of the iteration and
$\poppname{Nmbr_{until}}$ for the until guard).
\begin{prooftree}
\AxiomC{$\vdots$}
\UnaryInfC{$
	\pjdg [\C]
		[\envmv,\psnames \smv s:{\exG \pst {body}}]
 		[\poppname{Nmbr_{body}}]
		[@]
		[c:\dt{Str}]
$}
\AxiomC{$\vdots$}
\UnaryInfC{$
	\pjdg [\C]
		[\envmv,\psnames \smv s:{\exG \pst {until}}]
 		[\poppname{Nmbr_{until}}]
		[@]
		[c:\dt{Str}]
$}
\RightLabel{\myrule{VLoop}}
\BinaryInfC{$
\pjdg [\C]
		[\envmv,\psnames \smv s: \exG \pst {then_2}]
 		[{\poppname{Nmbr}}]
		[@]
		[c:\dt{Str}]$}
\end{prooftree}
with $\envmv,\psnames \smv s: \exG \pst {body}$ and $\envmv,\psnames
\smv s: \exG \pst {until}$ passively compatible and $\exG \pst then_2
= \Tdef{(\exG \pst {body})};\exG \pst {until}$.  By following the same
approach, it can be shown that both premisses are derivable by taking
\[
\begin{array}{lcl}
  \exG \pst {body}
  & = &
  \guard[\C]{\Treceive{\popcname{fold}}{\dt{Str}}}.
  \guard[\C]{\Tsend{\popcname{r}}{\dt{Int}}};\guard[\C]{\Tend}
  +
  \guard[\C]{\Treceive{\popcname{read}}{\dt{Int}}}.
  \guard[\C]{\Tsend{\popcname{r}}{\dt{Int}};\exG \pst {SIZE}}
  \\
  \exG \pst {until}
  & = &
  \guard[\C]{\Treceive{\popcname{quit}}{}}.
  \guard[\C]{\Tsend{\popcname{bye}}{}};\guard[\C]{\Tend}
\end{array}
\]
for a suitable $\exG \pst {SIZE}$ such that $\rmg{\exG \pst {SIZE}} =
\exG \TT {SIZE}$.
It is straightforward to check that
$\envmv,\psnames \smv s: \exG \pst {body}$ and
$\envmv,\psnames \smv s: \exG \pst {until}$ are passively compatible
because the channels in $\exG \pst {body}$ are different from the ones
appearing in $\exG \pst {until}$.  Moreover,
$\rmg{\exG \pst {then_2}} = \rmg{ \Tdef{(\exG \pst {body})};\exG \pst
  {until}} = \exG \TT {Nmbr}$.

It remains to show that $\exG \pst {MBOX} = \exG \pst {then}
\ \peace\ \exG \pst {else}$ is well-defined. We first compute the
normal form of the pseudo types.
\[
\begin{array}{lll}
  \nfp{\exG \pst {then}}
  &=&
  \guard[\C]{\Tsend{\popcname{r}}{\dt {Int}}}; \nfp{\exG \pst {then_2}}
  \\
  \nfp{\exG \pst {else}}
  &=&
  \guard[\neg \C]{\Tsend{\popcname{e}}{}};
  \guard[\neg \C]{\Tsend{\popcname{bye}}{}};
  \guard[\neg \C]{\Tend}
\end{array}
\]
It is immediate to notice that $\nfp{\exG \pst {then}}$ and $\nfp{\exG
  \pst {else}}$ are mergeable because they are internal choices on
disjoint set of session channels. Therefore,
\[
\exG \pst {MBOX}
=
\exG \pst {then} \ \peace\ \exG \pst {else}
=
\guard[\C]{\Tsend{\popcname{r}}{\dt {Int}}};
\nfp{\exG \pst {then_2}}\
\oplus\
\guard[\neg \C]{\Tsend{\popcname{e}}{}};
\guard[\neg \C]{\Tsend{\popcname{bye}}{}};
\guard[\neg \C]{\Tend}
\]
Finally, note that $\rmg{\exG \pst {MBOX}} = \exG \TT {MBOX}$.

We now give the main types for the multiparty variant given in~\eqref{ex:1bproc} (\cf page~\pageref{ex:1bproc}).
Assume $\envmv(u)\equiv\GPOP'$ from \cref{sec:runGT} and consider the
following pseudo-type:
\begin{align*}
  \exG{\pst'}{\ptp{s}}
  & \eqdef
  \guard[\truek]{
    \Treceive{\popcname{quit}}{}.
             {\exG \pst {EXIT}}
  }
  ~+ ~ \guard[\truek]{\Treceive{\popcname{helo}}{\dt{Str}}.\exG\pst {AUTH}}
  \\
  \exG \pst {AUTH}
  & \eqdef
  \guard[\truek]{
    \Tsend{\popcname{req}}{\dt{Str}};
    \guard[\truek]{ \Treceive{\popcname{res}}{\dt{Bool}}}.
    \exG{\pst'}{MBOX}
  }
  \\
  \exG{\pst'}{MBOX}
  & \eqdef
  \guard[{\C\land a }]{
    \Tsend{\popcname{r}}{\dt{Int}}; \exG {\pst'} {NMBR}
  }
  ~\oplus ~\guard[{\neg(\C \land a)}]{
    \Tsend{\popcname{e}}{}; \exG {\pst'} {EXIT}
  }
\end{align*}
such that $\rmg{\pst'_{\ptp{s}}} = \TT'_\ptp{s}$, $ \exG {\pst}
{NMBR}$ is as $ \exG {\pst} {then_2}$ above except that all enabling
conditions are $(\C\land a)$, and $ \exG {\pst'} {EXIT}$ is as $ \exG
\pst {EXIT}$ except that all enabling conditions are $\neg(\C \land
a)$.
The typing judgement $\pjdg[\truek][\envmv, \psnames \smv s: \exG
  {\pst'} {\ptp s}][\poppname{Init'}][@][\emptyset]$ can be obtained
as in the previous case.

%%% Local Variables:
%%% mode: latex
%%% TeX-master: "main"
%%% End:

\section{Properties of the Type System}%
\label{sec:properties}
% !TEX root = main.tex
In this section we show that a well-typed process (i) behaves as
specified by the global type (\cref{thm:conformance}) and (ii) is a
WSI of the role played in the global type
(\cref{thm:impl-covers-local}).

\subsection{Conformance}\label{sec:conformance}
In order to show that any well-typed process adheres to the behaviour
defined by a global type, we relate the semantics of the process with
the one of its specification through a subject reduction result
(\cref{thm:sr}).
The operational semantics of specifications is generated by the rules
in \cref{fig:LTStypes}, where it is implicitly assumed that we work
up-to normal forms, namely the pseudo-types are normalised before and
after transitions.
Notice that the labels are as in~\eqref{eq:ltslabel} on
page~\pageref{eq:ltslabel} but for the fact that they cannot be
conditional actions $\lcond e \alpha$.
\begin{figure*}[t]\small
  \[\begin{array}{lr}
    \mathrule{
      \envmv(\shname) \equiv \G(\tuple \smv)
      \quad
      \rmg \pst = \nfp{\G(\tuple \smv) \proj \ptp 0}
    }{
      \spec \Gamma \envmv
      \Ttrans{\lreq{\shname}{n}{\tuple \smv}}
      \spec \Gamma {\envmv,\psnames  \smv 0 : \pst}
    }{TReq}
    &
    \mathrule{
      j \in I
    }{
      \envmv, \psnames \smv p : \Epst[\C_i]{i \in I}{\smv_i}{\sort_i}{\pst_i}
      \Ttrans{\lreceive{\smv_j}{\sort_j}}
      \envmv, \psnames \smv p : \pst_j
    }{TRcv}
    \\[3em]
    \mathrule{
      \envmv(\shname) \equiv \G(\tuple \smv)
      \quad \rmg \pst = \nfp{\G(\tuple \smv) \proj \ptp p}
    }
    {\spec \Gamma \envmv
      \Ttrans{\lacc{\shname}{p}{\tuple \smv}}
      \spec \Gamma {\envmv,\psnames  \smv p : \pst}
    }{TAcc}
    &
    \mathrule{
      j \in I
    }{
      \envmv, \psnames \smv p : \Ipst[\C_i]{i \in I}{\smv_i}{\sort_i}{\pst_i}
      \Ttrans{\lsend{\smv_j}{\sort_j}}
      \envmv, \psnames \smv p : \pst_j
    }{TSend}
    \\[3em]
    \multicolumn{2}{l}{
      \mathrule{
        \envmv_1
        \Ttrans{\alpha}
        \envmv_1'
      }{
        \envmv_1;\envmv_2
        \Ttrans{\alpha}
        \envmv_1';\envmv_2
      }{TSeq}
      \hspace{1.5cm}
      \mathrule{
        \envmv_2
        \Ttrans{\lreceive \smv \sort}
        \envmv_2'
      }{
        \Tdef{\envmv_1};\envmv_2
        \Ttrans{\lreceive \smv \sort}
        \envmv_2'}{TLoop_0}
      \hspace{1.5cm}
      \mathrule{
        \envmv
        \Ttrans{\alpha}
        \envmv'
      }{
        \Tdef{\envmv}
        \Ttrans{\alpha}
        \envmv'}{TLoop_1}
      \hspace{1.5cm}
      \mathrule{
        \envmv
        \Ttrans{\alpha}
        \envmv'
      }{
        \Tdef{\envmv}
        \Ttrans{\alpha}
        \envmv';\Tdef{\envmv}}{TLoop_2}
      }
      \\[25pt]
      \multicolumn 2 l {
      \mathrule{
      \envmv(\shname) \equiv \Geq[\G][\tuple \smv][\GT]
        \qquad
        \participants{\GT} = \{\ptp p_0, \ldots, \ptp p_n\}
        \qquad
        \rmg{\pst_i} =  \nfp{\GT\proj{\ptp p_i}} \ \forall i \in \{0,\ldots,n\}
      }{
        \envmv \Ttrans{\tau} \envmv, \psnames \smv {p_0} : \pst_0, \ldots, \psnames \smv {p_n} : \pst_n,
        \tuple \smv : \queue{}{}
      }
      {TInit}
}
    \\[15pt]
    \multicolumn 2 l {
      \mathrule{
        j \in I
      }{
        \envmv, \psnames \smv p : \Ipst[\C_i]{i \in I}{\smv_i}{\sort_i}{\pst_i}, \smv_j : \queue{}{\tuple \sort}
        \Ttrans{\tau}
        \envmv, \psnames \smv p : \pst_j, \smv_j : \queue{}{\tuple \sort \cdot \sort_j}
      }{TCom_1}
    }
    \\[15pt]
    \multicolumn 2 l {
      \mathrule{
        j \in I
      }{
        \envmv, \psnames \smv p : \Epst[\C_i]{i \in I}{\smv_i}{\sort_i}{\pst_i}, \smv_j : \queue{}{\sort_j \cdot \tuple \sort}
        \Ttrans{\tau}
        \envmv, \psnames \smv p : \pst_j, \smv_j : \queue{}{\tuple \sort}
      }{TCom_2}
    }
  \end{array}\]
  \caption{Labelled transitions for specifications}%
  \label{fig:LTStypes}%
  \label{fig:LTSruntimetypes}
\end{figure*}
Intuitively, the rules in \cref{fig:LTStypes}, barred the last three,
state how the specification of a single participant behaves in a
session $\tuple \smv$ and are instrumental to establish subject
reduction.

Rules \myrule{TReq} and \myrule{TAcc} account for a specification that
initiates a new session by projecting (on $\ptp 0$ and $\ptp p$,
resp.) the global type associated with the shared name $\shname$ in
$\dom \envmv$.
Note that $\pst$ can use arbitrary guards in the projections
$\G(\tuple \smv) \proj \ptp p$ as long as the normal form of $\pst$
matches the one of the projection.
Rule \myrule{TRcv} accounts for the reception of a message.
Dually, rule \myrule{TSend} accounts for an endpoint that performs one
of its outputs.
Rule \myrule{TSeq} relies on the definition of sequential composition
of specifications (\cf page~\pageref{pag:envseq}); observe that the
case in which all pseudo-types in $\envmv_1$ are of the form $\guard
\Tend$ is precluded because we work up-to normal form of pseudo-types.
Finally, an iterative local type can be skipped (rule
\myrule{TLoop_0}), executed once (rule \myrule{TLoop_1}), or be
unfolded (rule \myrule{TLoop_2}).
The last three rules in \cref{fig:LTStypes} state how specifications
of systems behave.
Rule \myrule{TInit} initiates a new session $\tuple \smv$ by assigning
each participant with a type that matches the corresponding projection
of the global type.
Rules \myrule{TCom_1} and \myrule{TCom_2} establish how specifications
send and receive messages through queues.

\begin{figure*}
  \begin{math}
    \begin{array}{c@{\hspace{1cm}}c@{\hspace{1cm}}c}
      \mathrule{\consExp}{\consPst [ {\guard \Tend}]}{CEnd}
      &
      \mathrule{
        \consPst[\pst_1] \qquad \consPst[\pst_2]
      }{
        \consPst [ {\pst_1 ; \pst_2  }]
      }{CSeq}
      &
      \mathrule{\consPst[\pst]}{\consPst [{\Tdef{\pst}}]}{CLoop}
      \\[2em]
      \multicolumn 3 c {
        \mathrule{
          \exists i\in I \qst \consExp[\C_i] \wedge \consPst[\pst_i]
        }{
          \consPst[{\Ipst[\C_i]{i \in I}{\smv_i}{\sort_i}{\pst_i} }]
        }{CSend}
        \qquad\qquad
        \mathrule{
          \forall i\in I \qst \consExp[\C_i] \wedge \consPst[\pst_i]
        }{
          \consPst[{\Epst[\C_i]{i \in I}{\smv_i}{\sort_i}{\pst_i}}]
        }{CRcv}
      }
    \end{array}
  \end{math}
  \caption{\label{fig:consistency-store-pst}Consistency relation
    between stores and pseudo-types}
\end{figure*}

The behaviour of processes depends on the stores they run
on. Consequently, we compare process and specifications with respect
to stores; concretely, we only consider the behaviour of processes
running on stores that are consistent with the guards in the
pseudo-types of the specifications.
The consistency predicate $\consPst[\_][\_]$ relates stores with
pseudo-types and is inductively defined by the rules
in~\cref{fig:consistency-store-pst}.
Intuitively, $\consPst$ holds if $\sigma$ does not falsify any guard
in $\pst$; which is checked by rules $\myrule{CEnd}$, $\myrule{CSend}$
and $\myrule{CRcv}$, where $\consExp$ means $\eval \C \sigma = \truek$
or $\eval \C \sigma$ undefined.

The notion of consistency is then extended to type judgments as follows.

\begin{defi}[Consistency]\label{def:consistent}
  A store $\sigma$ is \emph{consistent} with a judgement $\pjdg[\C][@][S][@]$,
  written $\consSt$, if
  \begin{enumerate}
  \item\label{cons2:eq} $\dom{\envmv|_{\chset}} \cup \dom \Gamma
    \subseteq \dom \sigma$
  \item\label{cons3:eq} $\forall x \in \dom{\Gamma} \qst
    \vjdg[][\sigma(x)][\Gamma(x)]$
  \item\label{cons4:eq} $\eval {\C } \sigma = \truek$
  \item\label{cons5:eq} $\forall \psnames{\smv}{p}\in \dom{\envmv}
    \qst \consPst[\envmv \psnames{\smv}{p}]$.
  \end{enumerate}
\end{defi}
\noindent
A store $\sigma$ is consistent with a type judgement
$\pjdg[\C][@][S][@]$ if $\sigma$ contains an assignment for any free
name of $S$ (\cref{cons2:eq}); the values assigned to variables should
match the type assigned by the environment $\Gamma$
(\cref{cons3:eq}). Besides, the typing assumption $\C$ and the guards
in the pseudo-types in $\envmv$ should hold when evaluated over
$\sigma$ (\cref{cons4:eq} and \cref{cons5:eq}).

\begin{restatable}[Subject reduction]{thm}{subjred}%
\label{thm:sr}
If $\pjdg[\C][@][S][@]$, $\consSt$, and $\state[S][\sigma]
\Ptrans{\lcond{\C'}{\alpha}}\state[S'][\sigma']$ then there exist
$\Gamma'$ and $\envmv'$ such that
\begin{enumerate}

\item\label{sr:in}
  if $\alpha = \lreceive \smv {\val v}$ then $\envmv \Ttrans{\lreceive
    \smv \sort} \envmv'$ for a sort $\sort$; moreover, if
  $\vjdg[][\val v][\sort]$ then there is $x \in \varset$ such that
  \[
  \pjdg[\C \land \C'][\envmv'][S'][][\Gamma', x : \sort],
  \qquad
  \sigma'(x) = \val v,
  \quad\text{and}\quad
  \consSt[\sigma'][S'][\envmv'][\C\land\C'][\Gamma', x : \sort]
  \]
\item\label{sr:out}
  if $\alpha = \lsend \smv {\val v}$ then $\envmv \Ttrans{\lsend \smv
    \sort} \envmv'$ with $\vjdg[][\val v][\sort]$, $\pjdg[\C \land
    \C'][\envmv'][S'][][\Gamma']$, and
  $\consSt[\sigma'][S'][\envmv'][\C \land \C'][\Gamma']$

\item\label{sr:other}
  otherwise $\envmv \Ttrans \alpha \envmv'$, and $\pjdg[\C \land
    \C'][\envmv'][S'][][\Gamma']$,
  $\consSt[\sigma'][S'][\envmv'][\C\land\C'][\Gamma']$.
\end{enumerate}
\end{restatable}

\noindent
The typing rules in \cref{sec:typing} ensure the semantic conformance
of processes with the behaviour prescribed by their types.
Here, we define conformance in terms of \emph{conditional that}
simulation relates states and specifications.
Our definition is standard, except for input actions, for which
specifications have to simulate only inputs of messages with the
expected type (i.e., systems are not responsible when receiving
ill-typed messages).

Define $\Strans{\alpha} = \Ttrans\tau^*\Ttrans\alpha$.
Let $\envmv \Strans \smv$ mean that there are $\envmv'$ and $\val v$
such that $\envmv \Strans{\smv \val v} \envmv'$.

\begin{defi}[Conditional simulation]\label{def:simulation}
  A relation $\R$ between pairs state-specification is a
  \emph{conditional simulation} if for any
  $\CWSRel{\state[S][\sigma]}{\envmv}$, if
  $\state[S][\sigma]\Ptrans{\lcond{e}{\alpha}}\state[S'][\sigma']$
  then
  \begin{enumerate}
  \item\label{it:cs1} if $\alpha = \lreceive{\smv}{\val v}$ then
    there exists $\envmv'$ such that
    $\envmv \Strans{\lreceive{\smv}{\sort}} \envmv'$ and if
    $\vjdg[][\val v][\sort]$ then there exists $x$ such that
    $\sigma' = \sigma\upd {\val v}{x}$ and
    $\CWSRel{\state[S'][\sigma']}{\envmv'}$
  \item\label{it:cs2} otherwise, $\envmv \Strans{\alpha} \envmv'$ and
    $\CWSRel{\state[S'][\sigma']}{\envmv'}$.
  \end{enumerate}
  We write $\state[S][\sigma] \precsim \envmv$ if there is a
  conditional simulation $\R$ such that
  $\CWSRel{\state[S][\sigma]}{\envmv}$.
\end{defi}

\noindent
By~\eqref{it:cs1}, only inputs of $S$ with the expected type have to
be matched by $\envmv$ (recall rule \myrule{TRec} in
\cref{fig:LTStypes}), while it is no longer expected to conform to the
specification after an ill-typed input (i.e., not allowed by
$\envmv$).

Conformance follows by straightforward coinduction from \cref{thm:sr}.
\begin{restatable}[Conformance]{thm}{notrestated}\label{thm:conformance}
  If $\pjdg[@][@][S]$ and $\consSt$ then $\state[S][\sigma] \precsim
  \spec{\Gamma}{\envmv}$.
\end{restatable}
\begin{proof}
  Using \cref{thm:sr} it is straightforward to show that
  \[
    \R = \{(\state[S][\sigma],\spec{\Gamma}{\envmv}) \sst
    \pjdg[@][@][S] \text{ and } \consSt\}
  \]
  is a conditional simulation.
\end{proof}

%%% Local Variables:
%%% mode: latex
%%% TeX-master: "main"
%%% End:

\subsection{WSI by Typing}%
\label{sec:wsi-by-typing}
% !TEX root = main.tex

We show that well-typed processes are WSIs (\cref{def:wsr} on
page~\pageref{def:wsr}).
First, we relate the runs of a global type with those of its
corresponding specifications.
Then, we state the correspondence between the runs of specifications
and well-typed implementations.
A set of implementations covering a global type $\G$ can exhibit more
behaviour than the runs of $\G$.
Nonetheless, we use WSI with our subject reduction property
(\cf\ \cref{def:simulation} and \cref{thmt@@notrestated}) to
characterise valid implementations.

Given a specification $\envmv$ such that $\envmv\psnames\smv p$ is in
normal form for all $\psnames\smv p \in\dom{\envmv}$, we let
$\typeruns{\tuple \smv}{\envmv}$, inductively defined by the rules in
\cref{fig:def-runs-local-types}, to be of the set of runs of session
$\tuple \smv$ generated by $\envmv$.
\begin{figure*}[t]
  \[
  \begin{array}{l}
    \mathrule{
      j \in I
      \qquad
      r\in\typeruns{\tuple \smv}{\envmv, \psnames \smv p : \pst_j,
        \smv_j : \queue{}{\tuple \sort \cdot \sort_j}}
    }{
      \state[\ptp p][\lsend {\smv_j} {\sort_j} ]
      r\in\typeruns{
        \tuple \smv
      }{
        \envmv, \psnames \smv p : \Ipst[\C_i]{i \in I}{\smv_i}{\sort_i}{\pst_i},
        \smv_j : \queue{}{\tuple \sort}}
    }{RTCom_1}
    \\[25pt]
    \mathrule{
      j \in I
      \qquad
      r\in\typeruns{\tuple \smv}
      {\envmv, \psnames \smv p : \pst_j, \smv_j : \queue{}{\tuple \sort}}
    }{
      \state[\ptp p][\lreceive {\smv_j} {\sort_j} ]
      r\in\typeruns{\tuple \smv}
      {\envmv, \psnames \smv p : \Epst[\C_i]{i \in I}{\smv_i}{\sort_i}{\pst_i},
        \smv_j : \queue{}{\sort_j \cdot \tuple \sort}}
    }{RTCom_2}
    \\[20pt]
    \mathrule{
      r \in \typeruns{\tuple \smv}{\envmv}
    }{
      r \in \typeruns{\tuple \smv}{\Tdef{\envmv}}
    }{RTIt_1}
    \hfill
    \mathrule{
      r_1 \in \typeruns{\tuple \smv}{\envmv_1}
      \qquad
      r_2 \in \typeruns{\tuple \smv}{\Tdef{\envmv_1}}
    }{
      r_1[r_2] \in \typeruns{\tuple \smv}{\Tdef{\envmv_1}}
    }{RTIt_2}
    \\[25pt]
    \mathrule{
      \envmv \ \text{$\Tend$-only}
    }{
      \epsilon\in\typeruns{\tuple \smv}{\envmv,\tuple \smv : \queue{}{}}
    }{RTEnd}
    \hspace{1cm}
    \mathrule{r_1 \in \typeruns{\tuple \smv}{\envmv_1}
      \qquad
      r_2 \in \typeruns{\tuple \smv}{\envmv_2}
    }{
      r_1r_2 \in \typeruns{\tuple \smv}{\envmv_1;\envmv_2}
    }{RTSeq}
  \end{array}
  \]
  \caption{Runs of specifications}%
  \label{fig:def-runs-local-types}
\end{figure*}
Rule~\myrule{RTCom_1} accounts for runs starting with output actions
$\state[\ptp p][\lsend {\smv_j} {\sort_j} ]$ performed by an endpoint
$\psnames \smv p$ and followed by a run $r$ of the continuation.
Rule~\myrule{RTCom_2} analogously deals with inputs.
Rules~\myrule{RTIt_1} and~\myrule{RTIt_2} unfold an iterative type;
one iteration is mandatory while the additional ones are optional.
The remaining rules are self-explanatory.
Since all types in $\envmv$ are in normal form (and therefore
inconsistent guards may appear only in sub terms $\guard\Tend$), the
rules in \cref{fig:def-runs-local-types} do not generate runs with
actions that cannot be fired.
The following two results establish the correspondence between the
denotational and operational semantics of specifications.

\begin{restatable}{lem}{speclem}\label{lem:corr-semantics-local-types}
  Let $\envmv$ be a specification such that for all $\psnames \smv
  p\in\dom\envmv$, $\envmv \psnames \smv p$ is in normal form. If
  $\spec{}{\envmv}\Ttrans{\tau}\spec{}{\envmv'}$, then for all
  $r\in\typeruns{\tuple \smv}{\envmv'}$ either:
  \begin{itemize}
  \item
    $r\in\typeruns{\tuple \smv}{\envmv}$, or
  \item
    $\state[\ptp p][\MAsend{\smv}{\sort}] r \in \typeruns{\tuple
    \smv}{\envmv}$, or else
  \item
    $\state[\ptp q][\Treceive{\smv}{\sort}] r \in \typeruns{\tuple
    \smv}{\envmv}$.
  \end{itemize}
\end{restatable}

\begin{restatable}{lem}{specbacklem}%
  \label{lem:corr-semantics-local-types-back}
  Let $\envmv$ a specification such that for all $\psnames \smv
  p\in\dom\envmv$, $\envmv \psnames \smv p$ is in normal form. If
  $r\in\typeruns{\tuple \smv}{\envmv}$ and $r\neq \epsilon$ then
  $\envmv\Ttrans{\tau}\envmv'$ and either
  \begin{itemize}
  \item $r\in\typeruns{\tuple \smv}{\envmv'}$, or
  \item $r=\state[\ptp p][\MAsend{\smv}{\sort}] r'$ and
    $r'\in\typeruns{\tuple \smv}{\envmv'}$, or
  \item $r= \state[\ptp q][{\smv}{\sort}] r'$, or else and
    $r'\in\typeruns{\tuple \smv}{\envmv'}$.
  \end{itemize}
\end{restatable}

\noindent
\cref{lem:runs-types-covers-runs-global} below ensures that
well-formed global types are covered by their projections
(\cref{def:runGTT} on page~\pageref{def:runGTT}).
\cref{thm:impl-covers-local} ensures that a local specification can be
covered by a set of implementations where a role $\ptp p_i$ that is
played by the same well-typed process $P$, as is WSI of the role $\ptp
p$ played by $P$ in the global type. As a corollary we have that a
well-typed process $P$ is a WSI of the role it plays.

\begin{restatable}[Coverage \& projections]{thm}{coveragethm}%
  \label{lem:runs-types-covers-runs-global}
  Let $\Geq$ be a global type with $\participants{\GT} = \{\ptp
  {p_0},\ldots, \ptp {p_n}\}$ and $\envmv = \psnames \smv {p_0} :
  \pst_0, \ldots, \psnames \smv {p_n} : \pst_n, \tuple\smv :
  \queue{}{}$ such that $\rmg{\pst_i}= \nfp{\GT\proj{\ptp p_i}}$ for
  all $0 \leq i \leq n$.
  Then $\imruns{\GT} \ \mysubseteq\ \typeruns{\tuple \smv}{\envmv}$.
\end{restatable}

\begin{restatable}[Typeability \& coverage]{thm}{coveragetwothm}%
  \label{thm:impl-covers-local}
  Let $\Geq$ be a global type with $\participants{\GT} = \{\ptp
  {p_0},\ldots, \ptp {p_n}\}$ and $\envmv = \psnames \smv {p_0} :
  \pst_0, \ldots, \psnames \smv {p_n} : \pst_n, \tuple\smv :
  \queue{}{}$ such that $\rmg{\pst_i}= \nfp{\GT\proj{\ptp p_i}}$ for
  all $0 \leq i \leq n$.
  If $\pjdg[\C][{\envmv', \psnames \smv {\ptp p_i} :
      \pst_i}][P][@][@]$ then the set $\mathbb{I} = \{\Gimp\ |
  \ \pjdg[\C][{\envmv, \envmv''}][\Gimp][@][@] \ \wedge\ \iota(\ptp
  p_i) = P \}$ covers $\typeruns{\tuple \smv}{\envmv}$, i.e.,
  $\typeruns{\tuple \smv}{\envmv} \mysubseteq \imrunsu{\mathbb{I}}$.
\end{restatable}

The result above relies on an auxiliary result that shows that each
run $r$ of an specification $\envmv$ can be covered by a well-typed
implementation $\Gimp$ of $\G(\tuple \smv)$ in which $P$ plays role
$\ptp p$ (details are provided in \cref{app:imp-cover-spec},
\cref{lem-impl-covers-local}).  Since $\mysubseteq$ is transitive, we
conclude that any well-typed process is a WSI of a role in a
choreography.

\begin{cor}[WSI of well-typed processes]\label{cor:impl-covers-local}
 Let $\G(\tuple \smv)$ be a global type, $\ptp p \in \participants
 \G$. If $\pjdg[\C][\envmv][P][@][@]$ and $\psnames \smv p \in \dom
 \envmv$, $P$ is a whole-spectrum implementation of $\ptp p$.
\end{cor}

%%% Local Variables:
%%% mode: latex
%%% TeX-master: "main"
%%% End:

\section{Conclusion and Related Work}%
\label{sec:related}
% !TEX root = main.tex

WSI rejects implementations of a role that persistently avoids the
execution of some branches in a choreography.
Although WSI is defined as a relation between the traces of a global
type and those of its candidate implementations, it can be checked by
using multiparty session types.
As standard, the soundness of our type system --guaranteed by the
conformance of the typing (\cref{thm:conformance})-- ensures that the
behaviour of well-typed implementations follows the protocol described
by the global type (i.e., global types are interpreted as
constraints).
Moreover, we show that $(i)$ the sets of the projections of a global
type $\G$ preserves all the traces in $\G$
(\cref{lem:runs-types-covers-runs-global}); and $(ii)$ a well-typed
process can be used to obtain any trace of the projections of a global
type when interacting in a proper context
(\cref{thm:impl-covers-local}).
These two results and the fact that the covering relation is
transitive allow us to conclude that any well-typed process is a WSI
of a role in a choreography (\cref{cor:impl-covers-local}), i.e.,
global types are interpreted as obligations.

%
% Behavioural types
%

\subsection{Behavioural types}
In this paper we have followed the rich line of research fostering the
application of behavioural types to concurrent programs.
Examples of similar approaches are those to guarantee properties of
complex concurrent systems such as in the seminal work of Kobayashi on
deadlock freedom for the $\pi$-calculus~\cite{koba04,koba06} or
progress analysis for choreographies~\cite{cdyp16}, information flow
analysis~\cite{koba05,cd16}, design-by-contract for message-passing
systems~\cite{bhty10}, or self-adaptation~\cite{cdv15,cdp16}. Our type
system is more restrictive
than~\cite{HondaYC08,BettiniCDLDY08,DBLP:journals/mscs/BravettiZ09a,CP09,CairesV09}
as it rules out sound but not exhaustive implementations, which
previous type systems considered well-typed.
To the best of our knowledge, the only proposal dealing with complete
(i.e., exhaustive) realisations in a behavioural context
is~\cite{DBLP:journals/corr/abs-1203-0780} but this approach focuses
on non-deterministic implementation languages.
WSI coincides with projection
realisability~\cite{lanese2008bridging,salaun2009realizability,DBLP:journals/corr/abs-1203-0780}
when implementation languages feature non-deterministic internal
choices.
On the contrary, WSI provides a finer criterion to distinguish
deterministic implementations, as illustrated by the motivating
example in the introduction.

\cref{asec:communicating} illustrates that WSI can be recast into
other computational models or settings, e.g., in the context of
guarded automata. In such context, WSI admits a more succinct
characterisation, and hence, would appear as a more amenable
definition for WSI (when compared against the definition given in
Section 6). However, we remark that several technicalities in our
proposal arise when developing a static and modular technique for
ensuring WSI\@. MSTs are perhaps the most widely accepted technique for
developing static and modular nal verification techniques for
multiparty interactions, and to link their formal specification to
programming languages, even at the expenses of dense technical
definition (e.g., three different languages, static and run-time
semantics, well-formed conditions, etc).
It may be the case that similar techniques could be developed for
ensuring WSI implementations in other contexts, e.g., guarded
automata, but the starting point is not that obvious. It could be the
case that some model-checking solution could be developed but the
existence of a feasible and simpler approach is an interesting
question that requires future work (note that the definitions in
\cref{asec:communicating} are based on an existential quantification
over possible contexts, which could be problematic for obtaining an
effective procedure).

\subsection{WSI and subtyping}
The standard subtyping relation~\cite{GH} is not suitable for WSI
because the liberal elimination of internal choices prevents WSI\@.  For
example, standard subtyping for output prefixes allows process $
\Preceive{\smv}{v}{} \Psend{\smv_1}{e}$ to have local type
$\Treceive{\smv}{\sort};(\Tsend{\smv_1}{\sort}+\Tsend{\smv_2}{\sort})$,
and this clearly violates WSI\@.  The problematic aspect is that
subtyping allows for a liberal implementation of internal choices,
whereas WSI requires a precise implementation of all branches in a
choice.  This means that a WSI of a type $\TT$, may not be a WSI of a
subtype of $\TT$.
The investigation of suitable forms of subtyping for WSIs is scope for
future work.
To some extent our proposal is related to the fair subtyping approach
in~\cite{DBLP:conf/coordination/Padovani11}, where refinement is
studied under the fairness assumption: fair subtyping differs from
usual subtyping when considering infinite computations but WSI differs
from partial implementation also when considering finite computations.

\subsection{Languages for types and processes} Our
  notation for session types combines interaction and branching in one
  single operation. This choice was made for the sake of conciseness.
  The combination of interaction and branching in one single operation
  has been used in several works, albeit in different flavours
  (e.g.,~\cite{DBLP:conf/tgc/ChenBDHY11},~\cite{lt12},~\cite{DBLP:conf/tgc/BocchiDY12}
  and~\cite{DBLP:conf/esop/DenielouY12}).
In~\cite{lt12} the notation is more flexible than ours, allowing
(syntactically) global types of the form
\begin{equation}\label{ex_rel}
\Gint[p][q][\smv_1][d_1];\GT_1 + \Gint[r][s][\smv_2][d_2];\GT_2
\end{equation}
which our syntax does not allows (as it requires $\ptp p=\ptp
r$). However, the work in~\cite{lt12} rules out types as the one in
(\ref{ex_rel}) when $\ptp p\not =\ptp r$ by the well-formedness rules
(i.e., knowledge of choice).
Our notation simplifies well-formedness by having this requirement in
the syntax. Another line of
work~\cite{DBLP:conf/tgc/ChenBDHY11,DBLP:conf/tgc/BocchiDY12} uses a
syntax more restrictive than ours: in types of the form~\eqref{ex_rel},~\cite{DBLP:conf/tgc/ChenBDHY11,DBLP:conf/tgc/BocchiDY12} require
$\ptp p=\ptp r$ and $\ptp q=\ptp s$. Such restriction prevents types
of the form
\[
\Gint[p][q][\smv_1][d_1];\Gint[r][s][\smv_2][d_2]; \Gend +
\Gint[r][s][\smv_2][d_2]; \Gint[p][q][\smv_1][d_1] ; \Gend
\]
which are instead well-formed in our theory.  The notation introduced
in this paper differs from usual syntaxes for session types in two
other significant ways, which we highlight below.  Firstly, we use
iteration instead of recursion. In order to verify WSI, our proof
system needs to statically determine that the body of each iteration:
(1) happens at least once (if interactions in the body are not
executed the implementation is not WSI), and (2) terminates (to ensure
that interactions in the continuation of an iteration will be
executed). For this reason, we need \emph{finitary interactions} in
the language for processes. However, usual interpretation of recursive
types (where maximal fixpoints are assumed) is not finitary.
Therefore, the use of standard recursive session types would introduce
a mismatch between the semantics of programs (where we need finitary
iterations) with the usual interpretation of recursive types (where
maximal fixpoints are assumed).

Therefore, the static verification of WSI requires a form of recursion
more restrictive than the one in previous work on session
types~\cite{HondaYC08,BettiniCDLDY08}, where the number of iterations
is limited. This restriction is on the lines
of~\cite{DBLP:journals/corr/abs-1203-0780} that also considers finite
traces, and is close to for-loops of programming languages. The
extension of our theory with a more general form of iteration is scope
for future work.

Because of the chosen notation for interaction-choice, iteration
requires some care to ensure that all participants agree on whether
another iteration should be executed. We do this by requiring that
each iteration has exactly one controller-participant as
in~\cite{DBLP:journals/corr/abs-1203-0780}.

Secondly, we use sequential composition to allow clear delimitation of
structured constructs (loops and conditional statements) hence ease
static verification and its modularity.

We disregarded delegation in our framework because its addition would
greatly increase the complexity of our framework. A type system
guaranteeing WSI under delegation in our context should not be
problematic to device at the cost of a higher technical complexity.

\subsection{Choreography languages}
In the literature, the term `choreography' is used in two different ways.
The first one, which we follow, considers choreographies as abstract
specifications~\cite{w3c:cho} or formal models (e.g., Petri
Nets~\cite{DBLP:conf/icsoc/LohmannW11} or session
types~\cite{HondaYC08,BettiniCDLDY08}) where choices are
non-deterministic, in the sense that branches are not associated to
conditions (e.g., as in if-then-else statements).  The W3C's Web
Services Choreography Description Language (WS-CDL)~\cite{w3c:cho} is
an XML-based specification language where choices are enumerated
without expressing conditions. In the context of multiparty session
types~\cite{HondaYC08,BettiniCDLDY08}, choreographies are modelled as
global types which, by projection, produce abstractions of processes
(with non-deterministic choices in the sense specified above) that can
be used for static and dynamic verification. On this thread, we also
mention works that study partial vs complete
realisations~\cite{DBLP:conf/icsoc/LohmannW11,DBLP:journals/corr/abs-1203-0780}.%
Remarkably, WSI coincides with projection
realisability~\cite{lanese2008bridging,salaun2009realizability,DBLP:journals/corr/abs-1203-0780}
when the language adopted to implement choreographies features
non-deterministic internal choices.
On the contrary, WSI provides a finer criterion to distinguish
deterministic implementations, as illustrated by the motivating
example in the introduction.  The second usage is in the context of
choreographic programming (e.g.,~\cite{DBLP:conf/popl/CarboneM13},~\cite{DBLP:conf/birthday/Lluch-LafuenteN15},~\cite{preda})
where a choreography is a built in construct for programs (hence
deterministic) used to produce correct-by-design code. In such
contexts there is no need to resolve non-deterministic choices since
the language is deterministic.

\subsection{POP 2 specification}
We captured (to the best of our understanding) the most salient
aspects of the communication of the POP2 protocol according to its
official informal specification in the state-machines in pages 16--18
of RFC937~\cite{pop2}. We have focussed on the interaction structure
and left out non-functional aspects such as timeouts that are beyond
the scope of this paper. We have adopted a simplification on the
interaction structure for the sake of a simpler presentation: POP2 has
a branch `quit' from state `SIZE' (see~\cite{pop2}, page 16) whereas
our model $\exG \GT {SIZE}$ in \cref{sec:runGTp} does not have a
`quit' option. To quit the protocol from state $\exG \GT {SIZE}$ one
needs to go back to state $\exG \GT {FOLD}$ (which offers the option
`quit' as also the corresponding state in the RFC). In our formalism,
the end of a loop has to be signalled by exactly one message (in our
case message `fold'). We could have encoded this extra `quit' option
from state `SIZE' by either:
\begin{itemize}
\item extending our choreography language to allow a set of possible
  messages to signal termination of a loop, which would not add
  considerable challenge but would increase the technicalities in the
  presentation;
\item encoding the exact POP2 pattern by introducing intermediary
  states with additional messages which would have made the
  presentation of the protocol itself less clear.
\end{itemize}
Our goal was to demonstrate that our typing language can model
realistic communications, not to provide an analysis of POP2 on its
own sake (which we leave as future work), hence our simplification of
the protocol.

%%% Local Variables:
%%% mode: latex
%%% TeX-master: "main"
%%% End:

\section*{Acknowledgement}
We thank the anonymous reviewers for reading the paper carefully and providing thoughtful comments.

\bibliographystyle{abbrv}
\bibliography{bib}

\begin{thebibliography}{10}

\bibitem{BettiniCDLDY08}
L.~Bettini, M.~Coppo, L.~{D'Antoni}, M.~{De Luca}, M.~{Dezani-Ciancaglini}, and
  N.~Yoshida.
\newblock Global progress in dynamically interleaved multiparty sessions.
\newblock In {\em CONCUR}, volume 5201 of {\em LNCS}, pages 418--433. Springer,
  2008.

\bibitem{DBLP:conf/tgc/BocchiDY12}
L.~Bocchi, R.~Demangeon, and N.~Yoshida.
\newblock A multiparty multi-session logic.
\newblock In C.~Palamidessi and M.~D. Ryan, editors, {\em TGC 2012}, volume
  8191 of {\em LNCS}, pages 97--111. Springer, 2012.

\bibitem{bhty10}
L.~Bocchi, K.~Honda, E.~Tuosto, and N.~Yoshida.
\newblock A theory of design-by-contract for distributed multiparty
  interactions.
\newblock In {\em CONCUR}, volume 6269 of {\em LNCS}, pages 162--176, 2010.

\bibitem{ESOP14}
L.~Bocchi, H.~Melgratti, and E.~Tuosto.
\newblock Resolving non-determinism in choreographies.
\newblock In Z.~Shao, editor, {\em Programming Languages and Systems}, pages
  493--512, Berlin, Heidelberg, 2014. Springer Berlin Heidelberg.

\bibitem{DBLP:journals/mscs/BravettiZ09a}
M.~Bravetti and G.~Zavattaro.
\newblock A theory of contracts for strong service compliance.
\newblock {\em Mathematical Structures in Computer Science}, 19(3):601--638,
  2009.

\bibitem{pop2}
M.~Butler, J.~Postel, D.~Chase, J.~Goldberger, and J.~Reynoldsa.
\newblock Post office protocol - version 2.
\newblock RFC 918, available at \url{http://tools.ietf.org/html/rfc937},
  February 1985.

\bibitem{CairesV09}
L.~Caires and H.~T. Vieira.
\newblock Conversation types.
\newblock In {\em ESOP}, volume 5502 of {\em LNCS}, pages 285--300. Springer,
  2009.

\bibitem{cd16}
S.~Capecchi, I.~Castellani, and M.~Dezani{-}Ciancaglini.
\newblock Information flow safety in multiparty sessions.
\newblock {\em Mathematical Structures in Computer Science}, 26(8):1352--1394,
  2016.

\bibitem{DBLP:conf/popl/CarboneM13}
M.~Carbone and F.~Montesi.
\newblock Deadlock-freedom-by-design: multiparty asynchronous global
  programming.
\newblock In {\em {POPL} '13}, pages 263--274. {ACM}, 2013.

\bibitem{DBLP:journals/corr/abs-1203-0780}
G.~Castagna, M.~Dezani-Ciancaglini, and L.~Padovani.
\newblock On global types and multi-party session.
\newblock {\em Logical Methods in Computer Science}, 8(1), 2012.

\bibitem{CP09}
G.~Castagna and L.~Padovani.
\newblock Contracts for mobile processes.
\newblock In {\em {CONCUR} 2009}, number 5710 in LNCS, pages 211--228, 2009.

\bibitem{cdp16}
I.~Castellani, M.~Dezani{-}Ciancaglini, and J.~A. P{\'{e}}rez.
\newblock Self-adaptation and secure information flow in multiparty
  communications.
\newblock {\em Formal Asp. Comput.}, 28(4):669--696, 2016.

\bibitem{DBLP:conf/tgc/ChenBDHY11}
T.-C. Chen, L.~Bocchi, P.-M. Deni{\'e}lou, K.~Honda, and N.~Yoshida.
\newblock Asynchronous distributed monitoring for multiparty session
  enforcement.
\newblock In R.~Bruni and V.~Sassone, editors, {\em TGC}, volume 7173 of {\em
  Lecture Notes in Computer Science}, pages 25--45. Springer, 2011.

\bibitem{ChenH12}
T.-C. Chen and K.~Honda.
\newblock Specifying stateful asynchronous properties for distributed programs.
\newblock In {\em CONCUR}, 2012.

\bibitem{cdv15}
M.~Coppo, M.~Dezani{-}Ciancaglini, and B.~Venneri.
\newblock Self-adaptive multiparty sessions.
\newblock {\em Service Oriented Computing and Applications}, 9(3-4):249--268,
  2015.

\bibitem{cdyp16}
M.~Coppo, M.~Dezani{-}Ciancaglini, N.~Yoshida, and L.~Padovani.
\newblock Global progress for dynamically interleaved multiparty sessions.
\newblock {\em Mathematical Structures in Computer Science}, 26(2):238--302,
  2016.

\bibitem{rfc0822}
D.~Crocker.
\newblock Standard for the format of arpa internet text messages.
\newblock RFC 822, available at \url{www.ietf.org/rfc/rfc0822.txt}, February
  1982.

\bibitem{preda}
M.~Dalla~Preda, M.~Gabbrielli, S.~Giallorenzo, I.~Lanese, and J.~Mauro.
\newblock {Dynamic Choreographies: Theory And Implementation}.
\newblock {\em {Logical Methods in Computer Science}}, 13:1 -- 57, May 2017.

\bibitem{DBLP:conf/esop/DenielouY12}
P.~Deni{\'{e}}lou and N.~Yoshida.
\newblock Multiparty session types meet communicating automata.
\newblock In H.~Seidl, editor, {\em Programming Languages and Systems - 21st
  European Symposium on Programming, {ESOP} 2012, Held as Part of the European
  Joint Conferences on Theory and Practice of Software, {ETAPS} 2012, Tallinn,
  Estonia, March 24 - April 1, 2012. Proceedings}, volume 7211 of {\em Lecture
  Notes in Computer Science}, pages 194--213. Springer, 2012.

\bibitem{DBLP:conf/wsfm/Dezani-Ciancaglinid09}
M.~Dezani-Ciancaglini and U.~de'Liguoro.
\newblock Sessions and session types: An overview.
\newblock In C.~Laneve and J.~Su, editors, {\em WS-FM}, volume 6194 of {\em
  Lecture Notes in Computer Science}, pages 1--28. Springer, 2009.

\bibitem{fbs05}
X.~Fu, T.~Bultan, and J.~Su.
\newblock Realizability of conversation protocols with message contents.
\newblock {\em Int. J. Web Service Res.}, 2(4):68--93, 2005.

\bibitem{GH}
S.~Gay and M.~Hole.
\newblock {Types and Subtypes for Client-Server Interactions}.
\newblock In {\em Proc.~of ESOP'99}, volume 1576 of {\em {LNCS}}, pages 74--90.
  Springer-Verlag, 1999.

\bibitem{GH05}
S.~Gay and M.~Hole.
\newblock {Subtyping for Session Types in the Pi-Calculus}.
\newblock {\em Acta Informatica}, 42(2/3):191--225, 2005.

\bibitem{gt16}
R.~Guanciale and E.~Tuosto.
\newblock An abstract semantics of the global view of choreographies.
\newblock In {\em Proceedings 9th Interaction and Concurrency Experience, {ICE}
  2016, Heraklion, Greece, 8-9 June 2016.}, pages 67--82, 2016.

\bibitem{gt17}
R.~Guanciale and E.~Tuosto.
\newblock Semantics of global views of choreographies.
\newblock {\em Journal of Logic and Algebraic Methods in Programming}, 2017.
\newblock Revised and extended version of \cite{gt16}. Accepted for
  publication. To appear; version with proof available at
  \url{http://www.cs.le.ac.uk/people/et52/jlamp-with-proofs.pdf}.

\bibitem{HondaYC08}
K.~Honda, N.~Yoshida, and M.~Carbone.
\newblock Multiparty asynchronous session types.
\newblock In G.~C. Necula and P.~Wadler, editors, {\em POPL}, pages 273--284.
  ACM, 2008.

\bibitem{event}
R.~Hu, D.~Kouzapas, O.~Pernet, N.~Yoshida, and K.~Honda.
\newblock Type-safe eventful sessions in {J}ava.
\newblock In {\em ECOOP 2010}, volume 6183 of {\em LNCS}, pages 329--353.
  Springer-Verlag, 2010.

\bibitem{HuY16}
R.~Hu and N.~Yoshida.
\newblock Hybrid session verification through endpoint api generation.
\newblock In {\em Fundamental Approaches to Software Engineering}, pages
  401--418, Berlin, Heidelberg, 2016. Springer.

\bibitem{koba04}
A.~Igarashi and N.~Kobayashi.
\newblock A generic type system for the pi-calculus.
\newblock {\em Theor. Comput. Sci.}, 311(1-3):121--163, 2004.

\bibitem{w3c:cho}
N.~Kavantzas, D.~Burdett, G.~Ritzinger, T.~Fletcher, and Y.~Lafon.
\newblock \url{http://www.w3.org/TR/2004/WD-ws-cdl-10-20041217}, 2004.

\bibitem{koba05}
N.~Kobayashi.
\newblock Type-based information flow analysis for the pi-calculus.
\newblock {\em Acta Inf.}, 42(4-5):291--347, 2005.

\bibitem{koba06}
N.~Kobayashi.
\newblock A new type system for deadlock-free processes.
\newblock In {\em {CONCUR} 2006 - Concurrency Theory, 17th International
  Conference, {CONCUR} 2006, Bonn, Germany, August 27-30, 2006, Proceedings},
  volume 4137 of {\em Lecture Notes in Computer Science}, pages 233--247.
  Springer, 2006.

\bibitem{lanese2008bridging}
I.~Lanese, C.~Guidi, F.~Montesi, and G.~Zavattaro.
\newblock Bridging the gap between interaction-and process-oriented
  choreographies.
\newblock In {\em SEFM}, 2008.

\bibitem{lt12}
J.~Lange and E.~Tuosto.
\newblock Synthesising choreographies from local session types.
\newblock In M.~Koutny and I.~Ulidowski, editors, {\em CONCUR}, volume 7454 of
  {\em LNCS}, pages 225--239, 2012.

\bibitem{lty15}
J.~Lange, E.~Tuosto, and N.~Yoshida.
\newblock {From Communicating Machines to Graphical Choreographies}.
\newblock In {\em POPL}, pages 221--232, 2015.

\bibitem{DBLP:conf/birthday/Lluch-LafuenteN15}
A.~Lluch{-}Lafuente, F.~Nielson, and H.~R. Nielson.
\newblock Discretionary information flow control for interaction-oriented
  specifications.
\newblock In {\em Logic, Rewriting, and Concurrency - Essays dedicated to
  Jos{\'{e}} Meseguer on the Occasion of His 65th Birthday}, volume 9200 of
  {\em LNCS}, pages 427--450. Springer, 2015.

\bibitem{DBLP:conf/icsoc/LohmannW11}
N.~Lohmann and K.~Wolf.
\newblock Decidability results for choreography realization.
\newblock In G.~Kappel, Z.~Maamar, and H.~R.~M. Nezhad, editors, {\em ICSOC},
  volume 7084 of {\em Lecture Notes in Computer Science}, pages 92--107.
  Springer, 2011.

\bibitem{mil89}
R.~Milner.
\newblock {\em {Communication and Concurrency}}.
\newblock Prentice Hall, 1989.

\bibitem{NeubauerThiemann04}
M.~Neubauer and P.~Thiemann.
\newblock {An Implementation of Session Types}.
\newblock In {\em Practical Aspects of Declarative Languages (PADL)}, volume
  3057 of {\em LNCS}, pages 56--70. Springer, 2004.

\bibitem{NeykovaHYA18}
R.~Neykova, R.~Hu, N.~Yoshida, and F.~Abdeljallal.
\newblock A session type provider: compile-time {API} generation of distributed
  protocols with refinements in f{\#}.
\newblock In {\em {CC} 2018}, pages 128--138. {ACM}, 2018.

\bibitem{DBLP:conf/coordination/Padovani11}
L.~Padovani.
\newblock Fair subtyping for multi-party session types.
\newblock In {\em COORDINATION}, volume 6721 of {\em LNCS}, pages 127--141,
  2011.

\bibitem{salaun2009realizability}
G.~Sala{\"u}n and T.~Bultan.
\newblock Realizability of choreographies using process algebra encodings.
\newblock In {\em Integrated Formal Methods}, 2009.

\bibitem{DBLP:conf/wsfm/SuBFZ07}
J.~Su, T.~Bultan, X.~Fu, and X.~Zhao.
\newblock Towards a theory of web service choreographies.
\newblock In M.~Dumas and R.~Heckel, editors, {\em WS-FM}, volume 4937 of {\em
  Lecture Notes in Computer Science}, pages 1--16. Springer, 2007.

\end{thebibliography}

\newpage
\appendix
\section{Auxiliary Properties of Typing}\label{sec:aux}
% !TEX root = main.tex

In this section we provide technical details and auxiliary properties
of the type system, which are used in the proof of the main results of
the paper.

\subsection{Normal form \texorpdfstring{$\nform[\_]{\_}$}{}}%
\label{app:normal-form}

\newcommand{\weight}[1]{\omega(#1)}

Below we state useful results about the normal form of
pseudo-types. We start by introducing a well-founded relation on
pseudo-types, which will be used for inductive proofs. The relation
$<$ on pseudo-types is defined in terms of the following function
$\omega :\pst \to \nat$:

 \[
 \begin{array}{r@{\ =\ }l}
 \weight{\guard[\C] \Tend} & 1
 \\
 \weight{\Ipst[e_i]{i \in I}{\smv_i}{\sort_i}{\pst_i}} &1 + \max\ {\{\weight{\pst_i}\}}_{i\in I}
 \\
 \weight{\Epst[e_i]{i \in I}{\smv_i}{\sort_i}{\pst_i}} &1 + \max\ {\{\weight{\pst_i}\}}_{i\in I}
 \\
 \weight{\pst_1;\pst_2} &2*\weight{\pst_1} + \weight{\pst_2}
 \\
 \weight{\Tdef\pst} & 1 + \weight{\pst}
 \end{array}
 \]

 We say $\pst_1<\pst_2$ iff $\weight{\pst_1} < \weight{\pst_2}$. It is
 straightforward to check that $\weight\pst > 0$ for all
 $\pst$. Consequently, $(\pst,<)$ is well-founded.

\begin{lem}\label{lem:nf-termination}
  $\nfp{\pst}$ is defined for any $\pst$ (i.e., it terminates).
\end{lem}

\begin{proof} The proof follows by showing that the function
  $f(\C,\pst) = \weight\pst$ is a variant function for the definition
  of $\nform[\C]\pst$ in~\cref{fig:normalization-pseudo}. We proceed
  by case analysis on the equations in~\cref{fig:normalization-pseudo}
  (we illustrate the interesting cases below).
  \begin{itemize}[align=left]
  \item[(2)]
    $f(\C,{\Ipst[\C_i]{i \in I}{\smv_i}{\sort_i}{\pst_i}}) =
    1 + \max\ {\{\weight{\pst_i}\}}_{i\in I} > \weight{\pst_i} =
    f(\C_i\wedge\C,\pst_i)$
    for all $i \in I\setminus J$.
  \item[(3)] Analogous to $(2)$.
  \item[(5)]
    \[
    \begin{array}{lcl}
      f(\C,{\big(\Ipst[\C_i]{i \in I}{\smv_i}{\sort_i}{\pst_i}\big);\pst})
      & = &
      \weight{\big(\Ipst[\C_i]{i \in I}{\smv_i}{\sort_i}{\pst_i}\big);\pst}
      \\
      & = &
      2*(1 + \max\ {\{\weight{\pst_i}\}}_{i\in I}) + \weight\pst
      \\
      & = &
      2 +  \max\ {\{2* \weight{\pst_i} + \weight\pst\}}_{i\in I}
      \\
      & > &
      1 + \max\ {\{2*\weight{\pst_i}+ \weight\pst\}}_{i\in I}
      \\
      & = &
      \weight{\big(\Ipst[\C_i]{i \in I}{\smv_i}{\sort_i}{\pst_i;\pst}\big)}
      \\
      & = & f(\C,{\big(\Ipst[\C_i]{i \in I}{\smv_i}{\sort_i}{\pst_i;\pst}\big)})
      \\
    \end{array}
    \]
  \item[(6)]  Analogous to $(5)$.

  \item[(7)]
    \begin{align*}
    \begin{array}{lcl}
      f(\C,(\pst_1;\pst_2);\pst_3)
      &= &\weight{(\pst_1;\pst_2);\pst_3} \\
      &=& 2*(2*\weight{\pst_1}+\weight{\pst_2})+\weight{\pst_3} \\
      &=& 4*\weight{\pst_1}+ 2* \weight{\pst_2}+\weight{\pst_3} \\
      &>& 2*\weight{\pst_1}+ 2* \weight{\pst_2}+\weight{\pst_3} \\
      &=& \weight{\pst_1;(\pst_2;\pst_3)}\\
      &=& f(\C,{\pst_1;(\pst_2;\pst_3)}) \\
      \end{array}
      \\[\dimexpr-1.0\baselineskip+\dp\strutbox]&\qedhere
    \end{align*}
  \end{itemize}
\end{proof}

\begin{lem}\label{lem:nf-stronger-condition}
  For all $\C,\C'$, $\pst$, if $\nform[\C]{\pst} = \guard[\C'']\Tend$,
  then $\nform[\C\wedge\C']{\pst} = \guard[\C''\wedge\C']\Tend$.
\end{lem}
\begin{proof}
  By well-founded induction on $(\pst,<)$. The proof follows by case
  analysis on the structure of $\pst$.
\end{proof}

\begin{lem}\label{lem:nf-aux}
 For all $\C,\C'$, $\pst$, $\nform[\C]{\nform[\C']{\pst}} =
 \nform[\C\wedge\C']{\pst}$.
\end{lem}
\begin{proof}
  By well-founded induction on $(\pst,<)$. The proof follows by case
  analysis on the structure of $\pst$.
  \begin{itemize}
  \item $\pst = \guard[\C''] \Tend $: Then,
    \[
    \begin{array}{rll}
      \nform[\C]{\nform[\C']\pst}
      &
      =
      \nform[\C]{\nform[\C']{\guard[\C'']\Tend}}
      &
      \textit{by def.\ of $\pst$}
      \\
      &
      =
      \nform[\C]{\guard[\C'\wedge\C'']\Tend}
      &
      \textit{by \cref{fig:normalization-pseudo}(1)}
      \\
      &
      =
      {\guard[\C\wedge\C'\wedge\C'']\Tend}
      &
      \textit{by \cref{fig:normalization-pseudo}(1)}
      \\
      &
      =
      \nform[\C\wedge\C']{\guard[\C'']\Tend}
      &
      \textit{by \cref{fig:normalization-pseudo}(1)}
      \\
      &
       =
      \nform[\C\wedge\C']{\pst}
      &
      \textit{by def.\ of $\pst$}
    \end{array}
    \]
  \item $\pst = \Ipst[e_i]{i \in I}{\smv_i}{\sort_i}{\pst_i}$: Then,
    \[
    \begin{array}{rll}
      \nform[\C]{\nform[\C']\pst}
      &
      =
      \nform[\C]{\nform[\C'] {\Ipst[e_i]{i \in I}{\smv_i}{\sort_i}{\pst_i}}}
      &
      \textit{ by def.\ of $\pst$}
      \\
      &
      =
      \nform[\C]{\Ipst[e_i\wedge\C']{i \in I\setminus J'}{\smv_i}{\sort_i}
        {\nform[\C_i\wedge\C']{\pst_i}}}
      &
      \textit{by \cref{fig:normalization-pseudo}(2)}
      \\
      \multicolumn 3 r { J' = \{i \in I \sst (\C \land \C_i')
        \iff \falsek\} \neq I}
      \\
      &
      =
      {\Ipst[e_i\wedge\C'\wedge\C]{i \in (I\setminus J')\setminus J}{\smv_i}
        {\sort_i}{\nform[\C]{\nform[\C_i\wedge\C']{\pst_i}}}}
      &
      \textit{by \cref{fig:normalization-pseudo}(2)}
      \\
      \multicolumn 3 r { J = \{i \in I\setminus J' \sst (\C \land \C_i')
        \iff \falsek\} \neq I\setminus J'}
      \\
      &
      =
      {\Ipst[e_i\wedge\C'\wedge\C]{i \in (I\setminus J')\setminus J}{\smv_i}
        {\sort_i}{\nform[\C_i\wedge\C'\wedge\C]{\pst_i}}}
      &
      {\it by\ ind.\ hyp.}
      \\
      &
      =
      \nform[\C\wedge\C']{\Ipst[e_i]{i \in I}{\smv_i}{\sort_i}{\pst_i}}
      &
      \textit{by \cref{fig:normalization-pseudo}(2)}
      \\
      &
      =
      \nform[\C\wedge\C']{\pst}
      &
      \textit{by def.\ of $\pst$}
    \end{array}
    \]
    The cases in which $I =J'$ or $I\setminus J' = J$ follow
    immediately because $\nform[\C]{\nform[\C']\pst} =
    \guard[\falsek]\Tend = \nform[\C\wedge\C']{\pst}$.

  \item
    $\pst = \Epst[e_i]{i \in I}{\smv_i}{\sort_i}{\pst_i}$: It follows
    analogously to the previous one.
  \item
    $\pst = \pst_1 ; \pst_2 $: We proceed by case analysis on the
    structure of $\pst_1$. The cases $\pst_1 = \guard[\C''] \Tend$,
    $\pst_1 = \Ipst[e_i]{i \in I}{\smv_i}{\sort_i}{\pst_i}$ and
    $\pst_1 = \Epst[e_i]{i \in I}{\smv_i}{\sort_i}{\pst_i}$ follow
    analogously to the previous cases. The case for $\pst_1 =
    \pst';\pst''$ is as follows:
    \[
    \begin{array}{rll}
      \nform[\C]{\nform[\C']\pst}
      &
       =
      \nform[\C]{\nform[\C']{(\pst';\pst'');\pst_2}}
      &
      \textit{by def.\ of $\pst$}
      \\
      &
      =
      \nform[\C]{\nform[\C']{\pst';(\pst'';\pst_2)}}
      &
      \textit{by \cref{fig:normalization-pseudo}(7)}
      \\
      &
      =
      \nform[\C\wedge\C']{\pst';(\pst'';\pst_2)}
      &
      \textit{by ind.\ hyp.}
      \\
      &
      =
      \nform[\C\wedge\C']{(\pst';\pst'');\pst_2}
      &
      \textit{by \cref{fig:normalization-pseudo}(7)}
      \\
      &
       =
      \nform[\C\wedge\C']{\pst}
      &
      \textit{by def.\ of $\pst$}
    \end{array}
    \]
    The case $\pst_1 = \Tdef{(\pst')}$, is as follows
    \[
    \nform[\C]{\nform[\C']\pst}
    =
    \nform[\C]{\nform[\C']{\Tdef{(\pst')};\pst_2}}
    \]
    There are two cases, when $\nform[\C']{\pst'} =
    {\guard[\C'']\Tend}$, by \cref{fig:normalization-pseudo}(8) % chktex 36
    \[
    \nform[\C]{\nform[\C']\pst}
    = \guard[\C\wedge\C'']\Tend
    \]
    Then, the proof is completed by~\cref{lem:nf-stronger-condition}.
    Otherwise ($\nform[\C']{\pst'} \neq \guard[\C'']\Tend$), we
    proceed as follows
    \[
    \begin{array}{rll}
      \nform[\C]{\nform[\C']\pst}
      &
      =
      \nform[\C]{\nform[\C']{\Tdef{(\pst')};\pst_2}}
      &
      \textit{by def.\ of $\pst$}
      \\
      &
      =
      \nform[\C]{\nform[\C']{\Tdef{(\pst')}};\nform[\C']{\pst_2}}
      &
      \textit{by \cref{fig:normalization-pseudo}(9)}
      \\
      &
      =
      \nform[\C]{\Tdef{\nform[\C']{\pst'}};\nform[\C']{\pst_2}}
      &
      \textit{by \cref{fig:normalization-pseudo}(11)}
      \\
      &
      =
      \nform[\C]{\Tdef{\nform[\C']{\pst'}}};\nform[\C]{\nform[\C']{\pst_2}}
      &
      \textit{by \cref{fig:normalization-pseudo}(9)}
      \\
      &
      =
      \Tdef{\nform[\C]{\nform[\C']{\pst'}}};\nform[\C]{\nform[\C']{\pst_2}}
      &
      \textit{by \cref{fig:normalization-pseudo}(11)}
      \\
      &
      =
      \Tdef{\nform[\C\wedge\C']{\pst'}};\nform[\C\wedge\C']{\pst_2}
      &
      \textit{by ind.\ hyp.}
      \\
      &
      =
      \nform[\C\wedge\C']{\Tdef{(\pst')}};\nform[\C\wedge\C']{\pst_2}
      &
      \textit{by \cref{fig:normalization-pseudo}(11)}
      \\
      &
      =
      \nform[\C\wedge\C']{\Tdef{(\pst')};\pst_2}
      &
      \textit{by \cref{fig:normalization-pseudo}(9)}
      \\
      &
      =
      \nform[\C\wedge\C']{\pst}
      &
      \textit{by def.\ of $\pst$}
    \end{array}
    \]

    \item $\pst = \Tdef{\pst_1}$: Then,
     \[
     \nform[\C]{\nform[\C']\pst}
     =
     \nform[\C]{\nform[\C']{\Tdef{\pst_1}}}
     \]

     If $\nform[\C']{\Tdef{\pst_1}} =\guard[\C'']\Tend$, the proof
     follows by~\cref{lem:nf-stronger-condition}, otherwise:
    \begin{align*}
     \begin{array}{rll}
       \nform[\C]{\nform[\C']\pst}
       &
       =
       \nform[\C]{\nform[\C']{\Tdef{\pst_1}}}
       &
       \textit{by def.\ of $\pst$}
       \\
       &
       =
       \nform[\C]{\nform[\C']{\Tdef{\pst_1}}}
       &
       \textit{by \cref{fig:normalization-pseudo}(9)}
       \\
       &
       =
       \nform[\C]{\Tdef{\nform[\C']{\pst_1}}}
       &
       \textit{by \cref{fig:normalization-pseudo}(11)}
       \\
       &
       =
       \Tdef{\nform[\C]{\nform[\C']{\pst_1}}}
       &
       \textit{by \cref{fig:normalization-pseudo}(13)}
       \\
       &
       =
       \Tdef{\nform[\C\wedge\C']{\pst_1}}
       &
       {\it by\ ind.\ hyp.}
       \\
       &
       =
       \nform[\C\wedge\C']{\Tdef{\pst_1}}
       &
       \textit{by \cref{fig:normalization-pseudo}(13)}
       \\
       &
       =
       \nform[\C\wedge\C']{\pst}
       &
       \textit{by def.\ of $\pst$}
     \end{array}
     \\[\dimexpr-1.0\baselineskip+\dp\strutbox]&\qedhere
     \end{align*}
  \end{itemize}
\end{proof}

\begin{lem}\label{lem:normalisation-with-false-cond}
  If $\C \iff \falsek$ then $\nform\pst = \guard[\C']\Tend$ and $\C'
  \iff \falsek$.
\end{lem}
\begin{proof}
  By well-founded induction on $(\pst,<)$. The proof follows by case
  analysis on the structure of $\pst$.
\end{proof}

\begin{lem}\label{lem:nf-seq-composition}
 For all $\C$, $\pst,\pst'$,
 $\nform[\C]{\pst;\nform[\C]{\pst'}}=\nform[\C]{\pst;\pst'}$.
\end{lem}
\begin{proof}
  By well-founded induction on $(\pst,<)$. The proof follows by case
  analysis on the structure of $\pst$.

  \begin{itemize}
  \item $\pst = \guard[\C'] \Tend $: Then,
    \[
    \begin{array}{rll}
      \nform[\C]{\pst;\nform[\C]{\pst'}}
      &
      =
      \nform[\C]{\guard[\C']\Tend;\nform[\C]{\pst'}}
      &
      \textit{by def.\ of $\pst$}
      \\
      &
      =
      \nform[\C\wedge\C']{\nform[\C]{\pst'}}
      &
      \textit{by \cref{fig:normalization-pseudo}(1)}
      \\
      &
      =
      \nform[\C\wedge\C']{\pst'}
      &
      \textit{by \cref{lem:nf-aux}}
      \\
      &
      =
       \nform[\C]{\guard[\C']\Tend;{\pst'}}
      &
      \textit{by \cref{fig:normalization-pseudo}(1)}
      \\
      &
       =
       \nform[\C]{\pst;{\pst'}}
      &
      \textit{by def.\ of $\pst$}
    \end{array}
    \]
  \item
    $\pst = \Ipst[e_i]{i \in I}{\smv_i}{\sort_i}{\pst_i}$: Then,
    \[
    \begin{array}{rll}
      \nform[\C]{\pst;\nform[\C]{\pst'}}
      &
      =
      \nform[\C]{(\Ipst[e_i]{i \in I}{\smv_i}{\sort_i}{\pst_i});
        \nform[\C]{\pst'}}
      &
      \textit{by def.\ of $\pst$}
      \\
      &
      =
      \nform[\C]{\Ipst[e_i]{i \in I}{\smv_i}{\sort_i}{\pst_i};
        \nform[\C]{\pst'}}
      &
      \textit{by \cref{fig:normalization-pseudo}(6)}
      \\
    \end{array}
    \]

   The case in which $J = \{i \in I \sst (\C \land \C_i') \iff
   \falsek\} = I$ follows immediately. Otherwise, we proceed as
   follows
   \[
   \begin{array}{rll}
     &
     =
     {\Ipst[e_i\wedge\C]{i \in I\setminus J}{\smv_i}{\sort_i}
       {\nform[\C_i\wedge\C]{\pst_i;\nform{\pst'}}}}
     &
     \textit{by \cref{fig:normalization-pseudo}(2)}
     \\
     &
     =
     {\Ipst[e_i\wedge\C]{i \in I\setminus J}{\smv_i}{\sort_i}
       {\nform[\C_i]{\nform[\C]{\pst_i;\nform{\pst'}}}}}
     &
     \textit{by \cref{lem:nf-aux}}
     \\
     &
     =
     {\Ipst[e_i\wedge\C]{i \in I\setminus J}{\smv_i}{\sort_i}
       {\nform[\C_i]{\nform[\C]{\pst_i;{\pst'}}}}}
     &
     \textit{by ind.\ hyp.}
     \\
     &
     =
     {\Ipst[e_i\wedge\C]{i \in I\setminus J}{\smv_i}{\sort_i}
       {\nform[\C_i\wedge\C]{{\pst_i;{\pst'}}}}}
     &
     \textit{by \cref{lem:nf-aux}}
     \\
     &
     =
     \nform[\C]{\Ipst[e_i]{i \in I}{\smv_i}{\sort_i}{\pst_i};{\pst'}}
     &
     \textit{by \cref{fig:normalization-pseudo}(2)}
     \\
     &
     =
     \nform[\C]{(\Ipst[e_i]{i \in I}{\smv_i}{\sort_i}{\pst_i});{\pst'}}
     &
     \textit{by \cref{fig:normalization-pseudo}(6)}
     \\
     &
     =
     \nform[\C]{\pst;\pst'}
     &
     \textit{by def.\ of $\pst$}
   \end{array}
   \]

  \item $\pst = \Epst[e_i]{i \in I}{\smv_i}{\sort_i}{\pst_i}$: It
    follows analogously to the previous one.

  \item $\pst = \pst_1; \pst_2 $:
    \[
    \begin{array}{rll}
      \nform[\C]{\pst;\nform[\C]{\pst'}}
      &
      =
      \nform[\C]{(\pst_1 ; \pst_2);\nform[\C]{\pst'}}
      &
      \textit{by def.\ of $\pst$}
      \\
      &
      =
      \nform[\C]{\pst_1; (\pst_2;\nform[\C]{\pst'})}
      &
      \textit{by \cref{fig:normalization-pseudo}(7)}
      \\
      &
      =
      \nform[\C]{\pst_1; \nform{\pst_2;\nform[\C]{\pst'}}}
      &
      \textit{by ind.\ hyp.}
      \\
      &
      =
      \nform[\C]{\pst_1; \nform{\pst_2;{\pst'}}}
      &
      \textit{by\ ind.\ hyp.}
      \\
      &
      =
      \nform[\C]{\pst_1; ({\pst_2;{\pst'}})}
      &
      \textit{by\ ind.\ hyp.}
      \\
      &
      =
      \nform[\C]{(\pst_1; \pst_2);{\pst'}}
      &
      \textit{by \cref{fig:normalization-pseudo}(7)}
      \\
      &
       =
      \nform[\C]{\pst;{\pst'}}
      &
      \textit{by def.\ of $\pst$}
    \end{array}
    \]

  \item $\pst =\Tdef \pst_1$:
    \[
    \begin{array}{rll}
      \nform[\C]{\pst;\nform[\C]{\pst'}}
      &
      =
      \nform[\C]{\Tdef \pst_1;\nform[\C]{\pst'}}
      &
      \textit{by def.\ of $\pst$}
    \end{array}
    \]

    If $\nform{\pst} = \guard[\C']\Tend$, then the proof follows
    straightforwardly
    from~\cref{fig:normalization-pseudo}(8). Otherwise, % chktex 36
    \begin{align*}
    \begin{array}{rll}
      \\
      &
      =
      \nform[\C]{\Tdef{\pst_1}};\nform[\C]{\nform{\pst'}}
      &
      \textit{by \cref{fig:normalization-pseudo}(9)}
      \\
      &
      =
      \nform[\C]{\Tdef{\pst_1}};\nform[\C]{{\pst'}}
      &
      \textit{by \cref{lem:nf-aux}}
      \\
      &
      =
      \nform[\C]{\Tdef{\pst_1};\pst'}
      &
      \textit{by \cref{fig:normalization-pseudo}(9)}
      \\
      &
      =
      \nform[\C]{\pst;{\pst'}}
      &
      \textit{by def.\ of $\pst$}
    \end{array}
    \\[\dimexpr-1.0\baselineskip+\dp\strutbox]&\qedhere
    \end{align*}
  \end{itemize}
\end{proof}

\begin{lem}\label{lem:nf-neutral-end-seq-right}
  If $\C\implies \C'$, then $\nform {\pst;\guard[\C'] \Tend} =
  \nform{\pst}$.
\end{lem}
\begin{proof}
  By well-founded induction on $(\pst,<)$. The proof follows by case
  analysis on the structure of $\pst$.
  \begin{itemize}
  \item $\pst = \guard[\C''] \Tend $: Then,
    \[
    \begin{array}{rll}
      \nform {\pst;\guard[\C']\Tend}
      &
      =
      \nform {\guard[\C'']\Tend;\guard[\C']\Tend}
      &
      \textit{by def.\ of $\pst$}
      \\
      &
      =
      \nform[\C\land\C'']{\guard[\C'] \Tend}
      &
      \textit{by \cref{fig:normalization-pseudo}(6)}
      \\
      &
      =
      \guard[\C\land\C''\land \C']\Tend
      &
      \textit{by \cref{fig:normalization-pseudo}(1)}
      \\
      &
      =
      \guard[\C\land\C'']\Tend
      &
      \C\implies \C'
      \\
      &
      =
      \nform[\C]{\guard[\C'']\Tend}
      &
      \textit{by \cref{fig:normalization-pseudo}(6)}
      \\
      &
      =
      \nform {\pst}
      &
      \textit{by def.\ of $\pst$}
    \end{array}
    \]
  \item $\pst = \Ipst[e_i]{i \in I}{\smv_i}{\sort_i}{\pst_i}$: Then,
    \[
    \begin{array}{rll}
      \nform {\pst;\guard[\C']\Tend}
      &
      =
      \nform {(\Ipst[e_i]{i \in I}{\smv_i}{\sort_i}{\pst_i});\guard[\C']\Tend}
      &
      \textit{by def.\ of $\pst$}
      \\
      &
      =
      \nform {\Ipst[e_i]{i \in I}{\smv_i}{\sort_i}{(\pst_i;\guard[\C']\Tend)}}
      &
      \textit{by \cref{fig:normalization-pseudo}(7)}
   \end{array}
   \]
   If $I =J$, then the case follows immediately. Otherwise,
    \[
    \begin{array}{rll}
      &
      = {\Ipst[\C_i\wedge\C]{i \in I\backslash J}{\smv_i}{\sort_i}
        {\nform[\C_i\wedge\C]{\pst_i;\guard[\C']\Tend}}}
      &
      \textit{by \cref{fig:normalization-pseudo}(4)}
      \\
      &
      = {\Ipst[\C_i\wedge\C]{i \in I\backslash J}{\smv_i}{\sort_i}
        {\nform[\C_i\wedge\C]{\pst_i}}}
      &
      \textit{by ind.\ hyp.}
      \\
      &
      =
      \nform {\Ipst[e_i]{i \in I}{\smv_i}{\sort_i}{\pst_i}}
      &
      \textit{by \cref{fig:normalization-pseudo}(4)}
      \\
      &
      =
      \nform[\C]{\pst}
      &
      \textit{by def.\ of $\pst$}
    \end{array}
    \]
  \item $\pst = \Epst[e_i]{i \in I}{\smv_i}{\sort_i}{\pst_i}$: It
    follows analogously to the previous one.
  \item $\pst = \pst_1 ; \pst_2 $: Then,
    \[
    \begin{array}{rll}
      \nform {\pst;\guard[\C']\Tend}
      &
      =
      \nform {(\pst_1;\pst_2);\guard[\C']\Tend}
      &
      \textit{by def.\ of $\pst$}
      \\
      &
      =
      \nform {\pst_1;(\pst_2;\guard[\C']\Tend)}
      &
      \textit{by \cref{fig:normalization-pseudo}(9)}
      \\
      &
      =
      \nform {\pst_1; \nform{\pst_2;\guard[\C']\Tend}}
      &
      \textit{by \cref{lem:nf-seq-composition}}
      \\
      &
      =
      \nform {\pst_1; \nform{\pst_2}}
      &
      \textit{by ind.\ hyp.}
      \\
      &
      =
      \nform {\pst_1; \pst_2}
      &
      \textit{by \cref{fig:normalization-pseudo}(9)}
      \\
      &
      =
      \nform[\C]{\pst}
      &
      \textit{by def.\ of $\pst$}
    \end{array}
    \]
  \item $\pst = \Tdef{\pst_1}$: Then,
    \begin{align*}
    \begin{array}{rll}
      \nform {\pst;\guard[\C']\Tend}
      &
      =
      \nform {\Tdef{\pst_1};\guard[\C']\Tend}
      &
      \textit{by def.\ of $\pst$}
      \\
      &
      =
      \nform {\Tdef{\pst_1}};\nform{\guard[\C']\Tend}
      &
      \textit{by \cref{fig:normalization-pseudo}(11)}
      \\
      &
      =
      \nform {\Tdef{\pst_1}};\nform[\C\land \C']{\Tend}
      &
      \textit{by \cref{fig:normalization-pseudo}(6)}
      \\
      &
      =
      \nform {\Tdef{\pst_1}};\nform{\Tend}
      &
      \C\implies \C'
      \\
      &
      =
      \nform {\Tdef{\pst_1};\Tend}
      &
      \textit{by \cref{fig:normalization-pseudo}(11)}
      \\
      &
      =
      \nform[\C]{\pst; \Tend}
      &
      \textit{by def.\ of $\pst$}
      \\
      &
      =
      \nform[\C]{\pst}
      &
    \end{array}
    \\[\dimexpr-1.0\baselineskip+\dp\strutbox]&\qedhere
    \end{align*}
  \end{itemize}
\end{proof}

\begin{lem}\label{lem:nf-neutral-end-seq}
  Equality $\nform {\pst;(\guard[\C'] \Tend; \pst')} =
  \nform{\pst;\pst'}$ holds for all $\C$ and $\C'$ such that $\neg
  (\C\iff\falsek)$ and $\C \implies \C'$.
\end{lem}
\begin{proof}
  By well-founded induction on $(\pst,<)$. The proof follows by case
  analysis on the structure of $\pst$.
  \begin{itemize}
  \item $\pst = \guard[\C''] \Tend $: Then,
    \[
    \begin{array}{rll}
      \nform {\pst;(\guard[\C']\Tend; \pst')}
      &
      =
      \nform {\guard[\C''] \Tend;(\guard[\C']\Tend; \pst')}
      &
      \textit{by def.\ of $\pst$}
      \\
      &
      =
      \nform[\C\land\C'']{\guard[\C'] \Tend; \pst'}
      &
      \textit{by \cref{fig:normalization-pseudo}(6)}
      \\
      &
      =
      \nform[\C\land\C'' \land \C']{\pst'};
      &
      \textit{by \cref{fig:normalization-pseudo}(6)}
      \\
      &
      =
      \nform[\C\land\C'']{\pst'};
      &
      \C\land\C'' \land \C' \iff \C\land\C''
      \\
      &
      =
      \nform[\C]{\guard[\C'']  \Tend; \pst'}
      &
      \textit{by \cref{fig:normalization-pseudo}(6)}
      \\
      &
      =
      \nform[\C]{\pst; \pst'}
      &
      \textit{by def.\ of $\pst$}
    \end{array}
    \]
  \item $\pst = \Ipst[e_i]{i \in I}{\smv_i}{\sort_i}{\pst_i}$: Then,
    \[
    \begin{array}{rll}
      \nform {\pst;(\guard[\C']\Tend; \pst')}
      &
      =
      \nform {(\Ipst[e_i]{i \in I}{\smv_i}{\sort_i}{\pst_i});
        (\guard[\C'] \Tend; \pst')}
      &
      \textit{by def.\ of $\pst$}
      \\
      &
      =
      \nform {\Ipst[e_i]{i \in I}{\smv_i}{\sort_i}{(\pst_i;
          (\guard[\C'] \Tend; \pst'))}}
      &
      \textit{by \cref{fig:normalization-pseudo}(7)}
    \end{array}
    \]
    If $I =J$, then the case follows immediately. Otherwise,
    \[
    \begin{array}{rll}
      &
      = {\Ipst[\C_i\wedge\C]{i \in I\backslash J}{\smv_i}{\sort_i}
        {\nform[\C_i\wedge\C]{(\pst_i;(\guard[\C'] \Tend; \pst'))}}}
      &
      \textit{by \cref{fig:normalization-pseudo}(4)}
      \\
      &
      = {\Ipst[\C_i\wedge\C]{i \in I\backslash J}{\smv_i}{\sort_i}
        {\nform[\C_i\wedge\C]{\pst_i;\pst'}}}
      &
      {\it by\ ind.\ hyp.}\
      \\
      &
      =
      \nform {\Ipst[e_i]{i \in I}{\smv_i}{\sort_i}{(\pst_i; \pst')}}
      &
      \textit{by \cref{fig:normalization-pseudo}(4)}
      \\
      &
      =
      \nform {(\Ipst[e_i]{i \in I}{\smv_i}{\sort_i}{\pst_i}); \pst'}
      &
      \textit{by \cref{fig:normalization-pseudo}(7)}
      \\
      &
      =
      \nform[\C]{\pst; \pst'}
      &
      \textit{by def.\ of $\pst$}
    \end{array}
    \]
  \item $\pst = \Epst[e_i]{i \in I}{\smv_i}{\sort_i}{\pst_i}$: It
    follows analogously to the previous one.
  \item $\pst = \pst_1 ; \pst_2 $: Then,
    \[
    \begin{array}{rll}
      \nform {\pst;(\guard[\C']\Tend; \pst')}
      &
      =
      \nform {(\pst_1;\pst_2);(\guard[\C']\Tend; \pst')}
      &
      \textit{by def.\ of $\pst$}
      \\
      &
      =
      \nform {\pst_1; ({\pst_2;(\guard[\C']\Tend; \pst')})}
      &
      \textit{by \cref{fig:normalization-pseudo}(9)}
      \\
      &
      =
      \nform {\pst_1; \nform{\pst_2;(\guard[\C']\Tend; \pst')}}
      &
      \textit{by \cref{lem:nf-seq-composition}}
      \\
      &
      =
      \nform {\pst_1;\nform{\pst_2;\pst'}}
      &
      {\it by\ ind.\ hyp.}
      \\
      &
      =
      \nform {\pst_1;({\pst_2;\pst'})}
      &
      \textit{by \cref{lem:nf-seq-composition}}
      \\
      &
      =
      \nform {(\pst_1; \pst_2); \pst'}
      &
      \textit{by \cref{fig:normalization-pseudo}(9)}
      \\
      &
      =
      \nform[\C]{\pst; \pst'}
      &
      \textit{by def.\ of $\pst$}
    \end{array}
    \]
  \item $\pst = \Tdef{\pst_1}$: Then,
    \begin{align*}
    \begin{array}{rll}
      \nform {\pst;(\guard[\C']\Tend; \pst')}
      &
      =
      \nform {\Tdef{\pst_1};(\guard[\C']\Tend; \pst')}
      &
      \textit{by def.\ of $\pst$}
      \\
      &
      =
      \nform {\Tdef{\pst_1}};\nform{\guard[\C']\Tend; \pst'}
      &
      \textit{by \cref{fig:normalization-pseudo}(11)}
      \\
      &
      =
      \nform {\Tdef{\pst_1}};\nform[\C\land \C']{\pst'}
      &
      \textit{by \cref{fig:normalization-pseudo}(6)}
      \\
      &
      =
      \nform {\Tdef{\pst_1}};\nform{\pst'}
      &
      \C\implies \C'
      \\
      &
      =
      \nform {\Tdef{\pst_1};\pst'}
      &
      \textit{by \cref{fig:normalization-pseudo}(11)}
      \\
      &
      =
      \nform[\C]{\pst; \pst'}
      &
      \textit{by def.\ of $\pst$}
    \end{array}
    \\[\dimexpr-1.0\baselineskip+\dp\strutbox]&\qedhere
    \end{align*}
  \end{itemize}
\end{proof}

\begin{lem}\label{lem:sr-merge-normalization}
  $\nform{\pst\peace\pst'} = \nform{\pst}\peace\nform{\pst'}$.
\end{lem}
\begin{proof}
  By straightforward induction on the derivation of $\pst\peace\pst'$.
\end{proof}

\subsection{Typing}
We write $\varpst{\pst}$ to denote the variables appearing in the
expressions occurring on the pseudo-type $\pst$.  It is
straightforwardly extended to specifications $\varpst{\envmv}$.

\begin{lem}\label{lem:no-new-var-in-types}
  If $\pjdg[@][@][S]$, then $\var \envmv\subseteq \var \C \cup \fX[S]
  \cup \bX[S]$.
\end{lem}
\begin{proof}
  By straightforward structural induction on the typing judgment.
\end{proof}

\begin{lem}\label{lem:typing-weakening}
  If $\pjdg$ then $\pjdg[@][@][@][@][\Gamma, x : \sort]$.
\end{lem}
\begin{proof}
  By induction on the structure of the proof $\pjdg$.
\end{proof}

\begin{lem}\label{lem:typing:and:fn}
  If $\pjdg$ then
  \begin{itemize}
  \item $\fX[P] \cup \fY[P] \cup \fX[\C] \subseteq \Gamma$
  \item $\forall \smv \in \fY[P] \qst \smv \in \dom \envmv$.
  \end{itemize}
\end{lem}
\begin{proof}
  By induction of the derivation of $\pjdg$ and inspecting the rules
  in \cref{fig:typingproc} and~\ref{fig:typingsys}.
\end{proof}

\begin{lem}\label{lem:typing-preservation-structural}
  If $\pjdg$ and $P\equiv Q$, then $\pjdg[@][\envmv'][Q]$ and
  $\nfp\envmv = \nfp{\envmv'}$.
\end{lem}
\begin{proof} By induction on the proof of $P\equiv Q$.
  \begin{itemize}
  \item $P;(Q;R) \equiv (P;Q);R$ follows by inductive hypothesis and
    associativity of $;$ over pseudo types % chktex 40
    (\cref{fig:normalization-pseudo} (9)).
  \item $\Pend; P \equiv P; \Pend \equiv P$ follows by using
    \lemref{lem:nf-neutral-end-seq} and
    \lemref{lem:nf-neutral-end-seq-right}.
    \qedhere
  \end{itemize}
\end{proof}

\begin{lem}\label{lem:typing-with-false-env}
  If $\pjdg[@][\envmv, \psnames \smv p: \pst]$ and $\C\iff\falsek$,
  then $\nfp \pst = \guard[\falsek]\Tend$.
\end{lem}
\begin{proof}
  Follows by induction on the structure of $P$. The interesting cases
  are
  \begin{itemize}
  \item \myrule{VSeq}: Then, $P = P_1;P_2$, $\pjdg[@][\envmv_1][P_1]$,
    and $\pjdg[@][\envmv_2][P_2]$ and $\envmv_1;\envmv_2$ defined. By
    inductive hypothesis, if $\psnames \smv p \in\dom{\envmv_i}$ then
    $\nfp {\envmv_i \psnames \smv p} = \guard[\falsek]\Tend$ for
    $i=1,2$. Then, the result follows by the definition of $;$ and % chktex 40
    rule (4) in~\cref{fig:normalization-pseudo}.

  \item \myrule{VIf}: Then, $P = \Pif{\C'}{P_1}{P_2}$,
    $\pjdg[\C\wedge\C'][\envmv_1][P_1]$, and
    $\pjdg[\C\wedge\neg\C'][\envmv_2][P_2]$ and
    $\envmv_1\peace\envmv_2$.  By inductive hypothesis, if $\psnames
    \smv p \in\dom{\envmv_i}$ then $\nfp {\envmv_i \psnames \smv p} =
    \guard[\falsek]\Tend$ for $i=1,2$. Then, the results follows by
    definition of $\peace$.
    \qedhere
  \end{itemize}
\end{proof}

\begin{lem}\label{thm:distribution-end-pseudo-types}
  If $\pjdg[@][\envmv]$, then for all $\C'$ there is $\envmv'$ such
  that:
  \begin{itemize}
  \item $\dom\envmv = \dom{\envmv'}$, and
  \item $\pjdg[\C\land\C'][\envmv']$, and
  \item for all $\psnames \smv p\in\dom\envmv$, $\envmv' \psnames \smv p
  = \guard[\C\land \C']\Tend;\envmv \psnames \smv p$.
  \end{itemize}
\end{lem}
\begin{proof}
  Let $\envmv\psnames \smv p = \pst$ and $\envmv'\psnames \smv p=
  \pst'$. We show that $\nfp{\pst'}= \nfp{\guard[\C\land
      \C']\Tend;\pst}$ by induction on the structure of the proof for
  the judgment $\pjdg[@][\envmv, \psnames \smv p: \pst]$. We first
  assume that $\neg (\C\land \C' \iff \falsek)$, and show
  \begin{itemize}
  \item \myrule{VReq}, \myrule{VAcc}: Follow by inductive hypothesis.
  \item \myrule{VRcv}: Then, $\pst = \displaystyle{\sum_{i\in I}
    \Tbra{\guard\smv_i}{\sort_{\mbox{\scriptsize$i$}}}{\pst_i}}$, and
    $\forall i \in I \qst \pjdg[@][\envmv, \psnames \smv p :
    {\pst_i}][P_i][@][\Gamma, x_i : \sort_i]$.
    Similarly,
    $\pst' = \displaystyle{\sum_{i\in I}
      \Tbra{\guard[\C\land\C']{\smv_i}}{\sort_{\mbox{\scriptsize$i$}}}{\pst_i'}}$
    and
    $\forall i \in I \qst
    \pjdg[\C\land\C'][\envmv', \psnames \smv p : {\pst_i'}][P_i][@]
         [\Gamma, x_i : \sort_i]$
    \[
    \begin{array}{rll}
      \nform[\truek] {\pst'}
      &
      =
      \displaystyle{\sum_{i\in I}
        \Tbra{\guard[\C\land\C']{\smv_i}}{\sort_{\mbox{\scriptsize$i$}}}{\pst_i'}}
      &
      \textit{by def.\ of $\pst'$}
      \\
      &
      =
      \displaystyle{\sum_{i\in I}
        \Tbra{\guard[\C\land\C']\smv_i}{\sort_{\mbox{\scriptsize$i$}}}
             {\nform[\C\land\C']{\pst_i'}}}
      &
      \textit{by \cref{fig:normalization-pseudo}(5)}
      \\
      &
      =
      \displaystyle{\sum_{i\in I}
        \Tbra{\guard[\C\land\C']\smv_i}{\sort_{\mbox{\scriptsize$i$}}}
             {\nform[\C\land\C']{\guard[\C\land\C']\Tend;\pst_i}}}
      &
      {\it by\ ind.\ hyp.}\
      \\
      &
      =
      \displaystyle{\sum_{i\in I}
        \Tbra{\guard[\C\land\C']\smv_i}{\sort_{\mbox{\scriptsize$i$}}}
             {\nform[\C\land\C']{\pst_i}}}
      &
      \textit{by \cref{fig:normalization-pseudo}(6)}
      \\
      &
      =
      \nform[\C\land\C']{\displaystyle{\sum_{i\in I}
          \Tbra{\guard\smv_i}{\sort_{\mbox{\scriptsize$i$}}}{{\pst_i}}}}
      &
      \textit{by \cref{fig:normalization-pseudo}(5)}
      \\
      &
      =
      \nform[\C\land\C']\pst
      &
      \textit{by def.\ of $\pst$}
      \\
      &
      =
      \nform[\truek] {\guard[\C\land\C']\Tend;\pst}
      &
      \textit{by \cref{fig:normalization-pseudo}(6)}
    \end{array}
    \]
  \item \myrule{VSend}: Then,
    $\pst = \guard{\Tsend \smv \sort};\guard\Tend$ and
    $\pst' = \guard[\C\land\C']{\Tsend \smv
      \sort};\guard[\C\land\C']\Tend$.
    \[
    \begin{array}{rll}
      \nform[\truek]{\pst'}
      & =
      \nform[\truek]{\guard[\C\land\C']{\Tsend \smv \sort};
        \guard[\C\land\C']{\Tend}}
      &
      \textit{by def.\ of $\pst'$}
      \\
      &
      =
      \nform[\truek]{\guard[\C\land\C']{\Tsend \smv \sort};
        \nform[\C\land\C']{\guard\Tend}}
      &
      \textit{by \cref{fig:normalization-pseudo}(6)}
      \\
      &
      =
      \nform[\C\land\C']{\guard[\C]{\Tsend \smv \sort};{\guard\Tend}}
      &
      \textit{by \cref{fig:normalization-pseudo}(4)}
      \\
      &
      =
      \nform[\C\land\C']{\pst}
      &
      \textit{by def.\ of $\pst$}
      \\
      &
      =
      \nform[\truek]{\guard[\C\land\C']\Tend;\pst}
      &
      \textit{by \cref{fig:normalization-pseudo}(6)}
    \end{array}
    \]
  \item \myrule{VEnd}: Then, $\pst = \guard\Tend$ and
    $\pst' = \guard[\C\land\C']\Tend$. By normalising,
    \[
    \begin{array}{rll}
      \nform[\truek]{\pst'}
      & =
      \nform[\truek]{\guard[\C\land\C']\Tend}
      &
      \textit{by def.\ of $\pst'$}
      \\
      &
      =
      \nform[\C\land\C']{\guard\Tend};
      &
      \textit{by \cref{fig:normalization-pseudo}(1)}
      \\
      &
      =
      \nform[\C\land\C']{\pst};
      &
      \textit{by def.\ of $\pst'$}
      \\
      &
      =
      \nform[\truek]{\guard[\C\land\C']\Tend;\pst};
      &
      \textit{by \cref{fig:normalization-pseudo}(6)}
    \end{array}
    \]
  \item \myrule{VSeq} There are two cases:
    \begin{itemize}
    \item $\psnames \smv i\in (\dom{\envmv_1}\cup\dom{\envmv_2})$:
      Then $\pst= \pst_1;\pst_2$ and $\pst'= \pst_1';\pst_2'$. By
      inductive hypothesis,
      $\pst'_1= \guard[\C\land \C']\Tend;{\pst_1}$ and
      $\pst'_2= \guard[\C\land \C']\Tend;{\pst_2}$.
      \[
      \begin{array}{rll}
        \nform[\truek]{\pst'}& =
        \nform[\truek]{\pst'_1;\pst'_2}
        &
        \textit{by def.\ of $\pst'$}
        \\
        &
        =
        \nform[\truek]{(\guard[\C\land \C']\Tend;{\pst_1});
          (\guard[\C\land \C']\Tend;{\pst_2})}
        &
        {\it by\ ind.\ hyp.}
        \\
        &
        =
        \nform[\truek]{\guard[\C\land \C']\Tend;
          ({\pst_1};(\guard[\C\land \C']\Tend;{\pst_2}))}
        &
        \textit{by \cref{fig:normalization-pseudo}(6)}
        \\
        &
        =
        \nform[\C\land \C']{{\pst_1};(\guard[\C\land \C']\Tend;{\pst_2})}
        &
        \textit{by \cref{fig:normalization-pseudo}(9)}
        \\
        &
        =
        \nform[\C\land \C']{{\pst_1};{\pst_2}}
        &
        \textit{by \lemref{lem:nf-neutral-end-seq}}
        \\
        &
        =
        \nform[\C\land \C']{{\pst}}
        &
        \textit{by def.\ of $\pst$}
        \\
        &
        =
        \nform[\truek]{\guard[\C\land\C']\Tend;\pst}
        &
        \textit{by \cref{fig:normalization-pseudo}(6)}
      \end{array}
      \]
    \item $\psnames \smv i \not\in \dom{\envmv_2}$: Then
      $\pst= \pst_1$ and $\pst'= \pst_1'$. By inductive hypothesis,
      $\pst'_1= \guard[\C\land \C']\Tend;{\pst_1}$. Hence,
      $\pst' = \guard[\C\land \C']\Tend;{\pst}$.
    \end{itemize}
  \item \myrule{VIf}: Then $\pst = \pst_1\peace\pst_2$ and
    $\pst' = \pst_1'\peace\pst_2'$.
    \[
    \begin{array}{rll}
      \nform[\truek]{\pst'}& =
      \nform[\truek]{\pst_1'\peace\pst_2'}
      &
      \textit{by def.\ of $\pst'$}
      \\
      &
      =
      \nform[\truek]{\pst_1'}\peace\nform[\truek]{\pst_2'}
      &
      \textit{by \cref{lem:sr-merge-normalization}}
      \\
      &
      =
      \nform[\truek]{(\guard[\C\land \C']\Tend;
        {\pst_1})}\peace\nform[\truek]{{(\guard[\C\land \C']\Tend;{\pst_2})}}
      &
      {\it by\ ind.\ hyp.}
      \\
      &
      =
      \nform[\C\land \C']{\pst_1}\peace\nform[\C\land \C']{\pst_2}
      &
      \textit{by \cref{fig:normalization-pseudo}(6)}
      \\
      &
      =
      \nform[\C\land \C']{\pst_1'\peace\pst_2'}
      &
      \textit{by \cref{lem:sr-merge-normalization}}
      \\
      &
      =
      \nform[\C\land \C']{{\pst}}
      &
      \textit{by def.\ of $\pst$}
      \\
      &
      =
      \nform[\truek]{\guard[\C\land\C']\Tend;\pst}
      &
      \textit{by \cref{fig:normalization-pseudo}(6)}
    \end{array}
    \]
  \item \myrule{VFor}: Then $\pst = \Tdef{(\pst_1)}$,
    $\pst' = \Tdef{(\pst_1')}$. Then,
    \[
    \begin{array}{rll}
      \nform[\truek]{\pst'}& =
      \nform[\truek]{ \Tdef{(\pst_1')}}
      &
      \textit{by def.\ of $\pst'$}
      \\
      &
      =
      \Tdef{\nform[\truek]{\pst_1'}}
      &
      \textit{by \cref{fig:normalization-pseudo}(13)}
      \\
      &
      =
      \Tdef{\nform[\truek]{\guard[\C\land \C']\Tend;{\pst_1}}}
      &
      {\it by\ ind.\ hyp.}
      \\
      &
      =
      \Tdef{\nform[\C\land \C']{\pst_1}}
      &
      \textit{by \cref{fig:normalization-pseudo}(6)}
      \\
      &
      =
      \nform[\C\land \C']{\Tdef{(\pst_1)}}
      &
      \textit{by \cref{fig:normalization-pseudo}(13)}
      \\
      &
      =
      \nform[\C\land \C']{{\pst}}
      &
      \textit{by def.\ of $\pst$}
      \\
      &
      =
      \nform[\truek]{\guard[\C\land\C']\Tend;\pst}
      &
      \textit{by \cref{fig:normalization-pseudo}(6)}
    \end{array}
    \]
  \item \myrule{VForEnd}: It follows analogously to the case
    \myrule{VEnd}.
  \item \myrule{Vloop}: Since ${\envmv_1}$ and ${\envmv_2}$ are
    passively compatible, then $\dom{\envmv_1} = \dom {\envmv_2}$,
    then $\psnames \smv i\in (\dom{\envmv_1}\cap\dom{\envmv_2})$.
    Consequently, $\pst = \Tdef{(\pst_1)};\pst_2$, $\pst' =
    \Tdef{(\pst_1')};\pst_2$.  By inductive hypothesis, $\pst'_1=
    \guard[\C\land \C']\Tend;{\pst_1}$ and $\pst'_2= \guard[\C\land
      \C']\Tend;{\pst_2}$.
    \[
    \begin{array}{rll}
      \nform[\truek]{\pst'}& =
      \nform[\truek]{\Tdef{(\pst'_1)};\pst'_2}
      &
      \textit{by def.\ of $\pst'$}
      \\
      &
      =
      \nform[\truek]{\Tdef{(\guard[\C\land \C']\Tend;\pst_1)};
        (\guard[\C\land \C']\Tend;{\pst_2})}
      &
      {\it by\ ind.\ hyp.}
      \\
      &
      =
      \nform[\truek]{\Tdef{(\guard[\C\land \C']\Tend;\pst_1)}};
      \nform[\truek]{\guard[\C\land \C']\Tend;{\pst_2}}
      &
      \textit{by \cref{fig:normalization-pseudo}(11)}
      \\
      &
      =
      \Tdef{\nform[\truek]{(\guard[\C\land \C']\Tend;\pst_1)}};
      \nform[\truek]{\guard[\C\land \C']\Tend;{\pst_2}}
      &
      \textit{by \cref{fig:normalization-pseudo}(13)}
      \\
      &
      =
      \Tdef{\nform[\C\land \C']{\pst_1}};\nform[\C\land \C']{\pst_2}
      &
      \textit{by \cref{fig:normalization-pseudo}(6)}
      \\
      &
      =
      \nform[\C\land \C']{\Tdef{\pst_1}};\nform[\C\land \C']{\pst_2}
      &
      \textit{by \cref{fig:normalization-pseudo}(13)}
      \\
      &
      =
      \nform[\C\land \C']{\Tdef{(\pst_1)};\pst_2}
      &
      \textit{by \cref{fig:normalization-pseudo}(11)}
      \\
      &
      =
      \nform[\C\land \C']{{\pst}}
      &
      \textit{by def.\ of $\pst$}
      \\
      &
      =
      \nform[\truek]{\guard[\C\land\C']\Tend;\pst}
      &
      \textit{by \cref{fig:normalization-pseudo}(6)}
    \end{array}
    \]
  \end{itemize}
  If $\C\land \C' \iff \falsek$, we proceed as follows
  \begin{align*}
  \begin{array}{rll}
    \nform[\truek] {\guard[\C\land\C']\Tend; \pst }
    & =
    \nform[\C\land\C'] \pst
    &
    \textit{by \cref{fig:normalization-pseudo}(6)}
    \\
    &
    =
    \guard[\falsek]\Tend
    &
    \textit{by \lemref{lem:normalisation-with-false-cond}}
    \\
    &
    =
    \pst'
    &
    \textit{by \lemref{lem:typing-with-false-env}}
    \end{array}
    \\[\dimexpr-1.0\baselineskip+\dp\strutbox]&\qedhere
  \end{align*}
\end{proof}

\begin{lem}\label{lem:typing-strenghten-condition}
  If $\pjdg[\C][\envmv][P][@][@]$, then
  \begin{itemize}
  \item
    $\pjdg[\C\land\C'][\envmv'][P][@][@]$; and
  \item for all $\psnames \smv p \in\dom {\envmv'}$, $\envmv' \psnames
    \smv p = \nform[\C']{\envmv \psnames \smv p}$
  \end{itemize}
\end{lem}

\begin{proof}
  Directly from \lemref{thm:distribution-end-pseudo-types}.
\end{proof}

%%% Local Variables:
%%% mode: latex
%%% TeX-master: "main"
%%% End:

% !TEX root = main.tex
\subsection{Consistency}

\begin{lem}\label{lem:consistency-preservation-under-store-extension}
  If $\consPst[\pst]$, then $\consPst[@][\sigma']$ for any $\sigma'$
  such that $\sigma(x) = \sigma'(x)$ for all $x\in \varpst{\pst}$.
\end{lem}
\begin{proof}
  By structural induction on the structure of $\consPst$ and by
  noticing that $\eval \C {(\sigma \cap \sigma')}= \eval \C \sigma$
  for every expression $\C$ in $\pst$.
\end{proof}

\begin{lem}\label{lem:const-pseudo-extended-true-guard}
  If $x \not\in \var \pst$ and $\consPst$ then
  $\consPst[@][\sigma|_{\dom \sigma \setminus \{x\}}]$.
\end{lem}
\begin{proof}
  It follows
  from~\cref{lem:consistency-preservation-under-store-extension}.
\end{proof}

\begin{lem}\label{lem:const-merge-if}
  If\ $\pst \peace \pst'$ is defined and either $\consPst$ or
  $\consPst[\pst']$, then $\consPst[(\pst\peace\pst')]$.
\end{lem}
\begin{proof}
  By induction on the structure of $\pst$.
\end{proof}

\begin{lem}\label{lem:const-pseudo-typing-and-store-correspondence}
  Let $\sigma$, $\envmv$, $\C$, and $\Gamma$ be such that conditions~\eqref{cons2:eq}--\eqref{cons4:eq} in \defref{def:consistent} hold.
  Then
  \[\pjdg \quad\implies\quad \forall \psnames \smv p \in \dom{\envmv}
  \qst \consPst[\envmv \psnames \smv p]\]
\end{lem}
\begin{proof}
  By structural induction on the derivation of the typing
  judgement. We proceed by case analysis on the last rule applied in
  the derivation of the judgment $\pjdg$.
  \begin{itemize}
  \item \myrule{VReq} The thesis directly follows from the inductive
    hypothesis.
  \item \myrule{VAcc} The thesis directly follows from the inductive
    hypothesis.
  \item \myrule{VRcv} Then,
    $\envmv = \envmv' , \psnames \smv p : \displaystyle{\sum_{i\in
        I}
      \Tbra{\guard\smv_i}{\sort_{\mbox{\scriptsize$i$}}}{\pst_i}}$.
    The inductive hypothesis applied to the premiss of~\myrule{VRcv}
    implies that $\consPst[\pst_i][\sigma,x_i\mapsto \val {v_i}]$ for
    all $i \in I$ and that
    $\consPst[\envmv' \psnames{\smv'} q][\sigma,x_i\mapsto \val
    {v_i}]$.
    Since $\eval {\C } \sigma = \truek$, then
    $\eval {\C } \sigma, x_i \mapsto \val {v_i} = \truek$ for all
    $i\in I$.  By inductive hypothesis on all the premiss, we
    conclude that $\consPst[\pst_i][\sigma,x_i\mapsto \val {v_i}]$.
    By \lemref{lem:const-pseudo-extended-true-guard},
    $\consPst[\pst_i][\sigma]$ for all $i\in I$, so we can apply
    rule~\myrule{CRcv} and obtain the thesis.
  \item \myrule{VSend} Then,
    $\envmv = \envmv' , \psnames \smv p: \guard{\Tsend \smv
      \sort};\guard\Tend$
    and, by the premiss of~\myrule{VSend},
    $\envmv' \psnames {\smv'} q = \guard\Tend$ and
    $\consPst[\guard \Tend]$ by rule~\myrule{CRcv}.
    We just need to prove that $\consPst[\guard{\Tsend \smv \sort}]$
    so to apply rule~\myrule{CSeq} and obtain the thesis.
    This is indeed the case as $\eval {\C } \sigma = \truek$ by
    assumption~\eqref{cons4:eq} of \defref{def:consistent} which
    implies $\consExp$ and so, by rule~\myrule{CSend} we conclude the
    proof.
  \item\myrule{VEnd} Then, $\psnames \smv p \in \dom \envmv$ and
    $\envmv\psnames \smv p = \guard\Tend$. Moreover, $\eval {\C }
    \sigma = \truek$ by assumption. Therefore, $\consPst[\guard
      \Tend]$ by rule~\myrule{CRcv}.
  \item\myrule{VSeq} We have
    \[ \vseqrule{} \quad \text{with } \envmv = \envmv_1;\envmv_2\]
    For all $\psnames \smv p \notin \dom{\envmv_1} \cap
    \dom{\envmv_2}$ the thesis follows directly by the inductive
    hypothesis on one of the premiss of the rule above.
    If $\psnames \smv p \in \dom{\envmv_1} \cap \dom{\envmv_2}$ then
    $\envmv \psnames \smv p = \envmv_1\psnames \smv p ;
    \envmv_2\psnames \smv p$ and, by inductive hypothesis, both
    $\consPst[\envmv_1\psnames \smv p]$ and $\consPst[\envmv_2\psnames
      \smv p]$ hold. Hence, the thesis follows by rule~\myrule{CSeq}.
  \item\myrule{VIf} We have
    \[\vifrule{} \quad \text{with } \envmv = \envmv_1 \peace \envmv_2\]
    Since $\eval {\C } \sigma = \truek$, we have either
    $\eval {(\C\land\C')} \sigma = \truek$ or
    $\eval {(\C\land\neg \C')} \sigma = \truek$; hence, we can apply
    the inductive hypothesis to one of the premisise of the rule
    above.
    So, we have that $\consPst[\pst_1]$ or $\consPst[\pst_2]$ and
    $\consPst[\pst_1\peace \pst_2]$ hold by
    \lemref{lem:const-merge-if}.
  \item\myrule{VFor} We have that $\envmv = \envmv_1^*$ for some
    specification $\envmv_1$ such that
    $\pjdg[\C \land x \in \ell][\envmv_1][@][@][x : \sort,\Gamma]$.
    By the inductive hypothesis, we have
    $\consPst[\envmv_1 \psnames \smv p][\sigma,x\mapsto \val {v}]$ for
    all $\psnames \smv p \in \dom{\envmv_1} = \dom \envmv$.
    The thesis then follows by
    \lemref{lem:const-pseudo-extended-true-guard} and applying rule~\myrule{CLoop}.
  \item \myrule{VForEnd} Similarly to the case~\myrule{VEnd}.
  \item \myrule{Vloop} We have that $\envmv = \envmv_1^*;\envmv_2$
    for two specifications $\envmv_1$ and $\envmv_2$ such that
    \[ \vlooprule{} \]
    For all
    $\psnames \smv p \in \dom{\envmv_1} \cup \dom{\envmv_2} = \dom
    \envmv$,
    we have $\consPst[\psnames \smv p]$ by the inductive hypothesis.
    The thesis follows by using rules~\myrule{CLoop} and~\myrule{CSeq}.\qedhere
  \end{itemize}
\end{proof}

\begin{lem}\label{lem:red-active-process}
  If $\pjdg[\C][\envmv][@][@][\Gamma]$, $\consSt[@][P]$ and
  $\envmv$ active, then $\state[P]\Ptrans{\alpha}\state[P'][\sigma']$
  implies $\alpha \neq \lreceive{\smv}{\sort}$.
\end{lem}
\begin{proof} By induction on the structure of the typing judgement.
  \begin{itemize}
  \item \myrule{VReq}: Then, $P = \Preq{\shname}{n}{\tuple
      \smv}{P'}$.
    By inspecting reduction rules, the only possibility is
    $\alpha = \lreq{\shname}{n}{\tuple \smv}$.
  \item \myrule{VAcc}: Then, $P = \Pacc{\shname}{i}{\tuple
      \smv}{P'}$.
    By inspecting reduction rules, the only possibility is
    $\alpha = \lacc \shname i {\tuple \smv}$.
  \item \myrule{VSend}: Then, $P = \Psend{\smv}{e'}$. By inspecting
    reduction rules, the only possibility is
    $\alpha = \lsend \smv {\val{v}}$.
  \item \myrule{VEnd}: Then, $P = \Pend$. Hence, there is no
    reduction.
  \item \myrule{VSeq}: Then, $P = P_1;P_2$,
    $\envmv=\envmv_1;\envmv_2$, and $\pjdg[@][\envmv_1][P_1]$.  Note
    that $\state[P]\Ptrans{\alpha}\state[P'][\sigma']$ implies
    $\state[P_1]\Ptrans{\alpha}\state[P'_1][\sigma']$. Moreover,
    $\envmv$ active implies $\envmv_1$ active. Then, by inductive
    hypothesis, $\alpha\neq \lsend{\smv}{\sort}$.
  \item \myrule{VIf}: Then, $P = \Pif{\C'}{P_1}{P_2}$,
    $\envmv=\envmv_1\peace\envmv_2$. Moreover,
    $\state[P]\Ptrans{\alpha}\state[P'][\sigma']$ implies
    $\state[P_1]\Ptrans{\alpha}\state[P'_1][\sigma']$ or
    $\state[P_1]\Ptrans{\alpha}\state[P'_1][\sigma']$. Moreover,
    $\envmv$ active implies $\envmv_1$ and $\envmv_2$ active. Then, by
    inductive hypothesis, $\alpha\neq \lsend{\smv}{\sort}$.
  \item \myrule{VRcv},~\myrule{VEnd},~\myrule{VFor},~\myrule{VForEnd}
    and~\myrule{Vloop} implies $\envmv$ not active, hence the cases
    trivially hold.
    \qedhere
  \end{itemize}
\end{proof}

\begin{lem}\label{lem:inv-subj-red}
  Let $\pjdg[\C][@][P][\envmv]$ and $\tuple \smv \in \dom \envmv$. If
  $\envmv\Ttrans{\alpha}\envmv'$ with $\alpha \neq \tau$ and
  $\names[\alpha] \subseteq \tuple \smv$, then
  for all stores $\sigma$ such that $\consSt[@][P]$ the following conditions hold
  \begin{enumerate}
  \item
    $\state[P][\sigma]
    \Wtrans{\lcond{\CC}{\seq\beta}}\Ptrans{\lcond{\C'}{\alpha}}
    \state[P'][\sigma']$
    with
    $\names[\seq{\beta}]\cap\tuple\smv =\emptyset$

  \item
    $\pjdg[\C\land\CC\land\C'][\envmv''][P'][@][\Gamma']$ with
    $\Gamma \subseteq \Gamma'$

  \item $\envmv'\psnames\smv p = \envmv''\psnames\smv p$ for all
    $\psnames \smv p \in \dom \envmv$

  \item
    $\consSt[\sigma'][P'][\Gamma'][\envmv''][\C\land\CC\land\C']$.
  \end{enumerate}
\end{lem}

\begin{proof}
  The proof follows by induction on the structure of the derivation of
  $\pjdg[\C][@][P][\envmv]$. We report the representative cases.

  \begin{itemize}
  \item\myrule{VReq}Then, $P=\Preq{\shname}{n}{\tuple \smv}{P''}$,
    $\envmv(\shname) \equiv \G(\tuple \smv)$ and
    $\pjdg[@][\envmv, \psnames \smv 0: \pst][P'']$ with
    $\rmg{ \pst} = \rmg{\G(\tuple\smv)\proj{\ptp 0}}$.
    There are two cases:
    \begin{enumerate}
    \item
      $\alpha = \lreq\shname n {\tuple \smv}$. Then,
      $\envmv' = \envmv,\psnames  \smv 0 : \pst$.
      Take $P'=P''$, $\seq\beta=\epsilon$, $\CC = \C' = \truek$,
      $\sigma' = \sigma\upd\shname{\tuple\smv}$, $\Gamma'= \Gamma$
      and $\envmv''=\envmv'$.
      Note that
      $\state[P][\sigma] \Ptrans{\lcond{\C'}{\alpha}}\state[P'][\sigma']$
      by $\myrule{SReq}$.
      Condition $\consSt[\sigma'][P''][\Gamma][\envmv'][\C]$ follows
      straightforwardly from $\consSt[\sigma][P][\Gamma][\envmv][\C]$.

    \item
      $\alpha \neq \lreq\shname n {\tuple\smv}$.
      Hence, $\names[\alpha] \cap \tuple{\smv} = \emptyset$.
      By~\lemref{lem:env-weakening}, $\envmv\Ttrans{\alpha}\envmv'$
      implies $\envmv,\psnames \smv 0:
      \pst\Ttrans{\alpha}\envmv',\psnames \smv 0: \pst$.
      By inductive hypothesis on
      $\pjdg[@][\envmv, \psnames\smv 0:\pst][P'']$ and
      $\envmv,\psnames\smv 0: \pst\Ttrans{\alpha}\envmv',\psnames\smv 0: \pst$,
      we conclude that
      $\state[P''][\sigma\upd{\shname}{\tuple \smv}]
      \Wtrans{\lcond{\CC'}{\seq\beta'}}
      \Ptrans{\lcond{\C'}{\alpha}}\state[P'][\sigma']$.
      Then, take
      $\seq\beta = \lreq\shname n{\tuple\smv}\seq{\beta'}$
      and $\CC=\CC'$.
      Conditions (2)--(4) follow straightforwardly from inductive % chktex 36
      hypothesis.
    \end{enumerate}

  \item\myrule{TRcv}
    Then, $P = \Pechoice{i\in I}{\smv_i}{x_i}{P_i}$, and
    $\envmv = \psnames\smv p : \displaystyle{\sum_{i\in I} \Tbra{\guard\smv_i}{\sort_{\mbox{\scriptsize$i$}}}{\pst_i}}, \envmv_0$
    and
    $\forall i \in I \qst \smv_i \in \tuple\smv$  and
    $\pjdg[@][\envmv,\psnames \smv p : {\pst_i}][P_i][@][\Gamma,x_i : \sort_i]$.

    There are two cases,
	\begin{enumerate}
	   \item $\alpha = \lreceive{\smv_j}{\sort_j}$ with  $j\in I$. Then,
	    $\envmv' = \envmv, \psnames \smv p : \pst_j$.
	   The proof is completed by taking $P'=P_j$, and $\seq\beta=\epsilon$, and $\CC = \C' = \truek$,  and
	   $\sigma' = \sigma\upd{\val v}{x_j}$, and
	    $\Gamma'= \Gamma,x_i : \sort_i$, and $\envmv''=\envmv'$. Note that $\state[P][\sigma] \Ptrans{\lcond{\C'}{\alpha}}\state[P'][\sigma']$
	   by  $\myrule{SRec}$.  Condition $\consSt[\sigma'][P'][\Gamma'][\envmv''][\C]$ follows straightforwardly
	   from  $\consSt[\sigma][P][\Gamma][\envmv][\C]$.

	   \item $\alpha \neq \lreq{\shname}{n}{\tuple \smv}$. Hence, $\names[\alpha] \cap \tuple{\smv} = \emptyset$. The case follows analogously
	   to rule~\myrule{VReq}.
       \qedhere
	\end{enumerate}
\end{itemize}
\end{proof}

\begin{lem}\label{lem:consistent-reduction-with-store}
  If $\state[P] \Ptrans{\lcond {\C} \alpha} \state[P'][\sigma']$, then
  $\eval \C \sigma = \truek$.
\end{lem}

\begin{proof}
  By straightforward induction on the structure of the proof
  $\state[P] \Ptrans{\lcond {\C'} \alpha} \state[P'][\sigma']$.
\end{proof}

%%% Local Variables:
%%% mode: latex
%%% TeX-master: "main"
%%% End:

% !TEX root = main.tex
\subsection{Subject reduction}

We first state an auxiliary property about the semantics of systems.
\begin{lem}\label{monotone:lem}
  If $\state[S] \Ptrans{\lcond e \alpha} \state[S'][\sigma']$ then
  $\dom \sigma \subseteq \dom{\sigma'}$.
\end{lem}
\begin{proof}
  By induction on the derivation of $\state[S] \Ptrans{\lcond e
    \alpha} \state[S'][\sigma']$ inspecting the rules in
  \cref{fig:LTSproc} and~\ref{fig:LTSsys}.
\end{proof}

\subjred*
\begin{proof}
  The proof is by induction on the derivation of $\state[S][\sigma]
  \Ptrans{\lcond{\C'}{\alpha}}\state[S'][\sigma']$ which may end with
  the application of one of the rules in \cref{fig:LTSproc} (on
  page~\pageref{fig:LTSproc}) or in \cref{fig:LTSsys} (on
  page~\pageref{fig:LTSsys}).
  It is a simple observation that if $\alpha \neq \tau$ then $S$
  must be a process, so the proof ends with one of the rules in
  \cref{fig:LTSproc}.
  \subsection*{Base cases}
  \begin{itemize}
  \item{\myrule{SReq}} In this case
    $\alpha = \lreq \shname n {\tuple \smv}$,
    $\C' = \truek$  and
    \[
    \state[\underbrace{\Preq \shname n {\tuple \smv} P}_{=S}]
    \Ptrans{\lreq{\shname}{n}{\tuple \smv}}
    \state[P][\underbrace{\sigma \upd{\shname}{\tuple \smv}}_{=\sigma'}]
    \qquad
    \text{ with }
    S' = P
    \]
    under the hypothesis that
    $\tuple \smv \cap \dom \sigma = \emptyset$.
    By inspecting the typing rules in \cref{fig:typingproc}
    (page~\pageref{fig:typingproc}), $S$ can be typed only by
    using rule \myrule{VReq}, whose hypothesis yields
    \begin{gather}\label{eq:SR:sreq}
      \envmv(\shname) \equiv \G(\tuple \smv)
      \qquad
      \pjdg[@][\envmv, \psnames \smv 0:  \pst]
      \qquad
      \rmg{ \nfp{\pst}} = \rmg{\nfp{\G(\tuple \smv)\proj{\ptp 0}}}
    \end{gather}
    We show~\eqref{sr:other} taking
    $\envmv' = \envmv, \psnames \smv 0: \pst$.
    We have that
    \[
    \spec \Gamma \Delta \Ttrans{\lreq \shname n {\tuple \smv}}
    \spec \Gamma {\Delta'}
    \qquad\text{and}\qquad
    \pjdg[\C\land\C'][\envmv'][S']
    \]
    respectively by rule \myrule{TReq} and by the judgment
    in~\eqref{eq:SR:sreq} observing that $\C' = \truek$ and
    $S' = P$.
    It remains to prove that
    $\consSt[\sigma'][S'][\envmv'][\C\land\C']$:

    \begin{enumerate}[(1)]
    \item We have $ \dom \Gamma \subseteq \dom
      \sigma
      \subseteq \dom \sigma \cup\tuple\smv = \dom{\sigma'}$
      where the first inclusion is implied by $\consSt$.
      By definition,
      $\dom {\envmv'} = \dom \envmv \cup\tuple\smv$, hence
      $\dom{\envmv'|_{\chset}} = (\dom{\envmv|_{\chset}})
      \cup \tuple \smv$.
      Moreover, $\consSt$ implies
      $(\dom{\envmv|_{\chset}}) \subseteq \dom \sigma$.
      Consequently,
      \[
      (\dom {\envmv'|_{\chset}}) = (\dom {\envmv|_{\chset}})
      \cup \tuple\smv \subseteq \dom \sigma \cup\tuple\smv =
      \dom {\sigma'}\]
    \item By definition of $\sigma'$, $\sigma'|_{\varset} =
      \sigma|_{\varset}$, hence $\vjdg[][\sigma'(x)][\Gamma(x)]$ for all
      $x \in \dom \Gamma$ since $\consSt$ implies
      $\vjdg[][\sigma(x)][\Gamma(x)]$ for all $x \in \dom{\Gamma}$.
    \item From $\sigma'|_{\varset} = \sigma|_{\varset}$ and $\consSt$,
      which implies $\eval \C \sigma = \truek$ by~\eqref{cons4:eq} in
      \defref{def:consistent}; we have $\eval {\C \land \C'} \sigma' =
      \truek$.
    \item Note that $\consPst$ implies $\consPst[\pst][\sigma']$ because
      $\sigma'|_{\varset} = \sigma|_{\varset}$. Therefore, for all
      $\psnames{\smv}{p}\in \dom{\envmv}$, we have that $\consPst[\envmv
      \psnames{\smv}{p}][\sigma']$ because $\consPst[\envmv
      \psnames{\smv}{p}]$ (which is ensured by $\consSt$).  Finally,
      $\consPst[\envmv'\psnames \smv 0][\sigma']$ follows from
      \lemref{lem:const-pseudo-typing-and-store-correspondence}.
    \end{enumerate}

  \item{\myrule{SAcc}} Analogous to the previous case.

  \item{\myrule{SRcv}} We have
    \[
    \state[\underbrace{\Pechoice{i\in I}{\smv_i}{x_i}{P_i}}_{=S}]
    \Ptrans{\lreceive \smv_j {\val v}}
    \state[P_j][\underbrace{\sigma\upd{\val v}{x_j}}_{=\sigma'}]
    \qquad\text{where}\qquad S' = P_j
    \]
    with $\alpha = \lreceive \smv_j {\val v}$ for some $j\in I$.
    As above $\C' = \truek$ and, by typing rule \myrule{VRcv}, we have
    \begin{gather}\label{eq:SR:srcv}
      \envmv = \envmv'', \psnames \smv p : \displaystyle{\sum_{i\in I}
        \Tbra{\guard\smv_i}{\sort_{\mbox{\scriptsize$i$}}}{\pst_i}}
      \qquad
      \forall i \in I \qst \pjdg[@][\envmv'_i][P_i][@][\Gamma,x_i : \sort_i]
    \end{gather}
    where $\envmv'_i = \envmv'', \psnames \smv p : {\pst_i}$
    for each $i \in I$.
    We show~\eqref{sr:in} taking $\envmv' = \envmv'_j$.
    From rule \myrule{TRcv}, we have
    $\envmv \Ttrans{\lreceive{\smv_j}{\sort_j}} \envmv'$.
    In addition, if $\vjdg[][\val v][\sort_j]$ then take
    $x = x_j$ and $\sort = \sort_j$.
    Judgment $\pjdg[\C \land \C'][\envmv'][S']$ is derived from the
    $j^{th}$ judgment in~\eqref{eq:SR:srcv} by noticing that $\C' =
    \truek$.
    Then, observing that $\dom {\sigma'} = \dom \sigma \cup\{x\}$, we
    prove that $\consSt [\sigma'][S'][\envmv'][\C \land \C']$ as
    follows:
    \begin{enumerate}[(1)]
    \item We have $\dom {\Gamma, x:\sort} = \dom \Gamma \cup \{x\}
      \subseteq \dom \sigma \cup \{x\} = \dom {\sigma'}$ since $\consSt$
      implies $\dom \Gamma \subseteq \dom \sigma$.
      Noting that $\dom{\envmv'} = \dom \envmv$, we have
      $\dom{\envmv'|_{\chset}} = \dom{\envmv|_{\chset}} \subseteq \dom
      \sigma$ because $\consSt$.
    \item Let $\Gamma' = \Gamma \cup \{x\}$.  If $\vjdg[][\val
      v][\sort]$, then $\sigma'(x) = \val v: \sort = \Gamma'(x)$.
      We conclude that $\sigma(x') = \sigma(x') : \Gamma'(x')$ for all
      $x' \neq x \in \dom{\Gamma'}$ by $\consSt$.
      If it is not the case that $\vjdg[][\val v][\sort]$, then there is
      nothing to prove.
    \item Since $\C' = \truek$, $\eval {\C \land \C'} {\sigma'} =
      \truek$ trivially holds and $\consSt$ implies $\eval \C \sigma =
      \truek$.
      In addition, $\sigma'(x)=\sigma(x)$ for all $x\in\dom \sigma$.
      Hence, $\eval \C {\sigma'} = \truek$, which implies $\eval
      {(\C\land \C')} {\sigma'} = \truek$.
    \item By \lemref{lem:const-pseudo-typing-and-store-correspondence}
      we have the thesis.
    \end{enumerate}

  \item{\myrule{SSend}} In this case $\alpha = \lsend \smv {\val v}$
    with $S = \Psend \smv \C_1 {}$, $\eval {\C_1} \sigma = \val v$,
    $\C' = \truek$, $S' = 0$ and $\sigma' = \sigma$.
    Hence, $S$ can be typed only by applying rule \myrule{VSend},
    which yields
    \begin{gather*}
      \vjdg[\Gamma][e_1][\sort]
      \qquad
      \smv \in \tuple \smv
      \qquad
      \envmv  =  \envmv'', \psnames \smv p: \guard{\Tsend \smv \sort};
      \guard\Tend
    \end{gather*}
    where $\envmv'' \psnames {\smv'} p = \guard\Tend$ for all
    $\psnames {\smv'} p \in \dom{\envmv''}$.
    We show~\eqref{sr:out} taking $\envmv' = \envmv'', \psnames \smv p
    : \Tend$.
    Note that, from the hypothesis $\consSt$, we have
    $\vjdg[][\sigma(x)][\Gamma(x)]$ for all $x \in \dom{\Gamma}$,
    hence $\vjdg[][\val v][\sort]$.
    From rule \myrule{TSend}, $\envmv \Ttrans{\lsend{\smv}{\sort}}
    \envmv'$ and, by rule \myrule{VEnd},
    $\pjdg[\C\land\C'][\envmv'][0]$.

    It is straightforward to conclude that
    $\consSt[\sigma'][0][\envmv'][\C\land\C']$ because $\consSt$,
    $\dom{\envmv'} = \dom \envmv$, $\sigma' = \sigma$, and $\C' =
    \truek$.

  \item{\myrule{SInit}} We have
    \[
    \state[\underbrace{\Preq \shname n {\tuple \smv}{P_0} \sparop
        \Pacc \shname 1 {\tuple \smv}{P_1} \sparop \ldots \sparop \Pacc
        \shname n {\tuple \smv}{P_n}}_{=S}] \Ptrans{\tau}
    \state[\underbrace{(\nu \AT{\tuple \smv} \shname ) (P_0 \sparop
        \ldots \sparop P_n \sparop \tuple \smv : \emptyset)}_{=
      S'}][\underbrace{\sigma\upd{\shname}{\tuple \smv}}_{=\sigma'}]
    \]
    The judgment $\pjdg[@][@][S]$ is obtained by repeated
    application of rule \myrule{VPar} on
    \[
    \qquad
    \irule{ \envmv_0(\shname) \equiv \G(\tuple \smv) \qquad
      \pjdg[@][\envmv_0, \psnames \smv 0: \pst_0][P_0] \qquad
      \rmg{\pst_0} = \rmg{\G(\tuple \smv)\proj{\ptp 0}} }{
      \pjdg[@][\envmv_0][\Preq \shname n {\tuple
          \smv}{P_0}][\envmv_0'][\Gamma] }\qquad\myrule{VReq}
    \]
    and
    \[
    \qquad
    \irule{
      \envmv_p(\shname) \equiv \G(\tuple \smv) \qquad
      \pjdg[@][\envmv_p, \psnames \smv p : \pst_p][P_p] \qquad
      \rmg{\pst_p} = \rmg{\G(\tuple \smv) \proj {\ptp p}}
    }{
      \pjdg[@][\envmv_p][\Pacc \shname p {\tuple \smv}{P_p}][\envmv_p'][\Gamma]
    }\qquad\myrule{VAcc}
    \]
    with $p \in \{1,\ldots,n\}$, $\envmv = \envmv_0 \cup \envmv_1 \cup
    \cdots \cup \envmv_n$, under the assumption that $\envmv_1,\ldots,
    \envmv_n$ are pairwise independent.
    We show~\eqref{sr:other} by taking $\envmv' = \envmv$ and
    observing that, by rule \myrule{VNew}, we have
    \[
    \irule{
      \pjdg[@][\envmv,
        \big\{\psnames \smv p \mapsto \pst_p \sst p \in \{0,\ldots,n\} \big\}]
           [P_0 \ \sparop\  \ldots \ \sparop\  P_n \ \sparop\  \tuple \smv :
             \queue{}{}][@][\Gamma]
    }{
      \pjdg[@][@][S'][\envmv' \setminus
        \big\{\psnames \smv p \mapsto \G\proj{\ptp p} \sst
                       {\ptp p}\in\participants{G} \big\}][\Gamma]
    }
    \]
    with the premiss of the above derivation obtained by rule
    \myrule{VEmpty} and repeated applications of rule \myrule{VPar}.
    To show that $\consSt[\sigma'][S'][\envmv]$ we note that
    $\dom{\sigma'} = \dom \sigma \cup \tuple \smv$ by construction and
    that
    \[
    \fY[S'] = \bigcup_{p=0}^n \fY[P_p] \setminus
    \tuple\smv = \bigcup_{p=0}^n(\fX[P_p] \setminus \tuple \smv) = \fY[S]
    \]
    therefore conditions~\eqref{cons2:eq}, and~\eqref{cons5:eq} of
    \defref{def:consistent} hold by inductive hypothesis while
    conditions~\eqref{cons3:eq} and~\eqref{cons4:eq} hold because
    $\sigma'|_{\varset} = \sigma|_{\varset}$.

  \item{\myrule{SForEnd}} We have $\isemptyL{\eval \ell \sigma}$ and
    \[
    \state[\underbrace{\Pfor x \ell  P}_{=S}] \Ptrans{\tau} \state[\Pend]
    \qquad \text{with }
    S' = \Pend, \ \sigma' = \sigma, \text{ and } \C' = \truek
    \]
    Since ${\eval \C \sigma} = \truek$ (because $\consSt$) we have
    $\consistent[\C][\ell = \emptylist]$ and, consequently,
    \myrule{VForEnd} is the only applicable rule to type $S$.
    Therefore we have $\envmv = \Tdef{\envmv'}$ for an $\Tend$-only
    $\envmv'$ (by the premiss of \myrule{VForEnd}).
    Finally, we have that~\eqref{sr:other} holds by rule
    \myrule{TLoop_1} and $\consSt[@][\Pend][\envmv']$ directly follows
    by $\consSt$.
  \end{itemize}

\subsection*{Inductive step}
\begin{itemize}
\item {\myrule{SThen}} In this case
  \[
  \mathrule{
    \eval {\C'} \sigma = \truek
    \qquad
    \state[P_1] \Ptrans{\lcond{\C_1}{\alpha}} \state[P_1'][\sigma']
  }{
    \state[\underbrace{\Pif{\C'}{P_1}{P_2}}_{= S}]
    \Ptrans{\lcond{\C' \land \C_1}\alpha}
    \state[P_1'][\sigma']
  }{}
  \]
  with $S' = P_1'$.
  To type $S$ we must apply rule \myrule{VIf}
  \[
  \irule{
    \begin{array}{c}
      \pjdg[\C \land
        \C'][\envmv_1][P_1][@][\Gamma] \qquad \pjdg[\C\land\neg
        \C'][\envmv_2][P_2][@][\Gamma]
    \end{array}
  }{
    \pjdg[@][\underbrace{\envmv_1 \peace \envmv_2}_{=\envmv}]
         [\Pif{\C'}{P_1}{P_2}][\envmv_1'\restriction_{\dom \envmv}]
  }\myrule{Vif}
  \]
  recall that $\envmv_1 \peace \envmv_2$ is defined when $\envmv_1$
  and $\envmv_2$ are mergeable (i.e.,
  $\dom{\envmv_1} = \dom{\envmv_2}$ and the pseudo-types
  $\envmv_1 \psnames \smv p$ and $\envmv_2 \psnames \smv p$ are
  mergeable for all $\psnames \smv p \in \dom{\envmv_1}$).

  By the inductive hypothesis, there exist $\Gamma'$ and $\envmv_1'$
  such that
  \begin{enumerate}
  \item[\eqref{sr:in}] if $\alpha = \lreceive{\smv}{\val v}$ then
    $\envmv_1 \Strans{\lreceive{\smv}{\sort}}\envmv_1'$; moreover, if
    $\vjdg[][\val v][\sort]$ then
    \begin{gather*}
      \pjdg[\C \land \C_1 \land \C'][\envmv_1'][P_1'][@][\Gamma', x :
        \sort]
    \end{gather*}
    and
    $\consSt[\sigma'][P_1'][\envmv_1'][\C \land \C_1 \land
      \C'][\Gamma', x : \sort]$.
    By \lemref{lem:merge-reduction},
    $\envmv = \envmv_1 \peace \envmv_2 \Strans{\lreceive{\smv}{\sort}}
    \envmv_1'$.
    Then, take $\envmv' = \envmv_1'$ and note that
    \[\pjdg[\C \land (\C_1 \land
      \C')][\envmv'][P_1'][@][\Gamma', x : \sort]\]
    and
    $\consSt[\sigma'][P_1'][\envmv'][\C \land (\C'\land
      \C')][\Gamma,x:\sort]$ hold by associativity of $\land$.
  \item[\eqref{sr:out}] and~\eqref{sr:other} analogous to the previous
    case.
  \end{enumerate}

\item{\myrule{SElse}} Analogous to \myrule{SThen}.

\item{\myrule{SSeq}} We have
  \[
  \mathrule{
    \state[P_1] \Ptrans{\lcond{\C'} \alpha} \state[P'_1][\sigma']
  }{
    \state[\underbrace{P_1;P_2}_{=S}]
    \Ptrans{\lcond{\C'} \alpha}
    \state[\underbrace{P'_1;P_2}_{=S'}][\sigma']
  }{}
  \]
  As in the previous case, we can type $S$ only applying
  \myrule{VSeq}.
  This implies that $\envmv = \envmv_1 ; \envmv_2$ for some
  $\envmv_1$ and $\envmv_2$ such that
  \begin{itemize}
  \item $\dom{\envmv_2} \subseteq \dom{\envmv_1}$ and
    $\envmv_1 |_{\shset \cup \chset} = \envmv_2 |_{\shset \cup
    \chset} = \envmv|_{\shset \cup \chset}$
  \item $\envmv : \psnames \smv p \mapsto
    \begin{cases}
      \envmv_1 \psnames \smv p ; \envmv_2 \psnames \smv p
      & \psnames \smv p \in \dom {\envmv_2}
      \\
      \envmv_1 \psnames \smv p
      & \psnames \smv p \in \dom{\envmv_1} \setminus \dom {\envmv_2}
      \\
      \text{undef} & \text{otherwise}
    \end{cases}$
  \item
    $\pjdg[@][\envmv_1][P_1]$ and $\pjdg[@][\envmv_2][P_2]$.
  \end{itemize}
  By the inductive hypothesis, there are $\Gamma'$ and $\envmv_1'$
  such that
  \begin{enumerate}
  \item[\eqref{sr:in}] if $\alpha = \lreceive \smv {\val v}$ then
    $\envmv_1 \Ttrans{\lreceive \smv \sort} \envmv_1'$ for a sort
    $\sort$; also, if $\vjdg[][\val v][\sort]$ then there is
    $x \in \varset$ such that $\sigma'(x) = \val v$,
    $\pjdg[\C \land \C'][\envmv_1'][P_1'][@][\Gamma', x : \sort]$ and
    $\consSt[\sigma'][P_1'][\envmv_1'][\C\land\C'][\Gamma', x : \sort]$
  \item[\eqref{sr:out}] if $\alpha = \lsend \smv {\val v}$ then
    $\envmv_1 \Ttrans{\lsend \smv \sort} \envmv_1'$ with
    $\vjdg[][\val v][\sort]$,
    $\consSt[\sigma'][P_1'][\envmv_1'][\C \land \C'][\Gamma']$, and
    $\pjdg[\C \land \C'][\envmv_1'][P_1'][@][\Gamma']$
  \item[\eqref{sr:other}] otherwise
    $\envmv_1 \Ttrans \alpha \envmv_1'$,
    $\consSt[\sigma'][P_1'][\envmv_1'][\C \land \C'][\Gamma']$, and
    $\pjdg[\C \land \C'][\envmv_1'][P_1'][@][\Gamma']$
  \end{enumerate}
  In any case, by rule \myrule{TSeq} we have that
  \[
    \envmv = \envmv_1;\envmv_2
    \Ttrans \alpha
    \envmv_1';\envmv_2 = \envmv'
  \]
  and the proof concludes with the application of rule \myrule{VSeq}
  to the judgement typing $P_1'$ and
  $\pjdg[@][\envmv_2][P_2][][\Gamma']$ observing that
  $\dom{\Gamma'} \subseteq \dom{\Gamma}$ and
  $\dom{\envmv_2} \subseteq \dom{\envmv_1} \subseteq \dom{\envmv_1'}$
  where the latter inclusion holds because the domain of a
  specification may only grow after transitions (by inspection of the
  rules in \cref{fig:LTSruntimetypes}) and, for the same reason,
  $\envmv_2|_{\shset \cup \chset} = \envmv_1|_{\shset \cup \chset} =
  \envmv_1'|_{\shset \cup \chset}$ since $\envmv_1'$ and $\envmv_1$
  differ at one participant's session only.

\item{\myrule{SFor}}
  \newcommand{\gupd}[1][\tuple \sort]{}
  We have
  \begin{gather}\label{eq:SR:sfor2-red-a}
    \mathrule{
      \eval \ell \sigma \neq \emptylist\
      \qquad
      \state[P][\sigma\upd{\headL{\eval \ell \sigma}}{x}]
      \Ptrans{\lcond{\C'} \alpha}
      \state[P'][\sigma']
    }{
      \state[\underbrace{\Pfor{x}{\ell}{P}}_{=S}]
      \Ptrans{\lcond{\C'} \alpha}
      \state[\underbrace{P';\Pfor{x}{\tailL{\ell}}{P}}_{S'}][\sigma']
    }
    {}
  \end{gather}
  The first premiss of~\eqref{eq:SR:sfor2-red-a} implies
  $\consistent[\C][\ell \not = \emptylist]$, hence the only applicable
  rule to type $S$ is $\myrule{VFor}$ which we instantiate as:
  \begin{gather}\label{eq:SR:sfor2-typing-a}
    \mathrule{
      \begin{array}{c}
        \vjdg[\Gamma][\ell][[\sort]]
        \qquad
        \pjdg[\C][\envmv_1][@][@][\Gamma\gupd, x : \sort]
        \qquad \envmv_1 \mbox{ active}
        \qquad x \not\in \var{\envmv_1}
      \end{array}
    }{
      \pjdg[@][\underbrace{\envmv_1^*}_{=\envmv}][\Pfor x \ell P]
    }{}
  \end{gather}
  for some $\envmv_1$.
  In order to use the inductive hypothesis, we check that
  \[
    \consSt[\sigma\upd{\headL{\eval \ell
        \sigma}}{x}][P][\envmv_1][\C][\Gamma\gupd, x : \sort]
  \]
  In fact, from $\consSt$ we have that
  \begin{enumerate}[(1)]
  \item
    $\dom{\Gamma\gupd, x : \sort} = \dom
    \Gamma \cup \{x\} \subseteq \dom \sigma \cup \{x\} =
    \dom{\sigma\upd{\headL{\eval \ell \sigma}}{x}}$; also,
    $\dom{\envmv_1|_{\chset}} \subseteq \dom \sigma \subseteq
    \dom{\sigma\upd{\headL{\eval \ell \sigma}} x}$
  \item for all $z \in \dom{\Gamma\gupd}$ we have
    $\vjdg[][\sigma\upd{\headL{\eval \ell
        \sigma}}{x}][\Gamma\gupd(z)]$ since
    $\sigma\upd{\headL{\eval \ell \sigma}}{x}(z) = \sigma(z)$;
    moreover, $\vjdg[\Gamma][\ell][[\sort]]$ implies that
    $\vjdg[\Gamma][\headL{\eval \ell \sigma}][\sort]$ and therefore
    $\headL{\eval \ell
      \sigma} = \sigma\upd{\headL{\eval \ell \sigma}}{x}(x)$ has sort $\sort$
  \item since $\eval {\C } \sigma =\truek$, $x\not\in\var \C$ and
    therefore
    $\eval {\C } {\sigma\upd{\headL{\eval \ell \sigma}}{x}} =\truek$
  \item $\forall \psnames{\smv}{p}\in \dom{\envmv_1},
    \consPst[\envmv_1 \psnames{\smv}{p}][\sigma\upd{\headL{\eval
        \ell \sigma}}{x}]$ follows from
    \lemref{lem:const-pseudo-typing-and-store-correspondence}.
  \end{enumerate}

  We proceed by case analysis on $\alpha$.
  First note that By \lemref{lem:red-active-process}, we can exclude
  the case $\alpha = \lreceive{\smv}{\val v}$ because $\envmv_1$ is
  active.
  We consider the case $\alpha = \lsend{\smv}{\val v}$ (the proof in
  the other cases is analogous and simpler and therefore omitted).
  Assume that $\vjdg[][\val v][\sort']$.
  By the inductive hypothesis on~\eqref{eq:SR:sfor2-red-a} and~\eqref{eq:SR:sfor2-typing-a}, there exists
  $\envmv_1 \Ttrans{\lsend \smv {\sort'}} \envmv_1'$ such that
  \begin{gather}\label{eq:SR:sfor2-continuation-P}
    \pjdg[\C \land \C'][\envmv_1'][P'][@][\Gamma\gupd, x:\sort]
    \quad
    \text{with } \consSt[\sigma'][P'][\envmv_1'][\C \land
    \C'][\Gamma\gupd,x:\sort]
  \end{gather}
  We show how to find an environment $\envmv'$ required in~\eqref{sr:out} depending on the possible typings of
  $\Pfor{x}{\tailL{\ell}}{P}$.

  If $\consistent[\C \land \C'][\tailL\ell = \emptylist]$ then, using rule
  \myrule{VForEnd}, we have
  \begin{gather}\label{eq:SR:sfor2-continuation-for}
    \pjdg[\C \land \C'][\envmv_1'|_{\shset \cup \chset}][\Pfor{x}{\tailL
      \ell}{P}][@][{\Gamma\gupd, x : \sort}]
  \end{gather}
  since
  $\forall \psnames \smv p \in\dom{\envmv_1'|_{\shset \cup \chset}} \qst
  (\envmv_1'|_{\shset \cup \chset}) \psnames \smv p = \Tend$ vacuously holds.
  Then, from~\eqref{eq:SR:sfor2-continuation-P} and~\eqref{eq:SR:sfor2-continuation-for}, we derive
  $\pjdg[\C \land
  \C'][\envmv_1';(\envmv_1'|_{\shset \cup \chset})][S'][@][\Gamma\gupd, x : \sort]$
  by using rule \myrule{VSeq}.
  Finally,
  $\envmv = \Tdef{\envmv_1} \Ttrans{\lsend{\smv}{\sort'}}
  \envmv_1'$ by \myrule{TLoop_1} and the thesis follows by
  observing that
  $\envmv_1' = \envmv_1';(\envmv_1'|_{\shset \cup \chset})$
  since
  $\dom{\envmv_1'} \cap \dom{\envmv_1';(\envmv_1'|_{\shset \cup
      \chset})} \subseteq \shset \cup \chset$ by definition of
  $\_ ; \_$ on specifications (\cf\
  page~\pageref{pag:envseq}).

  If $\consistent[\C \land \C'][\tailL\ell \not = \emptylist]$ then,
  using rule $\myrule{VFor}$,
  \begin{gather}\label{eq:sfor2-non-empty-iteration-retyping-a}
    \mathrule{
      \begin{array}{c}
        \vjdg[\Gamma][\tailL \ell][[\sort]]
        \qquad
        \pjdg[@][\envmv_1][@][@][\Gamma\gupd,x : \sort]
        \qquad
        \envmv_1 \text{ active}
        \qquad
        x \not\in \var{\envmv_1}
      \end{array}
    }{
      \pjdg[@][\envmv_1^*][\Pfor{x}{\tailL \ell}{P}][@][\Gamma\gupd, x : \sort]
    }{}
  \end{gather}
  where $\vjdg[\Gamma][\tailL\ell][[\sort]]$ holds because
  $\vjdg[\Gamma][\ell][[\sort]]$ and the conditions in the rest of the
  hypothesis of~\eqref{eq:sfor2-non-empty-iteration-retyping-a} hold
  by the typing~\eqref{eq:SR:sfor2-typing-a}.
  By \lemref{thm:distribution-end-pseudo-types} and
  \lemref{lem:typing-weakening} on the conclusion, from~\eqref{eq:sfor2-non-empty-iteration-retyping-a}
  \begin{gather}\label{eq:sfor2-non-empty-iteration-retyping-e}
    \pjdg[\C\land\C'][\envmv''][\Pfor{x}{\tailL\ell}{P}][@][\Gamma\gupd,x:\sort]
  \end{gather}
  where
  $\envmv'' = \Tdef{\envmv_1}\upd{\guard[\C']\Tend; \Tdef{\envmv_1}
    \psnames \smv p}{\psnames \smv p}_{\psnames \smv p
    \in\dom{\Tdef{\envmv_1}}}$.

  We now show that $\envmv_1';\envmv''$ and $\envmv_1';\envmv$ are
  defined and equivalent.
  Note that $\dom{\envmv''} = \dom\envmv = \dom{\envmv_1}
  \subseteq\dom {\envmv_1'}$ because $\envmv_1
  \Ttrans{\lsend{\smv}{\sort}} \envmv_1'$.
  For the equivalence, we just need to consider session names; let
  $\psnames \smv p \in\dom{\envmv_1';\envmv'' }$, we have two cases:
  if $\psnames \smv p\not \in\dom{\envmv'' }$ then
  $(\envmv_1';\envmv'' ) \psnames \smv p = {\envmv_1' } \psnames
  \smv p = (\envmv_1';\envmv ) \psnames \smv p$ by the definition of
  $\_ ; \_$ and if $\psnames \smv p \in\dom{\envmv'' }$ then
  \[
    \begin{array}{rll}
      (\envmv_1';\envmv'' ) \psnames \smv p
      =
      &
      {\envmv_1' } \psnames \smv p; {\envmv''} \psnames \smv p;
      &
      \text{by def.\ of } \_;\_
      \\
      =
      &
      {\envmv_1' } \psnames \smv p; (\guard[\C']\Tend;
      {\envmv} \psnames \smv p);
      &
      \text{by def.\ of } \envmv''
      \\
      =
      &
      {\envmv_1' } \psnames \smv p; {\envmv} \psnames \smv p;
      &
      \text{by \lemref{lem:nf-neutral-end-seq}}
    \end{array}
  \]
  By rule \myrule{VSeq}, with the judgments in~\eqref{eq:SR:sfor2-continuation-P} and~\eqref{eq:sfor2-non-empty-iteration-retyping-e}, we have
  \begin{gather}\label{eq:sfor2-non-empty-iteration-retyping-d}
    \pjdg[\C \land \C'][\envmv_1';\envmv][S'][@][\Gamma\gupd,x:\sort]
  \end{gather}
  and
  $\envmv = \Tdef{\envmv_1} \Ttrans{\lreceive \smv \sort}
  \envmv_1';\envmv = \envmv'$ by rule \myrule{TLoop_2}.

  It remains to show that
  $\consSt[\sigma'][S'][\envmv'][\C
  \land \C'][\Gamma\gupd,x:\sort]$ holds:
  \begin{enumerate} [(1)]
  \item It follows immediately from~\eqref{eq:SR:sfor2-continuation-P} and the fact that
    $\dom{\envmv'} = \dom{\envmv_1'}$.
  \item Trivially from~\eqref{eq:SR:sfor2-continuation-P}.
  \item Immediate from~\eqref{eq:SR:sfor2-continuation-P};
  \item By \lemref{lem:const-pseudo-typing-and-store-correspondence}.
  \end{enumerate}

  \item{\myrule{SLoopEnd}} We have
  \[
    \mathrule{
      \state[M] \Ptrans{\lcond{\C'} \alpha}\state[P][\sigma']
    }{
      \state[\underbrace{\PloopU N M}_{=S}]
      \Ptrans{\lcond{\C'} \alpha}
      \state[P][\sigma']
    }{}
    \qquad
    \text{with } S' = P
  \]
  and the only applicable rule to type $S$ is \myrule{VLoop}:
  \[
    \vlooprule{}
  \]
  with $\envmv = \Tdef{\envmv_1};\envmv_2$.
  The inductive hypothesis requires to have $\consSt[@][M][\envmv_2]$,
  which easily follows because~\eqref{cons2:eq}-\eqref{cons4:eq} are
  immediate from $\consSt$ and~\eqref{cons5:eq} follows from
  \lemref{lem:const-pseudo-typing-and-store-correspondence}.
  Note that $\alpha$ has to be an input label
  $\lreceive{\smv}{\val v}$ because $M$ is a choice process.
  Hence, there exist $\Gamma'$ and
  $\envmv_2 \Ttrans{\lreceive \smv \sort} \envmv_2'$ and if
  $\vjdg[][\val v][\sort]$ then
  \[
    \pjdg[\C \land \C'][\envmv_2'][P][@][\Gamma', x : \sort]
    \qquad\text{and}\qquad
    \consSt[\sigma'][P][\envmv_2'][\C \land  \C'][\Gamma', x : \sort]
  \]
  Then, take $\envmv' = \envmv_2'$ and note that by \myrule{TLoop_0}
  $\envmv = \Tdef{\envmv_1};\envmv_2 \Ttrans{\lreceive{\smv}{\sort}}
  \envmv'$.

  \item{\myrule{SLoop}} We have
  \begin{gather}\label{eq:SR:sloop2-red-a}
    \mathrule{
      \state[N] \Ptrans{\lcond{\C'} \alpha} \state[P][\sigma']
      \quad
      M = \Pechoice{i \in I}{\smv_i}{x_i}{P_i}
      \quad
      \forall i\in I. \smv_i \not\in \fY[\alpha]
    }{
      \state[\underbrace{\PloopU NM}_{=S}]
      \Ptrans{\lcond{\C'} \alpha}
      \state[\underbrace{P;\PloopU N M}_{=S'}][\sigma']
    }{}
  \end{gather}
  and, as in the case \myrule{SLoopEnd},
  the only applicable rule to type $S$ is \myrule{VLoop}:
  \[
    \vlooprule{}
  \]
  with $\envmv = \Tdef{\envmv_1};\envmv_2$.
  The inductive hypothesis requires to show
  $\consSt[@][N][\envmv_1]$, which holds because~\eqref{cons2:eq}-\eqref{cons4:eq} are immediate from $\consSt$ and~\eqref{cons5:eq} follows from
  \lemref{lem:const-pseudo-typing-and-store-correspondence}.
  Note that $\alpha$ has to be an input label
  $\lreceive{\smv}{\val v}$ because $N$ is a choice process.
  Hence, there exist $\Gamma'$ and
  $\envmv_1 \Ttrans{\lreceive \smv \sort} \envmv_1'$ and if
  $\vjdg[][\val v][\sort]$ then
  \[
    \pjdg[\C \land \C'][\envmv_1'][P][@][\Gamma', x : \sort]
    \qquad\text{and}\qquad
    \consSt[\sigma'][P][\envmv_1'][\C \land  \C'][\Gamma', x : \sort]
  \]
  Then, take $\envmv' = \envmv_1';\Tdef{\envmv_1}$ and note that by
  \myrule{TLoop_2}
  $\envmv = \Tdef{\envmv_1};\envmv_2 \Ttrans{\lreceive{\smv}{\sort}}
  \envmv'$.
  Finally, the thesis follows by applying \myrule{VSeq}.

\item{\myrule{SCom_1}} In this case:
    \[
    \irule{
      \state[P] \Ptrans{\lcond {e'} {\lsend{\smv}{\val v}}}\state[P'][\sigma']
    }{
      \state[\underbrace{P \sparop \queue \smv \qmv}_{=S}]
      \Ptrans{\lcond {e'} \tau}
      \state[\underbrace{P' \sparop \queue \smv {\qmv \cdot \val v}}_{= S'}]
            [\sigma']
    }
    \]
    and, by rule \myrule{VPar}, we have the following typing for $S$:
    \[\qquad
    \irule{
      \begin{array}{c}
        \pjdg[@][\overbrace{\envmv \setminus
            \{\smv : \tuple \sort\}}^{=\envmv_1}][@][\envmv_1'][\Gamma]
        \quad
        \pjdg[@][\overbrace{\envmv|_{\shset},\smv : \tuple \sort}^{=\envmv_2}]
             [\queue \smv \qmv][\envmv_2']
             [\Gamma]
        \quad
        \envmv_1 \text{ and } \envmv_2 \text{ independent}
      \end{array}
    }{
      \pjdg[e][@][S][@][\Gamma]
    }
  \]
  where the typing of the queue $\queue \smv \qmv$ is obtained by
  repeated applications of the rule \myrule{VQueue} ending with an
  application of \myrule{VEmpty} and \cref{lem:typing-weakening};
  hence, $\vjdg[][\tuple \qmv][\tuple \sort]$.
  Note that $\envmv = \envmv_1 \cup \envmv_2$.
  We now observe that $\consSt[@][P][\envmv_1][@]$ and
  $\consSt[@][P][\envmv_2][@]$ directly follow from $\consSt$ (since
  $\dom{\envmv_1} \cup \dom{\envmv_2} \subseteq \dom{\envmv}$) so, by
  inductive hypothesis,
  \[\qquad
    \envmv_1 \Ttrans {\lsend{\smv}{\sort}} \envmv_1'
    \quad \text{with}\quad
    \vjdg[][\val v][\sort]
    \quad
    \pjdg[\C \land \C'][\envmv'_1][P']
    \quad\text{and}\quad
    \consSt[\sigma'][P'][\envmv'_1][\C\land\C'][\Gamma']
  \]
  for some environment $\Gamma'$.
  We now show that we can type $S'$ with
  $\envmv' = \envmv'_1 \cup \{\smv : \tuple \sort \cdot \sort \}$:
  \[\qquad
    \irule{
      \begin{array}{c}
        \pjdg[\C \land \C'][\envmv_1'][P'][@][\Gamma']
        \qquad
        \irule{
        \vjdg[{}][\val v][\sort]
        \qquad
        \pjdg[\C \land \C'][\smv : \queue{}{\tuple \sort}]
             [\queue{\smv}{\tuple{\val v}}][@][\Gamma']
        }{
          \pjdg[\C \land \C'][\smv : \tuple \sort \cdot \sort]
               [\queue \smv {\qmv \cdot \val v}][][\Gamma']
        } \ \myrule{VQueue}
      \end{array}
    }{
      \pjdg[\C \land \C'][\envmv'][S'][@][\Gamma']
    }\ \myrule{VPar}
  \]
  where the judgement in the premiss of \myrule{VQueue} holds by
  \cref{lem:typing-weakening} and the application of \myrule{VPar} is
  possible because $\envmv'_1$ and $\smv : \tuple \sort \cdot \sort$
  are independent.
  The proof of this case ends by showing that
  $\consSt[\sigma'][P'][\envmv'][\C\land\C'][\Gamma']$: we have that~\eqref{cons2:eq} of \defref{def:consistent} holds because
  $\fY[S'] = \fY[P'] \cup \{ \smv \}$ while
  $\dom{\envmv'|_{\chset}} = \dom{\envmv'_1|_{\chset}} \cup \{\smv\}$
  by definition and, by \lemref{monotone:lem},
  $\dom{\sigma} \subseteq \dom{\sigma'}$; also, conditions~\eqref{cons3:eq} and~\eqref{cons4:eq} hold because
  $\sigma'|_{\varset} = \sigma|_{\varset}$ while~\eqref{cons5:eq}
  holds because for each $\psnames \smv p$ we have
  $\envmv' \psnames \smv p = \envmv \psnames \smv p$ by construction.

  \item{\myrule{SCom_2}} In this case:
  \[\irule{
    \state[P]
    \Ptrans{\lcond {e'} {\lreceive \smv {\val v}}}
    \state[P'][\sigma']
  }{
    \state[\underbrace{P \sparop \queue \smv {\val v \cdot \qmv}}_{=S}]
    \Ptrans{\lcond {e'} \tau}
    \state[\underbrace{P' \sparop \queue \smv {\qmv}}_{= S'}][\sigma']
  }
  \]
  and, assuming
  $\envmv_1 = \envmv \setminus \{\smv : \sort \cdot \tuple \sort\}$,
  we observe that $\consSt[@][P][\envmv_1][@]$ directly follows from
  $\consSt$ (since $\dom{\envmv_1} \subseteq \dom{\envmv}$) so, by
  inductive hypothesis.
  We can type $S$ as follows
  \[
    \irule{
      \pjdg[@][\envmv_1][@][\envmv_1'][\Gamma]
      \qquad
      \pjdg[@][\envmv|_{\shset}, \smv : \sort \cdot \tuple \sort]
           [\queue \smv {\val v \cdot \qmv}][\envmv_2'][\Gamma]
    }{
      \pjdg[e][@][S][@][\Gamma]
    }
  \]
  by using rule \myrule{VPar} since $\envmv_1$ and
  $\envmv|_{\shset}, \smv : \sort \cdot \tuple \sort$ are independent
  and the typing of the queue $\queue \smv {\val v \cdot \qmv}$ is
  obtained by repeated applications of the rule \myrule{VQueue} ending
  with an application of \myrule{VEmpty}, which yields
  $\vjdg[][\qmv][\sort]$ and $\vjdg[][\tuple \qmv][\tuple \sort]$.
  By the inductive hypothesis
  $\envmv_1 \Ttrans {\lreceive \smv \sort} \envmv_1'$ and there are
  $\Gamma'$ and $x \in \varset$ such that
  $\pjdg[@][\envmv'_1][P'][][\Gamma', x : \sort]$ and
  $\consSt[\sigma'][P'][\envmv'_1][@][\Gamma',x : \sort]$.

  We now show that we can type $S'$ with
  $\envmv' = \envmv'_1 \cup \{\smv : \tuple \sort \}$:
  \[\qquad
    \irule{
      \begin{array}{c}
        \pjdg[\C \land \C'][\envmv_1'][P'][@][\Gamma', x : \sort]
        \qquad
        \pjdg[\C \land \C'][\smv : \tuple \sort][\queue \smv \qmv][@]
             [\Gamma',x : \sort]
      \end{array}
    }{
      \pjdg[\C \land \C'][\envmv'][S'][@][\Gamma',x : \sort]
    }\quad \myrule{VPar}
  \]
  where the second judgment in the premiss holds by
  \cref{lem:typing-weakening} and the application of \myrule{VPar} is
  possible because $\envmv'_1$ and $\smv : \tuple \sort$ are
  independent since $\smv \not\in \dom{\envmv'_1}$ otherwise
  $\envmv_1$ would not be independent of $\envmv|_{\shset}, \smv :
  \tuple \sort$.
  The proof of this case ends by showing that
  $\consSt[\sigma'][P'][\envmv'][\C\land\C'][\Gamma', x : \sort]$.
  Since $\fY[S'] = \fY[P'] \cup \{ \smv \}$, we have that~\eqref{cons2:eq} of \defref{def:consistent} holds; also,
  $\dom{\envmv'|_{\chset}} = \dom{\envmv'_1|_{\chset}} \cup \{\smv\}$
  by definition and, by \lemref{monotone:lem}, $\dom{\sigma} \subseteq
  \dom{\sigma'}$; also, conditions~\eqref{cons3:eq}
  and~\eqref{cons4:eq} hold because $\sigma'|_{\varset} =
  \sigma|_{\varset}$ and $\sigma'(x) = \val v$ has sort $\sort =
  \Gamma'(x)$; finally,~\eqref{cons5:eq} holds because for each
  $\psnames \smv p$ we have $\envmv' \psnames \smv p = \envmv \psnames
  \smv p$ by construction.

  \item{\myrule{SPar}} In this case we have:
  \[
    \irule{
      \state[S_1]
      \Ptrans{\lcond{\C'} \alpha}
      \state[S_1'][\sigma']
      \quad
      \bn \alpha \cap \dom{\sigma} = \emptyset
      \quad
      \fX[S_2] \cap (\dom{\sigma'} \setminus \dom{\sigma}) = \emptyset
    }{
      \state[\underbrace{S_1 \sparop S_2}_{=S}]
      \Ptrans{\lcond{\C'} \alpha}
      \state[\underbrace{S'_1 \sparop S_2}_{=S'}][\sigma']
    }
  \]
  and the only applicable rule to type $S$ is \myrule{VPar}:
  \[\mathrule{
      \pjdg[\C_1][\envmv_1][S_1][\envmv_1'][\Gamma]
      \qquad
      \pjdg[\C_2][\envmv_2][S_2][\envmv_2'][\Gamma]
      \qquad
      \envmv_1\text{ and }\envmv_2\text{ independent}
    }{
    \pjdg[\underbrace{\C_1 \land \C_2}_{=\C}]
         [\underbrace{\envmv_1 \cup \envmv_2}_{=\envmv}][S_1 \sparop S_2]
         [@][\Gamma]
    }{}
  \]
  Noting that $\consSt$ trivially implies
  $\consSt[@][S_1][\envmv_1][\C_1][@]$, by the inductive hypothesis
  there $\Gamma'$ and $\envmv_1'$ such that
  \begin{enumerate}
  \item[\eqref{sr:in}] if $\alpha = \lreceive \smv {\val v}$ then
    $\envmv_1 \Ttrans{\lreceive \smv \sort} \envmv_1'$ for a sort
    $\sort$; also, if $\vjdg[][\val v][\sort]$ then there is
    $x \in \varset$ such that $\sigma'(x) = \val v$,
    $\pjdg[\C_1 \land \C'][\envmv_1'][S_1'][][\Gamma', x : \sort]$ and
    $\consSt[\sigma'][S_1'][\envmv_1'][\C_1 \land
    \C'][\Gamma', x : \sort]$
  \item[\eqref{sr:out}] if $\alpha = \lsend \smv {\val v}$ then
    $\envmv_1 \Ttrans{\lsend \smv \sort} \envmv_1'$ with
    $\vjdg[][\val v][\sort]$,
    $\consSt[\sigma'][S_1'][\envmv_1'][\C_1 \land \C']$, and
    $\pjdg[\C_1 \land \C'][\envmv_1'][S_1']$
  \item[\eqref{sr:other}] otherwise
    $\envmv_1 \Ttrans \alpha \envmv_1'$,
    $\consSt[\sigma'][S_1'][\envmv_1'][\C_1 \land \C'][\Gamma']$, and
    $\pjdg[\C_1 \land \C'][\envmv_1'][S_1'][][\Gamma']$
  \end{enumerate}
  In the first two cases $\envmv_1'$ and $\envmv_2$ are independent
  because $\dom{\envmv_1'} = \dom{\envmv_1}$ and $\envmv_1$ and
  $\envmv_2$ are independent.
  In the latter case, we observe that
  $\consSt[\sigma'][S_1'][\envmv_1'][\C_1 \land \C'][\Gamma']$ implies
  $\dom{\envmv_1'} \subseteq \dom{\sigma'}$; moreover, $\sigma'$ can
  contain a new session, say $\tuple \smv$ only if the rule
  \myrule{SInit} had been used to derive the transition.
  In this case, the side condition of \myrule{SInit} guarantees that
  $\tuple \smv \not\in \dom \sigma$ and therefore they do not occur in
  $S_2$ and therefore they are not in $\dom{\envmv_2}$ since
  $\consSt$.
  The judgement $\pjdg[\C_2][\envmv_2][S_2][\envmv_2'][\Gamma']$ holds
  because condition
  $\fX[S_2] \cap (\dom{\sigma'} \setminus \dom{\sigma}) = \emptyset$
  makes $\fX[S_2] \cap \dom{\Gamma'} = \emptyset$, hence we apply
  \cref{lem:weakening-env-in-runs}.

  \item{\myrule{SNew}} In this case we have:
  \[\qquad\qquad
  \irule{
      \state[S_1]
      \Ptrans{\lcond {e'} \alpha}
      \state[S'_1][\sigma']
      \qquad
      \tuple \smv \cap \fY[\alpha] = \emptyset
    }{
      \state[\underbrace{(\nu \tuple \smv \At \shname) S_1}_{=S}]
      \Ptrans{\lcond {e'} \alpha}
      \state[\underbrace{(\nu \tuple \smv \At \shname) S'_1}_{=S'}][\sigma']
    }
    \qquad\qquad
    \mathrule{
      \pjdg[@][\envmv'][S][][\Gamma]
      \qquad
      \envmv = \envmv'|_{-\tuple \smv}
    }{
      \pjdg[@][\envmv][(\nu \tuple \smv@\shname)S][@][\Gamma]
    }{}
  \]
  where on the right we have the typing of $S$ obtained by applying
  rule \myrule{VNew}.
  The thesis immediately follows from the inductive hypothesis since
  $\fY[S] = \fY[S_1] \setminus \tuple \smv$ and
  $\fY[S'] = \fY[S'_1] \setminus \tuple \smv$.

  \item{\myrule{SStr}}
  The thesis is immediate by using the inductive hypothesis.
  \qedhere
  \end{itemize}
\end{proof}

%%% Local Variables:
%%% mode: latex
%%% TeX-master: "main"
%%% End:

\section{WSI by typing}\label{sec:proof-wsr-typing}
% !TEX root = main.tex
\subsection{Correspondence between semantics of specifications}%
\label{sec:corrSpec}
The next results show that the denotational semantics of
specifications coincides with the operational rules given in
\cref{fig:LTStypes}.

\begin{lem}\label{lem:env-weakening}
  If $\envmv\Ttrans{\alpha}\envmv'$ then
  $\envmv,\envmv''\Ttrans{\alpha}\envmv',\envmv''$.
\end{lem}
\begin{proof}
  By straightforward induction on the structure of the proof
  $\envmv\Ttrans{\alpha}\envmv'$.
\end{proof}

\begin{lem}\label{lem:merge-reduction}
  % Let $\envmv_1$ and $\envmv_2$ be mergeable specifications. If
  $\envmv_1 \Ttrans{\alpha}\envmv_1'$, then
  ${\envmv_1}{\peace}{\envmv_2} \Ttrans{\alpha} \envmv_1'$.
\end{lem}
\begin{proof}
  By induction on the structure of the derivation $\envmv_1
  \Ttrans{\alpha} \envmv_1'$.
\end{proof}

\begin{lem}\label{lem:weakening-env-in-runs}
  $\typeruns{\tuple \smv}{\envmv_1, \envmv_2} = \typeruns{\tuple
    \smv}{\envmv_1}$ if $\dom{\envmv_2} \cap \tuple\smv = \emptyset$.
\end{lem}
\begin{proof} By straightforward induction on the structure of the
  proof $r\in\typeruns{\tuple \smv}{\envmv_1,\envmv_2}$.
\end{proof}

\begin{lem}\label{lem:preservation-of-disjoint-endpoint-types}
  Let $\envmv$ be a specification and $\tuple \smv \in \dom\envmv$.  If
  $\envmv \Ttrans{\alpha} \envmv'$ and
  $\names[\alpha]\cap\tuple\smv=\emptyset$, then
  $\envmv\psnames\smv p = \envmv'\psnames\smv p$ for all
  $\psnames \smv p\in\envmv$.
\end{lem}
\begin{proof}
  By straightforward induction on the structure of the proof
  $\envmv\Ttrans{\alpha}\envmv'$.
\end{proof}

\speclem*
\begin{proof}
  By induction on the structure of the proof
  $\envmv\Ttrans{\tau}\envmv'$.
  \begin{itemize}
  \item \myrule{TCom_1}:
    $\envmv = \envmv_1, \psnames \smv p : \Ipst[\C_i]{i \in
    I}{\smv_i}{\sort_i}{\pst_i}, \smv_j : \queue{}{\tuple \sort}$
    and
    $\envmv' = \envmv_1, \psnames \smv p : \pst_j, \smv_j :
    \queue{}{\tuple \sort \cdot \sort_j}$.  Then, by
    \myrule{RTCom_1}, $r\in\typeruns{\tuple \smv}{\envmv'}$ implies
    $\state[\ptp p][\MAsend{\smv}{\sort}]r\in\typeruns{\tuple
      \smv}{\envmv}$.

  \item \myrule{TCom_2}: $\envmv = \envmv_1, \psnames \smv p :
    \Epst[\C_i]{i \in I}{\smv_i}{\sort_i}{\pst_i}, \smv_j :
    \queue{}{\sort_j \cdot \tuple \sort}$ and $\envmv' = \envmv_1,
    \psnames \smv p : \pst_j, \smv_j : \queue{}{\tuple \sort}$.  Then,
    by \myrule{RTCom_2}, $r\in\typeruns{\tuple \smv}{\envmv'}$ implies
    $\state[\ptp p][\lreceive {\smv_j} {\sort_j} ]r\in\typeruns{\tuple
      \smv}{\envmv}$.

  \item \myrule{TInit}: Then, $\envmv'=\envmv, \psnames {\smv'} {p_0}
    : \pst_0, \ldots, \psnames {\smv'} {p_n} : \pst_n, \tuple {\smv'}
    : \queue{}{}$.  By well-formedness conditions on specifications,
    $\tuple \smv \subseteq \dom\envmv $ implies $\tuple \smv \cap
    \tuple \smv'$. The case follows by using
    \lemref{lem:weakening-env-in-runs} to conclude that
    $\typeruns{\tuple \smv}{\envmv'}=\typeruns{\tuple \smv}{\envmv}$.

  \item \myrule{TSeq}: Then, $\envmv =\envmv_1;\envmv_2$, and
    $\envmv'=\envmv_1';\envmv_2$ with
    ${\envmv_1}\Ttrans{\tau}{\envmv_1'}$.  By rule \myrule{RTSeq},
    $r=r_1r_2$ with $r_1\in\typeruns{\tuple \smv}{\envmv_1'}$ and
    $r_2\in\typeruns{\tuple \smv}{\envmv_2}$.  By inductive
    hypothesis, $r_1\in\typeruns{\tuple \smv}{\envmv_1'}$ implies
    $r_1\in\typeruns{\tuple \smv}{\envmv_1}$, $\state[\ptp
      p][\MAsend{\smv}{\sort}] r_1\in \typeruns{\tuple
      \smv}{\envmv_1}$, $\state[\ptp q][\Treceive{\smv}{\sort}] r_1\in
    \typeruns{\tuple \smv}{\envmv_1}$, or $r_1 = r_1'r_2'r_3'$ and
    $r_1'[r_2']r_3'\in \typeruns{\tuple \smv}{\envmv_1}$. Then, the
    proof is completed by using rule \myrule{RTSeq}.

  \item \myrule{TLoop_0}: Then, $\envmv' = \envmv \proj\shset$.  By
    inspection of rules in~\cref{fig:def-runs-local-types}, we
    conclude that $r\in\typeruns{\tuple \smv}{\envmv'}$ implies $r =
    \epsilon$. Then, $r\in\typeruns{\tuple \smv}{\envmv}$ by
    \myrule{RTEnd}.

  \item \myrule{TLoop_1}: Then, $\envmv = \Tdef{\envmv_1}$ and
    $\envmv' = \envmv_1$.  By \myrule{RTIt1}, $r\in\typeruns{\tuple
    \smv}{\envmv'}$ implies $r\in\typeruns{\tuple \smv}{\envmv}$.

  \item \myrule{TLoop_2}: Then, $\envmv = \Tdef{\envmv_1}$ and
    $\envmv' = \envmv_1;\Tdef{\envmv_1}$.  By inspection of rules
    in~\cref{fig:def-runs-local-types}, we conclude that
    $r\in\typeruns{\tuple \smv}{\envmv'}$ implies $r = r_1r_2$ with
    $r_1\in \typeruns{\tuple \smv}{\envmv_1}$ and $r_2\in
    \typeruns{\tuple \smv}{\Tdef{\envmv_1}}$.  By \myrule{RTIt2},
    $r_1[r_2]\in\typeruns{\tuple \smv}{\envmv}$.
    \qedhere
  \end{itemize}
\end{proof}

\specbacklem*

\begin{proof}
  The proof follows by induction on the structure of proof
  $r\in\typeruns{\tuple \smv}{\envmv}$.

  \begin{itemize}
  \item \myrule{RTCom_1} follows immediately by taking $r= \state[\ptp
    p][\lsend {\smv} {\sort} ]r'$ with $r' \in\typeruns{\tuple
    \smv}{\envmv'}$ and using \myrule{TCom_1}.

  \item \myrule{RTCom_2} follows immediately by taking $r= \state[\ptp
    p][\lreceive {\smv} {\sort} ]r'$ with $r' \in\typeruns{\tuple
    \smv}{\envmv'}$ and using \myrule{TCom_2}.

  \item \myrule{RTSeq} Then $\envmv = \envmv_1;\envmv_2$, $r=r_1r_2$
    with $r_1\in\typeruns{\tuple \smv}{\envmv_1}$ and
    $r_2\in\typeruns{\tuple \smv}{\envmv_2}$. The proof is completed
    by using inductive hypothesis on $r_1\in\typeruns{\tuple
      \smv}{\envmv_1}$ to conclude that
    $\envmv_1\Ttrans{\tau}\envmv'_1$. Then, by \myrule{TSeq}, $\envmv
    = \envmv_1;\envmv_2 \Ttrans{\tau} \envmv_1';\envmv_2$. Finally,
    take $\envmv = \envmv_1';\envmv_2$ and use rule \myrule{RTSeq} to
    conclude that either $r\in\typeruns{\tuple \smv}{\envmv'}$ or
    $r'\in\typeruns{\tuple \smv}{\envmv'}$.

  \item \myrule{RTIt1}. Then $\envmv = \Tdef{\envmv_1}$ and
    $r\in\typeruns{\tuple \smv}{\envmv_1}$.  The proof is completed by
    taking $\envmv' = \envmv_1$ and using rule \myrule{TLoop_1} to
    derive $\envmv\Ttrans{\tau}\envmv'$.

  \item \myrule{RTIt2}. Then $\envmv = \Tdef{\envmv_1}$ and $r =
    r_1[r_2]$ with $r_1\in\typeruns{\tuple \smv}{\envmv_1}$ and
    $r_1\in\typeruns{\tuple \smv}{\Tdef{\envmv_1}}$.  The proof is
    completed by taking $\envmv' = \envmv_1;\Tdef{\envmv_1}$ and using
    rule \myrule{TLoop_2} to derive $\envmv\Ttrans{\tau}\envmv'$.
    \qedhere

  \end{itemize}

\end{proof}

%%% Local Variables:
%%% mode: latex
%%% TeX-master: "main"
%%% End:

% !TEX root = main.tex

\subsection{Runs of Specifications}

\begin{lem}%
\label{lem:runs-preserved-by-extchoice}
Let
$r \in \typeruns{\tuple \smv}{\envmv, \psnames \smv p:
  \Tbbra{\smv_i}{\sort_i}{\pst_i}}$
and
$\smv\in\tuple \smv$ and $\smv \neq \smv_i$ for all $i\in I$, then
$r \in \typeruns{\tuple \smv}{\envmv, \psnames \smv p:
\Treceive\smv\sort;\pst +\Tbbra{\smv_i}{\sort_i}{\pst_i}}$.
\end{lem}

\begin{proof} By straightforward induction on the structure of the proof.
\end{proof}

\begin{lem}%
\label{lem:runs-preserved-by-intchoice}
Let
$r \in \typeruns{\tuple \smv}{\envmv, \psnames \smv p:
  \Tssel{\smv_i}{\sort_i}{\pst_i}}$
and
$\smv\in\tuple \smv$ and $\smv \neq \smv_i$ for all $i\in I$, then
$r \in \typeruns{\tuple \smv}{\envmv, \psnames \smv p:
  \Treceive\smv\sort;\pst \oplus\Tssel{\smv_i}{\sort_i}{\pst_i}}$.
\end{lem}

\begin{proof} By straightforward induction on the structure of the proof.
\end{proof}

\begin{lem}
  Let $\envmv$ and $\envmv'$ be two specifications such that
  $\rmg{\envmv \psnames \smv p} =\rmg{\envmv \psnames \smv p}$ for all
  $\psnames \smv p \in\dom {\envmv'}$.  Then, $\typeruns{\tuple
    \smv}{\envmv}=\typeruns{\tuple \smv}{\envmv'}$.
\end{lem}

\begin{proof}
  It follows by straightforward induction on the derivation of
  $r\in\typeruns{\tuple \smv}{\envmv}$ after noticing that guards are
  irrelevant for deriving runs.
\end{proof}

%%% Local Variables:
%%% mode: latex
%%% TeX-master: "main.tex"
%%% End:

% !TEX root = main.tex
\subsection{Specifications cover Global types}

\begin{lem}\label{lem:aux-conv-recursive}
  Let  $\Geq$ be a global type with
  $\participants{\GT} = \{\ptp {p_0},\ldots, \ptp {p_n}\}$
  and
  $\envmv = \psnames \smv {p_0} : \GT\proj{\ptp p_0}, \ldots,
  \psnames \smv {p_n} : \GT\proj{\ptp p_n}, \tuple\smv : \queue{}{}$.
  If  $\imruns{\GT} \ \subseteq\ \typeruns{\tuple \smv}{\envmv}$, then
  $r\in\imrunsaux{\Gdef \f \GT}\subseteq\ \typeruns{\tuple \smv}{\Tdef{\envmv}}$.
\end{lem}

\begin{proof}
  By induction on the structure of the derivation
  $r\in\imrunsaux{\Gdef \f \GT}$
  \begin{itemize}
  \item \myrule{RG^*_1}: then $r\in\imruns{\GT}$. By hypothesis, $r\in
    \typeruns{\tuple \smv}{\envmv}$. By \myrule{RTIt1}, $r\in
    \typeruns{\tuple \smv}{\Tdef{\envmv}}$.
  \item \myrule{RG^*_1}: then $r = r_1[r_2]$, $r_1\in\imruns{\GT}$ and
    $r_2\in\imrunsaux{\Gdef \f \GT}$. By hypothesis,
    $r_1\in \typeruns{\tuple \smv}{\envmv}$.
    By inductive hypothesis $r_2\in \typeruns{\tuple \smv}{\Tdef{\envmv}}$.
    By \myrule{RTIt2}, $r_1[r_2]\in\typeruns{\tuple \smv}{\Tdef{\envmv}}$.
    \qedhere
  \end{itemize}
\end{proof}

\coveragethm*
\begin{proof}
  We show a stronger result proving that $ \imruns{\GT} \subseteq
  \typeruns{\tuple \smv}{\envmv}$, which obviously implies
  $\imruns{\GT} \ \mysubseteq\ \typeruns{\tuple \smv}{\envmv}$.
  We proceed by induction on the derivation of $r\in\imruns{\GT}$. We
  proceed by case analysis on the last applied rule in the derivation
  of $r\in\imruns{\GT}$.

  \begin{itemize}

  \item Case \myrule{RGEnd}: Follows straightforwardly from rule
    \myrule{RTEnd} observing that $\GT = \Gend$ and $\envmv$ is
    $\Tend$-only by the definition of projection (\cf\ on
    page~\pageref{page:projection}).

  \item Case \myrule{RGCom}: Then, $\GT = \Gchoice$ and
    $r = \state[\ptp p,\lsend {\smv_h}{\sort_h}] \state[\ptp
    q_h][\lreceive{\smv_h}{\sort_h}] r'$ for an $h \in I$ and an
    $r' \in \imruns{\GT_h}$.

    Then, letting $J_i = \{j \in I \sst \ptp q_j = \ptp q_i\}$ and
    $K_i = I \setminus J_i$ for all $i \in I$
    \[
    \qquad\qquad
    \begin{array}{l@{\ =\ }l@{\hspace{2cm}}l@{\ =\ }l}
      \GT\proj{\ptp p}
      &
        \TSsel{i\in I}{\smv_i}{\sort_i}{\GT_i\proj{\ptp p}}
      &
        \GT\proj{\ptp q_i}
      &
        \TBbra{j}{J_i}{\smv_j}{\sort_j}{\GT_j\proj{\ptp r}}
        + \displaystyle{\sum_{k\in K_i}{\GT_k\proj{\ptp r}} }
    \end{array}
    \]
    By inductive hypothesis, we know that
    \begin{equation}\label{eq-th4-local-covers-global-ch}
      r'
      \in\typeruns{
        \smv
      }{
        \psnames \smv {p_0} : \GT_h\proj{\ptp p_0}, \ldots,
	\psnames \smv {p_n} : \GT_h\proj{\ptp p_n}, \tuple\smv : \queue{}{}
      }
    \end{equation}

    By repeatedly applying~\lemref{lem:runs-preserved-by-extchoice}
    over~\eqref{eq-th4-local-covers-global-ch}, for each $\ptp
    p_i\neq \ptp p_h$, we conclude that
    \begin{equation*}
      r'\in\typeruns{\smv}{\envmv'}
      \qquad\text{where}\qquad
      \envmv'
      :
      \begin{cases}
        \tuple\smv \mapsto \queue{}{}
        \\
        \psnames \smv p \mapsto \GT_h\proj{\ptp p}
        \\
        \psnames \smv {q_h} \mapsto \GT_h \proj {\ptp q_h}
        \\
        \psnames \smv q \mapsto \GT \proj {\ptp q} &
        \text{if } \ptp q \not \in \{\ptp p, \ptp q_h\}
        \\
        \text{undefined} & \text{otherwise}
      \end{cases}
    \end{equation*}
    The proof is completed by using rules \myrule{RTCom_1} and
    \myrule{RTCom_2} as follows:
    \[
    \irule{
      \irule
      {
        \irule
        {}
        {r'\in\typeruns{\smv}
          {\envmv'}
        }
      }
      {
        \begin{array}{l}
          \state[\ptp q_h][\lreceive {\smv_h} {\sort_h} ]r'\
          \in
          \typeruns{\smv}{
          \envmv'
          \upd{\queue{}{\sort_h}, \ {\GT\proj{\ptp q_h}}}
              {\tuple \smv, \ {\psnames \smv {q_h}}}
          }
        \end{array}
      }
      {\myrule{RTCom_2}}
    }
    {
      \state[\ptp p,\lsend {\smv_h}{\sort_h}]
      \state[\ptp q_h][\lreceive{\smv_h}{\sort_h}]r'
      \in\typeruns{\smv}{\envmv}
    }
    \myrule{RTCom_1}
    \]

  \item
    Case \myrule{RGSeq}: Then $\GT = \GT_1;\GT_2$ and $r = r_1r_2$
    with $r_i \in \imruns{\GT_i}$ for $i \in \{1,2\}$ and
    \begin{equation*}
      r_i \in
      \typeruns{\smv}
      {\psnames \smv {p_0} : \GT_i\proj{\ptp p_0}, \ldots,
        \psnames \smv {p_n} : \GT_i\proj{\ptp p_n}, \tuple\smv : \queue{}{}}
    \end{equation*}
    by inductive hypothesis; hence, by rule \myrule{RTSeq}
    \begin{equation*}
      r_1r_2\in
      \typeruns{\smv}
      {\psnames \smv {p_0} : \GT\proj{\ptp p_0}, \ldots,
        \psnames \smv {p_n} : \GT\proj{\ptp p_n}, \tuple\smv : \queue{}{}}
    \end{equation*}

  \item Case \myrule{RGIter}: Then, $\GT = \Gdef{\f}{\GT_1}$ with
    $\participants{\GT_1} = \{\ptp p_0, \ptp p_1,\ldots,\ptp p_n\}$,
    $\ready{\GT_1} = \ptp p_0$, $f(\ptp p_i) = \smv_i\ \sort_i$ for
    all $1 \leq i \leq n$, and
    $r = r' \state[\ptp p_0][\lsend{\smv_1}{\sort_1}]
    \state[\ptp p_1][\lreceive{\smv_1}{\sort_1}] \ldots
    \state[\ptp p_0][\lsend{\smv_n}{\sort_n}]
    \state[\ptp p_n][\lreceive{\smv_n}{\sort_n}]$
    with $r'\in\imrunsaux{\Gdef{\f}{\GT}}$.  Therefore,
    \[
    \begin{array}{l}
      \qquad\envmv =
      \psnames \smv {p_0} : \Tdef{(\GT_1\proj{\ptp p_0})};
      \Tsend{\smv_1}{\sort_1};\ldots ; \Tsend{\smv_n}{\sort_n},
      \psnames \smv {p_1} : \Tdef{(\GT_1\proj{\ptp {p_1}})};
      \Treceive{\smv_1}{\sort_1}, \ldots, \qquad\qquad
      \\
      \hfill
      \psnames \smv {p_n} : \Tdef{(\GT_1\proj{\ptp {p_1}})};
      \Treceive{\smv_n}{\sort_n},
      \tuple\smv : \queue{}{}
    \end{array}
    \]

    Note that $\envmv$ can be written as the sequential composition
    $\envmv = \Tdef{\envmv_1};\envmv_2$ where
    \[
    \begin{array}{ll}
      \envmv_1 =
      &
      \psnames \smv {p_0} : {(\GT_1\proj{\ptp p_0})},
      \psnames \smv {p_1} : {(\GT_1\proj{\ptp {p_1}})}, \ldots,
      \psnames \smv {p_n} : {(\GT_1\proj{\ptp {p_n}})},
      \tuple\smv : \queue{}{}
      \\
      \envmv_2=
      &
      \psnames \smv {p_0} : \Tsend{\smv_1}{\sort_1};\ldots;
      \Tsend{\smv_n}{\sort_n},
      \psnames \smv {p_1} : \Treceive{\smv_1}{\sort_1}, \ldots,
      \psnames \smv {p_n} : \Treceive{\smv_n}{\sort_n},
      \tuple\smv : \queue{}{}
    \end{array}
    \]

    Then, by repeated used of rules \myrule{RTCom_1} and
    \myrule{RTCom_2} we can build a proof for
    $r''
    =
    \state[\ptp p_0][\lsend{\smv_1}{\sort_1}]
    \state[\ptp p_1][\lreceive{\smv_1}{\sort_1}]
    \ldots
    \state[\ptp p_0][\lsend{\smv_n}{\sort_n}]
    \state[\ptp p_n][\lreceive{\smv_n}{\sort_n}]
    \in
    \typeruns{\tuple \smv}{\envmv_2}$,
    as illustrated by the following sketch
    \[\qquad
    \irule
    {
      \irule{
        \irule{
          \irule{}{
            \begin{array}{c}
              \epsilon\in\typeruns{
              \tuple \smv
              }{
              \psnames \smv {p_0} : \Tend, \ldots,
              \psnames \smv {p_n} : \Tend,
              \tuple\smv : \queue{}{}
              }
              \\[15pt]
              \vdots
              \\[15pt]
            \end{array}
          }
          \myrule{RTEnd}
        }{
          \begin{array}{l}
            \state[\ptp p_0][\lsend{\smv_2}{\sort_2}]
            \state[\ptp p_2][\lreceive{\smv_2}{\sort_2}]
            \ldots
            \state[\ptp p_0][\lsend{\smv_n}{\sort_n}]
            \state[\ptp p_n][\lreceive{\smv_n}{\sort_n}]\in
            \hspace{2.3cm}
            \\
            \hfill
            \typeruns{
              \tuple \smv
            }{
              \envmv_2\upd{\Tend,\
                \Tsend{\smv_2}{\sort_2};\ldots ;
                \Tsend{\smv_n}{\sort_n}}{\psnames \smv {p_1},
                \ \psnames \smv {p_0}}
            }
          \end{array}
        }
        \myrule{RTCom_1}
      }{
        \begin{array}{l}
          \state[\ptp p_1][\lreceive{\smv_1}{\sort_1}]
          \state[\ptp p_0][\lsend{\smv_2}{\sort_2}]
          \state[\ptp p_2][\lreceive{\smv_2}{\sort_2}]
          \ldots
          \state[\ptp p_0][\lsend{\smv_n}{\sort_n}]
          \state[\ptp p_n][\lreceive{\smv_n}{\sort_n}]\in
          \hspace{2.3cm}
          \\
          \hfill
          \in\typeruns{
            \tuple \smv
          }{
            \envmv_2\upd{\queue{}{\sort_1\AT\ptp p_0},\
              \Tsend{\smv_2}{\sort_2};\ldots ; \Tsend{\smv_n}{\sort_n}}
                  {\smv_1,\ \psnames \smv {p_0}}
          }
        \end{array}
      }
      \myrule{RTCom_2}
    }
    {r'' \in\typeruns{\tuple \smv}{\envmv_2}}
    \myrule{RTCom_1}
    \]

    We now show that $r'\in\imrunsaux{\Gdef \f \GT_1}$ implies
    $r'\in\typeruns{\tuple \smv}{\Tdef{\envmv_1}}$.
    By inductive hypothesis, $ \imruns{\GT_1}\subseteq
    \typeruns{\tuple \smv}{\envmv_1}$.
    By \lemref{lem:aux-conv-recursive}, $\imrunsaux{\Gdef \f
      \GT}\subseteq \typeruns{\tuple \smv}{\Tdef{\envmv_1}}$.
    Therefore, $r'\in\imrunsaux{\Gdef \f \GT}$ implies
    $r'\in\typeruns{\tuple \smv}{\Tdef{\envmv_1}}$.
    By \myrule{RTSeq},
    $r = r'r''\in\typeruns{\tuple \smv}{\Tdef{\envmv_1};\envmv_2}$
    since $r'\in\typeruns{\tuple \smv}{\Tdef{\envmv_1}}$ and
    $r'' \in\typeruns{\tuple \smv}{\envmv_2}$.
    \qedhere
  \end{itemize}
\end{proof}

\subsection{Implementations cover specifications}%
\label{app:imp-cover-spec}
\begin{defi}[Viable types]
  A pseudo-type $\pst$ is \emph{viable} if either of the
  following holds
  \begin{itemize}
  \item
    $\nfp{\pst} = \guard \Tend$
  \item
    $\nfp{\pst} = \Ipst[e_i]{i \in I}{\smv_i}{\sort_i}{\pst_i}$ or
    $\nfp{\pst} = \Epst[e_i]{i \in I}{\smv_i}{\sort_i}{\pst_i}$ with
    $\pst_i$ viable for all $i \in I$
  \item $\nfp{\pst} = \Tdef{\pst_1}; \pst_2$ with $\pst_1$ and
    $\pst_2$ viable, and either $\pst_1$ and $\pst_2$ passively
    compatible or $\pst_1 = \Ipst[e_i]{i \in
      I}{\smv_i}{\sort_i}{\pst_i}$.
  \end{itemize}
  A specification $\envmv$ is \emph{viable} when every type in $\envmv$ is
  viable.
\end{defi}

\begin{lem}
  Let $\G(\tuple \smv)$ a global type and $\envmv$ a specification
  such that (i) $\dom \envmv = \{\psnames \smv q \sst \ptp q \in
  \participants \G\}$ and (ii) $\envmv : \psnames \smv q \mapsto
  \G\proj{\ptp q}$ for all $\ptp q \in \participants \G$.
  Then $\envmv$ is viable.
\end{lem}

In what follows we write
$\state[S][\sigma] \Wtrans{\lcond\CC{\seq\alpha}}\state[S'][\sigma']$
with
$\CC=\C_1\land\ldots\land\C_n$ and $\seq\alpha= \alpha_1\ldots\alpha_n$
for the sequence
$\state[S][\sigma]
\Ptrans{\lcond{\C_1}{\alpha_1}}
\state[S_2][\sigma_2]\Ptrans{\lcond{\C_2}{\alpha_2}}
\ldots \Ptrans{\lcond{\C_n}{\alpha_n}}\state[S'][\sigma]$.
Let $\names[\alpha] = \fY[\alpha] \cup \bn\alpha$; the definitions of
$\names$, $\fY$ and $\bn{\_}$ straightforwardly extend to sequences of
labels.

\begin{lem}\label{lem:inhabitant-types}
  Let $\pst$ be a viable pseudo-type and $\C$ such that $\nform{\pst}
  = {\pst}$. Then, there exist $\Gamma$ and $P$ and $\envmv$ and
  $\pst'$ such that for any $\tuple \smv \supseteq \fY[\pst]$ and
  $\ptp p$ it holds that
  $\pjdg[\C][\envmv,\psnames \smv {p} :\pst'][P][@][@]$ and
  $\rmg\pst = \rmg{\pst'}$.
\end{lem}

\begin{proof}
  The proof follows by induction on the structure of $\pst$.

  \begin{itemize}
  \item
    $\pst = \guard[e_j]\Tend$. It follows immediately by using rule
    $\myrule{VEnd}$.

  \item
    $\pst = \Ipst[\C_i]{i \in I}{\smv_i}{\sort_i}{\pst_i}$. Note that
    $\nform{\pst} = {\pst}$ implies $\nform{\pst_i} = {\pst_i}$ for
    all $i\in I$.
    Then, by inductive hypothesis there exist $\Gamma_i$, $P_i$,
    $\envmv_i$, $\pst_i$ for any $\tuple \smv \supseteq \fY[\pst_i]$
    and $\ptp p$ it holds that
    $\pjdg[\C][\envmv_i,\psnames \smv {p} :\pst_i'][P_i][@][\Gamma_i]$
    and $\rmg{\pst_i} = \rmg{\pst_i'}$.
    Then, take $s$ and $x$ fresh, i.e.,
    $s\not\in\dom{\envmv'_0,\ldots,\envmv'_n}$ and
    $x\not\in\dom{\Gamma'_0,\ldots,\Gamma'_n}$ and note that
    \[
    \pjdg[\C][\envmv_i, \psnames\smv p :\pst_i', \psname s {p'} :
    \guard\Tend][P_i][@][\Gamma_i, x:\sort]
    \]
    for any $\sort$ and ${\ptp{p'}}$.
    By~\lemref{thm:distribution-end-pseudo-types},
    \[
    \pjdg[\C\wedge\C'_i][
    \envmv_i', \psnames\smv{p} :\pst_i', \psname s {p'} : {\guard[\C']\Tend}
    ][P_i][@][\Gamma_i,x:\sort]
    \]
    with $\C'_i = x\neq 0 \wedge\ldots \wedge x\neq i-1 \wedge x = i$
    and
    \[
    \begin{array}{lcl}
      \envmv_i' \psnames \smv p
      &
      =
      &
      \guard[\C \wedge \C_i']\Tend; \envmv_i \psnames \smv p
      \\
      \pst_i'
      &
      =
      &
      \guard[\C \wedge \C_i']\Tend;{\pst_i}
      \\
      \guard[\C']\Tend
      &
      =
      &
      \guard[\C \wedge \C_i']\Tend; \guard\Tend
    \end{array}
    \]
    By definition of $\nform[\_]{\_}$,
    $\pst_i' = \nform[\C \wedge\C_i']{\pst_i}$.
    Since $x\not\in\fn{\pst}$ and $\nform{\pst} = {\pst}$ we have that
    $\rmg{\pst_i'} = \rmg{\pst_i}$.
    Moreover, by typing rules~\myrule{VSend} and~\myrule{VSeq},
    \[
    \pjdg[\C\wedge \C_i'][\envmv_i''][\Psend{\smv_i}{c_i};P_i][@]
         [\Gamma_i,x:\sort]
    \]
    such that $c_i$ is some constant of type $\sort_i$ and
    \[
    \envmv''_i \psnames \smv p
    =
    \guard[\C \wedge \C_i']{\Tsend {\smv_i}{\sort_i};
      \envmv_i' \psnames \smv p}
    \]
    Then, take $\pst' = \Ipstaux{i\in I}{\envmv''_i\psnames\smv p}$.
    Note that $\pst'$ is well-defined (because the guards in
    all types are different) and $\rmg{\pst'} = \rmg{\pst}$ because
    $\rmg{\pst_i'} = \rmg{\pst_i}$ holds for every $i$.

    Finally, define
    $\envmv
    =
    \envmv_0, \ldots, \envmv_n, \psnames\smv p :\pst',
    \psname s {p'} : \guard\Tend$,
    and
    $\Delta
    =
    \Delta_0, \ldots, \Delta_n$ and
    $P
    =
    \Preceive{s}{x}{
      \Pif {x = 0}
           {P_0}
           {\Pif {x = 1} {P_1}{\ldots}}
    }$.
    The proof is concluded by straightforward use of typing rules.

  \item
    $\Epst[e_i]{i \in I}{\smv_i}{\sort_i}{\pst_i}$. It follows
    analogously to the previous case.

  \item
    $\Tdef{\pst_1}; \pst_2$ with $\pst_1$ with
    $\pst_1
    =
    \Ipst[e_i]{i \in I}{\smv_i}{\sort_i}{\pst_i}$.
    It follows by inductive hypothesis on both $\pst_1$ and $\pst_2$
    to conclude that there exist $P_1$ and $P_2$ and by taking $P =
    \Pfor{x}{\ell}{P_1}; P_2$.

  \item
    $\Tdef{\pst_1}; \pst_2$ and $\pst_1$ and $\pst_2$ passively
    compatible.  It follows by inductive hypothesis on both $\pst_1$
    and $\pst_2$ to conclude that there exist $P_1$ and $P_2$ and by
    taking $P =\PloopU{P_1} P_2$.

  \end{itemize}

\end{proof}

\begin{lem}\label{lem:subs}
  If $\pjdg$ and $x \not\in \fX[P]$ then,   for all $x' \in \varset$,
  $\pjdg[@][\envmv'][P\subs{x'}{x}][@][\Gamma|_{-x'}\upd{\Gamma(x')}{x}]$
  where $\envmv'|_{\shset \cup \chset} = \envmv|_{\shset \cup \chset}$ and
  for all $\psnames \smv p \in \dom \envmv$, $\envmv' \psnames \smv p
  = \envmv \psnames \smv p \subs{x'}{x}$.
\end{lem}

\begin{proof}
  Trivial, observing that typing does not depend on the identity of
  the variables used in processes.
\end{proof}

\begin{defi}\label{def:delta-iota-comp}
  Let $\Gimp$ be a $\iota$-implementation of a global type $\G$ at
  $\shname$, $\C$ a guard, $\Gamma$ an environemnt, and $\envmv$ a
  specification.
  We say that $\Gimp$, $\Gamma' \supseteq \Gamma$, and specification
  $\envmv ' = {(\envmv_{\ptp q})}_{\ptp q \in \participants \G}$ are
  \emph{compatible with $\C$, $\Gamma$, and $\envmv$} if
  \begin{itemize}
  \item the queue on each $\smv \in \tuple \smv$ in $\Gimp$ is typed
    as $\envmv(\smv)$ and
  \item for all $\ptp q \in \participants \G$
    \begin{equation}
      \label{lem-cov-cond-2}
      \pjdg[\C][\envmv_{\ptp q}][\iota({\ptp q})][][\Gamma'] \text{ and }
      \rmg{\envmv_{\ptp q} \psnames \smv q} = \rmg{\envmv\psnames
        \smv q}
    \end{equation}
  \end{itemize}
\end{defi}

\begin{lem}\label{lem-impl-covers-local}
  Let $\G(\tuple \smv)$ be a global type, $\envmv$ be a viable
  specification such that $\tuple \smv \subseteq \dom \envmv$ and
  $\psnames \smv q \in \dom \envmv$ for all
  $\ptp q \in \participants \G$, and let $\pjdg[\C][\envmv'][P][@][@]$
  be a judgement where
  $\envmv'{\psnames \smv p} = \envmv{\psnames \smv p}$ for a
  participant $\ptp p \in \participants \G$.
  For all $r \in \typeruns{\tuple \smv}{\envmv}$ there are $\Gimp$, an
  environment $\Gamma'$, and a specification
  $\envmv'$
  compatible with $\C$, $\Gamma$, and $\envmv$ such that
  \begin{enumerate}
  \item\label{lem-cov-cond-1}\label{lem-cov-cond-3}
    $\iota (\ptp p) = P$
  \item\label{lem-cov-cond-4} for all
    $\consSt[\sigma][\Gimp][\Gamma'][\envmv'][\C]$ there is
    $r' \in \imrunsu{\state[\Gimp][\sigma]}$
    such that $r\myequal r'$.
  \end{enumerate}
\end{lem}

\begin{proof} By induction on the structure of the derivation
  $r\in\typeruns{\tuple \smv}{\envmv}$.
  In the proof, $\vjdg[{}][\iota(\tuple \smv)][\envmv(\tuple \smv)]$
  shortens $\vjdg[{}][\queue \smv {\tuple{\val v}}][\envmv(\smv)]$ for all
  $\smv \in \tuple \smv$ where $\queue \smv {\tuple{\val v}}$ is the queue
  on $\smv$ in $\Gimp$.

  \begin{itemize}
  \item \myrule{RTEnd} Then, $r = \epsilon$ and for all
    $\ptp q \in \participants \G$,
    $\envmv \psnames \smv q = \guard[e_{\ptp q}]\Tend$ and
    $\envmv(\tuple \smv) = \queue{}{}$.
    Now we show that each condition holds.
    \begin{enumerate}
    \item Take $\iota (\ptp p) = P$, and $\iota(\ptp q) = \Pend$ for
      all $\ptp q \neq \ptp p$, and
      $\vjdg[{}][\iota(\tuple\smv)][\queue{}{}]$.
    \item Let $\Gamma' = \Gamma$, $\envmv_{\ptp p} = \envmv'$ and
      $\envmv_{\ptp q} : \psnames \smv q \mapsto \guard[\C] \Tend$ for
      all $\ptp q \neq \ptp p$.
      \begin{enumerate}
      \item Then condition~\eqref{lem-cov-cond-2} holds for
        participant $\ptp p$ by the hypothesis
        $\pjdg[\C][\envmv' = \envmv_{\ptp p}][P][@]$, and
        straightforwardly for participants $\ptp q \neq \ptp p$
        observing that
        $\pjdg[@][\envmv_{\ptp q}][\iota({\ptp q})][@][@]$ and guard
        removals coincide since all local types are $\Tend$.
      \item For condition~\eqref{lem-cov-cond-4} it is enough to take
        any store $\sigma$ mapping each free variable $x$ of $P$ in a
        value of type $\Gamma(x)$ so that $e$ evaluates to $\truek$
        in $\sigma$ (the existence of such values is guaranteed by
        the typing of $P$).
        This also entails~\eqref{cons5:eq} of \cref{def:consistent}.
        Finally, note that
        $r' = \epsilon \in \imrunsu{\state[\Gimp][\sigma]}$ from
        \myrule{REnd} and $r \myequal r'$ follows from
        \myrule{\myequal\text{-emp}}.
      \end{enumerate}
    \end{enumerate}

  \item \myrule{RTCom_1} There is $\ptp q \in \participants \G$
    such that
    $\envmv\psnames \smv q = \Ipst[\C_i]{i \in
      I}{\smv_i}{\sort_i}{\pst_i}$
    and $r = \state[\ptp q][\lsend {\smv_j} {\sort_j} ] r_1$ with
    $j \in I$ and $r_1 \in \typeruns{\tuple \smv}{\envmv''}$ where
    $\envmv'' = \envmv \upd{\pst_j,\ \queue{}{\tuple \sort \cdot
        \sort_j}} {\psname \smv q,\
      \tuple \smv}$.
    Note that $\envmv''$ is viable and that
    $\envmv \psnames \smv q \neq \guard[\falsek] \Tend$, which implies
    $\neg (\C_j \iff \falsek)$.

    We distinguish two cases.

    \paragraph{Case $\ptp q = \ptp p$} By \lemref{lem:inv-subj-red},
    we know that for any $\sigma_0$ such that
    $\consSt[\sigma_0][P][\envmv']$ it holds that
    $\state[P][\sigma_0]
    \Wtrans{\lcond{\CC}{\seq\beta}}\state[P''][\sigma_1]
    \Ptrans{\lcond{\C'}{\lsend{\smv_j}{\val
          v}}}\state[P'][\sigma_2]$
    with $\vjdg[][\val v][\sort_j]$, and
    $\names[\seq{\beta}]\cap\tuple\smv =\emptyset$, and
    $\neg (\C\land\CC\land\C' \iff \falsek)$.
    By applying \cref{thm:sr} to the hypothesis
    $\pjdg[\C][\envmv'][P][@][@]$ and $\sigma_0$, we have
    \begin{equation}\label{eq:lem-impl-covers-local-sent-b}
      \envmv' \Wtrans{\seq{\beta'}} \envmv'_1 \Ttrans{\lsend{\smv_j}{\sort_j}}
      \envmv'_2,  \quad
      \pjdg[\C \land\CC\land \C'][\envmv'_2][P'][][\Gamma''],
      \quad\text{and}\quad
      \Gamma \subseteq \Gamma''
    \end{equation}
    (where the sequence $\seq{\beta'}$ is the sequence $\seq \beta$
    where expressions are replaced by their sorts).
    By \lemref{lem:preservation-of-disjoint-endpoint-types},
    $\names[\seq{\beta'}] \cap \tuple \smv = \emptyset$ implies
    $\envmv'_1 \psnames \smv r = \envmv' \psnames \smv r= \envmv
    \psnames \smv r$ for all $\ptp r \in \participants \G$.
    Moreover, $\envmv'_2 \psnames \smv p = \pst_j$ because of the
    reduction $\envmv'_1 \Ttrans{\lsend{\smv_j}{\sort_j}}\envmv'_2$.
    By inductive hypothesis on
    $r_1\in \typeruns{\tuple \smv}{\envmv''}$ and~\eqref{eq:lem-impl-covers-local-sent-b}, we know that there are
    $\Gimp[@][\iota^\texttt{ih}]$,
    $\Gamma^\texttt{ih} \supseteq \Gamma''$, and
    $\envmv^\texttt{ih} = {(\envmv^\texttt{ih}_{\ptp r})}_{\ptp r
      \in \participants \G}$
    compatible with $\C \land E \land \C'$, $\Gamma''$, and
    $\envmv_2'$.
    Namely,
    \begin{enumerate}
    \item\label{it:iota} $\iota^\texttt{ih}(\ptp p) = P'$ and
    \item\label{lem-cov-cond-1-HI-send}\label{lem-cov-cond-3-HI-send}
      $\pjdg[\C \land\CC\land \C']
      [\envmv^\texttt{ih}_{\ptp r}]
      [\iota^\texttt{ih}({\ptp r})][][\Gamma^\texttt{ih}]$
      and
      $\rmg{\envmv^\texttt{ih}_{\ptp r} \psnames \smv r} =
      \rmg{\envmv''\psnames \smv r}$,
      for all $\ptp r \in \participants \G$
    \item\label{lem-cov-cond-4-HI-send}\label{lem-cov-cond-5-HI-send}
      for all
      $\consSt[\sigma^\texttt{ih}][{\Gimp[@][\iota^\texttt{ih}]}]
      [\Gamma^\texttt{ih}][\envmv^\texttt{ih}][\C
        \land \CC \land \C']$
      there is
      $r^\texttt{ih}_1 \in
      \imrunsu{\state[{\Gimp[@][\iota^\texttt{ih}]}][\sigma^\texttt{ih}]}$
      such that $r_1 \myequal r^\texttt{ih}_1$.
    \end{enumerate}
    Then, let
    \begin{itemize}
    \item $\iota(\ptp p) = P$, and $\iota(\ptp q) = \iota^\texttt{ih}(\ptp q)$
      for all $\ptp q \neq \ptp p \in \participants \G$, and
      $\iota(\smv_j) = \queue{}{\tuple{\val v}}$ where
      $\iota^\texttt{ih}(\smv_j) = \queue{}{\tuple{\val v} \cdot \val v}$ (note
      that~\eqref{it:iota} and the definition of $\envmv''$ imply
      that $\iota^\texttt{ih}(\smv_j)$ is a non-empty queue); this implies
      $\vjdg[{}][\iota^\texttt{ih}(\smv_j)][\queue{}{\tuple \sort \cdot \sort}]$
      since
      $\queue{}{\tuple \sort \cdot \sort} = \envmv''(\smv_j)$.
      Observe that for all
      $\smv \in \tuple \smv \setminus \{\smv_j\}$ we have
      $\vjdg[{}][\iota(\smv)][\envmv''(\smv)]$ since
      $\iota^\texttt{ih}(\smv) = \iota(\smv)$ and
      $\envmv(\smv) = \envmv''(\smv)$.
    \item Note that $\Gamma^\texttt{ih} \supseteq \Gamma$ since
      $\Gamma^\texttt{ih} \supseteq \Gamma''$ and
      $\Gamma'' \supseteq \Gamma$ by~\eqref{eq:lem-impl-covers-local-sent-b}.
      Let $\envmv_{\ptp p} = \envmv'$ and
      $\envmv_{\ptp r} = \envmv_{\ptp r}'$ for all
      $\ptp r \neq \ptp p$; then
      \begin{enumerate}
      \item by the validity of $\pjdg[@][\envmv']$ and the fact
        that $\Gamma^\texttt{ih} \supseteq \Gamma$, we have
        $\pjdg[\C][\envmv_{\ptp p}][\iota(\ptp p)][@][\Gamma^\texttt{ih}]$
        applying \cref{lem:typing-weakening}.
        For $\ptp r \neq \ptp p \in \participants \G$, the validity
        of the judgement in~\eqref{lem-cov-cond-1-HI-send} implies
        the validity of $\pjdg[\C][\envmv_{\ptp r}][\iota(\ptp r)]$
        because $\var{\CC \land \C'}$ are not bound in
        $\iota(\ptp r)$.
        Also,
        $\rmg{\envmv_{\ptp r}\psnames \smv r} = \rmg{\envmv''\psnames
          \smv r} = \rmg{\envmv\psnames \smv r}$
        where the first equality holds by the inductive hypothesis
        and the second by~\eqref{lem-cov-cond-1-HI-send}.
      \item Let $\sigma$ be $\sigma^\texttt{ih}$ restricted on
        $\bigcup_{\ptp r \in \participants \G}\dom{\envmv_{\ptp
          r}|_{\chset}} \cup \dom{\Gamma^\texttt{ih}}$.
        Then the first three conditions of \cref{def:consistent} are
        trivially satisfied, $\C$ evaluates to $\truek$ in $\sigma$
        since $\sigma|_{\var\C} = \sigma^\texttt{ih}|_{\var\C}$, and likewise
        for the last condition of consistency.
      \end{enumerate}
    \item Note that
      $r = \state[\ptp p][\lsend {\smv_j} {\sort_j} ]r_1 \myequal
      \state[\ptp p][\lsend {\smv_j} {\sort_j} ]r^\texttt{ih}_1 = r'$
      follows from \myrule{\myequal\text{-cmp}} and~\eqref{lem-cov-cond-4-HI-send}.
      The fact that $r' \in \Gimp$ follows by repeated applications
      of rule \myrule{RExt} (once for any $\beta \in \seq \beta$)
      followed by an application of \myrule{RSnd}.
    \end{itemize}

    \paragraph{Case $\ptp q \neq \ptp p$}
    By inductive hypothesis on
    $r_1\in \typeruns{\tuple \smv}{\envmv''}$ and~\eqref{eq:lem-impl-covers-local-sent-b}, we know that there are
    $\Gimp[@][\iota^\texttt{ih}]$,
    $\Gamma^\texttt{ih} \supseteq \Gamma''$, and
    $\envmv^\texttt{ih} = {(\envmv^\texttt{ih}_{\ptp r})}_{\ptp r
      \in \participants \G}$
    compatible with $\C \land \CC \land \C'$, $\Gamma''$, and
    $\envmv_2'$.
    Namely,
    \begin{enumerate}
    \item $\iota^\texttt{ih}(\ptp p) = P$
    \item\label{lem-cov-cond-3-HI-send:2}
      $\pjdg[@][\envmv^\texttt{ih}_{\ptp r}][\iota^\texttt{ih}({\ptp
        r})][][\Gamma^\texttt{ih}]$
      and
      $\rmg{\envmv^\texttt{ih}_{\ptp r} \psnames \smv r} =
      \rmg{\envmv''\psnames \smv r}$,
      for all $\ptp r \in \participants \G$
    \item\label{lem-cov-cond-4-HI-send:}\label{lem-cov-cond-5-HI-send:}
      for all store
      $\consSt[\sigma^\texttt{ih}][{\Gimp[@][\iota^\texttt{ih}]}]
      [\Gamma^\texttt{ih}][\bigcup_{\ptp
          r \in \participants \G}\envmv^\texttt{ih}_{\ptp r}][\C]$
      there is
      $r^\texttt{ih}_1 \in
      \imrunsu{\state[{\Gimp[@][\iota^\texttt{ih}]}][\sigma^\texttt{ih}]}$
      such that $r_1 \myequal r^\texttt{ih}_1$.
    \end{enumerate}
    Then, let
    \begin{itemize}
    \item $\iota(\smv_j) = \queue{}{\tuple{\val v}}$ where
      $\iota^\texttt{ih}(\smv_j) = \queue{}{\tuple{\val v} \cdot \val
      v}$ (the proof that queues $\smv \in \tuple \smv$ are typed
      by $\envmv''$ is as in the previous case $\ptp p = \ptp q$) and
      \[
      \iota : \begin{cases}
        \ptp p \mapsto P
        \\
        \ptp r \mapsto \iota^\texttt{ih}(\ptp r) & \forall
        \ptp r \in \participants \G \setminus \{\ptp p, \ptp q\}
        \\
        \ptp q \mapsto \Preceive s x \Pif x
             {\Psend{\smv_j}{\val v};\iota^\texttt{ih}(\ptp q)}{Q}
      \end{cases}
      \]
      with $s$ and $x$ fresh (formally,
      $s \not\in \bigcup_{\ptp r \in \participants
        \G}\dom{\envmv^\texttt{ih}_{\ptp r}}$, and
      $x \not\in \dom \Gamma$, and $\vjdg[][\val v][\sort_j]$) and $Q$
      a process such that the following judgement holds:
      \[
      \pjdg[@][\envmv_Q][Q][][\Gamma_Q] \quad\text{where}\quad
      \envmv_Q\psnames \smv q = \Ipst[e_i]{i \in I\setminus
        \{j\}}{\smv_i}{\sort_i}{\pst_i}
      \]
      The above judgment exists for the type
      $\pst = \Ipst[e_i]{i \in I\setminus
        \{j\}}{\smv_i}{\sort_i}{\pst_i}$ by
      \cref{lem:inhabitant-types}, since $\nform{\pst} = \pst$ applying
      \cref{lem:typing-strenghten-condition} to the hypothesis
      $\pjdg[\C][\envmv'][P][@][@]$.
        By \cref{lem:subs} wlog we can assume that $\fX[Q]$ are fresh.
      \item Let $\Gamma' = \Gamma\cup\Gamma_Q$ and
        \[
        \begin{array}{l@{\ = \ }l@{\quad}l}
          \envmv_{\ptp p} & \envmv'
          \\
          \envmv_{\ptp q}& \envmv^\texttt{ih}_{\ptp q}[\envmv_Q]
          \upd{\envmv \psname \smv q}{\psname \smv q}
          \upd{\Treceive{s}{\dt{bool};\Tend}}{\psname s \_}
          \\
          \envmv_{\ptp r}& \envmv^\texttt{ih}_{\ptp r} &
          \textit{if}\ \ptp r\not\in\{\ptp p, \ptp q\}
        \end{array}
        \]
        \begin{enumerate}
        \item by the validity of $\pjdg[@][\envmv'][P]$ and the fact
          that $\Gamma^\texttt{ih} \supseteq \Gamma$, we have
          $\pjdg[\C][\envmv_{\ptp p}][\iota(\ptp
          p)][@][\Gamma^\texttt{ih}]$ applying
          \cref{lem:typing-weakening}.
          For all $\psname \smv r \in \dom {\envmv'}$ with
          $\ptp r \in \participants \G \setminus \{\ptp p, \ptp q\}$ we
          have that $\envmv'' \psname \smv r = \envmv' \psname \smv r$ and
          $\rmg{\envmv_{\ptp r}\psnames \smv r} = \rmg{\envmv''\psnames
            \smv r} = \rmg{\envmv\psnames \smv r}$ by the
          inductive hypothesis.
          The validity of
          $\pjdg[@][\envmv_{\ptp q}][\iota(\ptp q)][@][\Gamma']$ holds
          by weakening (\cref{lem:env-weakening}) and the application of
          rules \myrule{VIf} and \myrule{VRcv} to the judgment
          $\pjdg[@][\envmv_Q][Q][@][\Gamma_Q]$ above.

        \item Let $\sigma$ be an extension of $\sigma^\texttt{ih}$ with
          assignments to the free variables of $Q$ so that no guard in
          $Q$ is falsified.
          Then the first three conditions of \cref{def:consistent}
          trivially follow from the inductive hypothesis, and the last
          condition follows from
          \cref{lem:const-pseudo-typing-and-store-correspondence}.
        \end{enumerate}
      \item From~\eqref{lem-cov-cond-5-HI-send} we know that
        $r^\texttt{ih}_1\in\imrunsu{\state[{\Gimp[@][\iota^\texttt{ih}]}]
        [\sigma^\texttt{ih}]}$.
        Consequently, we conclude that
        $r^\texttt{ih}_1\in\imrunsu{\state[{\Gimp[@][\iota^\texttt{ih}]}]
          [\sigma^\texttt{ih}\upd{\truek}{x}]}$
        holds.
        Moreover,
        \[
        \state[\iota(\ptp q)]
        \quad
        \Ptrans{\lcond{ }{\lreceive{s}{\truek}}}
        \quad
        \state[\Psend{\smv_j}{\val v};\iota^\texttt{ih}(\ptp q)]
              [\sigma\upd{\truek}{x}]
        \quad
        {\Ptrans{\lcond{}{\lsend{\smv_j}{\val v}}}}
        \quad
        \state[\iota^\texttt{ih}(\ptp q)][\sigma\upd{\truek}{x}]
        \]
        Hence,
        $\state[\ptp q][\lsend {\smv_j} {\sort} ]r^\texttt{ih}_1 \in
        \imrunsu{\state[{\Gimp[@][\iota]}][@]}$ which covers
        $\state[\ptp q][\lsend{\smv_j}{\sort}]r_1$ by the inductive hypothesis.
      \end{itemize}
    \item \myrule{RTCom_2} The proof follows analogously to the
      previous case.
    \item \myrule{RTSeq} Then $\envmv = \envmv_1; \envmv_2$ and
      $r = r_1 r_2$ with $r_1 \in \typeruns{\tuple \smv}{\envmv_1}$
      and $r_2 \in \typeruns{\tuple \smv}{\envmv_2}$.
      Note that $\envmv_1$ and $\envmv_2$ are viable, otherwise
      $\envmv$ would not be viable.
      Given the structure of $\envmv$, the typing of $P$ can be
      achieved with an application of  $\myrule{VSeq}$,  $\myrule{VLoop}$, or
      $\myrule{VSend}$.

      \item \myrule{VSeq} We have $P = P_1; P_2$; then there are
      specifications $\envmv_1'$ and $\envmv_2'$ such that
      $\pjdg[@][\envmv_i'][P_i]$ for $i \in \{1,2\}$,
      $\envmv' = \envmv_1' ; \envmv_2'$.
      We consider two cases depending on whether
      $\envmv'_1 = \Tdef{(\envmv'_0)}$ for some $\envmv'_0$.

      When there is no $\envmv'_0$ such that
      $\envmv'_1 \neq \Tdef{(\envmv'_0)}$ we proceed as follows.  For
      $i \in \{1,2\}$, by inductive hypothesis on
      $r_i \in \typeruns{\tuple \smv}{\envmv_i'}$ we know that there
      are $\Gimp[@][\iota_i^\texttt{ih}]$,
      $\Gamma^\texttt{ih}_i \supseteq \Gamma$, and
      $\envmv^\texttt{ih}_i = {(\envmv_{i,{\ptp r}})}_{\ptp r
        \in \participants \G}$
      compatible with $\C$, $\Gamma$, and $\envmv_i'$.
      Namely,
      \begin{enumerate}
      \item $\iota_i^\texttt{ih}(\ptp p) = P_i$ and
      \item
        $\pjdg[\C][\envmv_{i,{\ptp r}}][\iota_i^\texttt{ih}({\ptp
          r})][][\Gamma_i^\texttt{ih}]$ and
        $\rmg{\envmv_{i,{\ptp r}} \psnames \smv r} =
        \rmg{\envmv'_i\psnames \smv r}$, for all
        $\ptp r \in \participants \G$w
      \item for all
        $\consSt[\sigma_i^\texttt{ih}][{\Gimp[@][\iota_i^\texttt{ih}]}]
        [\Gamma_i^\texttt{ih}][\envmv_i^\texttt{ih}][\C]$
        there is
        $r^\texttt{ih}_i \in
        \imrunsu{\state[{\Gimp[@][\iota_i^\texttt{ih}]}]
          [\sigma_i^\texttt{ih}]}$
        such that $r_i \myequal r^\texttt{ih}_i$.
      \end{enumerate}
      Now let
      \begin{eqnarray*}
        \iota(\ptp r) & = & \iota_1^\texttt{ih}(\ptp r) ;
        \iota_2^\texttt{ih}(\ptp r)
                            \qquad \forall \ptp r \in \participants \G
        \\
        \iota(\smv) & = & \iota_1^\texttt{ih}(\smv)
                       \qquad \forall \smv \in \tuple \smv
        \\
        \Gamma' & = & \Gamma_1^\texttt{ih} \cup \Gamma_2^\texttt{ih}
        \\
        \envmv' & = & \envmv'_1 ; \envmv'_2
        \\
        \sigma & = & \sigma_1^\texttt{ih}
      \end{eqnarray*}
      and observe that $\envmv'$ is defined otherwise compatibility
      would be violated, contradicting the inductive hypothesis.

      If $\envmv'_1 = \Tdef{(\envmv'_0)}$ for some $\envmv'_0$, we
      note that $\envmv_1 = \Tdef{(\envmv_0)}$ for some $\envmv_0$
      because they are the same when restricted to participants'
      session of $\ptp p$.
      Then, $r_1$ has been obtained by using either $\myrule{RTIt_1}$
      or $\myrule{RTIt_2}$. Consequently, either
      $r_1 \in \typeruns{\tuple \smv}{\envmv'_0}$ or $r_1 = r_0[r_3]$
      with $r_0 \in \typeruns{\tuple \smv}{\envmv'_0}$.

      We consider $r_1 \in \typeruns{\tuple \smv}{\envmv'_0}$ (the
      other case follows analogously).  Since
      $\envmv_1' = \Tdef{(\envmv'_0)}$ and
      $\pjdg[@][\Tdef{(\envmv'_0)}][P_1]$, $P_1 = \Pfor x \ell P_0$
      and $\pjdg[@][\envmv'_0][P_0][\Gamma]$.
      As before, we apply inductive hypothesis on $\envmv_0$,
      $\envmv'_0$, $P_0$, and $r_o$; we hence obtain a system
      $\Gimp[@][\iota_1^\texttt{ih}]$, an environment
      $\Gamma_1^\texttt{ih}$, and a specification
      $\envmv_1^\texttt{ih}$ such that
      $\iota_1^\texttt{ih}(\ptp p) = P_0$ and for all stores
      $\consSt[\sigma_1^\texttt{ih}][{\Gimp[@][\iota_1^\texttt{ih}]}]
      [\Gamma_1^\texttt{ih}][\envmv_1^\texttt{ih}][\C]$
      there is
      $r_1^\texttt{ih} \in
      \imrunsu{\state[{\Gimp[@][\iota_1^\texttt{ih}]}][\sigma_1^\texttt{ih}]}$
      such that $r_1 \myequal r_1^\texttt{ih}$.
      Likewise, proceeding as in the case of \myrule{VSeq} above, there
      are a system $\Gimp[@][\iota_2^\texttt{ih}]$, an environment
      $\Gamma_2^\texttt{ih}$, and a specification
      $\envmv_2^\texttt{ih}$ such that
      $\iota_2^\texttt{ih}(\ptp p) = P_2$ and for all stores
      $\consSt[\sigma_2^\texttt{ih}][{\Gimp[@][\iota_2^\texttt{ih}]}]
      [\Gamma_2^\texttt{ih}][\envmv_2^\texttt{ih}][\C]$
      there is
      $r_2^\texttt{ih} \in
      \imrunsu{\state[{\Gimp[@][\iota_2^\texttt{ih}]}][\sigma_2^\texttt{ih}]}$
      such that $r_2 \myequal r_2^\texttt{ih}$.
      The proof ends by taking
      \begin{eqnarray*}
        \iota(\ptp p) & = & P
        \\
        \iota(\ptp r) & = & \PloopU{\iota_1^\texttt{ih}(\ptp r)}
        \iota_2^\texttt{ih}(\ptp r)
                  \qquad \forall \ptp r \neq \ptp p \in \participants \G
        \\
        \iota(\smv) & = & \iota_1^\texttt{ih}(\smv)
                       \qquad \forall \smv \in \tuple \smv
        \\
        \Gamma' & = & \Gamma_1^\texttt{ih} \cup \Gamma_2^\texttt{ih}
        \\
        \envmv' & = & \envmv'_1 ; \envmv'_2
        \\
        \sigma & = & \sigma_1^\texttt{ih}
      \end{eqnarray*}

      The case $r_1 = r_0[r_3]$ with
      $r_0 \in \typeruns{\tuple \smv}{\envmv'_0}$ is analogous.

      \item $\myrule{VLoop}$ The thesis follows as in the
      previous case noticing that $P = \PloopU M N$ (that is, $P$
      plays a passive role in the loop) and then using the inductive
      hypothesis on the premisses of the typing of $P$ and observing
      that $\envmv$ is viable, hence there it yields an active role
      deciding when to terminate the loop.

      \item \myrule{VSend} The thesis follows trivially as in
      the cases \myrule{RTCom_1} and \myrule{RTCom_2} by observing
      that $\envmv'$ assigns to the participant session
      $(\tuple \smv, \ptp p)$ an output on a session channel
      $\smv \in \tuple \smv$.
    \end{itemize}
    Note that \myrule{RTIt1} and \myrule{RTIt2} cannot be the last rules
    applied in a derivation of $r \in \typeruns{\tuple \smv}{\envmv}$.
  \end{proof}

\coveragetwothm*
\begin{proof} The proof follows directly from~\cref{lem-impl-covers-local}.
\end{proof}

  %%% Local Variables:
  %%% mode: latex
  %%% TeX-master: "main"
  %%% End:

\section{Whole-spectrum implementations and Guarded Automata}%
\label{asec:communicating}
\newcommand{\sss}[1]{\mathit{#1}}

In this section we briefly discuss how our notion of whole-spectrum
implementation (WSI) can be defined when specifications and
implementations are defined as Guarded Automata.

We first recall some basic definitions from~\cite{fbs05}:
$(P,M)$ is a composition schema where
$P=\{\ptp{p}_1,\ldots,\ptp{p}_n\}$ is a set of participants and $M$
are the messages (i.e., the alphabet),
$R=\langle (P,M),A\rangle $ is a conversation protocol where A is a
guarded automaton,
$W=\langle (P,M), A_1,\ldots,A_n\rangle$ is a web service composition,
$L(R)=L(A)$ is the language of a conversation protocol.
For a web service composition $W=\langle (P,M),A_1,\ldots,A_n\rangle$
we have runs, send sequences and conversations;
\begin{enumerate}
\item a run of $W$ is a sequence of configurations
  $\gamma=c_0,c_1,\ldots,c_n$ where:
\begin{itemize}
\item $c_0$ is an initial configuration
\item $c_i\rightarrow c_{i+1}\quad  (i=0\ldots n-1)$
\item $c_n$ is a final configuration
\end{itemize}
\item a send sequence $\sss{\gamma}$ on a run $\gamma$ is the sequence
  messages, one for each send action in $\gamma$, recorded in the
  order in which they are sent,
\item a conversation is a word $w$ over $M$ for which there is a run
  $\gamma$ of $W$ such that $w= \sss{\gamma}$,
\item the conversations of a web service $W$, written $C(W)$, is the
  set of all the conversations for $W$.
\end{enumerate}

\noindent
We are now ready to introduce a notion of WSI for guarded automata.

\begin{defi}[Whole-spectrum realisation of a guarded automaton]
Let $P$ be a set of participants defined as
$\{\ptp{p}_1,\ldots,\ptp{p}_n\}$. $A_i$ is a whole-spectrum
realisation of $\ptp{p}_i\in P$ in conversation protocol
$R=\langle(P,M),A\rangle$ if for all $w \in L(R)$ there exist
${\{A_j\}}_{j\in \{1,\ldots,n\}\setminus\{i\}}$ such that $w \in
C(<(P,M),A_1,\ldots, A_n>)$.
\end{defi}

\begin{defi}[WSI of guarded automaton]
  $A_i$ is a WSI of $\ptp{p}_i$ in conversation protocol
  $R=\langle(P,M),A\rangle$ if: (1) $A_i$ is a deterministic guarded
  automaton, and (2) $A_i$ is a whole-spectrum realisation of
  $\ptp{p}_i$ in $R=\langle(P,M),A\rangle$.
\end{defi}

%%% Local Variables:
%%% mode: latex
%%% TeX-master: "main.tex"
%%% End:

\end{document}

%%% Local Variables:
%%% mode: latex
%%% TeX-master: t
%%% End: